\documentclass[11pt]{article}

\usepackage[T1]{fontenc}
\usepackage{lmodern}
\usepackage[margin=1in]{geometry}
\usepackage{amsthm,amsmath,amsfonts,amssymb}
\usepackage[authoryear,round]{natbib}
\usepackage{mathtools,bm,booktabs,graphicx,placeins,enumitem,microtype}
\usepackage[colorlinks=true,citecolor=blue,urlcolor=blue,linkcolor=blue]{hyperref}
\hypersetup{
  pdftitle={Elliptical Regularized Hotelling Testing for High Dimensional Data},
  pdfauthor={Long Feng, Le Zhou, and Xiaoyi Wang},
  pdfkeywords={Cauchy combination, elliptical distribution, high-dimensional location testing, regularized Hotelling statistic, spatial median, spatial-sign covariance matrix}
}

\numberwithin{equation}{section}
\allowdisplaybreaks[3]
\theoremstyle{plain}
\newtheorem{theorem}{Theorem}[section]
\newtheorem{proposition}[theorem]{Proposition}
\newtheorem{corollary}[theorem]{Corollary}
\newtheorem{lemma}{Lemma}[section]

\theoremstyle{definition}
\newtheorem{assumption}{Assumption}

\newcommand{\R}{\mathbb R}
\newcommand{\E}{\mathbb E}
\newcommand{\Pp}{\mathbb P}
\newcommand{\Var}{\operatorname{Var}}
\newcommand{\Cov}{\operatorname{Cov}}
\newcommand{\tr}{\operatorname{tr}}
\newcommand{\diag}{\operatorname{diag}}
\newcommand{\rank}{\operatorname{rank}}
\newcommand{\op}{\mathrm{op}}
\newcommand{\F}{\mathrm F}
\newcommand{\I}{\mathbf I}
\newcommand{\1}{\mathbf 1}
\newcommand{\0}{\mathbf 0}
\newcommand{\btheta}{\bm\theta}
\newcommand{\bDelta}{\bm\Delta}
\newcommand{\bdelta}{\bm\delta}

\newcommand{\bg}{\bm g}

\newcommand{\mA}{\mathbf A}
\newcommand{\mB}{\mathbf B}
\newcommand{\mC}{\mathbf C}
\newcommand{\mD}{\mathbf D}
\newcommand{\mE}{\mathbf E}
\newcommand{\mF}{\mathbf F}
\newcommand{\mG}{\mathbf G}
\newcommand{\mGamma}{\boldsymbol\Gamma}
\newcommand{\mH}{\mathbf H}
\newcommand{\mJ}{\mathbf J}
\newcommand{\mK}{\mathbf K}
\newcommand{\mL}{\mathbf L}
\newcommand{\mM}{\mathbf M}
\newcommand{\mLambda}{\boldsymbol\Lambda}
\newcommand{\mP}{\mathbf P}
\newcommand{\mPi}{\boldsymbol\Pi}
\newcommand{\mQ}{\mathbf Q}
\newcommand{\mR}{\mathbf R}
\newcommand{\mS}{\mathbf S}
\newcommand{\mT}{\mathbf T}
\newcommand{\mU}{\mathbf U}
\newcommand{\mV}{\mathbf V}

\newcommand{\mX}{\mathbf X}
\newcommand{\mY}{\mathbf Y}
\newcommand{\mOmega}{\boldsymbol\Omega}

\newcommand{\wh}{\widehat}
\newcommand{\wt}{\widetilde}
\newcommand{\calF}{\mathcal F}

\newcommand{\calG}{\mathcal G}
\newcommand{\norm}[1]{\left\lVert#1\right\rVert}
\newcommand{\abs}[1]{\left\lvert#1\right\rvert}
\newcommand{\dd}{\,\mathrm d}

\newcommand{\dto}{\stackrel{d}{\to}}
\newcommand{\pto}{\xrightarrow{\Pp}}

\title{Elliptical Regularized Hotelling Testing for High Dimensional Data}
\author{Long Feng$^{1}$, Le Zhou$^{2}$, and Xiaoyi Wang$^{3}$\\[5pt]
\small $^{1}$School of Statistics and Data Science, Nankai University\\
\small $^{2}$Department of Mathematics, Hong Kong Baptist University\\
\small $^{3}$Advanced Institute of Natural Sciences, Beijing Normal University\\[3pt]
\small \texttt{flnankai@nankai.edu.cn; lezhou@hkbu.edu.hk; wangxy059@bnu.edu.cn}}
\date{}

\begin{document}
\maketitle

\begin{abstract}
We consider one-sample testing of a high-dimensional location parameter under elliptically symmetric distributions with heavy tails and pervasive cross-sectional dependence. We propose an elliptical regularized Hotelling test with Cauchy combination (ERHT--CC), based on the sample spatial median and the spatial-sign covariance matrix centered at that median. We derive its null asymptotic normality, consistent estimators of the centering and variance, and an explicit local power function. Since the power-optimal ridge parameter depends on the unknown alternative, we aggregate fixed-ridge $p$-values over a deterministic grid using the Cauchy rule. We establish a finite-grid joint Gaussian limit, justify the analytic combined $p$-value without estimating cross-ridge correlations, and characterize its local power. Simulation studies and an empirical analysis demonstrate the favorable finite-sample performance of ERHT--CC under heavy tails and pervasive dependence.
\end{abstract}



\section{Introduction}\label{sec:introduction}

Testing a high-dimensional mean or location vector is a fundamental problem in modern multivariate analysis.  It is encountered when a large collection of moment restrictions must be assessed jointly, as in asset-pricing models, portfolio evaluation, large cross-sectional panels, genomics, and imaging studies.  The main statistical difficulties are now well understood: the dimension may be comparable to or larger than the sample size, the coordinates can be strongly dependent, and the observations can have heavy radial tails.  These features make the classical Hotelling test unavailable or unstable and can invalidate procedures whose calibration relies strongly on empirical second and fourth moments.  Under an elliptically symmetric distribution, the center remains a meaningful target even when ordinary moments are weak or nonexistent, so robust location testing provides a natural extension of high-dimensional mean testing.  A broad overview of the latter literature is given by \citet{HuangLiLiYang2022}.

One major line of research develops quadratic-form or sum-type tests that avoid inverting the sample covariance matrix.  The trace-based approach dates back to \citet{Dempster1958}, and modern developments include \citet{BaiSaranadasa1996}, \citet{SrivastavaDu2008}, \citet{Srivastava2009}, \citet{ChenQin2010}, \citet{ParkAyyala2013}, \citet{SrivastavaKatayamaKano2013}, and \citet{GregoryCarrollBaladandayuthapaniLahiri2015}.  Scale-standardized and bias-corrected extensions were developed for unequal marginal variances and unequal population covariance matrices; see, among others, \citet{FengSun2015}, \citet{FengZouWangZhu2015}, and \citet{FengZouWangZhu2017}.  These methods aggregate many coordinatewise contributions and are therefore particularly effective against dense alternatives.  Their construction is also computationally attractive in large dimensions.  Nevertheless, off-diagonal dependence is typically used only in the variance calibration, or is replaced by diagonal standardization.  Consequently, the tests may fail to exploit favorable covariance directions when the dependence is strong, and procedures based directly on sample means and covariance moments may be sensitive to heavy-tailed observations.

A second line seeks to use the dependence structure more directly.  Random-projection and random-subspace procedures first reduce the dimension and then apply a low-dimensional Hotelling statistic; representative examples are \citet{LopesJacobWainwright2011} and \citet{Thulin2014}.  Regularized Hotelling tests instead replace the unstable covariance inverse by a ridge inverse, as in \citet{ChenPaulPrenticeWang2011} and \citet{LiAuePaulPengWang2020}.  Precision-transformed and thresholded methods improve detection of sparse signals under dependence; see \citet{ZhongChenXu2013}, \citet{CaiLiuXia2014}, and \citet{ChenLiZhong2019}.  Factor-adjusted procedures remove a low-dimensional common component before testing, as in \citet{MaLanWang2015} and \citet{HeZhangZhangZhou2020}.  These approaches demonstrate that incorporating covariance geometry can substantially increase power.  Their validity, however, generally requires a reliably estimable covariance, precision, or factor structure together with sufficiently regular moments.  For ridge-based tests, the tuning parameter creates a further difficulty: its power-optimal value depends on both the population spectrum and the unknown orientation of the alternative, so a single data-independent choice need not perform uniformly well.

A related literature addresses adaptation to signal sparsity.  Foundational detection-boundary results for sparse normal means and correlated signals motivate higher-criticism, maximum-type, and thresholding procedures; see \citet{DonohoJin2004}, \citet{HallJin2010}, and \citet{AriasCastroCandesPlan2011}.  In high-dimensional mean testing, multi-thresholding, sum--maximum combinations, and other adaptive constructions seek to bridge sparse and dense regimes; see \citet{ZhongChenXu2013}, \citet{XuLinWeiPan2016}, \citet{ChenLiZhong2019}, and the general developments reviewed by \citet{HuangLiLiYang2022}.  Robust versions based on ranks and spatial signs have also been proposed.  In particular, \citet{ZhangFeng2024} combine high-dimensional rank-based sum and maximum statistics, and \citet{LiuFengZhaoWang2027} develop a spatial-sign max--sum test for location parameters.  These methods reduce sensitivity to an unknown sparsity pattern, but their main objective is sparsity adaptation rather than covariance-aware normalization of dense signals under pervasive dependence.

Robust high-dimensional location testing has developed in parallel through spatial signs, spatial ranks, and inverse-distance weights.  The spatial median and spatial-sign covariance matrix remove or strongly attenuate radial magnitudes while preserving multivariate directional information; see \citet{MottoneOja1995}, \citet{VisuriKoivunenOja2000}, and \citet{Oja2010}.  High-dimensional sign asymptotics and systematic comparisons among mean-, sign-, and rank-based tests were developed by \citet{PaindaveineVerdebout2016} and \citet{ChakrabortyChaudhuri2017}.  For high-dimensional mean and location problems, \citet{WangPengLi2015} proposed a nonparametric spatial-sign test, and \citet{FengSun2016} developed a scalar-invariant one-sample spatial-sign procedure under elliptical distributions.  The effect of estimating location in two-sample sign procedures was treated through leave-one-out constructions by \citet{FengZouWang2016}, while \citet{FengZhangLiu2020} developed a high-dimensional spatial-rank test.  Inverse-norm weighting was used by \citet{FengLiuMa2021} for the one-sample problem and by \citet{HuangLiuZhouFeng2023} for the two-sample problem.  More recently, \citet{YanFengZhang2025} studied joint robust estimation of location and scatter for high-dimensional elliptical data.  This body of work establishes that directional procedures can retain useful efficiency while being considerably less sensitive to heavy radial tails.

Important limitations nevertheless remain.  Most existing robust location tests use scalar or diagonal standardization, or aggregate pairwise directional inner products without applying a regularized inverse of a full robust scatter matrix.  They therefore do not fully exploit the dependence geometry that motivates Hotelling-type methods.  Moreover, conventional normal calibrations for sum-type spatial-sign statistics are commonly derived under conditions that prevent a small number of eigen-directions from dominating.  Such conditions exclude pervasive dependence structures, including compound symmetry with a fixed nonzero correlation.  The recent analysis of \citet{ZhaoFeng2026} shows that the one-sample spatial-sign statistic of \citet{WangPengLi2015} can have a non-Gaussian limiting component when leading eigenvalues remain influential, which makes the strong-dependence issue substantive rather than technical.

To address these limitations, this paper develops a robust regularized Hotelling procedure.  The proposed statistic combines the sample spatial median with the spatial-sign covariance matrix centered at that median and applies a ridge inverse to the latter.  The shape matrix is allowed to contain a fixed number of pervasive eigenvalues proportional to the dimension, while the remaining spectrum forms a bounded bulk.  This formulation includes fixed-correlation compound symmetry and permits a diverging condition number.  The use of spatial quantities protects the procedure from large radial observations, whereas the regularized inverse retains covariance-direction information that is absent from purely scalar- or diagonal-standardized sign tests.

The paper makes four main contributions.  First, it derives the null asymptotic normal distribution of the robust regularized Hotelling statistic when the dimension and sample size increase proportionally under a genuine elliptical model with pervasive dependence.  The analysis explicitly retains the angular dependence of inverse distances and the first-order effect of centering the spatial-sign covariance matrix at the sample spatial median.  Fully data-driven estimators of the asymptotic centering and variance are constructed and shown to be consistent.  Second, the paper obtains a local power function for dense alternatives and identifies the signal-to-noise criterion that determines the oracle ridge parameter.  Third, because this oracle parameter depends on the unknown alternative, the ridge-specific tests are evaluated over a fixed deterministic grid and combined by the Cauchy rule of \citet{LiuXie2020}.  A finite-grid joint limit and the validity of the analytic combined p-value are established without estimating cross-ridge correlations.  Fourth, the theory accommodates pervasive spikes rather than treating strong common dependence as a negligible perturbation.

The main technical contribution is a proof strategy tailored to the simultaneous presence of an estimated spatial center, heavy-tailed elliptical radii, and pervasive eigenvalues.  Under pervasive dependence, the angular denominators do not concentrate at a constant, and the perturbation created by centering the spatial-sign covariance matrix at the sample spatial median is not negligible in operator norm.  We derive a quadratic-form-accurate expansion of the spatial median, isolate the estimated-center effect into a rank-two correction and factor-induced finite-rank terms, and show that ridge resolvents suppress the latter along the pervasive eigenspace.  This reduces the statistic to weighted companion-resolvent functionals for a separable random matrix with random column scales.  We then combine deterministic equivalents for these functionals with a conditional Rademacher quadratic-form central limit theorem.  The resulting hybrid argument yields the joint fixed-ridge Gaussian limit and the $O_{\Pp}(n^{-1})$ accuracy required for feasible centering; neither a standard spatial-median Bahadur expansion nor a classical linear spectral-statistic central limit theorem is sufficient for these conclusions.

The Monte Carlo study indicates that the proposed procedure can improve on several existing dense high-dimensional tests, especially when heavy tails or pervasive dependence weaken covariance-based or diagonally standardized competitors.  An analysis of paired lung-cancer gene-expression data further illustrates the behavior of the method in an application with many correlated coordinates.  

The remainder of the paper is organized as follows.  Section~\ref{sec:model} presents the elliptical model, the statistic, its null calibration, and the local power analysis.  Section~\ref{sec:cauchy} develops fixed-grid Cauchy aggregation and its joint asymptotic theory.  Section~\ref{sec:simulation} reports the simulation results, Section~\ref{sec:realdata} presents the empirical application, and Section~\ref{sec:discussion} concludes.  Auxiliary results and detailed proofs are provided in Appendices~\ref{app:angular}--\ref{app:cauchy}.

\section{Methodology and asymptotic theory}\label{sec:model}

\paragraph{Notation.}
For a vector, $\norm{\cdot}$ denotes the Euclidean norm.  For a matrix, $\norm{\cdot}_{\op}$, $\norm{\cdot}_{\F}$, $\tr(\cdot)$, and $\rank(\cdot)$ denote the operator norm, Frobenius norm, trace, and rank, respectively.  The symbols $\lambda_{\min}(\cdot)$ and $\lambda_{\max}(\cdot)$ denote the smallest and largest eigenvalues of a symmetric matrix, and $\delta_x$ denotes the unit point mass at $x$.  We write $O_{\Pp}(a_n)$ and $o_{\Pp}(a_n)$ for stochastic orders, $a_n\asymp b_n$ when their ratio is bounded above and away from zero, and $\dto$ and $\pto$ for convergence in distribution and probability, respectively.  The symbol $\Phi$ denotes the standard normal distribution function and $z_q=\Phi^{-1}(q)$ its $q$th quantile.

\subsection{Elliptical location model}

Let $\bm X_1,\ldots,\bm X_n\in\R^p$ be independent observations with location parameter $\btheta_p$.  We test
\[
 H_0:\btheta_p=\btheta_{0,p}
 \qquad\text{against}\qquad
 H_1:\btheta_p\ne\btheta_{0,p},
\]
where $\btheta_{0,p}$ is specified.  We assume the elliptically symmetric representation
\begin{equation}\label{eq:model}
 \bm X_i
 =\btheta_p+\sqrt p\,R_i\mOmega_p^{1/2}
 \frac{\bm G_i}{\norm{\bm G_i}},
 \qquad \bm G_i\sim N(\0,\I_p),
\end{equation}
where $R_i>0$ is independent of $\bm G_i$.  The positive-definite shape matrix is normalized by $p^{-1}\tr(\mOmega_p)=1$.  This normalization fixes the otherwise unidentified product of radial and shape scales.

For later use, define
\begin{align*}
 h_{i,p}&=p^{-1}\bm G_i^\top\bm G_i,
 &q_{i,p}&=p^{-1}\bm G_i^\top\mOmega_p\bm G_i,\\
 \ell_{i,p}&=(h_{i,p}/q_{i,p})^{1/2},
 &\xi_i&=R_i^{-1}.
\end{align*}
The population spatial sign and the scaled inverse Euclidean distance are
\begin{equation}\label{eq:elliptical_sign_weight}
 \bm U_i=
 \frac{\mOmega_p^{1/2}\bm G_i}
 {(\bm G_i^\top\mOmega_p\bm G_i)^{1/2}},
 \qquad
 w_i=\frac{\sqrt p}{\norm{\bm X_i-\btheta_p}}
 =\xi_i\ell_{i,p}.
\end{equation}
Thus $\norm{\bm U_i}=1$ and
\[
 \bm X_i-\btheta_p=\sqrt p\,w_i^{-1}\bm U_i.
\]
Although $\xi_i$ is independent of $\bm U_i$, the Euclidean inverse-distance weight $w_i$ is not, because $\ell_{i,p}$ is angular.

\subsection{Spatial estimators and regularized quadratic form}

For $\bm x\in\R^p$, define $\mathbf U(\bm x)=\bm x/\norm{\bm x}$ when $\bm x\ne\0$ and $\mathbf U(\0)=\0$.  The sample spatial median and its centered spatial-sign covariance matrix are
\begin{align*}
 \wh\btheta
 &\in\arg\min_{\bm t\in\R^p}
 \sum_{i=1}^n\norm{\bm X_i-\bm t},\\
 \wh{\bm Y}_i
 &=\sqrt p\,\mathbf U(\bm X_i-\wh\btheta),
 &\wh\mR_n&=\frac1n\sum_{i=1}^n
 \wh{\bm Y}_i\wh{\bm Y}_i^\top.
\end{align*}
We refer to $\wh\mR_n$ as the centered spatial-sign covariance matrix (SSCM).  Under the distributional conditions imposed below, the observations are not contained in an affine line with probability one for all sufficiently large $n$; hence the sample spatial median is unique almost surely.  Its score equation gives $n^{-1}\sum_i\wh{\bm Y}_i=\0$ and $\tr(\wh\mR_n)=p$.  For $\rho>0$, define
\begin{equation*}
 T_n(\rho)
 =n(\wh\btheta-\btheta_{0,p})^\top
 (\wh\mR_n+\rho\I_p)^{-1}
 (\wh\btheta-\btheta_{0,p}).
\end{equation*}
A deterministic location shift changes $\wh\btheta$ by the same shift and leaves $\wh\mR_n$ and all residual-based nuisance estimators unchanged; the test statistic itself additionally contains the resulting location-signal quadratic and cross terms.

\subsection{Null distribution and feasible calibration}\label{sec:main}

The following primitive conditions are imposed only for the distribution theory.

\begin{assumption}\label{ass:model}
The representation \eqref{eq:model} holds and the following conditions are satisfied.
\begin{enumerate}[label=(\roman*),leftmargin=1.8em]
\item The pairs $(R_i,\bm G_i)$ are independent and identically distributed across $i$.  The variables $R_i$ are independent of all Gaussian directions, their common distribution does not depend on $(n,p)$, and, with $\xi_i=R_i^{-1}$,
\begin{equation*}
 0<\xi_i\le \bar\xi\quad\text{almost surely},
 \qquad m_1:=\E\xi_i\ge m_->0.
\end{equation*}
Put $m_k=\E\xi_i^k$ for $k=0,1,2$, with $m_0=1$.
\item
\[
 0<c_-\le c_n:=p/n\le c_+<\infty,
 \qquad c_n\to c\in(0,\infty),
\]
and $\rho\in[\rho_0,\rho_1]$ for fixed $0<\rho_0<\rho_1<\infty$.
\item The deterministic shape matrix has the spectral decomposition
\begin{equation*}
 \mOmega_p
 =\mP_{F,p}\mLambda_{F,p}\mP_{F,p}^\top
 +\mP_{B,p}\mLambda_{B,p}\mP_{B,p}^\top.
\end{equation*}
Here $r\in\{0,1,\ldots\}$ is fixed, $\mP_{F,p}\in\R^{p\times r}$ and $\mP_{B,p}\in\R^{p\times(p-r)}$ have orthonormal columns spanning the pervasive and bulk eigenspaces, respectively.  When $r=0$, $\mP_{F,p}$ is the empty $p\times0$ matrix, $\mPi_{F,p}=\0$, $\mP_{B,p}$ is orthogonal, and every sum, block, and assertion indexed by a factor coordinate is understood to be vacuous.  In all cases,
\[
 \mLambda_{F,p}=\diag(\lambda_{1,p},\ldots,\lambda_{r,p}),
 \qquad
 \mLambda_{B,p}=\diag(\lambda_{r+1,p},\ldots,\lambda_{p,p}),
\]
and
\[
 \mPi_{F,p}=\mP_{F,p}\mP_{F,p}^\top,
 \qquad
 \mPi_{B,p}=\mP_{B,p}\mP_{B,p}^\top=\I_p-\mPi_{F,p}
\]
are the associated orthogonal projectors.  The eigenvalues satisfy
\[
 0<\tau_-\le\tau_{a,p}:=\lambda_{a,p}/p
 \le\tau_+<1,
 \qquad 1\le a\le r,
\]
and
\[
 0<\omega_-\le\lambda_{j,p}\le\omega_+<\infty,
 \qquad r<j\le p.
\]
Writing
\[
 \mu_{0,p}=p^{-1}\sum_{j=r+1}^p\lambda_{j,p}
 =1-\sum_{a=1}^r\tau_{a,p},
\]
there is a constant $\mu_->0$ such that $\mu_{0,p}\ge\mu_-$.
\item The empirical bulk spectral distribution
\[
 H_{B,p}=\frac1{p-r}\sum_{j=r+1}^p\delta_{\lambda_{j,p}}
\]
converges weakly to a probability distribution $H_B$ supported on $[\omega_-,\omega_+]$, and $\tau_{a,p}\to\tau_a\in[\tau_-,\tau_+]$.
\end{enumerate}
\end{assumption}

Assumption~\ref{ass:model} permits
$\lambda_{\max}(\mOmega_p)/\lambda_{\min}(\mOmega_p)=O(p)$ and includes fixed-correlation compound symmetry.  Since $\lambda_{\min}(\mOmega_p)\ge\omega_-$ for all sufficiently large $p$, \eqref{eq:elliptical_sign_weight} implies
\begin{equation*}
 0<w_i\le \bar w:=\bar\xi\,\omega_-^{-1/2}
 \qquad\text{almost surely}.
\end{equation*}
No positive moment of the upper radial tail $R_i$ is required.

\subsubsection{Theoretical null distribution}\label{subsec:theoretical_null}

Under $H_0$, define
\begin{equation*}
 \bm Y_i=\sqrt p\,\bm U_i,
 \qquad \mY=(\bm Y_1,\ldots,\bm Y_n),
 \qquad \mB_n=n^{-1}\mY\mY^\top,
 \qquad \mQ_n=(\mB_n+\rho\I_p)^{-1}.
\end{equation*}
The companion matrix is
\begin{equation*}
 \mA_n=\frac1n\mY^\top\mQ_n\mY
 =\I_n-\rho\left(\frac1n\mY^\top\mY+\rho\I_n\right)^{-1}.
\end{equation*}
Put
\begin{align*}
 \kappa_n&=\frac1n\tr(\mA_n),
 &e_n&=\frac1n\sum_{i=1}^nw_i,
 &t_n&=\frac1n\sum_{i=1}^nw_i^2,\\
 b_{1,n}&=\frac1n\sum_{i=1}^n(\mA_n)_{ii}w_i,
 &b_{2,n}&=\frac1n\sum_{i=1}^n(\mA_n)_{ii}w_i^2.
\end{align*}
The conditional denominator and centering are
\begin{equation}\label{eq:Dn_mun_ell}
 D_n=(e_n-b_{1,n})^2+\kappa_n(t_n-b_{2,n}),
 \qquad \mu_n=\frac{\kappa_n}{D_n}.
\end{equation}

For $a,b\in\{0,1,2\}$, define the weighted off-diagonal companion functionals
\begin{equation*}
 \psi_{ab,n}
 =\frac1n\sum_{i\ne j}(\mA_n)_{ij}^2w_i^aw_j^b.
\end{equation*}
They satisfy $\psi_{ab,n}=\psi_{ba,n}$.  Let
\begin{equation}\label{eq:Gamma_conditional}
 \mGamma_n=
 \begin{pmatrix}
 2\psi_{00,n}&2\psi_{01,n}&2\psi_{11,n}\\
 2\psi_{01,n}&\psi_{02,n}+\psi_{11,n}&2\psi_{12,n}\\
 2\psi_{11,n}&2\psi_{12,n}&2\psi_{22,n}
 \end{pmatrix}
\end{equation}
and
\begin{equation}\label{eq:g_conditional}
 \bg_n=D_n^{-2}
 \begin{pmatrix}
 (e_n-b_{1,n})^2\\
 2\kappa_n(e_n-b_{1,n})\\
 \kappa_n^2
 \end{pmatrix},
 \qquad
 \sigma_{D,n}^2=\bg_n^\top\mGamma_n\bg_n.
\end{equation}

For a nonzero vector $\bm y$, let $\mathfrak o(\bm y)$ be the unique member of $\{\bm y,-\bm y\}$ whose first nonzero coordinate is positive, and set
\[
 \calF_n^0=\sigma\{w_i,\mathfrak o(\bm Y_i):1\le i\le n\},
\]
where $\sigma\{\cdot\}$ denotes the sigma-field generated by its arguments.  Conditional on $\calF_n^0$, the unrecorded global signs of the $\bm Y_i$ are independent Rademacher variables, each taking the values $-1$ and $1$ with probability $1/2$.

\begin{theorem}\label{thm:theoretical_null}
Under Assumption~\ref{ass:model} and $H_0$, for every fixed $\rho\in[\rho_0,\rho_1]$,
\begin{equation}\label{eq:oracle_null_main}
 \sup_{x\in\R}
 \left|
 \Pp\left\{
 \left.
 \frac{T_n(\rho)-n\mu_n}{\sqrt{n\sigma_{D,n}^2}}
 \le x\right|\calF_n^0
 \right\}-\Phi(x)
 \right|\pto0
\end{equation}
Hence the same convergence holds unconditionally.  Moreover, there are constants $0<c_1<c_2<\infty$ such that
\[
 \Pp\{c_1\le D_n,\mu_n,\sigma_{D,n}^2\le c_2\}
 \to1.
\]
Consequently, $n\mu_n$ and $n\sigma_{D,n}^2$ are the conditional centering and variance scale in the Gaussian approximation, and both are of order $n$.
\end{theorem}

\subsubsection{Feasible centering and variance}\label{sec:feasible}

All quantities in \eqref{eq:Dn_mun_ell}--\eqref{eq:g_conditional} have direct sample analogues.  Define
\[
 \wh w_i=\frac{\sqrt p}{\norm{\bm X_i-\wh\btheta}},
 \qquad
 \wh\mQ_n=(\wh\mR_n+\rho\I_p)^{-1},
 \qquad
 \wh\mA_n=\frac1n\wh\mY^\top\wh\mQ_n\wh\mY,
\]
where $\wh\mY=(\wh{\bm Y}_1,\ldots,\wh{\bm Y}_n)$.  Put
\begin{align*}
 \wh\kappa_n&=n^{-1}\tr(\wh\mA_n),
 &\wh e_n&=n^{-1}\sum_i\wh w_i,
 &\wh t_n&=n^{-1}\sum_i\wh w_i^2,\\
 \wh b_{1,n}&=n^{-1}\sum_i(\wh\mA_n)_{ii}\wh w_i,
 &\wh b_{2,n}&=n^{-1}\sum_i(\wh\mA_n)_{ii}\wh w_i^2.
\end{align*}
Then
\[
 \wh D_n=(\wh e_n-\wh b_{1,n})^2
 +\wh\kappa_n(\wh t_n-\wh b_{2,n}),
 \qquad
 \wh\mu_n=\frac{\wh\kappa_n}{\wh D_n}.
\]
For $a,b\in\{0,1,2\}$, set
\[
 \wh\psi_{ab,n}
 =\frac1n\sum_{i\ne j}(\wh\mA_n)_{ij}^2
 \wh w_i^a\wh w_j^b,
\]
form $\wh\mGamma_n$ by replacing $\psi_{ab,n}$ with $\wh\psi_{ab,n}$ in \eqref{eq:Gamma_conditional}, and define
\[
 \wh\bg_n=\wh D_n^{-2}
 \begin{pmatrix}
 (\wh e_n-\wh b_{1,n})^2\\
 2\wh\kappa_n(\wh e_n-\wh b_{1,n})\\
 \wh\kappa_n^2
 \end{pmatrix},
 \qquad
 \wh\sigma_{D,n}^2=\wh\bg_n^\top\wh\mGamma_n\wh\bg_n.
\]
The fully feasible statistic is
\begin{equation*}
 Z_n(\rho)=
 \frac{T_n(\rho)-n\wh\mu_n}
 {\sqrt{n\wh\sigma_{D,n}^2}}.
\end{equation*}

\begin{theorem}\label{thm:feasible}
Under Assumption~\ref{ass:model} and $H_0$, for every fixed $\rho\in[\rho_0,\rho_1]$,
\begin{equation}\label{eq:feasible_CLT}
 Z_n(\rho)\dto N(0,1).
\end{equation}
More precisely,
\[
 \wh\mu_n-\mu_n=O_{\Pp}(n^{-1}),
 \qquad
 \wh\sigma_{D,n}^2-\sigma_{D,n}^2
 =O_{\Pp}(n^{-1/2}),
\]
and $\wh\sigma_{D,n}^2$ is bounded away from zero and infinity with probability tending to one.  
\end{theorem}
According to Theorem \ref{thm:feasible}, we rejects the null hypothesis when $Z_n(\rho)>z_{1-\alpha}$.

\subsection{Local alternatives, power, and the ridge parameter}\label{sec:power}

Let
\[
 \mu_0=1-\sum_{a=1}^r\tau_a=\int t\,\dd H_B(t),
 \qquad
 D^0=\left(\mu_0+\sum_{a=1}^r\tau_aG_a^2\right)^{-1},
\]
where $G_1,\ldots,G_r$ are independent $N(0,1)$ variables, and let $F_D$ denote the distribution of $D^0$.  For $\rho\in[\rho_0,\rho_1]$, let $(\delta_\rho,\wt\delta_\rho)$ be the unique positive solution of
\begin{align}
 \delta_\rho
 &=c\int\frac{t}{\rho(1+\wt\delta_\rho t)}\,\dd H_B(t),\label{eq:limit_delta}\\
 \wt\delta_\rho
 &=\int\frac{d}{\rho(1+\delta_\rho d)}\,\dd F_D(d).\label{eq:limit_tdelta}
\end{align}
Define
\begin{align*}
 \lambda_\rho(d)&=\frac{\delta_\rho d}{1+\delta_\rho d},
 &\kappa_\rho&=\int\lambda_\rho(d)\,\dd F_D(d),\\
 s_a&=\int d^{a/2}\,\dd F_D(d),
 &b_{a,\rho}&=\int d^{a/2}\lambda_\rho(d)\,\dd F_D(d),
 \qquad a=1,2,
\end{align*}
and
\begin{align*}
 u_\rho&=c\int\frac{t^2}{(1+\wt\delta_\rho t)^2}\,\dd H_B(t),
 &v_\rho&=\int\frac{d^2}{(1+\delta_\rho d)^2}\,\dd F_D(d),\\
 a_{k,\rho}&=\int\frac{d^{1+k/2}}{(1+\delta_\rho d)^2}\,\dd F_D(d),
 &&k=0,1,2.
\end{align*}
The deterministic denominator is
\begin{equation}\label{eq:D_rho_ell}
 D_\rho=m_1^2(s_1-b_{1,\rho})^2
 +m_2\kappa_\rho(s_2-b_{2,\rho}).
\end{equation}
For $a,b\in\{0,1,2\}$, put
\begin{equation}\label{eq:Psi_limit}
 \Psi_{ab}(\rho)
 =m_am_b\frac{u_\rho a_{a,\rho}a_{b,\rho}}
 {\rho^2-u_\rho v_\rho},
 \qquad m_0=1.
\end{equation}
Let $\mGamma_\rho$ be the matrix in \eqref{eq:Gamma_conditional} with $\psi_{ab,n}$ replaced by $\Psi_{ab}(\rho)$, and define
\begin{equation}\label{eq:sigma_E_limit}
 \bg_\rho=D_\rho^{-2}
 \begin{pmatrix}
 m_1^2(s_1-b_{1,\rho})^2\\
 2\kappa_\rho m_1(s_1-b_{1,\rho})\\
 \kappa_\rho^2
 \end{pmatrix},
 \qquad
 \sigma_E^2(\rho)=\bg_\rho^\top\mGamma_\rho\bg_\rho.
\end{equation}
Appendix~\ref{app:det_equiv} shows that $D_n\to D_\rho$ and $\sigma_{D,n}^2\to\sigma_E^2(\rho)$ in probability.  The stability inequality $u_\rho v_\rho<\rho^2$ makes \eqref{eq:Psi_limit} finite, and $\sigma_E^2(\rho)>0$.

\subsubsection{Bulk local alternatives and power}

Let $\bm p_{j,p}$ be the eigenvector of $\mOmega_p$ associated with the bulk eigenvalue $\lambda_{j,p}$, $j>r$.  Consider deterministic $\bm d_p$ such that
\[
 \mPi_{F,p}\bm d_p=\0,
 \qquad 0<d_-\le\norm{\bm d_p}\le d_+<\infty,
\]
and define
\begin{equation}\label{eq:signal_measure}
 \calG_p=\sum_{j=r+1}^p
 (\bm p_{j,p}^\top\bm d_p)^2\delta_{\lambda_{j,p}}.
\end{equation}
Assume $\calG_p\Rightarrow\calG$, where $\calG$ is a nonzero finite measure on $[\omega_-,\omega_+]$.  The local alternative is
\begin{equation}\label{eq:local_alternative}
 H_{1,n}(\calG):
 \qquad \btheta_p=\btheta_{0,p}+n^{-1/4}\bm d_p.
\end{equation}
Define
\begin{equation}\label{eq:q_Lambda_ell}
 q_\rho(\calG)=\frac1\rho
 \int\frac{1}{1+\wt\delta_\rho t}\,\dd\calG(t),
 \qquad
 \Lambda_\rho(\calG)=\frac{q_\rho(\calG)}{\sigma_E(\rho)}.
\end{equation}

\begin{theorem}\label{thm:power}
Suppose Assumption~\ref{ass:model} holds.  Under \eqref{eq:local_alternative}, for every fixed $\rho\in[\rho_0,\rho_1]$,
\begin{equation}\label{eq:local_normal_limit}
 Z_n(\rho)\dto N\{\Lambda_\rho(\calG),1\}.
\end{equation}
Consequently, the limiting power of the level-$\alpha$ test is
\begin{equation}\label{eq:power_function}
 \beta_\alpha(\rho;\calG)
 =1-\Phi\{z_{1-\alpha}-\Lambda_\rho(\calG)\}.
\end{equation}
More generally, under
\begin{equation}\label{eq:H1_model}
 \bm X_i=\btheta_{0,p}+\bDelta_p
 +\sqrt p\,R_i\mOmega_p^{1/2}
 \frac{\bm G_i}{\norm{\bm G_i}},
\end{equation}
if
\begin{equation}\label{eq:primitive_power}
 \sqrt n\left\{
 \norm{\mPi_{B,p}\bDelta_p}^2
 +p^{-1}\norm{\mPi_{F,p}\bDelta_p}^2
 \right\}\to\infty,
\end{equation}
then $\Pp_{H_1}\{Z_n(\rho)>z_{1-\alpha}\}\to1$.
\end{theorem}

For a diffuse signal with $\calG=s^2H_B$,
\[
 q_\rho(\calG)=\frac{s^2}{\rho}
 \int\frac{1}{1+\wt\delta_\rho t}\,\dd H_B(t),
\]
whereas concentration in a bulk eigenspace with eigenvalue $t_0$ gives
\[
 q_\rho(s^2\delta_{t_0})
 =\frac{s^2}{\rho(1+\wt\delta_\rho t_0)}.
\]
Corollary~\ref{cor:compound_symmetry} gives $H_B=\delta_{1-\varrho}$ and
$D^0=(1-\varrho+\varrho G^2)^{-1}$, so all integrals in
\eqref{eq:D_rho_ell}--\eqref{eq:q_Lambda_ell} reduce to one-dimensional Gaussian expectations.

\subsubsection{Oracle choice of the ridge parameter}

Since $\Phi$ is increasing, an oracle ridge parameter for signal profile $\calG$ is
\begin{equation}\label{eq:rho_opt}
 \rho_{\calG}^{\star}
 \in\arg\max_{\rho\in[\rho_0,\rho_1]}
 \frac{q_\rho(\calG)}{\sigma_E(\rho)}.
\end{equation}
The objective is continuous and the compact interval ensures existence.  Multiplying $\calG$ by a positive constant leaves the optimizer unchanged.  At an interior maximizer,
\[
 2\frac{q_\rho'(\calG)}{q_\rho(\calG)}
 =\frac{\{\sigma_E^2(\rho)\}'}{\sigma_E^2(\rho)}.
\]
The canonical derivatives solve
\begin{equation}\label{eq:canonical_derivative_system_main}
 \begin{pmatrix}\rho&u_\rho\\v_\rho&\rho\end{pmatrix}
 \begin{pmatrix}\delta_\rho'\\\wt\delta_\rho'\end{pmatrix}
 =-\begin{pmatrix}\delta_\rho\\\wt\delta_\rho\end{pmatrix},
\end{equation}
and
\[
 q_\rho'(\calG)
 =-\frac{q_\rho(\calG)}\rho
 -\frac{\wt\delta_\rho'}\rho
 \int\frac{t}{(1+\wt\delta_\rho t)^2}\,\dd\calG(t).
\]
The remaining derivatives are finite one-dimensional integrals; explicit formulas are given in Lemma~\ref{lem:power_derivatives} in Appendix~\ref{app:power}.  Thus \eqref{eq:rho_opt} can be evaluated by a one-dimensional numerical search.

\section{Cauchy aggregation over a fixed ridge grid}\label{sec:cauchy}

The power calculation in Section~\ref{sec:power} identifies an oracle ridge value $\rho_{\calG}^{\star}$, but the optimizer depends on the unknown signal measure $\calG$.  Estimating $\calG$ under a local alternative is difficult and can introduce a second tuning problem.  We therefore replace selection of one ridge value by aggregation over a deterministic finite set
\[
 \mathcal R_K=\{\rho^{(1)},\ldots,\rho^{(K)}\}\subset[\rho_0,\rho_1],
 \qquad \rho_0\le\rho^{(1)}<\cdots<\rho^{(K)}\le\rho_1,
\]
where $K$ is fixed as $n,p\to\infty$.  For each $\rho^{(k)}$, write
\[
 Z_{n,k}=Z_n(\rho^{(k)}),
 \qquad p_{n,k}=1-\Phi(Z_{n,k}),
 \qquad 1\le k\le K.
\]
Let $\varpi_1,\ldots,\varpi_K$ be fixed positive weights satisfying $\sum_{k=1}^K\varpi_k=1$.  The Cauchy statistic and its analytic combined $p$-value are
\[
 T_{\mathrm{CC},n}
 =\sum_{k=1}^K\varpi_k
 \tan\!\left[\pi\left\{\frac12-p_{n,k}\right\}\right],
 \qquad
 p_{\mathrm{CC},n}
 =\frac12-\frac1\pi\arctan(T_{\mathrm{CC},n}).
\]
The implementation below uses equal weights $\varpi_k=K^{-1}$.  We call the resulting procedure ERHT--CC.

The marginal limit in Theorem~\ref{thm:feasible} is not sufficient to justify aggregation, because the statistics at different ridge values are computed from the same spatial median, signs, and inverse distances.  Appendix~\ref{app:cauchy} therefore develops the cross-ridge weighted companion functionals.  For $\rho,\eta\in[\rho_0,\rho_1]$, it defines a deterministic cross-covariance matrix $\mGamma_{\rho,\eta}$ and
\[
 r(\rho,\eta)
 =\frac{\bg_\rho^\top\mGamma_{\rho,\eta}\bg_\eta}
 {\sigma_E(\rho)\sigma_E(\eta)}.
\]
Put $\mC_K=\{r(\rho^{(k)},\rho^{(\ell)})\}_{k,\ell=1}^K$.  The matrix $\mC_K$ is positive semidefinite, has unit diagonal, and has nonnegative entries.  Appendix~\ref{app:cauchy} also defines its finite-$p$ deterministic counterpart $\mC_{K,p}^{\star}=\{r_{p,k\ell}^{\star}\}_{k,\ell=1}^K$.

\begin{theorem}\label{thm:cauchy_joint}
Under Assumption~\ref{ass:model} and $H_0$, for every fixed deterministic grid $\mathcal R_K$,
\[
 (Z_{n,1},\ldots,Z_{n,K})^\top
 \dto N_K(\0,\mC_K).
\]
\end{theorem}

Theorem~\ref{thm:cauchy_joint} also defines a theoretical fixed-level benchmark.  Let $\bm V=(V_1,\ldots,V_K)^\top\sim N_K(\0,\mC_K)$ and define
\[
 T_{\mathrm{CC},\infty}^{0}
 =\sum_{k=1}^K\varpi_k
 \tan\!\left[\pi\{\Phi(V_k)-1/2\}\right].
\]
Let $c_{\alpha,K}$ be the unique value satisfying
$\Pp(T_{\mathrm{CC},\infty}^{0}>c_{\alpha,K})=\alpha$; its existence and uniqueness follow from Lemma~\ref{lem:cauchy_tail}.  This critical value is used only to characterize the oracle joint-limit benchmark.  The implemented analytic Cauchy procedure uses the standard-Cauchy cutoff $\cot(\pi\alpha)$ and requires no estimate of $\mC_K$.

\begin{theorem}\label{thm:cauchy_size}
Under the conditions of Theorem~\ref{thm:cauchy_joint},
\[
 \Pp\{T_{\mathrm{CC},n}>c_{\alpha,K}\}\to\alpha.
\]
For the analytic Cauchy $p$-value, at every fixed $\alpha\in(0,1/2)$,
\[
 \Pp(p_{\mathrm{CC},n}\le\alpha)
 \to
 \Pp\{T_{\mathrm{CC},\infty}^{0}\ge\cot(\pi\alpha)\}.
\]
Furthermore,
\[
 \lim_{\alpha\downarrow0}\lim_{n\to\infty}
 \frac{\Pp(p_{\mathrm{CC},n}\le\alpha)}{\alpha}=1.
\]
Thus the oracle joint-limit benchmark is asymptotically exact at an ordinary fixed level, whereas the implemented analytic Cauchy $p$-value has the dependence-robust small-tail validity associated with the Cauchy combination principle.
\end{theorem}

The distinction in Theorem~\ref{thm:cauchy_size} is important.  The critical value $c_{\alpha,K}$ depends on the unknown matrix $\mC_K$ and is retained only as a theoretical benchmark.  Except under special dependence structures, the analytic standard-Cauchy cutoff need not be exactly level $\alpha$ for a fixed conventional value such as $0.05$.  It is used because it is closed form and avoids estimating cross-ridge dependence; the Monte Carlo study evaluates its finite-sample behavior at the 5\% level.

The same finite-grid argument gives an explicit power function.  Put
\[
 \bm\Lambda_{\mathcal R}(\calG)
 =\{\Lambda_{\rho^{(1)}}(\calG),\ldots,
 \Lambda_{\rho^{(K)}}(\calG)\}^\top
\]
and, for $\bm V\sim N_K(\0,\mC_K)$, define
\[
 T_{\mathrm{CC},\infty}(\calG)
 =\sum_{k=1}^K\varpi_k
 \tan\!\left[
 \pi\{\Phi(V_k+\Lambda_{\rho^{(k)}}(\calG))-1/2\}
 \right].
\]

\begin{theorem}\label{thm:cauchy_power}
Suppose Assumption~\ref{ass:model} holds.  Under the local alternative \eqref{eq:local_alternative},
\[
 (Z_{n,1},\ldots,Z_{n,K})^\top
 \dto
 N_K\{\bm\Lambda_{\mathcal R}(\calG),\mC_K\}.
\]
Consequently, the limiting power of the analytic ERHT--CC test is
\[
 \beta_{\mathrm{CC},\alpha}^{\mathrm A}(\calG)
 =\Pp\{T_{\mathrm{CC},\infty}(\calG)
 \ge\cot(\pi\alpha)\},
\]
and the limiting power of the oracle joint-limit benchmark is
\[
 \beta_{\mathrm{CC},\alpha}^{\mathrm J}(\calG)
 =\Pp\{T_{\mathrm{CC},\infty}(\calG)>c_{\alpha,K}\}.
\]
Under the general shifted model \eqref{eq:H1_model}, if \eqref{eq:primitive_power} holds, then
\[
 p_{\mathrm{CC},n}\pto0,
 \qquad
 \Pp_{H_1}(p_{\mathrm{CC},n}\le\alpha)\to1.
\]
\end{theorem}

\section{Monte Carlo evidence}\label{sec:simulation}

We compare four one-sample location tests.  The fixed-ridge statistic $T_n(\rho)$ is referred to as the elliptical regularized Hotelling (ERHT) statistic.  The proposed ERHT--CC procedure combines its ten marginal $p$-values over $\mathcal R_{10}=\{0.1,0.2,\ldots,1.0\}$ with equal Cauchy weights.  CQ denotes the one-sample specialization of the diagonal-deleted quadratic $U$-statistic of \citet{ChenQin2010}.  WPL denotes the spatial-sign test of \citet{WangPengLi2015}; following its numerical implementation, each coordinate is standardized by its sample standard deviation before the spatial signs are formed.

RHT--CC is included as a covariance-based ridge benchmark.  For every $\rho\in\mathcal R_{10}$, we compute the regularized Hotelling statistic of \citet{ChenPaulPrenticeWang2011} from the sample mean and ordinary sample covariance matrix, convert it to a marginal $p$-value using its random-matrix calibration, and aggregate the ten $p$-values by exactly the same equal-weight Cauchy rule used for ERHT--CC.  Thus the comparison between ERHT--CC and RHT--CC isolates the effect of replacing the sample mean and covariance matrix by the sample spatial median and its centered spatial-sign covariance matrix; it is not driven by different ridge-selection rules.

All observations are generated from
\[
 \bm X_i=\btheta_p+\bm\varepsilon_i,\qquad i=1,\ldots,n,
\]
where $\mOmega_p$ is standardized by $\tr(\mOmega_p)=p$.  We consider the bounded-spectrum AR(1) matrix and the pervasive compound-symmetry matrix
\[
 \Omega_{p,jk}=0.5^{|j-k|},
 \qquad
 \mOmega_p=0.5\I_p+0.5\1_p\1_p^\top,
\]
respectively.  The latter has largest eigenvalue $1+0.5(p-1)$ and therefore directly examines the strong-dependence regime in Assumption~\ref{ass:model}.

Three centered elliptical error distributions are used:
\begin{align*}
 \text{Gaussian:}\quad
 &\bm\varepsilon_i\sim N_p(\0,\mOmega_p),\\
 \text{Multivariate t-distribution:}\quad
 &\bm\varepsilon_i\sim t_{5,p}\!\left(\0,\frac35\mOmega_p\right),\\
 \text{normal mixture:}\quad
 &\bm\varepsilon_i\sim
 0.8N_p\!\left(\0,\frac1{2.6}\mOmega_p\right)
 +0.2N_p\!\left(\0,\frac9{2.6}\mOmega_p\right).
\end{align*}
For the multivariate $t$ distribution, the second argument denotes its scale matrix, so its covariance is $\mOmega_p$.  The normal mixture also has covariance $\mOmega_p$.  Consequently, the three designs differ in radial-tail behavior while sharing the same covariance normalization.

For the size study, $H_0:\btheta_p=\0$, $n\in\{100,200\}$, and $p\in\{100,200,400\}$.  Both ridge-based methods use $\mathcal R_{10}$, and ERHT--CC uses the Bartlett center correction $a=8n/p$ implemented in the accompanying program.  The nominal level is $0.05$, and every entry is based on 1,000 Monte Carlo replications.

For the power study, we fix $(n,p)=(100,200)$ and draw a new location vector independently in each replication.  Three signal designs are considered:
\begin{align*}
 \text{isotropic dense:}\quad
 &\btheta_p\sim N_p(\0,c\I_p),\\
 \text{covariance-aligned dense:}\quad
 &\btheta_p\sim N_p(\0,c\mOmega_p),\\
 \text{four-sparse:}\quad
 &\theta_{p,j}=c\xi_j\,1\{j\in\mathcal S_p\},
 \qquad |\mathcal S_p|=4,
\end{align*}
where $\mathcal S_p$ is sampled uniformly from all four-element subsets of $\{1,\ldots,p\}$ and the nonzero signs $\xi_j$ are independent Rademacher variables.  For the two dense designs,
\[
 c\in\{0,0.00025,0.00050,\ldots,0.00250\},
\]
whereas for the sparse design,
\[
 c\in\{0,0.05,0.10,\ldots,0.50\}.
\]
The different grids account for the fact that $c$ is a variance multiplier in the dense alternatives and a coordinate amplitude in the sparse alternative.  Each power point is based on 2,000 replications; hence its Monte Carlo standard error is at most $\{0.25/2000\}^{1/2}=0.0112$.

\begin{table}[!htbp]
\centering
\caption{Empirical sizes (\%) under AR(1) dependence with $\varrho=0.5$.}\label{tab:size_ar}
\small
\setlength{\tabcolsep}{7pt}
\begin{tabular}{cc l ccc}
\toprule
$n$ & $p$ & Method & Gaussian & Multivariate $t_5$ & Normal mixture \\
\midrule
100 & 100 & ERHT--CC & 5.1 & 3.3 & 5.4 \\
100 & 100 & CQ & 6.4 & 5.6 & 5.2 \\
100 & 100 & RHT--CC & 7.5 & 4.4 & 5.4 \\
100 & 100 & WPL & 7.9 & 6.3 & 6.5 \\
\midrule
100 & 200 & ERHT--CC & 6.7 & 5.3 & 6.4 \\
100 & 200 & CQ & 6.3 & 5.4 & 6.8 \\
100 & 200 & RHT--CC & 9.0 & 3.3 & 5.6 \\
100 & 200 & WPL & 8.6 & 7.4 & 7.4 \\
\midrule
100 & 400 & ERHT--CC & 7.3 & 5.8 & 7.2 \\
100 & 400 & CQ & 6.4 & 4.6 & 4.8 \\
100 & 400 & RHT--CC & 7.8 & 2.5 & 5.3 \\
100 & 400 & WPL & 10.3 & 7.7 & 8.1 \\
\midrule
200 & 100 & ERHT--CC & 4.1 & 3.9 & 4.9 \\
200 & 100 & CQ & 6.7 & 6.8 & 6.6 \\
200 & 100 & RHT--CC & 5.9 & 4.9 & 4.7 \\
200 & 100 & WPL & 6.8 & 6.5 & 6.2 \\
\midrule
200 & 200 & ERHT--CC & 4.6 & 4.5 & 5.6 \\
200 & 200 & CQ & 6.6 & 4.7 & 5.7 \\
200 & 200 & RHT--CC & 6.5 & 4.4 & 5.2 \\
200 & 200 & WPL & 7.4 & 5.8 & 6.6 \\
\midrule
200 & 400 & ERHT--CC & 6.1 & 4.8 & 5.4 \\
200 & 400 & CQ & 6.0 & 4.8 & 6.4 \\
200 & 400 & RHT--CC & 7.1 & 2.8 & 5.0 \\
200 & 400 & WPL & 7.9 & 7.1 & 6.9 \\
\bottomrule
\end{tabular}
\end{table}

\begin{table}[!htbp]
\centering
\caption{Empirical sizes (\%) under compound-symmetry dependence with $\varrho=0.5$.}\label{tab:size_cs}
\small
\setlength{\tabcolsep}{7pt}
\begin{tabular}{cc l ccc}
\toprule
$n$ & $p$ & Method & Gaussian & Multivariate $t_5$ & Normal mixture \\
\midrule
100 & 100 & ERHT--CC & 7.0 & 4.8 & 5.1 \\
100 & 100 & CQ & 7.0 & 6.3 & 7.1 \\
100 & 100 & RHT--CC & 7.9 & 2.6 & 3.8 \\
100 & 100 & WPL & 6.9 & 6.5 & 6.2 \\
\midrule
100 & 200 & ERHT--CC & 6.7 & 5.6 & 7.8 \\
100 & 200 & CQ & 6.8 & 6.8 & 5.4 \\
100 & 200 & RHT--CC & 7.9 & 2.7 & 4.4 \\
100 & 200 & WPL & 6.7 & 6.7 & 6.2 \\
\midrule
100 & 400 & ERHT--CC & 6.2 & 5.3 & 6.9 \\
100 & 400 & CQ & 6.4 & 5.8 & 6.6 \\
100 & 400 & RHT--CC & 7.1 & 1.2 & 3.6 \\
100 & 400 & WPL & 6.2 & 6.3 & 6.7 \\
\midrule
200 & 100 & ERHT--CC & 4.3 & 4.8 & 4.9 \\
200 & 100 & CQ & 7.6 & 7.0 & 5.9 \\
200 & 100 & RHT--CC & 7.1 & 4.9 & 4.1 \\
200 & 100 & WPL & 7.3 & 6.5 & 7.0 \\
\midrule
200 & 200 & ERHT--CC & 6.0 & 5.3 & 5.2 \\
200 & 200 & CQ & 7.3 & 6.8 & 6.2 \\
200 & 200 & RHT--CC & 6.8 & 3.5 & 3.4 \\
200 & 200 & WPL & 7.5 & 6.3 & 6.8 \\
\midrule
200 & 400 & ERHT--CC & 6.8 & 5.8 & 4.5 \\
200 & 400 & CQ & 7.4 & 6.5 & 6.9 \\
200 & 400 & RHT--CC & 7.4 & 1.6 & 2.9 \\
200 & 400 & WPL & 7.6 & 6.1 & 7.0 \\
\bottomrule
\end{tabular}
\end{table}

Tables~\ref{tab:size_ar} and~\ref{tab:size_cs} show that ERHT--CC remains reasonably close to the 5\% target across the three radial distributions, both dependence structures, and all six $(n,p)$ combinations.  CQ and WPL are mildly liberal in a number of settings, with the largest distortion occurring for WPL when $p/n$ is large under AR(1).  RHT--CC is also liberal under Gaussian sampling but becomes conservative under the two heavy-tailed distributions, particularly under compound symmetry.  Since ERHT--CC and RHT--CC use the same ridge grid and aggregation rule, their different behavior under heavy tails reflects the spatial-median and spatial-sign construction rather than tuning-parameter selection.

Figures~\ref{fig:power_ar} and~\ref{fig:power_cs} report the power curves under AR(1) and compound-symmetry dependence, respectively.  In each figure, the rows correspond to $N_p(\0,\mOmega_p)$, $t_{5,p}(\0,3\mOmega_p/5)$, and the standardized normal mixture, while the columns correspond to the isotropic dense, covariance-aligned dense, and four-sparse alternatives.

\begin{figure}[htbp]
\centering
\includegraphics[width=\textwidth]{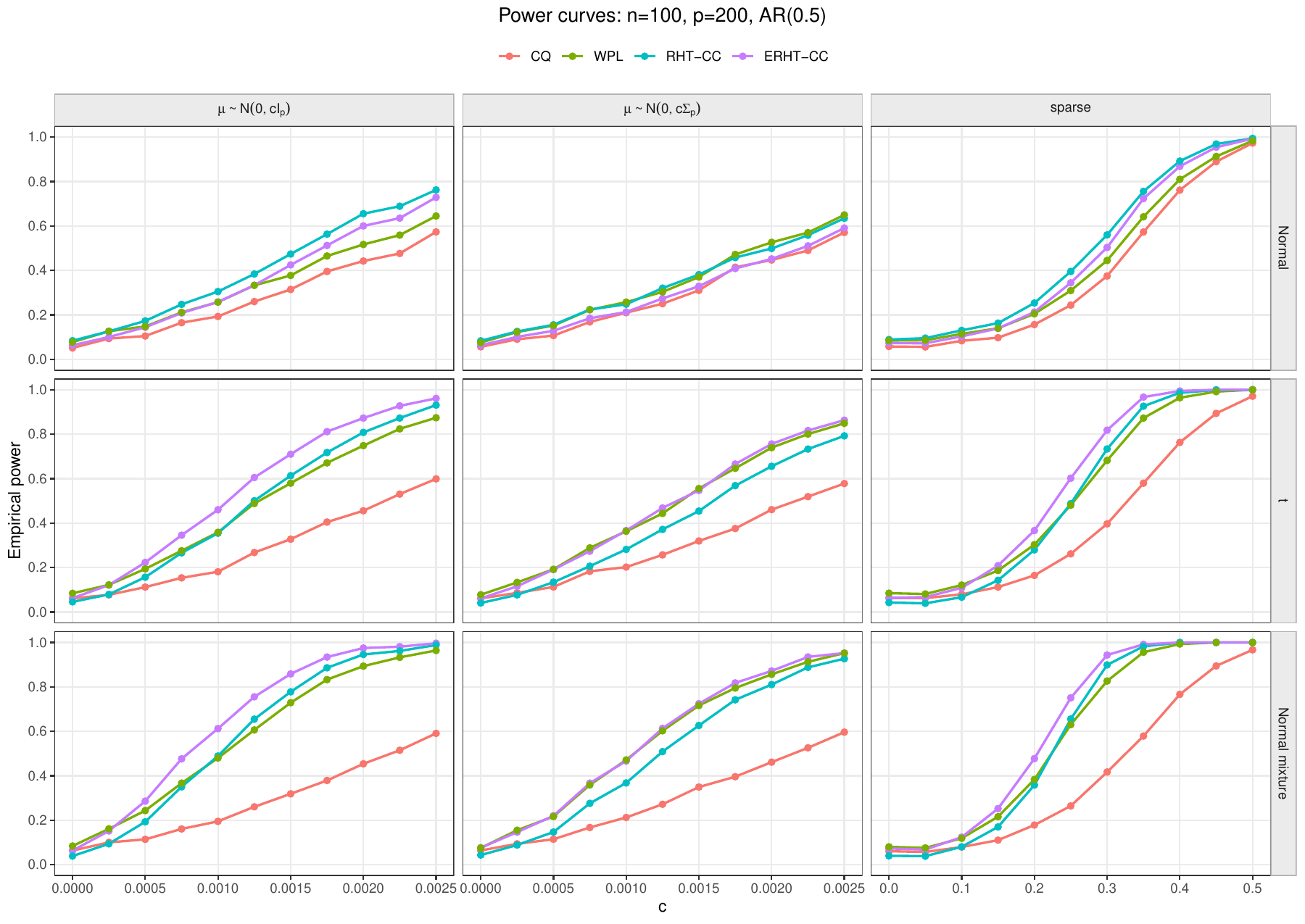}
\caption{Empirical power under AR(1) dependence, $\Omega_{p,jk}=0.5^{|j-k|}$, with $(n,p)=(100,200)$. }\label{fig:power_ar}
\end{figure}

Under AR(1), whose eigenvalues remain uniformly bounded, ERHT--CC and WPL generally lie in the same broad power range.  Their differences are especially modest for covariance-aligned and sparse alternatives, although the ranking can depend on the radial distribution and signal strength.  Under Gaussian sampling, RHT--CC is sometimes slightly more powerful, which is consistent with the efficiency of covariance-based regularization when second-moment modeling is correct.  Once the errors are Multivariate $t_5$ or normal mixtures, ERHT--CC is usually above RHT--CC and CQ, with a particularly visible gain for dense isotropic signals.  The robust procedure therefore pays little under bounded-spectrum dependence while protecting power against radial heavy tails.

\begin{figure}[!htbp]
\centering
\includegraphics[width=\textwidth]{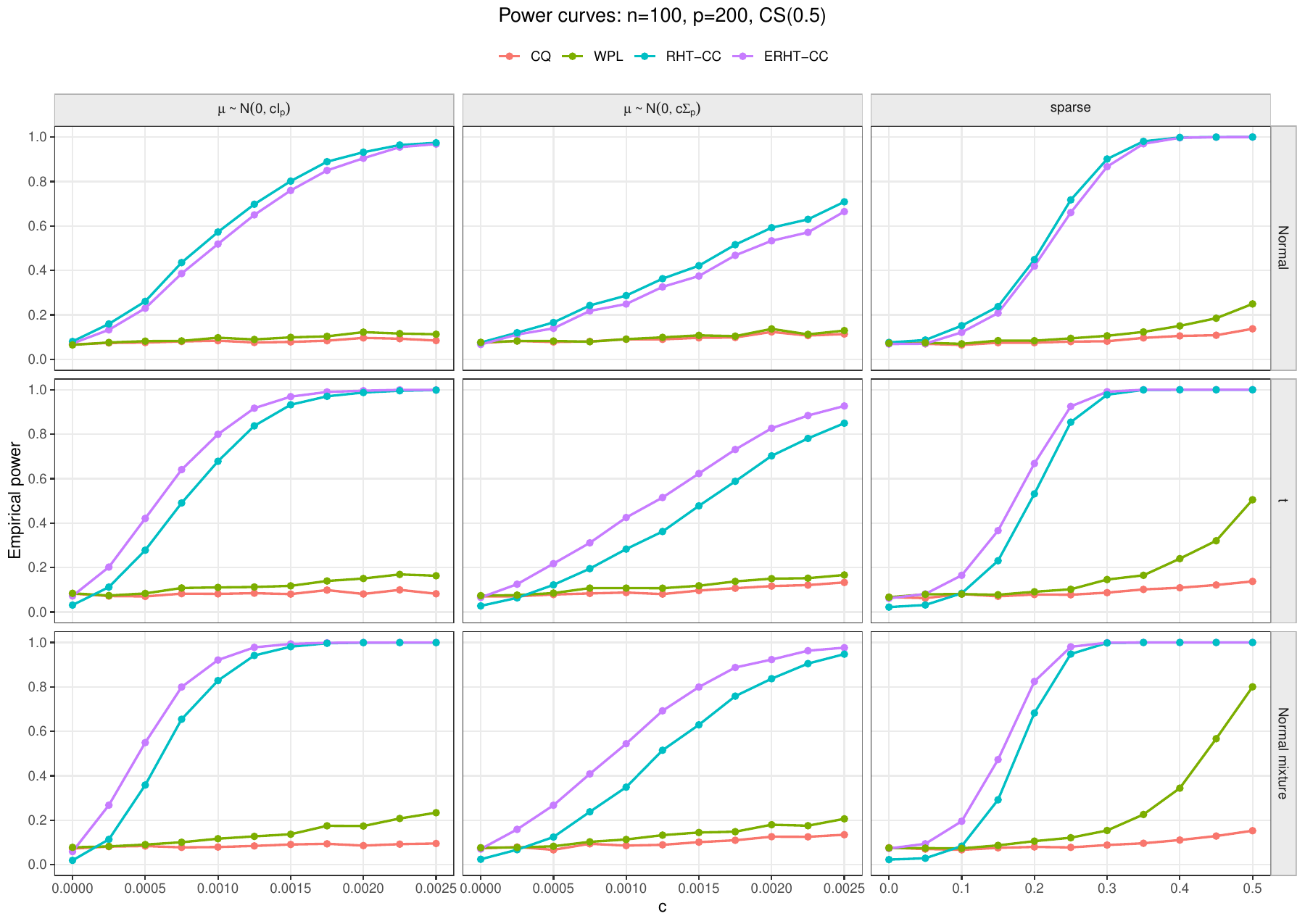}
\caption{Empirical power under compound-symmetry dependence, $\mOmega_p=0.5\I_p+0.5\1_p\1_p^\top$, with $(n,p)=(100,200)$. }\label{fig:power_cs}
\end{figure}

The distinction is much sharper under compound symmetry.  ERHT--CC and RHT--CC exploit the inverse angular-scatter structure and attain rapidly increasing power, whereas WPL can remain close to its null rejection probability for dense signals and increases much more slowly for sparse signals.  Thus, when the largest eigenvalue is proportional to $p$, ERHT--CC is substantially more powerful than WPL over most of the transition region.  The heavy-tailed rows also display the main robustness advantage: ERHT--CC uniformly dominates CQ and is appreciably above RHT--CC at moderate signal strengths for both dense alternatives and for the sparse alternative before the power curves reach one.  Under Gaussian errors, RHT--CC can retain a small advantage in some panels, but this advantage disappears or reverses under Multivariate $t_5$ and normal-mixture sampling.

Taken together, the power results agree with the theoretical motivation.  Under weak spectral dependence, the proposed test is competitive with the existing spatial-sign procedure; under pervasive dependence, ridge normalization of the centered SSCM produces a clear gain.  Replacing Euclidean magnitudes by the spatial median and spatial signs is most consequential under heavy tails, where ERHT--CC substantially improves on the covariance-based RHT--CC and the CQ quadratic test while maintaining the favorable dependence adaptation of a regularized Hotelling statistic.

\section{Real Data Application}\label{sec:realdata}

We illustrate the proposed procedure using the publicly available GSE19804 data set from the NCBI Gene Expression Omnibus
(\href{https://www.ncbi.nlm.nih.gov/geo/query/acc.cgi?acc=GSE19804}{GEO accession GSE19804}).
The study profiles transcriptional changes in lung cancer among nonsmoking women in Taiwan.  Expression was measured on the Affymetrix Human Genome U133 Plus 2.0 Array (GPL570) for 60 lung-tumor specimens and the corresponding 60 adjacent normal-tissue specimens, giving 60 matched pairs and 54,675 probe-level measurements per specimen.

Let $\bm T_i\in\R^p$ and $\bm N_i\in\R^p$ denote the tumor and matched normal expression vectors for patient $i$, respectively, and define the paired difference
\[
 \bm X_i=\bm T_i-\bm N_i,
 \qquad i=1,\ldots,60,
 \qquad p=54{,}675.
\]
The resulting one-sample problem tests whether the multivariate location of the paired expression changes is zero.  We apply CQ, WPL, RHT--CC, and ERHT--CC using the same implementations and, for the two ridge procedures, the same grid $\mathcal R_{10}$ as in Section~\ref{sec:simulation}.

To describe the marginal signal strength, for each probe $j$ we compute the paired Student statistic
\[
 t_j=\frac{\overline X_j}{s_j/\sqrt{60}},
\]
where $\overline X_j$ and $s_j$ are the sample mean and standard deviation of $X_{1j},\ldots,X_{60,j}$.  Positive values indicate higher expression in the tumor tissue.  Figure~\ref{fig:gse19804_t} shows a broad, asymmetric distribution with substantial mass in both tails, indicating that the full probe set contains many pronounced marginal changes.

\begin{figure}[htbp]
\centering
\includegraphics[width=0.82\textwidth]{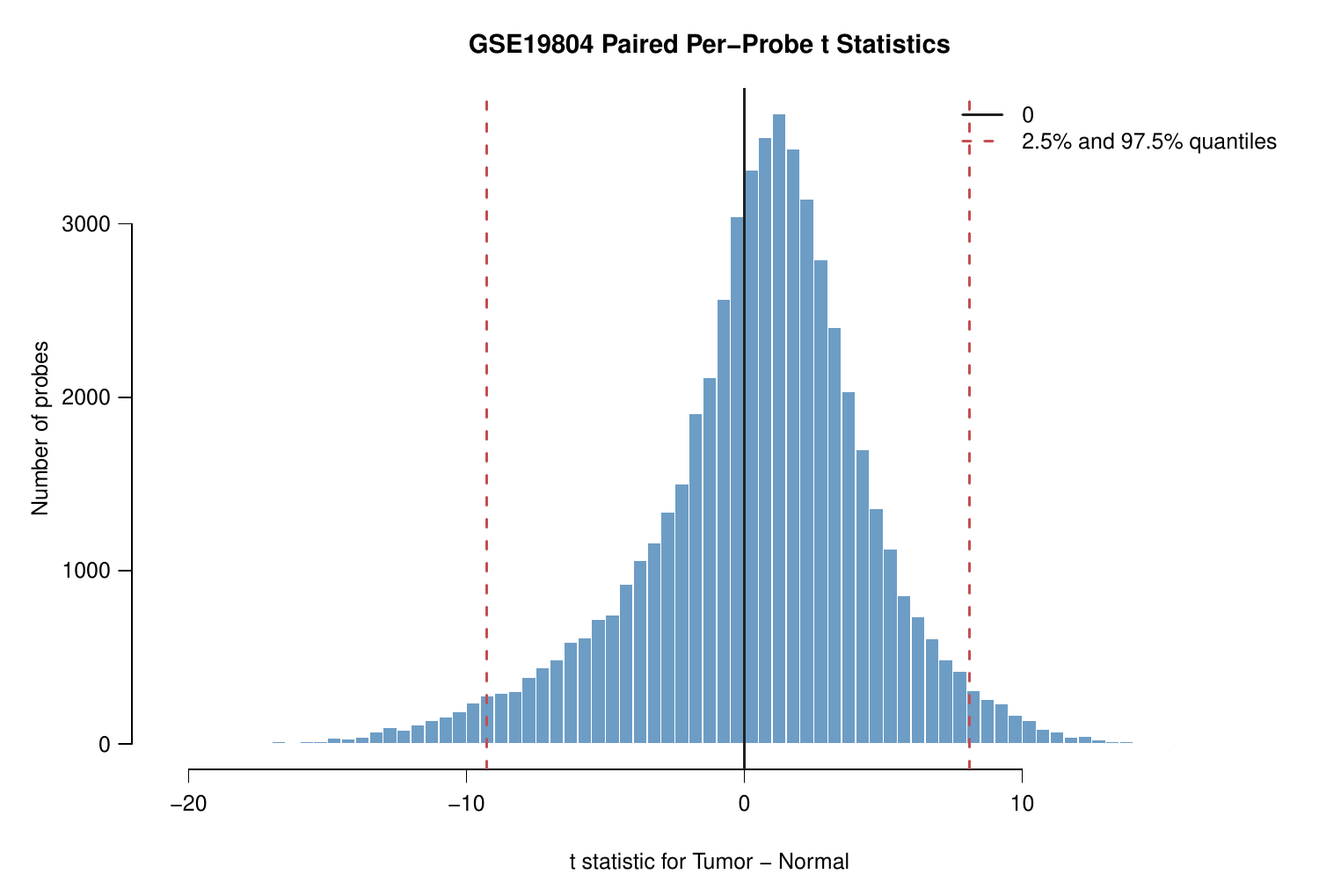}
\caption{Histogram of the paired probe-level Student-t statistics for the 54,675 probes in GSE19804.  The solid vertical line marks zero and the dashed lines mark the empirical 2.5\% and 97.5\% quantiles.}\label{fig:gse19804_t}
\end{figure}

All four global procedures strongly reject the null hypothesis when the complete probe set is used.  Table~\ref{tab:gse19804_full} reports the numerical results.  The Cauchy-combined $p$-values reach the numerical reporting floor used in the implementation.

\begin{table}[htbp]
\centering
\caption{Global tests using all 54,675 probes in GSE19804.}\label{tab:gse19804_full}
\small
\begin{tabular}{lcc}
\toprule
Method & Standardized statistic & $p$-value \\
\midrule
CQ      & 34.9413 & $8.77\times10^{-268}$ \\
WPL     & 35.2884 & $4.43\times10^{-273}$ \\
RHT--CC & --      & $<10^{-15}$ \\
ERHT--CC& --      & $<10^{-15}$ \\
\bottomrule
\end{tabular}
\end{table}

The overwhelming significance of the full analysis leaves little scope for distinguishing the four procedures.  We therefore conduct a descriptive subsampling experiment that focuses on probes with weak marginal evidence.  Specifically, we form a pool containing the 23,508 probes satisfying $|t_j|<2$.  For each $n\in\{30,40,50,60\}$ and $p\in\{40,50,60\}$, one repetition samples $n$ matched pairs and $p$ probes without replacement from the 60 pairs and the weak-signal pool, respectively, and then applies the four tests at level 0.05.  Each configuration is repeated 10,000 times.  Because the probe pool is selected from the same data, the resulting rejection frequencies are intended as a comparative sensitivity analysis rather than as estimates of type-I error or model-based power.

\begin{table}[htbp]
\centering
\caption{Rejection frequencies in the weak-marginal-signal subsampling experiment.}\label{tab:gse19804_subsample}
\small
\setlength{\tabcolsep}{9pt}
\begin{tabular}{cccccc}
\toprule
$n$ & $p$ & CQ & WPL & RHT--CC & ERHT--CC \\
\midrule
30 & 40 & 0.1140 & 0.2658 & 0.2461 & \textbf{0.3386} \\
30 & 50 & 0.1067 & 0.2850 & 0.2798 & \textbf{0.3997} \\
30 & 60 & 0.1095 & 0.3133 & 0.3235 & \textbf{0.4547} \\
\midrule
40 & 40 & 0.1201 & 0.2896 & 0.3104 & \textbf{0.4147} \\
40 & 50 & 0.1202 & 0.3268 & 0.3727 & \textbf{0.4935} \\
40 & 60 & 0.1193 & 0.3509 & 0.4286 & \textbf{0.5620} \\
\midrule
50 & 40 & 0.1308 & 0.3389 & 0.3887 & \textbf{0.5019} \\
50 & 50 & 0.1280 & 0.3807 & 0.4586 & \textbf{0.5904} \\
50 & 60 & 0.1267 & 0.4096 & 0.5327 & \textbf{0.6661} \\
\midrule
60 & 40 & 0.1191 & 0.4058 & 0.4534 & \textbf{0.5966} \\
60 & 50 & 0.1155 & 0.4656 & 0.5637 & \textbf{0.6996} \\
60 & 60 & 0.1298 & 0.5097 & 0.6449 & \textbf{0.7705} \\
\bottomrule
\end{tabular}
\end{table}

ERHT--CC has the largest rejection frequency in every configuration in Table~\ref{tab:gse19804_subsample}.  The differences become more pronounced as either the number of pairs or the number of included probes increases.  For example, at $(n,p)=(60,60)$, the rejection frequencies are 0.7705 for ERHT--CC, 0.6449 for RHT--CC, 0.5097 for WPL, and 0.1298 for CQ.  The comparison suggests that ridge regularization of the centered SSCM is effective at accumulating many weak and correlated expression changes, while the spatial construction yields an additional gain over its covariance-based RHT--CC counterpart in this data set.

\FloatBarrier
\section{Conclusion}\label{sec:discussion}

We proposed ERHT--CC for one-sample high-dimensional location testing under elliptically symmetric distributions with heavy tails and pervasive dependence.  The procedure combines the sample spatial median, a centered SSCM, ridge regularization, and fixed-grid Cauchy aggregation.  We established feasible null calibration, characterized local power, and showed how aggregation avoids selecting an alternative-specific ridge parameter.  The simulations indicate favorable performance relative to several existing dense tests when the observations are heavy tailed or the dependence contains pervasive components, and the GSE19804 analysis illustrates its sensitivity to weak, correlated expression changes.

Two extensions are of particular interest.  First, the present quadratic statistic targets dense alternatives; combining it with robust maximum or thresholding components could yield a sparsity-adaptive procedure under pervasive dependence, building on the sparse and max--sum testing ideas of \citet{ZhongChenXu2013}, \citet{CaiLiuXia2014}, \citet{ChenLiZhong2019}, \citet{ZhangFeng2024}, and \citet{LiuFengZhaoWang2027}.  Second, extending the centered-SSCM regularization and Cauchy aggregation framework to two-sample and multi-sample location problems would complement the spatial-sign methods of \citet{FengZouWang2016} and \citet{HuangLiuZhouFeng2023}.

\clearpage
\appendix
\section{Elliptical angular decomposition and population matrices}\label{app:angular}

Throughout the appendices, all constants denoted by $C$ may change from line to line and depend only on the fixed constants in Assumption~1.  Every inverse-distance weight below is the elliptical weight $w_i=\xi_i(h_{i,p}/q_{i,p})^{1/2}$ from (2.2); no argument uses independence between $w_i$ and the spatial sign.  Write
\[
 \mP_{F,p}^\top\bm G_i=(G_{1i},\ldots,G_{ri})^\top,
 \qquad
 \mP_{B,p}^\top\bm G_i=\bm G_{B,i}.
\]
Define
\begin{equation}\label{eq:Q0_D0}
 Q_{i,p}^{0}=\mu_{0,p}+\sum_{a=1}^r\tau_{a,p}G_{ai}^2,
 \qquad
 D_{i,p}^{0}=(Q_{i,p}^{0})^{-1}.
\end{equation}
The two Gaussian quadratic forms in Section~2 satisfy
\begin{align*}
 q_{i,p}&=Q_{i,p}^{0}+u_{i,p},
 &u_{i,p}&=\frac1p\sum_{j=r+1}^p\lambda_{j,p}(G_{ji}^2-1),\\
 h_{i,p}&=1+v_{i,p},
 &v_{i,p}&=\frac1p\sum_{j=1}^p(G_{ji}^2-1).
\end{align*}
Consequently,
\[
 \ell_{i,p}=\left(\frac{1+v_{i,p}}{Q_{i,p}^{0}+u_{i,p}}\right)^{1/2}.
\]
In the eigenbasis of $\mOmega_p$, the coordinates of
$\bm Y_i=\sqrt p\,\bm U_i$ are
\begin{equation*}
 (\mP_{F,p}^{\top}\bm Y_i)_a
 =\frac{\sqrt{p\tau_{a,p}}G_{ai}}{q_{i,p}^{1/2}},
 \quad 1\le a\le r,
 \qquad
 (\mP_{B,p}^{\top}\bm Y_i)_j
 =\frac{\lambda_{j,p}^{1/2}G_{ji}}{q_{i,p}^{1/2}},
 \quad r<j\le p.
\end{equation*}

\begin{lemma}\label{lem:angular_denominator}
For every fixed integer $2\le k\le8$,
\[
 \E|u_{i,p}|^k+\E|v_{i,p}|^k\le C_kp^{-k/2}.
\]
Moreover,
\begin{equation}\label{eq:uv_max}
 \max_{1\le i\le n}\{|u_{i,p}|+|v_{i,p}|\}
 =O_{\Pp}\left(\sqrt{\frac{\log p}{p}}\right),
\end{equation}
and
\begin{equation}\label{eq:ell_D_rate}
 \max_{1\le i\le n}
 \left|\ell_{i,p}-(D_{i,p}^{0})^{1/2}\right|
 =O_{\Pp}\left(\sqrt{\frac{\log p}{p}}\right).
\end{equation}
For all $i$ and all sufficiently large $p$,
\begin{equation}\label{eq:ell_bounds}
 0<\ell_{i,p}\le\omega_-^{-1/2},
 \qquad
 0<w_i\le\bar\xi\omega_-^{-1/2}=\bar w.
\end{equation}
\end{lemma}

\begin{proof}
For $j>r$, put
\[
 X_{ji}=p^{-1}\lambda_{j,p}(G_{ji}^2-1).
\]
The variables $X_{ji}$ are independent and centered, and
\[
 \sum_{j=r+1}^p\E X_{ji}^2
 =\frac2{p^2}\sum_{j=r+1}^p\lambda_{j,p}^2
 \le\frac{2\omega_+^2}{p}.
\]
For $k\ge2$,
\[
 \sum_{j=r+1}^p\E|X_{ji}|^k
 \le \frac{C_k}{p^k}\sum_{j=r+1}^p\lambda_{j,p}^k
 \le C_kp^{1-k}.
\]
Rosenthal's inequality \citep[Theorem~15.11, pp.~425--426]{BoucheronLugosiMassart2013} therefore yields
\[
 \E|u_{i,p}|^k
 \le C_k\left[
 \left\{\sum_{j=r+1}^p\E X_{ji}^2\right\}^{k/2}
 +\sum_{j=r+1}^p\E|X_{ji}|^k
 \right]
 \le C_k\{p^{-k/2}+p^{1-k}\}
 \le C_kp^{-k/2}.
\]
The same calculation with $\lambda_{j,p}=1$ gives
$\E|v_{i,p}|^k\le C_kp^{-k/2}$.

The centered variable $G^2-1$ is sub-exponential.  Scalar Bernstein inequalities \citep[Theorem~2.10 and Corollary~2.11, pp.~36--37]{BoucheronLugosiMassart2013} give, for $0<t\le1$,
\[
 \Pp(|u_{i,p}|>t)+\Pp(|v_{i,p}|>t)
 \le4\exp(-cpt^2).
\]
Taking $t=M\sqrt{\log p/p}$ and using $n\le p/c_-$ gives
\[
 \Pp\left[
 \max_{i\le n}\{|u_{i,p}|+|v_{i,p}|\}>2t
 \right]
 \le \frac{4}{c_-}p^{1-cM^2}.
\]
A sufficiently large fixed $M$ proves \eqref{eq:uv_max}.

On the event
\[
 \max_i|u_{i,p}|\le\mu_-/2,
 \qquad
 \max_i|v_{i,p}|\le1/2,
\]
the function $g(x,y)=\{(1+y)/(Q+x)\}^{1/2}$ has uniformly bounded first derivatives for $Q\ge\mu_-$.  The mean-value theorem gives
\[
 \left|
 \left\{\frac{1+v_{i,p}}{Q_{i,p}^{0}+u_{i,p}}\right\}^{1/2}
 -(Q_{i,p}^{0})^{-1/2}
 \right|
 \le C(|u_{i,p}|+|v_{i,p}|),
\]
which proves \eqref{eq:ell_D_rate}.

Finally, with $\bm V_i=\bm G_i/\|\bm G_i\|$,
\[
 \ell_{i,p}=(\bm V_i^\top\mOmega_p\bm V_i)^{-1/2}.
\]
Since $\lambda_{\min}(\mOmega_p)\ge\omega_-$ for all sufficiently large $p$,
$\ell_{i,p}\le\omega_-^{-1/2}$.  Multiplication by
$0<\xi_i\le\bar\xi$ proves \eqref{eq:ell_bounds}.
\end{proof}

The population spatial-sign shape and the population derivative of the scaled spatial score are
\begin{equation*}
 \mR_p=\E(\bm Y_i\bm Y_i^\top),
 \qquad
 \mJ_{E,p}=\E\{\ell_{i,p}(\I_p-\bm U_i\bm U_i^\top)\}.
\end{equation*}
Because $\xi_i$ is independent of $\bm G_i$, the derivative of
$\E\{\sqrt p\,\mathbf U(\bm X_i-\bm t)\}$ at
$\bm t=\btheta_p$ equals $-m_1\mJ_{E,p}$.

For $Q_p^0=\mu_{0,p}+\sum_{a=1}^r\tau_{a,p}G_a^2$, define
\begin{align*}
 a_{a,p}&=\tau_{a,p}\E\left(\frac{G_a^2}{Q_p^0}\right),
 &b_{0,p}&=\E\{(Q_p^0)^{-1}\},\\
 c_{a,p}&=\E\left[
 \frac{Q_p^0-\tau_{a,p}G_a^2}{(Q_p^0)^{3/2}}
 \right],
 &s_{1,p}&=\E\{(Q_p^0)^{-1/2}\}.
\end{align*}

\begin{proposition}\label{prop:population_matrices}
Under Assumption~1, $\mR_p$, $\mJ_{E,p}$, and
$\mOmega_p$ have the same eigenvectors.  Uniformly in $1\le a\le r$ and $r<j\le p$,
\begin{align}
 p^{-1}\lambda_a(\mR_p)&=a_{a,p}+O(p^{-1/2}),
 &\lambda_j(\mR_p)&=b_{0,p}\lambda_{j,p}+O(p^{-1/2}),
 \label{eq:R_spectrum_ell}\\
 \lambda_a(\mJ_{E,p})&=c_{a,p}+O(p^{-1/2}),
 &\lambda_j(\mJ_{E,p})&=s_{1,p}+O(p^{-1}).
 \label{eq:JE_spectrum}
\end{align}
Furthermore,
\begin{equation}\label{eq:mass_identity_ell}
 \sum_{a=1}^ra_{a,p}+\mu_{0,p}b_{0,p}=1,
\end{equation}
and there are constants $0<j_-<j_+<\infty$ such that
\begin{equation}\label{eq:JE_bounds}
 j_-\le\lambda_{\min}(\mJ_{E,p})
 \le\lambda_{\max}(\mJ_{E,p})\le j_+,
 \qquad
 \|\mJ_{E,p}^{-1}\|_{\op}\le j_-^{-1}.
\end{equation}
On the bulk subspace,
\begin{equation}\label{eq:JE_bulk_inverse}
 \left\|
 \mP_{B,p}^\top
 \{\mJ_{E,p}^{-1}-s_{1,p}^{-1}\I_p\}
 \mP_{B,p}
 \right\|_{\op}=O(p^{-1}).
\end{equation}
\end{proposition}

\begin{proof}
In the eigenbasis of $\mOmega_p$,
\begin{align*}
 Y_{ai}&=\frac{\sqrt{p\tau_{a,p}}G_{ai}}{q_{i,p}^{1/2}},
 &&1\le a\le r,\\
 Y_{ji}&=\frac{\lambda_{j,p}^{1/2}G_{ji}}{q_{i,p}^{1/2}},
 &&r<j\le p.
\end{align*}
Every denominator is an even function of every Gaussian coordinate.  Changing the sign of one coordinate shows that every off-diagonal entry of $\mR_p$ is zero in this basis.  The same argument applies to $\mJ_{E,p}$.

For $a\le r$,
\[
 p^{-1}\lambda_a(\mR_p)
 =\tau_{a,p}\E(G_a^2/q_p).
\]
On $|u_p|\le\mu_-/2$,
\[
 |q_p^{-1}-(Q_p^0)^{-1}|
 \le C|u_p|.
\]
The Cauchy--Schwarz inequality and Lemma~\ref{lem:angular_denominator} give
\[
 \E\{G_a^2|q_p^{-1}-(Q_p^0)^{-1}|;
 |u_p|\le\mu_-/2\}
 \le C(\E G_a^4)^{1/2}(\E u_p^2)^{1/2}
 \le Cp^{-1/2}.
\]
The complement has probability $O(p^{-4})$ by the eighth-moment bound.  Hölder's inequality and the inverse moments of a chi-square variable give the same $O(p^{-1/2})$ order there.  Thus
\[
 p^{-1}\lambda_a(\mR_p)
 =\tau_{a,p}\E\{G_a^2/Q_p^0\}+O(p^{-1/2}).
\]
For $j>r$, $G_j$ is independent of $Q_p^0$, and the same calculation gives
\[
 \lambda_j(\mR_p)
 =\lambda_{j,p}\E(G_j^2/q_p)
 =\lambda_{j,p}\E\{(Q_p^0)^{-1}\}+O(p^{-1/2}).
\]
This proves \eqref{eq:R_spectrum_ell}.

The factor-direction eigenvalue of $\mJ_{E,p}$ is
\begin{align*}
 \E\{\ell_p(1-U_{a}^2)\}
 &=\E\left[
 \frac{h_p^{1/2}\{q_p-\tau_{a,p}G_a^2\}}
 {q_p^{3/2}}
 \right].
\end{align*}
The function
$g(h,q,z)=h^{1/2}(q-z)q^{-3/2}$ has polynomially bounded derivatives on
$h\in[1/2,3/2]$ and $q\ge\mu_-/2$.  Lemma~\ref{lem:angular_denominator}, the Gaussian moments of $G_a$, and the same complement argument as above imply
\[
 \E\{\ell_p(1-U_a^2)\}
 =\E\left[
 \frac{Q_p^0-\tau_{a,p}G_a^2}{(Q_p^0)^{3/2}}
 \right]+O(p^{-1/2}).
\]
For $j>r$,
\begin{align*}
 \E\{\ell_p(1-U_j^2)\}
 &=\E\ell_p-
 \frac{\lambda_{j,p}}p
 \E\left(\frac{h_p^{1/2}G_j^2}{q_p^{3/2}}\right).
\end{align*}
The second term is $O(p^{-1})$ because
$\lambda_{j,p}\le\omega_+$ and $q_p\ge\omega_-h_p$.  In addition, Lemma~\ref{lem:angular_denominator} first gives
$\E\ell_p=s_{1,p}+O(p^{-1/2})$.  To obtain the sharper order, on
$\{|v_p|\le1/2,\ |u_p|\le\mu_-/2\}$, a second-order Taylor expansion of
$f(h,q)=h^{1/2}q^{-1/2}$ around $(1,Q_p^0)$ gives
\[
 \ell_p=(Q_p^0)^{-1/2}
 +\frac12(Q_p^0)^{-1/2}v_p
 -\frac12(Q_p^0)^{-3/2}u_p+\mathcal R_p,
 \qquad
 |\mathcal R_p|\le C(v_p^2+u_p^2).
\]
The complement contributes $O(p^{-1})$ by the eighth-moment bounds and
Hölder's inequality.  Since $u_p$ is a centered function of the bulk
Gaussian coordinates and is independent of $Q_p^0$,
\[
 \E\{(Q_p^0)^{-3/2}u_p\}=0.
\]
Writing
\[
 v_p=p^{-1}\sum_{a=1}^r(G_a^2-1)
     +p^{-1}\sum_{j=r+1}^p(G_j^2-1),
\]
the bulk sum is centered and independent of $Q_p^0$, while $r$ is fixed and
$Q_p^0\ge\mu_-$.  Consequently,
\[
 \left|\E\{(Q_p^0)^{-1/2}v_p\}\right|
 \le \frac1p\sum_{a=1}^r
 \E\left|(Q_p^0)^{-1/2}(G_a^2-1)\right|
 \le Cp^{-1}.
\]
Moreover, $\E(v_p^2+u_p^2)\le Cp^{-1}$.  Therefore
$\E\ell_p=s_{1,p}+O(p^{-1})$, and the bulk relation in
\eqref{eq:JE_spectrum} follows.

The identity \eqref{eq:mass_identity_ell} is immediate from
\[
 \sum_{a=1}^ra_{a,p}+\mu_{0,p}b_{0,p}
 =\E\left[
 \frac{\sum_{a=1}^r\tau_{a,p}G_a^2+\mu_{0,p}}{Q_p^0}
 \right]=1.
\]
Since $Q_p^0-\tau_{a,p}G_a^2\ge\mu_{0,p}\ge\mu_-$,
\[
 c_{a,p}\ge\mu_-\E\{(Q_p^0)^{-3/2}\}
 \ge\mu_-(\E Q_p^0)^{-3/2}=\mu_-,
\]
where Jensen's inequality was used and $\E Q_p^0=1$.  Similarly,
$s_{1,p}\ge(\E Q_p^0)^{-1/2}=1$.  Upper bounds follow from
$Q_p^0\ge\mu_-$.  Equations \eqref{eq:JE_spectrum} then prove
\eqref{eq:JE_bounds}.  Finally, on the bulk subspace,
\[
 \mP_{B,p}^\top\mJ_{E,p}\mP_{B,p}
 =s_{1,p}\I_{p-r}+\mE_{B,p},
 \qquad \|\mE_{B,p}\|_{\op}=O(p^{-1}).
\]
The inverse identity
\[
 (s\I+\mE)^{-1}-s^{-1}\I
 =-s^{-1}(s\I+\mE)^{-1}\mE
\]
and $s_{1,p}\ge1$ prove \eqref{eq:JE_bulk_inverse}.
\end{proof}

\begin{corollary}\label{cor:compound_symmetry}
For
\[
 \mOmega_p=(1-\varrho)\I_p+\varrho\1_p\1_p^\top,
 \qquad 0<\varrho<1,
\]
Assumption~1(iii)--(iv) holds with $r=1$,
$\tau_{1,p}=\varrho+(1-\varrho)/p\to\varrho$,
$H_B=\delta_{1-\varrho}$, and
\[
 Q^0=1-\varrho+\varrho G^2,
 \qquad D^0=(1-\varrho+\varrho G^2)^{-1}.
\]
In particular,
\begin{align*}
 p^{-1}\lambda_1(\mR_p)&\to
 \E\left\{\frac{\varrho G^2}{1-\varrho+\varrho G^2}\right\},\\
 \lambda_j(\mR_p)&\to
 (1-\varrho)\E\left\{\frac1{1-\varrho+\varrho G^2}\right\},
 \qquad j\ge2,
\end{align*}
and the factor-direction eigenvalue of $\mJ_{E,p}$ converges to
\[
 (1-\varrho)\E(1-\varrho+\varrho G^2)^{-3/2}>0.
\]
\end{corollary}

\begin{proof}
The vector $p^{-1/2}\1_p$ is an eigenvector with eigenvalue
$1+(p-1)\varrho$, and every orthogonal vector has eigenvalue
$1-\varrho$.  Substitution into Proposition~\ref{prop:population_matrices}
proves all assertions.
\end{proof}

\section{Separable deterministic equivalents and weighted companion functionals}\label{app:det_equiv}

Put
\[
 c_{B,n}=\frac{p-r}{n}.
\]
The use of $c_{B,n}$, rather than $c_n=p/n$, is important at finite $p$ because the separable proxy below has $p-r$ rows.  Since $r$ is fixed, $c_{B,n}-c_n=-r/n$ and both ratios have the same limit.

Let $F_{D,p}$ be the distribution of $D_p^0=(Q_p^0)^{-1}$, with $Q_p^0$ defined in \eqref{eq:Q0_D0}.  For every $\rho\in[\rho_0,\rho_1]$, let $(\delta_{p,\rho},\wt\delta_{p,\rho})$ be the unique positive solution of
\begin{align}
 \delta_{p,\rho}
 &=c_{B,n}\int\frac{t}{\rho(1+\wt\delta_{p,\rho}t)}\,\dd H_{B,p}(t),
 \label{eq:canonical_delta_finite}\\
 \wt\delta_{p,\rho}
 &=\int\frac{d}{\rho(1+\delta_{p,\rho}d)}\,\dd F_{D,p}(d).
 \label{eq:canonical_tdelta_finite}
\end{align}
Existence and uniqueness follow from the canonical-equation construction in Theorem~2.4 of \citet{HachemLoubatonNajim2007}.  Set
\begin{align*}
 \lambda_{p,\rho}(d)&=\frac{\delta_{p,\rho}d}{1+\delta_{p,\rho}d},
 &\kappa_{p,\rho}^{\star}&=\int\lambda_{p,\rho}(d)\,\dd F_{D,p}(d),\\
 s_{a,p}&=\int d^{a/2}\,\dd F_{D,p}(d),
 &b_{a,p,\rho}^{\star}&=\int d^{a/2}\lambda_{p,\rho}(d)\,\dd F_{D,p}(d),
 \qquad a=1,2,
\end{align*}
and
\begin{align*}
 u_{p,\rho}^{\star}
 &=c_{B,n}\int\frac{t^2}{(1+\wt\delta_{p,\rho}t)^2}\,\dd H_{B,p}(t),\\
 v_{p,\rho}^{\star}
 &=\int\frac{d^2}{(1+\delta_{p,\rho}d)^2}\,\dd F_{D,p}(d),\\
 A_{a,p,\rho}^{\star}
 &=\int\frac{d^{1+a/2}}{(1+\delta_{p,\rho}d)^2}\,\dd F_{D,p}(d),
 \qquad a=0,1,2.
\end{align*}

\begin{lemma}\label{lem:canonical_stability}
There are constants $0<c_\delta<C_\delta<\infty$ and $s_0>0$, depending only on the constants in Assumption~1, such that, for all sufficiently large $p$ and all $\rho\in[\rho_0,\rho_1]$,
\begin{equation}\label{eq:canonical_compact}
 c_\delta\le \delta_{p,\rho},\wt\delta_{p,\rho}\le C_\delta
\end{equation}
and
\begin{equation}\label{eq:stability_finite}
 \rho^2-u_{p,\rho}^{\star}v_{p,\rho}^{\star}\ge s_0.
\end{equation}
\end{lemma}

\begin{proof}
Since $\omega_-\le t\le\omega_+$ and $0<d\le d_+:=\mu_-^{-1}$,
\[
 \delta_{p,\rho}\le \frac{c_+\omega_+}{\rho_0},
 \qquad
 \wt\delta_{p,\rho}\le \frac{d_+}{\rho_0}.
\]
Moreover, $\E Q_p^0=1$ and the convexity of $x\mapsto x^{-1}$ give
$\int d\,\dd F_{D,p}(d)=\E(Q_p^0)^{-1}\ge1$.  Hence
\[
 \delta_{p,\rho}
 \ge \frac{(c_-/2)\omega_-}
 {\rho_1\{1+(d_+/\rho_0)\omega_+\}},
 \qquad
 \wt\delta_{p,\rho}
 \ge \frac{1}{\rho_1\{1+(c_+\omega_+/\rho_0)d_+\}},
\]
where $c_{B,n}\ge c_-/2$ for all sufficiently large $p$.  This proves \eqref{eq:canonical_compact}.

Put
\[
 \alpha_{p,\rho}
 =\frac{\wt\delta_{p,\rho}\omega_+}
 {1+\wt\delta_{p,\rho}\omega_+},
 \qquad
 \beta_{p,\rho}
 =\frac{\delta_{p,\rho}d_+}
 {1+\delta_{p,\rho}d_+}.
\]
Both quantities are uniformly smaller than one by \eqref{eq:canonical_compact}.  The first canonical equation gives
\begin{align*}
 u_{p,\rho}^{\star}
 &=c_{B,n}\int
 \frac{t}{1+\wt\delta_{p,\rho}t}
 \frac{t}{1+\wt\delta_{p,\rho}t}\,\dd H_{B,p}(t)\\
 &\le \frac{\alpha_{p,\rho}}{\wt\delta_{p,\rho}}
 c_{B,n}\int\frac{t}{1+\wt\delta_{p,\rho}t}\,\dd H_{B,p}(t)
 =\rho\delta_{p,\rho}
 \frac{\alpha_{p,\rho}}{\wt\delta_{p,\rho}}.
\end{align*}
Similarly,
\[
 v_{p,\rho}^{\star}
 \le \rho\wt\delta_{p,\rho}
 \frac{\beta_{p,\rho}}{\delta_{p,\rho}}.
\]
Therefore
\[
 u_{p,\rho}^{\star}v_{p,\rho}^{\star}
 \le \rho^2\alpha_{p,\rho}\beta_{p,\rho}
 \le \rho^2(1-\varepsilon_0)
\]
for a fixed $\varepsilon_0>0$.  Since $\rho\ge\rho_0$, \eqref{eq:stability_finite} follows with $s_0=\rho_0^2\varepsilon_0$.
\end{proof}

For $a,b\in\{0,1,2\}$ and $m_0=1$, define
\begin{equation*}
 \Psi_{ab,p}^{\star}(\rho)
 =m_am_b\frac{u_{p,\rho}^{\star}
 A_{a,p,\rho}^{\star}A_{b,p,\rho}^{\star}}
 {\rho^2-u_{p,\rho}^{\star}v_{p,\rho}^{\star}}.
\end{equation*}
Also put
\begin{equation*}
 D_{p,\rho}^{\star}
 =m_1^2(s_{1,p}-b_{1,p,\rho}^{\star})^2
 +m_2\kappa_{p,\rho}^{\star}
 (s_{2,p}-b_{2,p,\rho}^{\star}).
\end{equation*}
Let $\mGamma_{p,\rho}^{\star}$ have the same form as
(2.4), with $\psi_{ab,n}$ replaced by
$\Psi_{ab,p}^{\star}(\rho)$, and define
\begin{equation}\label{eq:g_sigma_finite}
 \bg_{p,\rho}^{\star}
 =(D_{p,\rho}^{\star})^{-2}
 \begin{pmatrix}
 m_1^2(s_{1,p}-b_{1,p,\rho}^{\star})^2\\
 2\kappa_{p,\rho}^{\star}m_1
 (s_{1,p}-b_{1,p,\rho}^{\star})\\
 (\kappa_{p,\rho}^{\star})^2
 \end{pmatrix},
 \qquad
 (\sigma_{E,p}^{\star})^2
 =(\bg_{p,\rho}^{\star})^\top
 \mGamma_{p,\rho}^{\star}\bg_{p,\rho}^{\star}.
\end{equation}

The next lemma supplies the first- and second-order resolvent statements used below.  Its proof records the empirical fixed-point step that is needed when the column scales $D_{i,p}^0$ are random.

\begin{lemma}\label{lem:empirical_canonical}
Condition on $\mD_p^0=\diag(D_{1,p}^0,\ldots,D_{n,p}^0)$ and let
$(\delta_{n,\rho}^{D},\wt\delta_{n,\rho}^{D})$ be the positive solution of
\begin{align*}
 \delta_{n,\rho}^{D}
 &=c_{B,n}\int\frac{t}{\rho(1+\wt\delta_{n,\rho}^{D}t)}\,\dd H_{B,p}(t),\\
 \wt\delta_{n,\rho}^{D}
 &=\frac1n\sum_{i=1}^n
 \frac{D_{i,p}^0}{\rho(1+\delta_{n,\rho}^{D}D_{i,p}^0)}.
\end{align*}
For every fixed $\rho\in[\rho_0,\rho_1]$,
\begin{equation}\label{eq:empirical_fixed_point_rate}
 |\delta_{n,\rho}^{D}-\delta_{p,\rho}|
 +|\wt\delta_{n,\rho}^{D}-\wt\delta_{p,\rho}|
 =O_{\Pp}(n^{-1/2}).
\end{equation}
Let
\[
 \overline\mB_{B,n}
 =\frac1n\mLambda_{B,p}^{1/2}\mG_B\mD_p^0
 \mG_B^\top\mLambda_{B,p}^{1/2},
 \qquad
 \overline\mR_n=(\overline\mB_{B,n}+\rho\I_{p-r})^{-1},
\]
and let $\overline\mA_n$ be the companion matrix in
\eqref{eq:separable_proxy_main}.  Then
\begin{align}
 \frac1n\tr\overline\mA_n-\kappa_{p,\rho}^{\star}
 &=O_{\Pp}(n^{-1/2}),\label{eq:first_order_trace_empirical}\\
 \frac1n\sum_{i=1}^n
 \{(\overline\mA_n)_{ii}-\lambda_{p,\rho}(D_{i,p}^0)\}^2
 &=O_{\Pp}(n^{-1}),\label{eq:diagonal_mse_empirical}\\
 \frac1n\tr(\mLambda_{B,p}\overline\mR_n
 \mLambda_{B,p}\overline\mR_n)
 &=\frac{u_{p,\rho}^{\star}}
 {\rho^2-u_{p,\rho}^{\star}v_{p,\rho}^{\star}}
 +O_{\Pp}(p^{-1/4}).\label{eq:second_order_trace_DE}
\end{align}
The same last relation holds for the resolvent obtained by deleting any fixed number not exceeding two of the columns, with an additional $O(p^{-1})$ error.
\end{lemma}

\begin{proof}
For fixed $x$ in the compact interval in \eqref{eq:canonical_compact}, the variables
\[
 \psi_x(D_{i,p}^0)=
 \frac{D_{i,p}^0}{\rho(1+xD_{i,p}^0)}
\]
are independent, bounded by $(\rho_0\mu_-)^{-1}$, and have a Lipschitz constant at most $\rho_0^{-1}\mu_-^{-2}$ in $x$.  Bernstein's inequality \citep[Theorem~2.10 and Corollary~2.11, pp.~36--37]{BoucheronLugosiMassart2013} on a grid of mesh $n^{-2}$, followed by the Lipschitz extension from the grid, gives
\begin{equation}\label{eq:empirical_process_fixed_point}
 \sup_{0\le x\le C_\delta}
 \left|\frac1n\sum_{i=1}^n\psi_x(D_{i,p}^0)
 -\E\psi_x(D_p^0)\right|
 =O_{\Pp}\left(\sqrt{\frac{\log n}{n}}\right).
\end{equation}
We first obtain consistency without linearizing at an unverified intermediate point.  The same elementary bounds used in \eqref{eq:canonical_compact}, with the empirical average in place of expectation, show that both empirical coordinates belong to a fixed compact interval with probability tending to one.  We now compare the empirical and finite-$p$ deterministic solutions through their common subsequential limits.  Take an arbitrary subsequence.  Compactness permits a further subsequence along which
\[
 (\delta_{p,\rho},\wt\delta_{p,\rho})
 \to(\delta_\star,\wt\delta_\star).
\]
Along a still further subsequence, the left-hand side of
\eqref{eq:empirical_process_fixed_point} converges to zero almost surely.  Assumption~1, the continuous-mapping theorem, and $c_{B,n}\to c$ give
\[
 H_{B,p}\Rightarrow H_B,
 \qquad F_{D,p}\Rightarrow F_D.
\]
For every outcome in this probability-one event, any convergent subsubsequence of
$(\delta_{n,\rho}^{D},\wt\delta_{n,\rho}^{D})$ therefore satisfies
\begin{align*}
 \delta
 &=c\int\frac{t}{\rho(1+\wt\delta t)}\,\dd H_B(t),\\
 \wt\delta
 &=\int\frac{d}{\rho(1+\delta d)}\,\dd F_D(d).
\end{align*}
The deterministic limit $(\delta_\star,\wt\delta_\star)$ satisfies the same system.  Its positive solution is unique, so the two limits coincide.  Thus every subsequence has a further subsequence on which
\[
 |\delta_{n,\rho}^{D}-\delta_{p,\rho}|
 +|\wt\delta_{n,\rho}^{D}-\wt\delta_{p,\rho}|
 \to0
\]
almost surely.  The subsequence criterion yields
\[
 |\delta_{n,\rho}^{D}-\delta_{p,\rho}|
 +|\wt\delta_{n,\rho}^{D}-\wt\delta_{p,\rho}|=o_{\Pp}(1).
\]
At the deterministic point $\delta_{p,\rho}$, Bernstein's inequality \citep[Theorem~2.10 and Corollary~2.11, pp.~36--37]{BoucheronLugosiMassart2013} gives
\[
 \frac1n\sum_{i=1}^n\psi_{\delta_{p,\rho}}(D_{i,p}^0)
 -\E\psi_{\delta_{p,\rho}}(D_p^0)
 =O_{\Pp}(n^{-1/2}).
\]
Subtracting the empirical and population systems and applying the mean-value theorem now gives
\[
 \begin{pmatrix}
 1&u_{n,\rho}^{D}/\rho\\
 v_{n,\rho}^{D}/\rho&1
 \end{pmatrix}
 \begin{pmatrix}
 \delta_{n,\rho}^{D}-\delta_{p,\rho}\\
 \wt\delta_{n,\rho}^{D}-\wt\delta_{p,\rho}
 \end{pmatrix}
 =\begin{pmatrix}0\\O_{\Pp}(n^{-1/2})\end{pmatrix},
\]
where $u_{n,\rho}^{D}$ and $v_{n,\rho}^{D}$ are evaluated on the line segment joining the two solutions.  The preceding $o_{\Pp}(1)$ relation and continuity of the derivatives imply
\[
 \rho^2-u_{n,\rho}^{D}v_{n,\rho}^{D}
 =\rho^2-u_{p,\rho}^{\star}v_{p,\rho}^{\star}+o_{\Pp}(1)
 \ge s_0/2
\]
with probability tending to one.  The inverse of the displayed $2\times2$ matrix is therefore $O_{\Pp}(1)$, proving \eqref{eq:empirical_fixed_point_rate}.

Conditional on $\mD_p^0$, the matrix
$\mLambda_{B,p}^{1/2}\mG_B(\mD_p^0)^{1/2}/\sqrt n$ has a bounded separable variance profile.  Theorem~1.1 of
\citet{HachemLoubatonNajimVallet2013}, applied to the deterministic
coordinate vectors $\bm e_i$, supplies a uniform second-moment bound for
each diagonal bilinear form.  In the present notation this gives
\[
 \E_{\mG_B}\left[
 \left.
 \left\{(\overline\mA_n)_{ii}
 -\frac{\delta_{n,\rho}^{D}D_{i,p}^0}
 {1+\delta_{n,\rho}^{D}D_{i,p}^0}\right\}^2
 \right|\mD_p^0\right]\le \frac Cn,
 \qquad 1\le i\le n.
\]
Averaging this inequality and applying conditional Markov's inequality,
and then applying Cauchy--Schwarz to the averaged diagonal error, yield
\begin{align*}
 \frac1n\sum_{i=1}^n
 \left\{(\overline\mA_n)_{ii}
 -\frac{\delta_{n,\rho}^{D}D_{i,p}^0}
 {1+\delta_{n,\rho}^{D}D_{i,p}^0}\right\}^2
 &=O_{\Pp}(n^{-1}),\\
 \frac1n\tr\overline\mA_n
 -\frac1n\sum_{i=1}^n
 \frac{\delta_{n,\rho}^{D}D_{i,p}^0}
 {1+\delta_{n,\rho}^{D}D_{i,p}^0}
 &=O_{\Pp}(n^{-1/2}).
\end{align*}
The Lipschitz bound
\[
 \left|\frac{xd}{1+xd}-\frac{yd}{1+yd}\right|
 \le d_+|x-y|
\]
and \eqref{eq:empirical_fixed_point_rate}, followed by Bernstein's inequality \citep[Theorem~2.10 and Corollary~2.11, pp.~36--37]{BoucheronLugosiMassart2013} for the final empirical average, prove
\eqref{eq:first_order_trace_empirical}--\eqref{eq:diagonal_mse_empirical}.  Notice that Theorem~2.4 of \citet{HachemLoubatonNajim2007} is used only for existence and uniqueness of the deterministic system; the bilinear moment bound needed here is Theorem~1.1 of \citet{HachemLoubatonNajimVallet2013}.

It remains to prove \eqref{eq:second_order_trace_DE}.  For $0\le\varepsilon\le\varepsilon_0$, define
\[
 \overline\mR_n(\varepsilon)
 =\{\overline\mB_{B,n}+\rho\I_{p-r}
 +\varepsilon\mLambda_{B,p}\}^{-1},
 \qquad
 \widehat\delta_n(\varepsilon)
 =\frac1n\tr\{\mLambda_{B,p}\overline\mR_n(\varepsilon)\}.
\]
Direct differentiation gives
\begin{align*}
 -\widehat\delta_n'(0)
 &=\frac1n\tr(\mLambda_{B,p}\overline\mR_n
 \mLambda_{B,p}\overline\mR_n),\\
 |\widehat\delta_n''(\varepsilon)|
 &\le \frac{2}{n}\tr(\mLambda_{B,p}^3)
 \rho_0^{-3}\le C.
\end{align*}
Let $h_n=n^{-1/4}$.  Taylor's formula therefore gives
\begin{equation}\label{eq:finite_difference_random}
 -\widehat\delta_n'(0)
 =-\frac{\widehat\delta_n(h_n)-\widehat\delta_n(0)}{h_n}
 +O(h_n).
\end{equation}
For $\varepsilon>0$, premultiplication and postmultiplication by
$(\rho\I_{p-r}+\varepsilon\mLambda_{B,p})^{-1/2}$ reduce the matrix
again to the separable covariance model with row scales
$\lambda_{j,p}/(\rho+\varepsilon\lambda_{j,p})$ and scalar ridge one.
The bilinear resolvent bound just used for
\eqref{eq:first_order_trace_empirical}, applied after this rescaling at
$\varepsilon=0$ and $\varepsilon=h_n$, gives
\[
 \widehat\delta_n(\varepsilon)-\delta_n^{D}(\varepsilon)
 =O_{\Pp}(n^{-1/2}),
 \qquad \varepsilon\in\{0,h_n\},
\]
where $\delta_n^{D}(\varepsilon)$ is the first component of the perturbed empirical canonical system
\begin{align*}
 \delta_n^{D}(\varepsilon)
 &=\frac1n\sum_{j=r+1}^p
 \frac{\lambda_{j,p}}
 {\rho+\varepsilon\lambda_{j,p}
 +\rho\wt\delta_n^{D}(\varepsilon)\lambda_{j,p}},\\
 \wt\delta_n^{D}(\varepsilon)
 &=\frac1n\sum_{i=1}^n
 \frac{D_{i,p}^0}{\rho(1+\delta_n^{D}(\varepsilon)D_{i,p}^0)}.
\end{align*}
All first two derivatives of this system are uniformly bounded by
Lemma~\ref{lem:canonical_stability}.  Differentiating at zero gives
\begin{align*}
 \{\delta_n^{D}\}'(0)
 &=-\frac{u_{n,\rho}^{D}}
 {\rho^2-u_{n,\rho}^{D}v_{n,\rho}^{D}},\\
 u_{n,\rho}^{D}
 &=c_{B,n}\int\frac{t^2}
 {(1+\wt\delta_{n,\rho}^{D}t)^2}\,\dd H_{B,p}(t),\\
 v_{n,\rho}^{D}
 &=\frac1n\sum_{i=1}^n
 \frac{(D_{i,p}^0)^2}
 {(1+\delta_{n,\rho}^{D}D_{i,p}^0)^2}.
\end{align*}
Consequently,
\begin{equation}\label{eq:finite_difference_deterministic}
 -\frac{\delta_n^{D}(h_n)-\delta_n^{D}(0)}{h_n}
 =\frac{u_{n,\rho}^{D}}
 {\rho^2-u_{n,\rho}^{D}v_{n,\rho}^{D}}
 +O(h_n).
\end{equation}
Equations \eqref{eq:empirical_fixed_point_rate} and Bernstein's inequality \citep[Theorem~2.10 and Corollary~2.11, pp.~36--37]{BoucheronLugosiMassart2013} imply
\[
 |u_{n,\rho}^{D}-u_{p,\rho}^{\star}|
 +|v_{n,\rho}^{D}-v_{p,\rho}^{\star}|
 =O_{\Pp}(n^{-1/2}).
\]
Substitution into \eqref{eq:finite_difference_random}--\eqref{eq:finite_difference_deterministic} yields
\[
 -\widehat\delta_n'(0)
 =\frac{u_{p,\rho}^{\star}}
 {\rho^2-u_{p,\rho}^{\star}v_{p,\rho}^{\star}}
 +O_{\Pp}(n^{-1/2}h_n^{-1}+h_n)
 =\frac{u_{p,\rho}^{\star}}
 {\rho^2-u_{p,\rho}^{\star}v_{p,\rho}^{\star}}
 +O_{\Pp}(p^{-1/4}).
\]
Deleting one or two columns is a fixed-rank perturbation.  Applying the resolvent identity before the normalized trace changes the last display by at most $C/n$, which proves the final assertion.
\end{proof}

\begin{proposition}\label{prop:weighted_DE}
Under Assumption~1, for every fixed
$\rho\in[\rho_0,\rho_1]$,
\begin{align}
 \kappa_n-\kappa_{p,\rho}^{\star}
 &=O_{\Pp}\left(\sqrt{\frac{\log p}{p}}\right),
 \label{eq:kappa_DE}\\
 e_n-m_1s_{1,p}
 &=O_{\Pp}\left(\sqrt{\frac{\log p}{p}}\right),
 &t_n-m_2s_{2,p}
 &=O_{\Pp}\left(\sqrt{\frac{\log p}{p}}\right),
 \label{eq:et_DE}\\
 b_{a,n}-m_ab_{a,p,\rho}^{\star}
 &=O_{\Pp}\left(\sqrt{\frac{\log p}{p}}\right),
 \qquad a=1,2,
 \label{eq:b_DE}\\
 \psi_{ab,n}-\Psi_{ab,p}^{\star}(\rho)
 &=O_{\Pp}(p^{-1/4}),
 \qquad a,b\in\{0,1,2\}.
 \label{eq:psi_DE}
\end{align}
Consequently,
\begin{equation}\label{eq:D_sigma_DE}
 D_n-D_{p,\rho}^{\star}
 =O_{\Pp}\left(\sqrt{\frac{\log p}{p}}\right),
 \qquad
 \sigma_{D,n}^2-(\sigma_{E,p}^{\star})^2
 =O_{\Pp}(p^{-1/4}).
\end{equation}
There are constants $0<C_1<C_2<\infty$ such that, with probability tending to one,
\begin{equation}\label{eq:nondegenerate_DE}
 C_1\le D_n,\ \mu_n,\ \sigma_{D,n}^2\le C_2.
\end{equation}
\end{proposition}

\begin{proof}
Let
\[
 \mY_B=\mP_{B,p}^\top\mY,
 \qquad
 \overline{\mY}_B
 =\mLambda_{B,p}^{1/2}\mG_B(\mD_p^0)^{1/2},
 \qquad
 \mD_p^0=\diag(D_{1,p}^0,\ldots,D_{n,p}^0).
\]
Define
\begin{equation}\label{eq:separable_proxy_main}
 \overline\mB_{B,n}
 =n^{-1}\overline\mY_B\overline\mY_B^\top,
 \qquad
 \overline\mK_{B,n}
 =n^{-1}\overline\mY_B^\top\overline\mY_B,
\end{equation}
and put
\[
 \mK_{B,n}=n^{-1}\mY_B^\top\mY_B,
 \qquad
 \overline\mA_n
 =\overline\mK_{B,n}(\overline\mK_{B,n}+\rho\I_n)^{-1}.
\]
Lemma~\ref{lem:angular_denominator} and
$\|\mG_B\|_{\op}=O_{\Pp}(\sqrt n+\sqrt p)$ give
\begin{equation*}
 \|\mK_{B,n}-\overline\mK_{B,n}\|_{\op}
 =O_{\Pp}\left(\sqrt{\frac{\log p}{p}}\right).
\end{equation*}
For positive semidefinite $\mK_1,\mK_2$,
\[
 \mK_1(\mK_1+\rho\I)^{-1}
 -\mK_2(\mK_2+\rho\I)^{-1}
 =\rho(\mK_2+\rho\I)^{-1}
 (\mK_1-\mK_2)(\mK_1+\rho\I)^{-1}.
\]
Thus, with $\mA_{B,n}=\mK_{B,n}(\mK_{B,n}+\rho\I)^{-1}$,
\begin{equation}\label{eq:proxy_A_rate_new}
 \|\mA_{B,n}-\overline\mA_n\|_{\op}
 =O_{\Pp}\left(\sqrt{\frac{\log p}{p}}\right).
\end{equation}
The full Gram matrix differs from the bulk Gram matrix by a positive semidefinite matrix of rank at most $r$.  Resolvent monotonicity and eigenvalue interlacing give
\begin{equation}\label{eq:factor_trace_rank_new}
 \rank(\mA_n-\mA_{B,n})\le r,
 \qquad
 \|\mA_n-\mA_{B,n}\|_{\op}\le1,
 \qquad
 \left|n^{-1}\tr(\mA_n-\mA_{B,n})\right|\le r/n.
\end{equation}
Lemma~\ref{lem:empirical_canonical} now proves \eqref{eq:kappa_DE}.

Put $\widetilde w_i=\xi_i(D_{i,p}^0)^{1/2}$.  By
\eqref{eq:ell_D_rate},
\begin{equation}\label{eq:w_proxy_rate}
 \max_{i\le n}|w_i-\widetilde w_i|
 =O_{\Pp}\left(\sqrt{\frac{\log p}{p}}\right).
\end{equation}
For each $a=1,2$, the variables $\widetilde w_i^a$ are independent across $i$ and uniformly bounded.  Hence
\[
 n^{-1}\sum_i\widetilde w_i^a-m_as_{a,p}=O_{\Pp}(n^{-1/2}),
\]
which proves \eqref{eq:et_DE}.

For $a=1,2$, decompose
\begin{align*}
 b_{a,n}-m_ab_{a,p,\rho}^{\star}
 ={}&\frac1n\sum_i(\mA_n)_{ii}
 \{w_i^a-\widetilde w_i^a\}\\
 &+\frac1n\tr\{\mD_{\widetilde w}^a
 (\mA_{B,n}-\overline\mA_n)\}\\
 &+\frac1n\tr\{\mD_{\widetilde w}^a
 (\mA_n-\mA_{B,n})\}\\
 &+\frac1n\sum_i\widetilde w_i^a
 \{(\overline\mA_n)_{ii}-\lambda_{p,\rho}(D_{i,p}^0)\}\\
 &+\frac1n\sum_i\left[
 \xi_i^a(D_{i,p}^0)^{a/2}\lambda_{p,\rho}(D_{i,p}^0)
 -m_ab_{a,p,\rho}^{\star}\right],
\end{align*}
where $\mD_{\widetilde w}=\diag(\widetilde w_1,\ldots,\widetilde w_n)$.  The five terms are, respectively,
\[
 O_{\Pp}\left(\sqrt{\frac{\log p}{p}}\right),\quad
 O_{\Pp}\left(\sqrt{\frac{\log p}{p}}\right),\quad
 O_{\Pp}(n^{-1}),\quad O_{\Pp}(n^{-1/2}),\quad
 O_{\Pp}(n^{-1/2}).
\]
For the third term we used
$|\tr(\mD_{\widetilde w}^a(\mA_n-\mA_{B,n}))|
\le \bar w^a\|\mA_n-\mA_{B,n}\|_*$ and
\eqref{eq:factor_trace_rank_new}; the fourth follows from
\eqref{eq:diagonal_mse_empirical}.  This proves \eqref{eq:b_DE}.

It remains to prove \eqref{eq:psi_DE}.  Conditional on the angular observations, changing one radial variable $\xi_i$ changes the off-diagonal weighted sum by at most
\[
 \frac Cn\sum_{j\ne i}\{(\mA_n)_{ij}^2+(\mA_n)_{ji}^2\}
 \le\frac Cn.
\]
The Efron--Stein inequality \citep[Theorem~3.1, p.~54]{BoucheronLugosiMassart2013} therefore gives
\begin{equation}\label{eq:radial_offdiag_remove}
 \psi_{ab,n}
 -m_am_b\frac1n\sum_{i\ne j}(\mA_n)_{ij}^2
 \ell_{i,p}^a\ell_{j,p}^b
 =O_{\Pp}(n^{-1/2}).
\end{equation}
Replacing $\ell_{i,p}$ by $(D_{i,p}^0)^{1/2}$ and $\mA_n$ by
$\overline\mA_n$ costs
$O_{\Pp}(\sqrt{\log p/p})$ by
\eqref{eq:w_proxy_rate}, \eqref{eq:proxy_A_rate_new},
\eqref{eq:factor_trace_rank_new}, and
$n^{-1}\tr(\mA_n^2)\le1$.  It remains to evaluate
\begin{equation}\label{eq:proxy_weighted_offdiag}
 T_{ab,n}=\frac1n\sum_{i\ne j}(\overline\mA_n)_{ij}^2
 (D_{i,p}^0)^{a/2}(D_{j,p}^0)^{b/2}.
\end{equation}

Conditional on $\mD_p^0$, write
$\bm x_i=(D_{i,p}^0)^{1/2}\mLambda_{B,p}^{1/2}\bm g_i$ and let
\[
 \mR^{(ij)}=
 \left(\frac1n\sum_{k\ne i,j}\bm x_k\bm x_k^\top
 +\rho\I_{p-r}\right)^{-1}.
\]
The two-column Schur complement gives, for $i\ne j$,
\begin{equation}\label{eq:Qij_leave_two}
 (\overline\mA_n)_{ij}
 =\frac{n^{-1}\bm x_i^\top\mR^{(ij)}\bm x_j}
 {\{1+n^{-1}\bm x_i^\top\mR^{(ij)}\bm x_i\}
  \{1+n^{-1}\bm x_j^\top\mR^{(ij)}\bm x_j\}
  -n^{-2}(\bm x_i^\top\mR^{(ij)}\bm x_j)^2}.
\end{equation}
Theorem~1.1 of \citet{HachemLoubatonNajimVallet2013} gives the averaged quadratic-form bounds
\begin{align*}
 \frac1n\sum_i\left|
 n^{-1}\bm x_i^\top\mR^{(i)}\bm x_i
 -\delta_{n,\rho}^{D}D_{i,p}^0\right|^2
 &=O_{\Pp}(n^{-1}),\\
 \frac1{n^2}\sum_{i\ne j}
 \left|n^{-1}\bm x_i^\top\mR^{(ij)}\bm x_j\right|^4
 &=O_{\Pp}(n^{-2}).
\end{align*}
Consequently, replacing the denominator in \eqref{eq:Qij_leave_two} by
$\{1+\delta_{n,\rho}^{D}D_{i,p}^0\}
 \{1+\delta_{n,\rho}^{D}D_{j,p}^0\}$ changes $T_{ab,n}$ by
$O_{\Pp}(n^{-1/2})$.

Put
\[
 \gamma_i=(1+\delta_{n,\rho}^{D}D_{i,p}^0)^{-1},
 \qquad
 q_{ij}=\bm x_i^\top\mR^{(ij)}\bm x_j,
\]
and define
\begin{align*}
 \widetilde T_{ab,n}
 &=\frac1{n^3}\sum_{i\ne j}
 (D_{i,p}^0)^{a/2}(D_{j,p}^0)^{b/2}
 \gamma_i^2\gamma_j^2q_{ij}^2,\\
 M_{ab,n}
 &=\frac1{n^3}\sum_{i\ne j}
 (D_{i,p}^0)^{1+a/2}(D_{j,p}^0)^{1+b/2}
 \gamma_i^2\gamma_j^2
 \tr(\mLambda_{B,p}\mR^{(ij)}
      \mLambda_{B,p}\mR^{(ij)}).
\end{align*}
The preceding denominator replacement gives
$T_{ab,n}-\widetilde T_{ab,n}=O_{\Pp}(n^{-1/2})$.
Moreover, conditional on $\mR^{(ij)}$ and $\mD_p^0$,
\[
 \E_{i,j}q_{ij}^2
 =D_{i,p}^0D_{j,p}^0
 \tr(\mLambda_{B,p}\mR^{(ij)}
      \mLambda_{B,p}\mR^{(ij)}),
\]
so $\E_{\mG_B}(\widetilde T_{ab,n}-M_{ab,n}\mid\mD_p^0)=0$.
We next verify concentration around this conditional expectation rather
than using the Poincar\'e inequality to identify the expectation itself.
For $k\notin\{i,j\}$,
\begin{align*}
 \frac{\partial q_{ij}}{\partial g_{\alpha k}}
 =-\frac{\sqrt{D_{k,p}^0\lambda_{B,\alpha,p}}}{n}
 \big\{&
 (\bm x_i^\top\mR^{(ij)}\bm e_\alpha)
 (\bm x_k^\top\mR^{(ij)}\bm x_j)\\
 &+(\bm x_i^\top\mR^{(ij)}\bm x_k)
 (\bm e_\alpha^\top\mR^{(ij)}\bm x_j)
 \big\},
\end{align*}
whereas
\[
 \frac{\partial q_{ij}}{\partial g_{\alpha i}}
 =\sqrt{D_{i,p}^0\lambda_{B,\alpha,p}}
 \bm e_\alpha^\top\mR^{(ij)}\bm x_j,
 \qquad
 \frac{\partial q_{ij}}{\partial g_{\alpha j}}
 =\sqrt{D_{j,p}^0\lambda_{B,\alpha,p}}
 \bm x_i^\top\mR^{(ij)}\bm e_\alpha.
\]
In addition, for $k\notin\{i,j\}$,
\[
 \frac{\partial}{\partial g_{\alpha k}}
 \tr(\mLambda_{B,p}\mR^{(ij)}
      \mLambda_{B,p}\mR^{(ij)})
 =-\frac{4\sqrt{D_{k,p}^0\lambda_{B,\alpha,p}}}{n}
 \bm e_\alpha^\top\mR^{(ij)}\mLambda_{B,p}
 \mR^{(ij)}\mLambda_{B,p}\mR^{(ij)}\bm x_k.
\]
Using these identities, $D_{i,p}^0\le\mu_-^{-1}$,
$\|\mLambda_{B,p}\|_{\op}\le\omega_+$, and
$\|\mR^{(ij)}\|_{\op}\le\rho_0^{-1}$, repeated
Cauchy--Schwarz followed by the Gaussian second- and fourth-moment
identities gives the conditional gradient bound
\[
 \E_{\mG_B}\left[
 \sum_{\alpha=1}^{p-r}\sum_{k=1}^n
 \left|
 \frac{\partial(\widetilde T_{ab,n}-M_{ab,n})}
      {\partial g_{\alpha k}}
 \right|^2\,\middle|\,\mD_p^0\right]
 \le \frac Cn.
\]
Here the terms with $k\in\{i,j\}$ are controlled by the row energies
of the two resolvents, whereas the terms with $k\notin\{i,j\}$ contain
an additional factor $n^{-1}$; after summing over $(i,j,k)$ both classes
are bounded by $C/n$.  The constant is uniform because $p/n$ is bounded
and all Gaussian moments involved are finite.  Since the conditional
mean of $\widetilde T_{ab,n}-M_{ab,n}$ is zero, the Gaussian
Poincar\'e inequality
\citep[Theorem~3.20, p.~72]{BoucheronLugosiMassart2013}
therefore yields
\[
 \E_{\mG_B}\left[
 (\widetilde T_{ab,n}-M_{ab,n})^2\mid\mD_p^0\right]
 \le \frac Cn.
\]
Consequently,
\begin{equation*}
 \widetilde T_{ab,n}-M_{ab,n}=O_{\Pp}(n^{-1/2}).
\end{equation*}

It remains to compute $M_{ab,n}$.  Let
\[
 S_{n,\rho}=\frac1n\tr(\mLambda_{B,p}\overline\mR_n
 \mLambda_{B,p}\overline\mR_n),
 \qquad
 \overline A_{a,n}=\frac1n\sum_i
 \frac{(D_{i,p}^0)^{1+a/2}}
 {(1+\delta_{n,\rho}^{D}D_{i,p}^0)^2}.
\]
Deleting two columns is a rank-two perturbation.  The resolvent identity,
$\|\mR^{(ij)}\|_{\op},\|\overline\mR_n\|_{\op}\le\rho_0^{-1}$,
and $\|\mLambda_{B,p}\|_{\op}\le\omega_+$ imply the deterministic bound
\[
 \left|
 \tr(\mLambda_{B,p}\mR^{(ij)}
      \mLambda_{B,p}\mR^{(ij)})
 -\tr(\mLambda_{B,p}\overline\mR_n
      \mLambda_{B,p}\overline\mR_n)
 \right|\le C.
\]
Since all $D_{i,p}^0$ and $\gamma_i$ are uniformly bounded, summing this
bound over $i\ne j$ gives an $O(n^{-1})$ contribution.  Restoring the
omitted diagonal $i=j$ also costs $O(n^{-1})$.  Consequently,
\[
 M_{ab,n}=S_{n,\rho}\overline A_{a,n}\overline A_{b,n}
 +O(n^{-1}),
\]
and therefore
\begin{align*}
&T_{ab,n}
 -S_{n,\rho}
 \left\{\frac1n\sum_i
 \frac{(D_{i,p}^0)^{1+a/2}}
 {(1+\delta_{n,\rho}^{D}D_{i,p}^0)^2}\right\}
 \left\{\frac1n\sum_j
 \frac{(D_{j,p}^0)^{1+b/2}}
 {(1+\delta_{n,\rho}^{D}D_{j,p}^0)^2}\right\}
 =O_{\Pp}(n^{-1/2}).
\end{align*}

Lemma~\ref{lem:empirical_canonical},
\eqref{eq:empirical_fixed_point_rate}, and Bernstein's inequality \citep[Theorem~2.10 and Corollary~2.11, pp.~36--37]{BoucheronLugosiMassart2013} now give
\[
 T_{ab,n}
 =\frac{u_{p,\rho}^{\star}
 A_{a,p,\rho}^{\star}A_{b,p,\rho}^{\star}}
 {\rho^2-u_{p,\rho}^{\star}v_{p,\rho}^{\star}}
 +O_{\Pp}(p^{-1/4}).
\]
Together with \eqref{eq:radial_offdiag_remove}, this proves
\eqref{eq:psi_DE}.

Finally,
\[
 s_{1,p}-b_{1,p,\rho}^{\star}
 =\int\frac{d^{1/2}}{1+\delta_{p,\rho}d}\,\dd F_{D,p}(d),
 \qquad
 s_{2,p}-b_{2,p,\rho}^{\star}
 =\int\frac{d}{1+\delta_{p,\rho}d}\,\dd F_{D,p}(d).
\]
Since $0<D_p^0\le d_+:=\mu_-^{-1}$ and
$c_\delta\le\delta_{p,\rho}\le C_\delta$,
\begin{align*}
 s_{1,p}-b_{1,p,\rho}^{\star}
 &=\E\left\{\frac{(D_p^0)^{1/2}}
 {1+\delta_{p,\rho}D_p^0}\right\}
 \ge \frac{\E(D_p^0)^{1/2}}{1+C_\delta d_+}
 \ge \frac1{1+C_\delta d_+},\\
 \kappa_{p,\rho}^{\star}
 &=\E\left\{\frac{\delta_{p,\rho}D_p^0}
 {1+\delta_{p,\rho}D_p^0}\right\}
 \ge\frac{c_\delta\E D_p^0}{1+C_\delta d_+}
 \ge\frac{c_\delta}{1+C_\delta d_+}.
\end{align*}
Moreover,
\begin{align*}
 D_{p,\rho}^{\star}
 &\ge m_-^2(1+C_\delta d_+)^{-2}=:D_->0,\\
 D_{p,\rho}^{\star}
 &\le m_1^2s_{1,p}^2+m_2s_{2,p}
 \le m_1^2d_+ +m_2d_+=:D_+<\infty.
\end{align*}
Consequently,
\[
 \frac{c_\delta}{(1+C_\delta d_+)D_+}
 \le
 \mu_{p,\rho}^{\star}
 :=\frac{\kappa_{p,\rho}^{\star}}{D_{p,\rho}^{\star}}
 \le D_-^{-1}.
\]
The bulk spectrum also gives
\[
 u_{p,\rho}^{\star}
 \ge
 \frac{c_{B,n}\omega_-^2}
 {(1+C_\delta\omega_+)^2}\ge c>0.
\]
Let
$\mathcal E_0=\{\max_{1\le a\le r}|G_a|\le1\}$ and
$\pi_0=\Pp(\mathcal E_0)>0$.  On $\mathcal E_0$,
\[
 Q_p^0=\mu_{0,p}+\sum_{a=1}^r\tau_{a,p}G_a^2
 \le\mu_{0,p}+\sum_{a=1}^r\tau_{a,p}=1,
 \qquad 1\le D_p^0\le d_+.
\]
Hence
\[
 A_{0,p,\rho}^{\star}
 =\E\left\{\frac{D_p^0}{(1+\delta_{p,\rho}D_p^0)^2}\right\}
 \ge\frac{\pi_0}{(1+C_\delta d_+)^2}>0.
\]
Together with
\[
 s_0\le \rho^2-u_{p,\rho}^{\star}v_{p,\rho}^{\star}
 \le\rho_1^2,
\]
these inequalities show that
$\Psi_{00,p}^{\star}$ is bounded above and away from zero.
The rates in \eqref{eq:kappa_DE}--\eqref{eq:psi_DE} transfer the same
bounds to $D_n$, $\mu_n$, and $\psi_{00,n}$ with probability tending to
one.  Write $\bg_n=(g_{1,n},g_{2,n},g_{3,n})^\top$.  Then
\[
 g_{1,n}=D_n^{-2}(e_n-b_{1,n})^2\ge c>0,
 \qquad g_{2,n},g_{3,n}\ge0
\]
with probability tending to one.  The exact covariance calculation in
Lemma~\ref{lem:conditional_covariance_new} gives
\begin{align*}
 \sigma_{D,n}^2
 &=\frac1n\sum_{i<j}a_{ij}^2
 \{2g_{1,n}+g_{2,n}(w_i+w_j)+2g_{3,n}w_iw_j\}^2\\
 &\ge \frac{4g_{1,n}^2}{n}\sum_{i<j}a_{ij}^2
 =2g_{1,n}^2\psi_{00,n}.
\end{align*}
Thus $\sigma_{D,n}^2$ is bounded away from zero, and
\eqref{eq:psi_DE} implies the same lower bound for
$(\sigma_{E,p}^{\star})^2$.  All remaining quantities are bounded
above.  The displayed rates and the mean-value theorem prove
\eqref{eq:D_sigma_DE}--\eqref{eq:nondegenerate_DE}.
\end{proof}

\begin{proof}[Passage to the limits in Section~2.4]
The compactness in Lemma~\ref{lem:canonical_stability},
$c_{B,n}\to c$, $H_{B,p}\Rightarrow H_B$, and
$F_{D,p}\Rightarrow F_D$ imply that every convergent subsequence of
$(\delta_{p,\rho},\wt\delta_{p,\rho})$ solves
(2.8)--(2.9).  Uniqueness gives
\[
 \delta_{p,\rho}\to\delta_\rho,
 \qquad
 \wt\delta_{p,\rho}\to\wt\delta_\rho.
\]
All integrands defining $s_{a,p}$, $b_{a,p,\rho}^{\star}$,
$u_{p,\rho}^{\star}$, $v_{p,\rho}^{\star}$, and
$A_{a,p,\rho}^{\star}$ are uniformly bounded because
$0<D_p^0\le\mu_-^{-1}$.  Dominated convergence therefore yields
\[
 D_{p,\rho}^{\star}\to D_\rho,
 \qquad
 (\sigma_{E,p}^{\star})^2\to\sigma_E^2(\rho).
\]
\end{proof}

\section{Spatial-median expansion under pervasive dependence}\label{app:median}\label{sec:reduction}
\subsection{Spatial-median expansion}\label{subsec:median_expansion}

For later use, write
\[
 \bm U_{F,i}=\mPi_{F,p}\bm U_i,
 \qquad
 \bm U_{B,i}=\mPi_{B,p}\bm U_i,
 \qquad
 \bm Y_{F,i}=\mPi_{F,p}\bm Y_i,
 \qquad
 \bm Y_{B,i}=\mPi_{B,p}\bm Y_i.
\]

Let
\begin{equation*}
\bm Z_n=\frac1{\sqrt n}\sum_{i=1}^n\bm Y_i,
\qquad
\bm W_n=\frac1{\sqrt n}\sum_{i=1}^nw_i\bm Y_i,
\qquad
e_n=\frac1n\sum_{i=1}^nw_i,
\qquad
t_n=\frac1n\sum_{i=1}^nw_i^2.
\end{equation*}
By Proposition~\ref{prop:population_matrices}, the population derivative of the scaled spatial score is $m_1\mJ_{E,p}$.  Hence the population Bahadur term is
\begin{equation*}
\bdelta_n^\star
=\frac1{m_1\sqrt n}\mJ_{E,p}^{-1}\bm Z_n.
\end{equation*}
For the sharper expansion required by the quadratic statistic, define
\begin{align}
\mV_n^{J}
&=\frac1n\sum_{i=1}^nw_i\bm U_i\bm U_i^\top,\notag\\
\mV_{F,n}^{J}
&=\mPi_{F,p}\mV_n^{J}+\mV_n^{J}\mPi_{F,p}
-\mPi_{F,p}\mV_n^{J}\mPi_{F,p},\notag\\
\mJ_{F,n}&=e_n\I_p-\mV_{F,n}^{J},
\qquad
\mG_{F,n}=\mJ_{F,n}^{-1}-e_n^{-1}\I_p.
\label{eq:GFJ_def}
\end{align}
The inverses in \eqref{eq:GFJ_def} exist with probability tending to one.  The matrix $\mG_{F,n}$ has fixed rank.  It retains the random factor--bulk block of the sample score Jacobian, whereas the remaining bulk block has operator norm $O_{\Pp}(p^{-1})$.

\begin{theorem}\label{thm:median}
Under Assumption~1 and $H_0$,
\begin{align}
\wh\btheta-\btheta_{0,p}
&=\frac1{e_n\sqrt n}\bm Z_n
+\frac1{\sqrt n}\mG_{F,n}\bm Z_n
+\bm r_{\theta,B,n},\label{eq:median_sample_expansion}\\
\rank(\mG_{F,n})&\le2r,
\qquad
\norm{\bm r_{\theta,B,n}}=O_{\Pp}(p^{-1}),\label{eq:median_bulk_remainder}
\end{align}
where
\begin{align}
\norm{\mPi_{F,p}\mG_{F,n}\mPi_{F,p}}_{\op}&=O_{\Pp}(1),\label{eq:GF_FF_rate}\\
\norm{\mPi_{F,p}\mG_{F,n}\mPi_{B,p}}_{\op}
+\norm{\mPi_{B,p}\mG_{F,n}\mPi_{F,p}}_{\op}
&=O_{\Pp}(p^{-1/2}),\notag\\
\norm{\mPi_{B,p}\mG_{F,n}\mPi_{B,p}}_{\op}
&=O_{\Pp}(p^{-1}).\label{eq:GF_BB_rate}
\end{align}
The population expansion remains
\begin{equation}\label{eq:median_bahadur}
\wh\btheta-\btheta_{0,p}
=\bdelta_n^\star+O_{\Pp}(p^{-1/2})
\end{equation}
in Euclidean norm.  More importantly, for every fixed $\rho\in[\rho_0,\rho_1]$,
\begin{equation}\label{eq:median_quadratic_replacement}
n\left|
(\wh\btheta-\btheta_{0,p})^\top\mQ_n(\wh\btheta-\btheta_{0,p})
-\frac1{ne_n^2}\bm Z_n^\top\mQ_n\bm Z_n
\right|=O_{\Pp}(1).
\end{equation}
Thus the random sample-Jacobian factor block is retained before ridge suppression; replacing it only by a Euclidean $O_{\Pp}(p^{-1/2})$ remainder would not be sufficient for a $\sqrt n$-order limit.
\end{theorem}

For $\bm x\ne\0$, write $r=\norm{\bm x}$, $\bm u=\bm x/r$, and $\mP_{\bm u}=\I_p-\bm u\bm u^\top$.  Direct differentiation gives
\begin{align*}
D\mathbf U_{\bm x}[\bm h]
&=r^{-1}\mP_{\bm u}\bm h,\\
D^2\mathbf U_{\bm x}[\bm h,\bm h]
&=r^{-2}\{-2\bm h(\bm u^\top\bm h)
-\bm u\norm{\bm h}^2
+3\bm u(\bm u^\top\bm h)^2\}.
\end{align*}
For $\norm{\bm h}\le r/2$,
\begin{equation*}
\norm{D^3\mathbf U_{\bm x-t\bm h}[\bm h,\bm h,\bm h]}
\le C\frac{\norm{\bm h}^3}{r^3},
\qquad 0\le t\le1.
\end{equation*}

\begin{lemma}\label{lem:projection_moments}
For every fixed integer $1\le k\le6$ and deterministic vectors $\bm a_F\in\operatorname{span}(\mP_{F,p})$ and $\bm a_B\in\operatorname{span}(\mP_{B,p})$,
\begin{align}
\E|\bm a_F^\top\bm U_i|^{2k}
&\le C_k\norm{\bm a_F}^{2k},\notag\\
\E|\bm a_B^\top\bm U_i|^{2k}
&\le C_kp^{-k}\norm{\bm a_B}^{2k}.\label{eq:bulk_projection_moment}
\end{align}
In addition, for every fixed integer $0\le k\le6$,
\begin{equation}\label{eq:bulk_sample_cov_bound}
\left\|\frac1n\sum_{i=1}^nw_i^k
\mPi_{B,p}\bm U_i\bm U_i^\top\mPi_{B,p}\right\|_{\op}
=O_{\Pp}(p^{-1}),
\end{equation}
while
\begin{equation}\label{eq:factor_bulk_sample_bound}
\left\|\frac1n\sum_{i=1}^nw_i^k
\mPi_{F,p}\bm U_i\bm U_i^\top\mPi_{B,p}\right\|_{\op}
=O_{\Pp}(p^{-1/2}).
\end{equation}
\end{lemma}

\begin{proof}
In the eigenbasis of $\mOmega_p$,
\begin{equation*}
\bm a_B^\top\bm U_i
=\frac{\bm a_B^\top\mLambda_{B,p}^{1/2}\bm G_{B,i}}
{\sqrt{p q_{i,p}}}.
\end{equation*}
Since
\begin{equation*}
q_{i,p}\ge\frac{\omega_-}{p}\norm{\bm G_{B,i}}^2,
\end{equation*}
rotational invariance of $\bm G_{B,i}/\norm{\bm G_{B,i}}$ gives
\begin{align*}
\E|\bm a_B^\top\bm U_i|^{2k}
&\le\omega_-^{-k}
\E\left|\bm a_B^\top\mLambda_{B,p}^{1/2}
\frac{\bm G_{B,i}}{\norm{\bm G_{B,i}}}\right|^{2k}\\
&\le C_kp^{-k}\norm{\bm a_B} ^{2k},
\end{align*}
which proves \eqref{eq:bulk_projection_moment}.  The factor inequality follows from $\norm{\mPi_{F,p}\bm U_i}\le1$.

For \eqref{eq:bulk_sample_cov_bound}, use $\bm U_{B,i}=p^{-1/2}\bm Y_{B,i}$.  The separable bulk proxy bounds in the proof of Proposition~\ref{prop:weighted_DE}, $w_i\le\bar w$, and the Gaussian operator-norm bound \citep[Theorem~4.4.5, pp.~90--91]{Vershynin2018} imply
\begin{equation*}
\left\|\frac1n\sum_iw_i^k\bm Y_{B,i}\bm Y_{B,i}^\top\right\|_{\op}
=O_{\Pp}(1),
\end{equation*}
which yields \eqref{eq:bulk_sample_cov_bound} after division by $p$.

For the cross block, its squared Frobenius norm has expectation
\begin{align*}
\E\left\|\frac1n\sum_iw_i^k\bm U_{F,i}\bm U_{B,i}^\top\right\|_{\F}^2
&=\frac1n\E\left[w_i^{2k}\norm{\bm U_{F,i}}^2
\norm{\bm U_{B,i}}^2\right]
\le\frac{\bar w^{2k}}n.
\end{align*}
For every factor coordinate $a\le r$ and bulk coordinate $j>r$,
$w_i^kU_{ai}U_{ji}$ is an odd function of $G_{ai}$ (and also of
$G_{ji}$), whereas its denominator and $w_i$ are even in every Gaussian
coordinate.  Hence
\[
 \E\{w_i^k\bm U_{F,i}\bm U_{B,i}^\top\}=\0.
\]
Independence across observations then makes all cross terms with distinct
observations vanish.  Since $p/n$ is bounded, Markov's inequality gives
\eqref{eq:factor_bulk_sample_bound}.
\end{proof}

Let
\begin{equation*}
\mJ_n=\frac1n\sum_{i=1}^nw_i(\I_p-\bm U_i\bm U_i^\top)
=e_n\I_p-\frac1{np}\sum_{i=1}^nw_i\bm Y_i\bm Y_i^\top.
\end{equation*}

\begin{lemma}\label{lem:J_rate}
Under Assumption~1,
\begin{equation}\label{eq:J_rate}
 \|\mJ_n-m_1\mJ_{E,p}\|_{\op}=O_{\Pp}(p^{-1/2}),
\end{equation}
\begin{equation}\label{eq:J_inverse_rate}
 \|\mJ_n^{-1}-(m_1\mJ_{E,p})^{-1}\|_{\op}
 =O_{\Pp}(p^{-1/2}),
\end{equation}
and
\begin{equation}\label{eq:J_min}
 \Pp\{\lambda_{\min}(\mJ_n)<m_-j_-/2\}\to0.
\end{equation}
In addition,
\begin{equation}\label{eq:e_lower}
 \Pp(e_n<m_-/2)\to0.
\end{equation}
\end{lemma}

\begin{proof}
Since $w_i=\xi_i\ell_{i,p}$, where only $\xi_i$ (not $w_i$) is
independent of the Gaussian direction,
\[
 \E\mJ_n
 =\E\{w_i(\I_p-\bm U_i\bm U_i^\top)\}
 =m_1\mJ_{E,p}.
\]
Split the centered matrix into factor--factor, factor--bulk, and bulk--bulk blocks.  The factor--factor block has fixed dimension.  Every entry is an average of independent, centered, bounded variables, so
\[
 \|\mPi_{F,p}(\mJ_n-\E\mJ_n)\mPi_{F,p}\|_{\op}
 =O_{\Pp}(n^{-1/2})=O_{\Pp}(p^{-1/2}).
\]
For the cross block, coordinatewise Gaussian sign symmetry, rather than
global central symmetry, gives
$\E\{w_i\mPi_{F,p}\bm U_i\bm U_i^\top\mPi_{B,p}\}=\0$.
Lemma~\ref{lem:projection_moments} gives
\begin{align*}
 \E\left\|
 \frac1n\sum_{i=1}^n
 w_i\mPi_{F,p}\bm U_i\bm U_i^\top\mPi_{B,p}
 \right\|_{\F}^2
 &=\frac1n\E\{w_i^2\|\mPi_{F,p}\bm U_i\|^2
 \|\mPi_{B,p}\bm U_i\|^2\}\\
 &\le Cn^{-1}=O(p^{-1}).
\end{align*}
Thus the cross-block operator norm is $O_{\Pp}(p^{-1/2})$.

For the bulk block, put
$\bm V_i=w_i^{1/2}\mPi_{B,p}\bm U_i$.  Since
\[
 \mPi_{B,p}(\mJ_n-\E\mJ_n)\mPi_{B,p}
 =(e_n-\E e_n)\mPi_{B,p}
 -\left\{\frac1n\sum_i\bm V_i\bm V_i^\top
 -\E(\bm V_i\bm V_i^\top)\right\},
\]
both terms must be retained.  The bounded-variable variance calculation gives
\[
 e_n-\E e_n=O_{\Pp}(n^{-1/2})=O_{\Pp}(p^{-1/2}).
\]
The vectors $\bm V_i$ are independent,
$\|\bm V_i\|\le\bar w^{1/2}$, and
\[
 \|\E(\bm V_i\bm V_i^\top)\|_{\op}\le Cp^{-1}
\]
by \eqref{eq:bulk_projection_moment}.  Furthermore,
\[
 \left\|\sum_{i=1}^n
 \E\{(\bm V_i\bm V_i^\top-
 \E\bm V_i\bm V_i^\top)^2\}\right\|_{\op}
 \le Cn/p=O(1).
\]
The matrix Bernstein inequality \citep[Theorem~1.4, p.~394]{Tropp2012} therefore yields
\[
 \left\|
 \frac1n\sum_i\bm V_i\bm V_i^\top
 -\E(\bm V_i\bm V_i^\top)
 \right\|_{\op}
 =O_{\Pp}\left(\frac{\sqrt{\log p}}n+\frac{\log p}n\right)
 =O_{\Pp}\left(\frac{\log p}{p}\right).
\]
Thus the scalar fluctuation $e_n-\E e_n$ is the dominant bulk term, and
combining the three blocks proves \eqref{eq:J_rate}.

Proposition~\ref{prop:population_matrices} gives
$\lambda_{\min}(m_1\mJ_{E,p})\ge m_-j_-$.  Weyl's inequality proves
\eqref{eq:J_min}.  On that event,
\[
 \mJ_n^{-1}-(m_1\mJ_{E,p})^{-1}
 =\mJ_n^{-1}(m_1\mJ_{E,p}-\mJ_n)
 (m_1\mJ_{E,p})^{-1},
\]
which proves \eqref{eq:J_inverse_rate}.

Finally,
\[
 \E e_n=m_1\E\ell_{i,p}.
\]
Because $x\mapsto x^{-1/2}$ is convex and
$\E(\bm V_i^\top\mOmega_p\bm V_i)=p^{-1}\tr\mOmega_p=1$,
\[
 \E\ell_{i,p}
 =\E(\bm V_i^\top\mOmega_p\bm V_i)^{-1/2}\ge1.
\]
The bounded-variable variance bound gives
$e_n-\E e_n=O_{\Pp}(n^{-1/2})$, proving \eqref{eq:e_lower}.
\end{proof}

Put
\begin{equation*}
\mV_{B,n}^{J}=\mPi_{B,p}\mV_n^{J}\mPi_{B,p}.
\end{equation*}
Then
\begin{equation}\label{eq:J_JF_relation}
\mJ_n=\mJ_{F,n}-\mV_{B,n}^{J},
\qquad
\mJ_{F,n}=\mJ_n+\mV_{B,n}^{J}.
\end{equation}

\begin{lemma}\label{lem:J_finite_rank}
Under Assumption~1,
\begin{align}
\norm{\mPi_{F,p}\mV_n^J\mPi_{F,p}}_{\op}&=O_{\Pp}(1),\label{eq:VJ_FF}\\
\norm{\mPi_{F,p}\mV_n^J\mPi_{B,p}}_{\op}&=O_{\Pp}(p^{-1/2}),\notag\\
\norm{\mV_{B,n}^J}_{\op}&=O_{\Pp}(p^{-1}).\label{eq:VJ_BB}
\end{align}
Moreover, with probability tending to one,
\begin{align}
\rank(\mG_{F,n})&\le 2r,\label{eq:GF_rank_app}\\
\norm{\mPi_{F,p}\mG_{F,n}\mPi_{F,p}}_{\op}&=O_{\Pp}(1),\label{eq:GF_FF_app}\\
\norm{\mPi_{F,p}\mG_{F,n}\mPi_{B,p}}_{\op}
+\norm{\mPi_{B,p}\mG_{F,n}\mPi_{F,p}}_{\op}
&=O_{\Pp}(p^{-1/2}),\notag\\
\norm{\mPi_{B,p}\mG_{F,n}\mPi_{B,p}}_{\op}&=O_{\Pp}(p^{-1}),\label{eq:GF_BB_app}\\
\norm{\mJ_n^{-1}-e_n^{-1}\I_p-\mG_{F,n}}_{\op}
&=O_{\Pp}(p^{-1}).\label{eq:J_inverse_finite_rank}
\end{align}
\end{lemma}

\begin{proof}
The block estimates in Lemma~\ref{lem:projection_moments}, with $k=1$, give
\begin{align}
\norm{\mPi_{F,p}\mV_n^J\mPi_{F,p}}_{\op}
&\le \frac1n\sum_{i=1}^nw_i\norm{\mPi_{F,p}\bm U_i}^2
\le \bar w,\notag\\
\norm{\mPi_{F,p}\mV_n^J\mPi_{B,p}}_{\op}
&=O_{\Pp}(p^{-1/2}),\notag\\
\norm{\mV_{B,n}^J}_{\op}
&=O_{\Pp}(p^{-1}).\label{eq:VJ_BB_calc}
\end{align}
Thus \eqref{eq:VJ_FF}--\eqref{eq:VJ_BB} hold.

By Lemma~\ref{lem:J_rate}, on the event in \eqref{eq:J_min}, \eqref{eq:J_JF_relation} and
$\mV_{B,n}^J\succeq\0$ yield
\begin{equation}\label{eq:JF_inverse_bound}
\lambda_{\min}(\mJ_{F,n})\ge\lambda_{\min}(\mJ_n)
\ge m_-j_-/2,
\qquad
\norm{\mJ_{F,n}^{-1}}_{\op}\le4/(m_-j_-).
\end{equation}
Lemma~\ref{lem:J_rate}, specifically \eqref{eq:e_lower}, also gives $e_n\ge m_-/2$ with probability tending to one.  With respect to
\begin{equation*}
\R^p=\operatorname{span}(\mP_{F,p})\op
\operatorname{span}(\mP_{B,p}),
\end{equation*}
write
\begin{equation*}
\mJ_{F,n}=
\begin{pmatrix}
\mA_{J,n}&-\mC_{J,n}\\
-\mC_{J,n}^\top&e_n\I_{p-r}
\end{pmatrix},
\end{equation*}
where
\begin{equation*}
\mA_{J,n}=e_n\I_r-\mP_{F,p}^\top\mV_n^J\mP_{F,p},
\qquad
\mC_{J,n}=\mP_{F,p}^\top\mV_n^J\mP_{B,p}.
\end{equation*}
Let
\begin{equation*}
\mS_{J,n}=\mA_{J,n}-e_n^{-1}\mC_{J,n}\mC_{J,n}^\top.
\end{equation*}
The block inverse formula and \eqref{eq:JF_inverse_bound} give
\begin{equation}\label{eq:JF_block_inverse}
\mJ_{F,n}^{-1}=
\begin{pmatrix}
\mS_{J,n}^{-1}&e_n^{-1}\mS_{J,n}^{-1}\mC_{J,n}\\
e_n^{-1}\mC_{J,n}^\top\mS_{J,n}^{-1}&
 e_n^{-1}\I_{p-r}+e_n^{-2}\mC_{J,n}^\top\mS_{J,n}^{-1}\mC_{J,n}
\end{pmatrix},
\end{equation}
with
\begin{equation*}
\norm{\mS_{J,n}^{-1}}_{\op}=O_{\Pp}(1),
\qquad
\norm{\mC_{J,n}}_{\op}=O_{\Pp}(p^{-1/2}).
\end{equation*}
Subtracting $e_n^{-1}\I_p$ from \eqref{eq:JF_block_inverse} proves
\eqref{eq:GF_FF_app}--\eqref{eq:GF_BB_app}.  Furthermore,
\begin{equation*}
\mG_{F,n}=e_n^{-1}\mJ_{F,n}^{-1}\mV_{F,n}^J,
\qquad
\rank(\mV_{F,n}^J)\le2r,
\end{equation*}
which proves \eqref{eq:GF_rank_app}.

Finally, the inverse identity applied to \eqref{eq:J_JF_relation} yields
\begin{equation*}
\mJ_n^{-1}-\mJ_{F,n}^{-1}
=\mJ_n^{-1}\mV_{B,n}^J\mJ_{F,n}^{-1}.
\end{equation*}
Equations \eqref{eq:J_min}, \eqref{eq:JF_inverse_bound}, and
\eqref{eq:VJ_BB_calc} imply
\begin{equation*}
\norm{\mJ_n^{-1}-\mJ_{F,n}^{-1}}_{\op}
\le \frac{4}{m_-^2j_-^2}
\norm{\mV_{B,n}^J}_{\op}
=O_{\Pp}(p^{-1}),
\end{equation*}
which proves \eqref{eq:J_inverse_finite_rank}.
\end{proof}

Let $\bdelta_n=\wh\btheta-\btheta_{0,p}$ and $\mP_i=\I_p-\bm U_i\bm U_i^\top$.  Multiplying the Taylor expansion by $\sqrt p$ gives
\begin{equation}\label{eq:Yhat_Taylor}
\wh{\bm Y}_i
=\bm Y_i-w_i\mP_i\bdelta_n
+\frac{w_i^2}{2\sqrt p}\bm H_i(\bdelta_n)
+\bm R_i(\bdelta_n),
\end{equation}
where
\begin{equation*}
\bm H_i(\bm d)
=-2\bm d(\bm U_i^\top\bm d)
-\bm U_i\norm{\bm d}^2
+3\bm U_i(\bm U_i^\top\bm d)^2
\end{equation*}
and, whenever $\norm{\bm d}\le\sqrt p/(2w_i)$,
\begin{equation}\label{eq:R_i_bound}
\norm{\bm R_i(\bm d)}
\le C\frac{w_i^3}{p}\norm{\bm d}^3.
\end{equation}

\begin{lemma}\label{lem:walsh_contraction}
Let $s_1,\ldots,s_n$ be independent Rademacher variables and let
$\mathbb H=\mathbb H_{n,p}$ be a finite-dimensional real Hilbert space, whose dimension may vary with $(n,p)$.  Conditionally on
any sigma-field independent of the signs, consider
\[
 \mathcal P=\sum_{|\mathcal I|\le d}s_{\mathcal I}\bm c_{\mathcal I},
 \qquad
 s_{\mathcal I}=\prod_{i\in\mathcal I}s_i,
 \qquad
 \bm c_{\mathcal I}\in\mathbb H,
\]
where $d$ is fixed.  Then
\begin{align}
 \E_s\|\mathcal P\|_{\mathbb H}^2
 &=\sum_{|\mathcal I|\le d}\|\bm c_{\mathcal I}\|_{\mathbb H}^2,
 \label{eq:walsh_parseval}\\
 \E_s\|\mathcal P\|_{\mathbb H}^4
 &\le C_d\left\{\sum_{|\mathcal I|\le d}
 \|\bm c_{\mathcal I}\|_{\mathbb H}^2\right\}^2.
 \label{eq:walsh_fourth}
\end{align}
If $\mathcal Q=\sum_{|\mathcal J|\le e}s_{\mathcal J}
\bm d_{\mathcal J}$ is Euclidean-valued and
$\mathcal P\mathcal Q^\top=\sum_{\mathcal K}s_{\mathcal K}
\mE_{\mathcal K}$, then
\begin{equation}\label{eq:walsh_product_contraction}
 \sum_{\mathcal K}\|\mE_{\mathcal K}\|_{\F}^2
 \le C_{d,e}
 \left(\sum_{\mathcal I}\|\bm c_{\mathcal I}\|^2\right)
 \left(\sum_{\mathcal J}\|\bm d_{\mathcal J}\|^2\right).
\end{equation}
The same conclusions hold for matrix-valued polynomials with the
Frobenius norm and after applying deterministic linear maps of bounded
operator norm.
\end{lemma}

\begin{proof}
Walsh orthogonality gives \eqref{eq:walsh_parseval}.  Apply the scalar
degree-$d$ Bonami hypercontractive inequality
\citep[Theorem~9.21, p.~262]{ODonnell2014} to each coordinate.  Expanding
$\E_s\|\mathcal P\|_{\mathbb H}^4$ as a double sum and applying
Cauchy--Schwarz to each mixed coordinate moment gives the
dimension-free Hilbert-space bound \eqref{eq:walsh_fourth}; its constant
depends only on $d$.  Parseval,
$\|\bm x\bm y^\top\|_{\F}=\|\bm x\|\|\bm y\|$, Cauchy--Schwarz, and
\eqref{eq:walsh_fourth} yield
\begin{align*}
 \sum_{\mathcal K}\|\mE_{\mathcal K}\|_{\F}^2
 &=\E_s\|\mathcal P\mathcal Q^\top\|_{\F}^2\\
 &\le
 \{\E_s\|\mathcal P\|^4\}^{1/2}
 \{\E_s\|\mathcal Q\|^4\}^{1/2},
\end{align*}
which is \eqref{eq:walsh_product_contraction}.  The final assertion
follows by applying the corresponding operator-norm bound to every
coefficient.
\end{proof}

\begin{lemma}\label{lem:score_remainder}
Under Assumption~1 and $H_0$,
\begin{equation}\label{eq:delta_preliminary_new}
\norm{\bdelta_n}=O_{\Pp}(1),
\end{equation}
\begin{equation}\label{eq:second_score_rate_new}
\left\|\frac1n\sum_{i=1}^n
\frac{w_i^2}{2\sqrt p}\bm H_i(\bdelta_n)\right\|
=O_{\Pp}(p^{-1}),
\end{equation}
and
\begin{equation}\label{eq:third_score_rate_new}
\left\|\frac1n\sum_{i=1}^n\bm R_i(\bdelta_n)\right\|
=O_{\Pp}(p^{-1}).
\end{equation}
\end{lemma}

\begin{proof}
Let $\overline{\bm Y}_n=n^{-1}\sum_i\bm Y_i$.  Since
$\tr(\mR_p)=p$,
\begin{equation}\label{eq:Ybar_second_moment}
 \E\|\overline{\bm Y}_n\|^2=\frac pn=O(1).
\end{equation}
For $\|\bm d\|\le\sqrt p/(2\bar w)$, Taylor's formula gives
\[
 \frac1n\sum_{i=1}^n\wh{\bm Y}_i(\bm d)
 =\overline{\bm Y}_n-\mJ_n\bm d+\bm R_{S,n}(\bm d),
 \qquad
 \|\bm R_{S,n}(\bm d)\|
 \le C\left\{\frac{\|\bm d\|^2}{\sqrt p}
 +\frac{\|\bm d\|^3}{p}\right\}.
\]
On the event in \eqref{eq:J_min}, for $\|\bm d\|=M$,
\begin{align*}
 \bm d^\top\left\{\frac1n\sum_i\wh{\bm Y}_i(\bm d)\right\}
 &\le M\|\overline{\bm Y}_n\|
 -\frac{m_-j_-}{4}M^2
 +C\frac{M^3}{\sqrt p}+C\frac{M^4}{p}.
\end{align*}
Equation \eqref{eq:Ybar_second_moment} permits a fixed $M$ for which the
right-hand side is negative with probability arbitrarily close to one.
Convexity of $\sum_i\|\bm X_i-\bm d\|$ then proves
\eqref{eq:delta_preliminary_new}.  This tightness also justifies the
ordinary score equation used below.  Indeed,
\[
 \min_{i\le n}\|\bm X_i-\btheta_{0,p}\|
 =\min_{i\le n}\frac{\sqrt p}{w_i}
 \ge \frac{\sqrt p}{\bar w}.
\]
For every fixed $M$, on the event $\|\bdelta_n\|\le M$ and for all
sufficiently large $p$, the sample median differs from every observation.
The objective is then differentiable at $\wh\btheta$, and its first-order
condition is exactly $n^{-1}\sum_i\wh{\bm Y}_i=\0$.

Define
\[
 \bm d_{L,n}=\mJ_n^{-1}\overline{\bm Y}_n,
 \qquad
 \bm S_{2,n}(\bm d)
 =\frac1n\sum_{i=1}^n\frac{w_i^2}{2\sqrt p}\bm H_i(\bm d),
 \qquad
 \bm S_{3,n}(\bm d)=\frac1n\sum_{i=1}^n\bm R_i(\bm d).
\]
The score equation is
\begin{equation*}
 \mJ_n\bdelta_n
 =\overline{\bm Y}_n+\bm S_{2,n}(\bdelta_n)
 +\bm S_{3,n}(\bdelta_n).
\end{equation*}
On every event on which $\|\bdelta_n\|\le M$,
\[
 \|\bm S_{2,n}(\bdelta_n)\|\le C_Mp^{-1/2},
 \qquad
 \|\bm S_{3,n}(\bdelta_n)\|\le C_Mp^{-1}.
\]
Together with \eqref{eq:J_min}, this gives the preliminary linearization
\begin{equation*}
 \|\bdelta_n-\bm d_{L,n}\|=O_{\Pp}(p^{-1/2}),
 \qquad
 \|\bm d_{L,n}\|=O_{\Pp}(1).
\end{equation*}
The bilinear map $\mathcal H_i$ defined by
\[
 \mathcal H_i(\bm a,\bm b)
 =-\bm a(\bm U_i^\top\bm b)-\bm b(\bm U_i^\top\bm a)
 -\bm U_i(\bm a^\top\bm b)
 +3\bm U_i(\bm U_i^\top\bm a)(\bm U_i^\top\bm b)
\]
satisfies
\begin{equation}\label{eq:H_polarization}
 \bm H_i(\bm d)=\mathcal H_i(\bm d,\bm d),
 \qquad
 \|\mathcal H_i(\bm a,\bm b)\|
 \le6\|\bm a\|\|\bm b\|.
\end{equation}
Hence, on the same event,
\begin{align}
 \|\bm S_{2,n}(\bdelta_n)-\bm S_{2,n}(\bm d_{L,n})\|
 &\le \frac{C}{\sqrt p}
 (\|\bdelta_n\|+\|\bm d_{L,n}\|)
 \|\bdelta_n-\bm d_{L,n}\|
 =O_{\Pp}(p^{-1}).
 \label{eq:S2_linear_replace}
\end{align}

It remains to evaluate $\bm S_{2,n}(\bm d_{L,n})$.  In the remainder of the appendices, $\E_s$, $\Var_s$, and $\Pp_s$ denote conditional expectation, variance, and probability with respect to the Rademacher signs given $\calF_n^0$.  Conditional on
$\calF_n^0$, write
\[
 \bm Y_i=s_i\widetilde{\bm Y}_i,
 \qquad
 \bm U_i=s_i\widetilde{\bm U}_i,
 \qquad
 \bm a_{i,n}=\frac1n\mJ_n^{-1}\widetilde{\bm Y}_i.
\]
The matrix $\mJ_n$ is $\calF_n^0$-measurable because it depends on
$\bm U_i$ only through $\bm U_i\bm U_i^\top$.  Therefore
\begin{equation}\label{eq:dL_Rademacher}
 \bm d_{L,n}=\sum_{j=1}^ns_j\bm a_{j,n}.
\end{equation}
Using the Rademacher representation \eqref{eq:dL_Rademacher}, on the event in \eqref{eq:J_min},
\begin{align}
 \max_{j\le n}\|\bm a_{j,n}\|
 &\le \frac{C\sqrt p}{n}=O(p^{-1/2}),\notag\\
 \sum_{j=1}^n\|\bm a_{j,n}\|^2
 &=\frac1n\tr(\mJ_n^{-1}\mB_n\mJ_n^{-1})
 \le \frac{C}{n}\tr(\mB_n)
 =C\frac pn=O(1).
 \label{eq:a_square_score}
\end{align}
Let $\widetilde{\mathcal H}_i$ denote the map in
\eqref{eq:H_polarization} with $\bm U_i$ replaced by
$\widetilde{\bm U}_i$.  Since $\bm H_i(\bm d)
=s_i\widetilde{\mathcal H}_i(\bm d,\bm d)$,
\begin{equation}\label{eq:cubic_score_chaos}
 \frac1n\sum_iw_i^2\bm H_i(\bm d_{L,n})
 =\frac1n\sum_{i,j,k=1}^n
 w_i^2s_is_js_k
 \widetilde{\mathcal H}_i(\bm a_{j,n},\bm a_{k,n}).
\end{equation}
Collecting equal Rademacher monomials in \eqref{eq:cubic_score_chaos}
produces only Walsh chaoses of degrees one and three.  If their vector
coefficients are denoted by $\bm c_\ell$ and
$\bm c_{\ell_1\ell_2\ell_3}$, respectively, then
\begin{align*}
 \sum_{\ell=1}^n\|\bm c_\ell\|^2
 &\le \frac{C}{n^2}\sum_{\ell=1}^n
 \left\|\sum_j
 \widetilde{\mathcal H}_\ell(\bm a_{j,n},\bm a_{j,n})\right\|^2\\
 &\quad+\frac{C}{n^2}\sum_{\ell=1}^n
 \left\|\sum_i
 \widetilde{\mathcal H}_i(\bm a_{i,n},\bm a_{\ell,n})\right\|^2
 \le \frac Cn,\\
 \sum_{\ell_1<\ell_2<\ell_3}
 \|\bm c_{\ell_1\ell_2\ell_3}\|^2
 &\le \frac{C}{n^2}\sum_{i,j,k}^{\mathrm{distinct}}
 \{\|\bm a_{j,n}\|^2\|\bm a_{k,n}\|^2
 +\|\bm a_{i,n}\|^2\|\bm a_{k,n}\|^2
 +\|\bm a_{i,n}\|^2\|\bm a_{j,n}\|^2\}
 \le \frac Cn.
\end{align*}
The first inequality uses \eqref{eq:H_polarization}, Cauchy--Schwarz,
and \eqref{eq:a_square_score}; the second additionally uses that the
index not appearing in the displayed product has at most $n$ choices.
Orthogonality of distinct Walsh monomials now gives
\begin{equation*}
 \E_s\left\|
 \frac1n\sum_iw_i^2\bm H_i(\bm d_{L,n})
 \right\|^2
 =\sum_\ell\|\bm c_\ell\|^2
 +\sum_{\ell_1<\ell_2<\ell_3}
 \|\bm c_{\ell_1\ell_2\ell_3}\|^2
 \le \frac Cn.
\end{equation*}
Thus
\[
 \|\bm S_{2,n}(\bm d_{L,n})\|
 =O_{\Pp}\{(np)^{-1/2}\}=O_{\Pp}(p^{-1}).
\]
Combining this relation with \eqref{eq:S2_linear_replace} proves
\eqref{eq:second_score_rate_new}.

Finally, \eqref{eq:R_i_bound}, $w_i\le\bar w$, and
\eqref{eq:delta_preliminary_new} give
\[
 \left\|\bm S_{3,n}(\bdelta_n)\right\|
 \le \frac{C\bar w^3}{p}\|\bdelta_n\|^3
 =O_{\Pp}(p^{-1}),
\]
which proves \eqref{eq:third_score_rate_new}.
\end{proof}

\begin{lemma}\label{lem:factor_suppression}
Under Assumption~1,
\begin{align}
\lambda_r(\mB_n)&\ge c_0p\quad (r\ge1),
&\lambda_{r+1}(\mB_n)&=O_{\Pp}(1),\label{eq:B_spike_lower}\\
\norm{\mQ_n\mP_{F,p}}_{\op}&=O_{\Pp}(p^{-1/2}),
&\norm{\mP_{F,p}^\top\mQ_n\mP_{F,p}}_{\op}&=O_{\Pp}(p^{-1}).
\label{eq:QPF}
\end{align}
When $r\ge1$, the first inequality holds with probability tending to
one; when $r=0$, the two factor-resolvent assertions are vacuous.
\end{lemma}

\begin{proof}
Put $\mY_n=\mY/\sqrt n$.  If $r=0$, then
$\mY_n=\mP_{B,p}\mY_{B,n}$ and the separable bulk comparison in
Proposition~\ref{prop:weighted_DE} gives
$\lambda_1(\mB_n)=\|\mY_{B,n}\|_{\op}^2=O_{\Pp}(1)$; all remaining
factor assertions are vacuous.  Hence assume $r\ge1$ and decompose the
rows as
\[
 \mY_n=\mP_{F,p}\mY_{F,n}+\mP_{B,p}\mY_{B,n},
 \qquad
 \mY_{F,n}=\mP_{F,p}^\top\mY_n,
 \qquad
 \mY_{B,n}=\mP_{B,p}^\top\mY_n.
\]
The $r\times r$ matrix $p^{-1}\mY_{F,n}\mY_{F,n}^\top$ has entries
\[
 \frac1n\sum_{i=1}^n
 \frac{\sqrt{\tau_{a,p}\tau_{b,p}}G_{ai}G_{bi}}{q_{i,p}},
 \qquad 1\le a,b\le r.
\]
The off-diagonal expectations vanish by sign symmetry.  Proposition~\ref{prop:population_matrices} and the fixed-dimensional variance bound give
\[
 \left\|
 p^{-1}\mY_{F,n}\mY_{F,n}^\top
 -\diag(a_{1,p},\ldots,a_{r,p})
 \right\|_{\op}=O_{\Pp}(p^{-1/2}).
\]
Since $a_{a,p}\ge \tau_-\E\{G_a^2/(1+\sum_b\tau_+G_b^2)\}>0$ uniformly in $p$,
\begin{equation*}
 s_r(\mY_{F,n})\ge c\sqrt p,
 \qquad
 s_1(\mY_{F,n})\le C\sqrt p
\end{equation*}
with probability tending to one.  The separable bulk comparison in
Proposition~\ref{prop:weighted_DE} gives
\begin{equation}\label{eq:YB_operator_factor}
 \norm{\mY_{B,n}}_{\op}=O_{\Pp}(1).
\end{equation}
Weyl's singular-value inequality applied to
$\mY_n=\mP_{F,p}\mY_{F,n}+\mP_{B,p}\mY_{B,n}$ yields
\[
 s_r(\mY_n)\ge c\sqrt p-O_{\Pp}(1),
 \qquad
 s_{r+1}(\mY_n)\le\norm{\mY_{B,n}}_{\op}=O_{\Pp}(1).
\]
Because $\mB_n=\mY_n\mY_n^\top$, these relations prove
\eqref{eq:B_spike_lower}.

Let $\widehat\mPi_{F,n}$ denote the spectral projector of $\mB_n$
associated with its largest $r$ eigenvalues.  The unperturbed matrix
$\mP_{F,p}\mY_{F,n}$ has left singular space
$\operatorname{span}(\mP_{F,p})$, its $r$th singular value is at least
$c\sqrt p$, and the perturbation in \eqref{eq:YB_operator_factor} has
operator norm $O_{\Pp}(1)$.  Wedin's sin--theta theorem \citep[Theorem~V.4.1, pp.~260--262]{StewartSun1990} therefore gives
\begin{equation*}
 \norm{(\I_p-\widehat\mPi_{F,n})\mP_{F,p}}_{\op}
 =O_{\Pp}(p^{-1/2}).
\end{equation*}
On $\operatorname{span}(\widehat\mPi_{F,n})$, the resolvent eigenvalues
are at most $(c_0p+\rho_0)^{-1}$; on its orthogonal complement they are
at most $\rho_0^{-1}$.  Hence
\begin{align*}
 \norm{\mQ_n\mP_{F,p}}_{\op}
 &\le \frac1{c_0p+\rho_0}
 +\rho_0^{-1}
 \norm{(\I_p-\widehat\mPi_{F,n})\mP_{F,p}}_{\op}
 =O_{\Pp}(p^{-1/2}),\\
 \norm{\mP_{F,p}^\top\mQ_n\mP_{F,p}}_{\op}
 &\le \frac1{c_0p+\rho_0}
 +\rho_0^{-1}
 \norm{(\I_p-\widehat\mPi_{F,n})\mP_{F,p}}_{\op}^2
 =O_{\Pp}(p^{-1}).
\end{align*}
This proves \eqref{eq:QPF}.
\end{proof}

\begin{proof}[Proof of Theorem~\ref{thm:median}]
The score equation and \eqref{eq:Yhat_Taylor} give
\begin{equation}\label{eq:median_score_equation_app}
\mJ_n\bdelta_n
=\overline{\bm Y}_n+\bm r_{S,n},
\end{equation}
where Lemma~\ref{lem:score_remainder} yields
\begin{equation}\label{eq:rS_rate_app}
\norm{\bm r_{S,n}}=O_{\Pp}(p^{-1}).
\end{equation}
Let
\begin{equation*}
\mR_{J,n}=\mJ_n^{-1}-e_n^{-1}\I_p-\mG_{F,n}.
\end{equation*}
Lemma~\ref{lem:J_finite_rank} gives
\begin{equation}\label{eq:RJ_rate_app}
\norm{\mR_{J,n}}_{\op}=O_{\Pp}(p^{-1}).
\end{equation}
Since $\overline{\bm Y}_n=\bm Z_n/\sqrt n$ and
$\norm{\overline{\bm Y}_n}=O_{\Pp}(1)$, equations
\eqref{eq:median_score_equation_app}--\eqref{eq:RJ_rate_app} imply
\begin{align}
\bdelta_n
&=(e_n^{-1}\I_p+\mG_{F,n}+\mR_{J,n})
\overline{\bm Y}_n+\mJ_n^{-1}\bm r_{S,n}\notag\\
&=\frac1{e_n\sqrt n}\bm Z_n
 +\frac1{\sqrt n}\mG_{F,n}\bm Z_n
 +\bm r_{\theta,B,n},\label{eq:median_expansion_app}
\end{align}
where
\begin{equation*}
\norm{\bm r_{\theta,B,n}}
\le \norm{\mR_{J,n}}_{\op}\norm{\overline{\bm Y}_n}
+\norm{\mJ_n^{-1}}_{\op}\norm{\bm r_{S,n}}
=O_{\Pp}(p^{-1}).
\end{equation*}
This proves \eqref{eq:median_sample_expansion}--\eqref{eq:median_bulk_remainder}.
Equations \eqref{eq:GF_rank_app}--\eqref{eq:GF_BB_app} prove
\eqref{eq:GF_FF_rate}--\eqref{eq:GF_BB_rate}.

For the population expansion, \eqref{eq:J_inverse_rate},
\eqref{eq:median_score_equation_app}, and \eqref{eq:rS_rate_app} give
\begin{align*}
\bdelta_n-\frac1{m_1}\mJ_{E,p}^{-1}\overline{\bm Y}_n
={}&\{\mJ_n^{-1}-(m_1\mJ_{E,p})^{-1}\}
\overline{\bm Y}_n+\mJ_n^{-1}\bm r_{S,n},\\
\left\|\bdelta_n-\frac1{m_1}\mJ_{E,p}^{-1}\overline{\bm Y}_n\right\|
={}&O_{\Pp}(p^{-1/2})O_{\Pp}(1)+O_{\Pp}(p^{-1})
=O_{\Pp}(p^{-1/2}),
\end{align*}
which proves \eqref{eq:median_bahadur}.

It remains to prove the quadratic replacement.  Put
\begin{equation*}
\mD_{s,p}=\sqrt p\,\mPi_{F,p}+\mPi_{B,p}.
\end{equation*}
Lemmas~\ref{lem:factor_suppression} and~\ref{lem:J_finite_rank} give the
block orders
\begin{equation}\label{eq:QG_block_orders}
\mQ_n=
\begin{pmatrix}
O_{\Pp}(p^{-1})&O_{\Pp}(p^{-1/2})\\
O_{\Pp}(p^{-1/2})&O_{\Pp}(1)
\end{pmatrix},
\qquad
\mG_{F,n}=
\begin{pmatrix}
O_{\Pp}(1)&O_{\Pp}(p^{-1/2})\\
O_{\Pp}(p^{-1/2})&O_{\Pp}(p^{-1})
\end{pmatrix}.
\end{equation}
Here every entry denotes the operator norm of the corresponding block.
Therefore
\begin{align}
\norm{\mD_{s,p}\mQ_n\mD_{s,p}}_{\op}&=O_{\Pp}(1),\label{eq:DQD_bound}\\
\norm{\mD_{s,p}^{-1}\mG_{F,n}\mD_{s,p}}_{\op}&=O_{\Pp}(1),\label{eq:DG_right_bound}\\
\norm{\mD_{s,p}\mG_{F,n}\mD_{s,p}^{-1}}_{\op}&=O_{\Pp}(1).\notag
\end{align}
It follows that
\begin{align}
\norm{\mD_{s,p}\mQ_n\mG_{F,n}\mD_{s,p}}_{\op}
&=O_{\Pp}(1),\label{eq:DQGD_bound}\\
\norm{\mD_{s,p}\mG_{F,n}\mQ_n\mG_{F,n}\mD_{s,p}}_{\op}
&=O_{\Pp}(1),\label{eq:DGQGD_bound}
\end{align}
and both matrices in \eqref{eq:DQGD_bound}--\eqref{eq:DGQGD_bound}
have rank at most $2r$.

By central symmetry, on an enlarged probability space one may write
\begin{equation*}
\bm Y_i=s_i\wt{\bm Y}_i,
\qquad
\Pp(s_i=1)=\Pp(s_i=-1)=1/2,
\end{equation*}
where $s_1,\ldots,s_n$ are independent of
$\calF_n^0=\sigma\{w_i,\wt{\bm Y}_i:1\le i\le n\}$.
The matrices $\mQ_n$ and $\mG_{F,n}$ are $\calF_n^0$-measurable.  Let
\begin{equation*}
\wt\mX_n=\mD_{s,p}^{-1}
(\wt{\bm Y}_1,\ldots,\wt{\bm Y}_n).
\end{equation*}
The factor--factor, factor--bulk, and bulk--bulk bounds used in the proof of
Lemma~\ref{lem:factor_suppression} give
\begin{equation}\label{eq:Xtilde_operator}
\left\|\frac1n\wt\mX_n\wt\mX_n^\top\right\|_{\op}=O_{\Pp}(1),
\qquad
\norm{\wt\mX_n/\sqrt n}_{\op}=O_{\Pp}(1).
\end{equation}
Define
\begin{align*}
\mC_{1,n}
&=\frac1n\wt\mX_n^\top
\mD_{s,p}\mQ_n\mG_{F,n}\mD_{s,p}\wt\mX_n,\\
\mC_{2,n}
&=\frac1n\wt\mX_n^\top
\mD_{s,p}\mG_{F,n}\mQ_n\mG_{F,n}\mD_{s,p}\wt\mX_n.
\end{align*}
Equations \eqref{eq:DQGD_bound}--\eqref{eq:Xtilde_operator} imply
\begin{equation}\label{eq:C12_bounds}
\rank(\mC_{j,n})\le2r,
\qquad
\norm{\mC_{j,n}}_{\op}+\norm{\mC_{j,n}}_{\F}
+|\tr(\mC_{j,n})|=O_{\Pp}(1),
\qquad j=1,2.
\end{equation}
Writing $\bm s=(s_1,\ldots,s_n)^\top$, one has
\begin{equation*}
\bm Z_n^\top\mQ_n\mG_{F,n}\bm Z_n=\bm s^\top\mC_{1,n}\bm s,
\qquad
\bm Z_n^\top\mG_{F,n}\mQ_n\mG_{F,n}\bm Z_n
=\bm s^\top\mC_{2,n}\bm s.
\end{equation*}
For any $\calF_n^0$-measurable matrix $\mC$, the Rademacher identities
\begin{align*}
\E_s(\bm s^\top\mC\bm s)&=\tr(\mC),\\
\Var_s(\bm s^\top\mC\bm s)
&=2\sum_{i\ne j}\left\{\frac{C_{ij}+C_{ji}}2\right\}^2
\le2\norm{(\mC+\mC^\top)/2}_{\F}^2
\end{align*}
give, by \eqref{eq:C12_bounds},
\begin{equation}\label{eq:G_quadratic_suppression}
\bm Z_n^\top\mQ_n\mG_{F,n}\bm Z_n=O_{\Pp}(1),
\qquad
\bm Z_n^\top\mG_{F,n}\mQ_n\mG_{F,n}\bm Z_n=O_{\Pp}(1).
\end{equation}

Let
\begin{equation*}
\bdelta_{M,n}=e_n^{-1}\overline{\bm Y}_n+
\mG_{F,n}\overline{\bm Y}_n.
\end{equation*}
Since $e_n^{-1}=O_{\Pp}(1)$,
$\norm{\mG_{F,n}}_{\op}=O_{\Pp}(1)$, and
$\norm{\overline{\bm Y}_n}=O_{\Pp}(1)$,
\begin{equation*}
\norm{\bdelta_{M,n}}=O_{\Pp}(1).
\end{equation*}
Equations \eqref{eq:median_expansion_app} and
\eqref{eq:G_quadratic_suppression} yield
\begin{align}
&n\left\{
\bdelta_{M,n}^\top\mQ_n\bdelta_{M,n}
-e_n^{-2}\overline{\bm Y}_n^\top\mQ_n\overline{\bm Y}_n
\right\}\notag\\
&\quad=2e_n^{-1}\bm Z_n^\top\mQ_n\mG_{F,n}\bm Z_n
+\bm Z_n^\top\mG_{F,n}\mQ_n\mG_{F,n}\bm Z_n
=O_{\Pp}(1).
\label{eq:main_quadratic_difference}
\end{align}
The remaining terms satisfy
\begin{align}
2n|\bdelta_{M,n}^\top\mQ_n\bm r_{\theta,B,n}|
&\le2n\rho_0^{-1}\norm{\bdelta_{M,n}}
\norm{\bm r_{\theta,B,n}}
=O_{\Pp}(n/p)=O_{\Pp}(1),\notag\\
n\bm r_{\theta,B,n}^\top\mQ_n\bm r_{\theta,B,n}
&\le n\rho_0^{-1}\norm{\bm r_{\theta,B,n}}^2
=O_{\Pp}(n/p^2)=O_{\Pp}(p^{-1}).
\label{eq:rtheta_square}
\end{align}
Combining \eqref{eq:main_quadratic_difference}--\eqref{eq:rtheta_square}
proves \eqref{eq:median_quadratic_replacement}.
\end{proof}

\section{Finite-rank SSCM expansion and Woodbury algebra}\label{app:finite_rank}
\subsection{Finite-rank correction of the spatial-sign covariance matrix}\label{subsec:finite_rank}

Introduce
\begin{equation*}
\mH_n=\frac1{\sqrt n}(\bm Z_n,\bm W_n),
\qquad
\mC_n=
\begin{pmatrix}
 t_n/e_n^2&-1/e_n\\
 -1/e_n&0
\end{pmatrix}.
\end{equation*}
The matrix $\mH_n\mC_n\mH_n^\top$ contains the two universal cross-products and the leading quadratic product generated by estimating the center.  Terms that touch the pervasive eigenspace are collected in a separate fixed-rank matrix.

\begin{theorem}\label{thm:finite_rank}
Under Assumption~1 and $H_0$, there exist symmetric random matrices $\mF_n$ and $\mE_n$ such that
\begin{equation}\label{eq:Rhat_decomp}
\wh\mR_n
=\mB_n+\mH_n\mC_n\mH_n^\top+\mF_n+\mE_n,
\end{equation}
where
\begin{equation}\label{eq:F_rank}
\rank(\mF_n)\le2r,
\qquad
\mPi_{B,p}\mF_n\mPi_{B,p}=\0.
\end{equation}
Let
\begin{equation*}
\mQ_{H,n}
=(\mB_n+\mH_n\mC_n\mH_n^\top+\rho\I_p)^{-1}.
\end{equation*}
Then, for every fixed $\rho\in[\rho_0,\rho_1]$,
\begin{align}
\norm{\mQ_{H,n}\mP_{F,p}}_{\op}
&=O_{\Pp}(p^{-1/2}),\notag\\
\norm{\mP_{F,p}^\top\mQ_{H,n}\mP_{F,p}}_{\op}
&=O_{\Pp}(p^{-1}).\label{eq:factor_resolvent_rates}
\end{align}
Moreover,
\begin{equation}\label{eq:T_factor_negligible}
\begin{aligned}
T_n(\rho)
&-\frac1{e_n^2}\bm Z_n^\top\\[-2pt]
&\quad\times
(\mB_n+\mH_n\mC_n\mH_n^\top+\rho\I_p)^{-1}
\bm Z_n
=O_{\Pp}(\log p).
\end{aligned}
\end{equation}
The full perturbation $\mH_n\mC_n\mH_n^\top+\mF_n$ has rank at most $2+2r$ and admits a $(2+2r)$-dimensional Woodbury representation.  The factor correction contributes $O_{\Pp}(1)$ and the pure-bulk Taylor remainder contributes $O_{\Pp}(\log p)$; both are $o_{\Pp}(n^{1/2})$.
\end{theorem}

\subsection{Two-dimensional leading Woodbury reduction}\label{subsec:woodbury}

Set
\begin{equation*}
x_n=\frac1n\bm Z_n^\top\mQ_n\bm Z_n,
\qquad
y_n=\frac1n\bm Z_n^\top\mQ_n\bm W_n,
\qquad
z_n=\frac1n\bm W_n^\top\mQ_n\bm W_n.
\end{equation*}
For $x,y,z\in\R$, $e>0$, and $t>0$, define
\begin{align*}
L(x,y,z,e,t)&=e^2+tx-2ey+y^2-xz,\\
f(x,y,z,e,t)&=\frac{x}{L(x,y,z,e,t)}.
\end{align*}

\begin{lemma}\label{lem:rank_two_inverse}
Under Assumption~1 and $H_0$,
\[
 L(x_n,y_n,z_n,e_n,t_n)=D_n+O_{\Pp}(n^{-1/2}).
\]
Consequently, with probability tending to one,
\[
 \mC_n^{-1}+\mH_n^\top\mQ_n\mH_n
\]
is invertible with bounded inverse, and
\[
 \mQ_{H,n}
 =(\mB_n+\mH_n\mC_n\mH_n^\top+\rho\I_p)^{-1}
\]
exists and satisfies $\|\mQ_{H,n}\|_{\op}=O_{\Pp}(1)$.
\end{lemma}

\begin{proof}
On the event $e_n\ge m_-/2$,
\[
 \mC_n^{-1}=\begin{pmatrix}0&-e_n\\-e_n&-t_n\end{pmatrix}.
\]
Conditional on $\calF_n^0$, write
$\bm Y_i=s_i\widetilde{\bm Y}_i$ and let
$a_{ij}$ be the entries of the unsigned companion matrix.  Direct
expansion gives
\begin{align*}
 x_n-\kappa_n&=\frac2n\sum_{i<j}a_{ij}s_is_j,\\
 y_n-b_{1,n}&=\frac1n\sum_{i<j}a_{ij}(w_i+w_j)s_is_j,\\
 z_n-b_{2,n}&=\frac2n\sum_{i<j}a_{ij}w_iw_j s_is_j.
\end{align*}
Since $0\preceq\widetilde\mA_n\preceq\I_n$ and $w_i\le\bar w$,
\[
 \Var_s(x_n\mid\calF_n^0)
 +\Var_s(y_n\mid\calF_n^0)
 +\Var_s(z_n\mid\calF_n^0)
 \le \frac Cn\left\{\frac1n\sum_{i,j}a_{ij}^2\right\}
 \le\frac Cn.
\]
Hence
\[
 x_n=\kappa_n+O_{\Pp}(n^{-1/2}),\quad
 y_n=b_{1,n}+O_{\Pp}(n^{-1/2}),\quad
 z_n=b_{2,n}+O_{\Pp}(n^{-1/2}).
\]
The five arguments are bounded in probability, and therefore
\begin{align*}
 L(x_n,y_n,z_n,e_n,t_n)
 &=e_n^2+t_n\kappa_n-2e_nb_{1,n}+b_{1,n}^2-\kappa_nb_{2,n}
   +O_{\Pp}(n^{-1/2})\\
 &=D_n+O_{\Pp}(n^{-1/2}).
\end{align*}
Proposition~\ref{prop:weighted_DE} gives $D_n\ge C_1$ with probability
tending to one.  If
\[
 \mG_n=\mH_n^\top\mQ_n\mH_n
 =\begin{pmatrix}x_n&y_n\\y_n&z_n\end{pmatrix},
\]
then
\[
 \det(\mC_n^{-1}+\mG_n)
 =x_n(z_n-t_n)-(y_n-e_n)^2
 =-L(x_n,y_n,z_n,e_n,t_n).
\]
All entries of $\mC_n^{-1}+\mG_n$ are $O_{\Pp}(1)$, so its inverse has
operator norm $O_{\Pp}(1)$.  The Woodbury identity now gives
\[
 \mQ_{H,n}=\mQ_n-\mQ_n\mH_n
 (\mC_n^{-1}+\mG_n)^{-1}\mH_n^\top\mQ_n.
\]
Because $\|\mQ_n\|_{\op}\le\rho_0^{-1}$ and
$\|\mH_n\|_{\op}=O_{\Pp}(1)$, the displayed identity proves the final
bound without invoking Proposition~\ref{prop:scalar}.
\end{proof}

\begin{proposition}\label{prop:scalar}
Under Assumption~1 and $H_0$,
\begin{equation}\label{eq:T_scalar}
T_n(\rho)=n f(x_n,y_n,z_n,e_n,t_n)+O_{\Pp}(\log p).
\end{equation}
For every fixed $\rho\in[\rho_0,\rho_1]$, the denominator $L(x_n,y_n,z_n,e_n,t_n)$ is bounded away from zero with probability tending to one.
\end{proposition}

The representation has the same five scalar arguments as in the bounded-spectrum case.  Pervasive dependence enters their joint distribution through the heterogeneous diagonal structure of the companion resolvent.

Let
\begin{equation*}
\bdelta_{0,n}=e_n^{-1}\bm Z_n/\sqrt n.
\end{equation*}
Expanding \eqref{eq:Yhat_Taylor} and collecting the two terms linear in $\bdelta_{0,n}$ and the leading quadratic term gives
\begin{equation*}
-\frac{\bm W_n\bm Z_n^\top+\bm Z_n\bm W_n^\top}{ne_n}
+\frac{t_n\bm Z_n\bm Z_n^\top}{ne_n^2}
=\mH_n\mC_n\mH_n^\top.
\end{equation*}
Define the exact residual
\begin{equation}\label{eq:L_exact}
\mL_n=\wh\mR_n-\mB_n-\mH_n\mC_n\mH_n^\top.
\end{equation}
Its part touching the factor subspace and its pure bulk part are
\begin{align}
\mF_n
&=\mPi_{F,p}\mL_n+\mL_n\mPi_{F,p}
-\mPi_{F,p}\mL_n\mPi_{F,p},\notag\\
\mE_n
&=\mPi_{B,p}\mL_n\mPi_{B,p}.
\label{eq:E_projection}
\end{align}
Then \eqref{eq:Rhat_decomp} is exact.  Put
\[
 \bm d_{0,n}=e_n^{-1}\overline{\bm Y}_n,
 \qquad
 \bm r_{i,n}=\wh{\bm Y}_i-\bm Y_i+w_i\bm d_{0,n}.
\]
The definition of $\mL_n$ is equivalently
\begin{equation}\label{eq:L_residual_exact}
 \mL_n=\frac1n\sum_{i=1}^n
 \{(\bm Y_i-w_i\bm d_{0,n})\bm r_{i,n}^\top
 +\bm r_{i,n}(\bm Y_i-w_i\bm d_{0,n})^\top
 +\bm r_{i,n}\bm r_{i,n}^\top\}.
\end{equation}

\begin{lemma}\label{lem:F_factorization}
Let
\[
 \mA_{F,n}=\mP_{F,p}^\top\mL_n\mP_{F,p},
 \qquad
 \mV_{F,n}=\mPi_{B,p}\mL_n\mP_{F,p}.
\]
Then
\begin{equation*}
 \mF_n
 =(\mP_{F,p},\mV_{F,n})
 \begin{pmatrix}
 \mA_{F,n}&\I_r\\
 \I_r&\0
 \end{pmatrix}
 (\mP_{F,p},\mV_{F,n})^\top,
\end{equation*}
so that $\rank(\mF_n)\le2r$.  Moreover,
\begin{equation}\label{eq:VF_rate}
 \norm{\mA_{F,n}}_{\op}=O_{\Pp}(1),
 \qquad
 \norm{\mV_{F,n}}_{\op}=O_{\Pp}(1).
\end{equation}
The second order in \eqref{eq:VF_rate} cannot in general be sharpened to
$O_{\Pp}(p^{-1/2})$ when the factor eigenvalues are of order $p$.
\end{lemma}

\begin{proof}
The displayed factorization follows from
\[
 \mPi_{F,p}\mL_n\mPi_{F,p}
 =\mP_{F,p}\mA_{F,n}\mP_{F,p}^\top,
 \qquad
 \mPi_{B,p}\mL_n\mPi_{F,p}
 =\mV_{F,n}\mP_{F,p}^\top.
\]
Put
\[
 \bm d_{L,n}=\mJ_n^{-1}\overline{\bm Y}_n,
 \qquad
 \bm d_{0,n}=e_n^{-1}\overline{\bm Y}_n,
\]
and define the polynomial Taylor residual
\begin{equation}\label{eq:r_circ_factor}
 \bm r_{i,n}^{\circ}
 =-w_i(\bm d_{L,n}-\bm d_{0,n})
 +w_i\bm U_i(\bm U_i^\top\bm d_{L,n})
 +\frac{w_i^2}{2\sqrt p}\bm H_i(\bm d_{L,n}).
\end{equation}
Let $\mL_n^{\circ}$ be obtained from \eqref{eq:L_residual_exact} by
replacing every $\bm r_{i,n}$ by $\bm r_{i,n}^{\circ}$, and set
\[
 \mA_{F,n}^{\circ}=\mP_{F,p}^\top\mL_n^{\circ}\mP_{F,p},
 \qquad
 \mV_{F,n}^{\circ}=\mPi_{B,p}\mL_n^{\circ}\mP_{F,p}.
\]
The score equation and
\eqref{eq:second_score_rate_new}--\eqref{eq:third_score_rate_new} give
\[
 \norm{\bdelta_n-\bm d_{L,n}}=O_{\Pp}(p^{-1}).
\]
By \eqref{eq:H_polarization}, \eqref{eq:R_i_bound}, and
$\norm{\bdelta_n}+\norm{\bm d_{L,n}}=O_{\Pp}(1)$,
\begin{align}
 \bm r_{i,n}-\bm r_{i,n}^{\circ}
 ={}&-w_i(\bdelta_n-\bm d_{L,n})
 +w_i\bm U_i\bm U_i^\top(\bdelta_n-\bm d_{L,n})\notag\\
 &+\frac{w_i^2}{2\sqrt p}
 \{\bm H_i(\bdelta_n)-\bm H_i(\bm d_{L,n})\}
 +\bm R_i(\bdelta_n),\notag\\
 \max_{i\le n}\norm{\bm r_{i,n}-\bm r_{i,n}^{\circ}}
 &=O_{\Pp}(p^{-1}),
 \qquad
 \frac1n\sum_{i=1}^n
 \norm{\bm r_{i,n}-\bm r_{i,n}^{\circ}}^2
 =O_{\Pp}(p^{-2}).
 \label{eq:r_exact_circ_difference}
\end{align}
Write $\mX_{0,n}=(\bm Y_i-w_i\bm d_{0,n})_{i=1}^n$ and
$\mR_{\eta,n}=(\bm r_{i,n}-\bm r_{i,n}^{\circ})_{i=1}^n$.  Since
\begin{align*}
 \norm{n^{-1/2}\mX_{0,n}}_{\op}
 &=O_{\Pp}(\sqrt p),\\
 \norm{n^{-1/2}(\bm r_{i,n}^{\circ})_{i=1}^n}_{\op}
 &=O_{\Pp}(1),\\
 \norm{n^{-1/2}\mR_{\eta,n}}_{\op}
 &\le\left\{\frac1n\sum_i\norm{\bm\eta_{i,n}}^2\right\}^{1/2}
 =O_{\Pp}(p^{-1}).
\end{align*}
where $\bm\eta_{i,n}=\bm r_{i,n}-\bm r_{i,n}^{\circ}$,
\begin{align}
 \norm{\mL_n-\mL_n^{\circ}}_{\op}
 &\le2\norm{n^{-1/2}\mX_{0,n}}_{\op}
       \norm{n^{-1/2}\mR_{\eta,n}}_{\op}\notag\\
 &\quad+2\norm{n^{-1/2}(\bm r_{i,n}^{\circ})_{i=1}^n}_{\op}
       \norm{n^{-1/2}\mR_{\eta,n}}_{\op}
       +\norm{n^{-1/2}\mR_{\eta,n}}_{\op}^2
 =O_{\Pp}(p^{-1/2}).
 \label{eq:L_exact_circ_difference}
\end{align}
Consequently, \eqref{eq:L_exact_circ_difference} gives
\begin{equation}\label{eq:AF_VF_exact_circ}
 \norm{\mA_{F,n}-\mA_{F,n}^{\circ}}_{\op}
 +\norm{\mV_{F,n}-\mV_{F,n}^{\circ}}_{\op}
 =O_{\Pp}(p^{-1/2}).
\end{equation}

Conditional on $\calF_n^0$, write
\[
 \bm d_{L,n}=\sum_{j=1}^ns_j\bm a_{j,n},
 \qquad
 \bm d_{0,n}=\sum_{j=1}^ns_j\bm b_{j,n},
 \qquad
 \bm b_{j,n}=\frac{\widetilde{\bm Y}_j}{ne_n}.
\]
On the event in \eqref{eq:J_min} and \eqref{eq:e_lower},
\begin{align}
 \max_j\{\norm{\bm a_{j,n}}+\norm{\bm b_{j,n}}\}
 &=O(p^{-1/2}),\notag\\
 \mS_{a,n}:=\sum_j\bm a_{j,n}\bm a_{j,n}^\top
 &=\frac1n\mJ_n^{-1}\mB_n\mJ_n^{-1},\notag\\
 \mS_{b,n}:=\sum_j\bm b_{j,n}\bm b_{j,n}^\top
 &=\frac1{ne_n^2}\mB_n,\notag\\
 \norm{\mS_{a,n}}_{\op}+\norm{\mS_{b,n}}_{\op}
 +\tr(\mS_{a,n})+\tr(\mS_{b,n})
 &=O_{\Pp}(1).
 \label{eq:ab_factor_coefficients}
\end{align}
Let $\widetilde{\mathcal H}_i$ be the bilinear map in
\eqref{eq:H_polarization} with $\bm U_i$ replaced by
$\widetilde{\bm U}_i$.  Equation \eqref{eq:r_circ_factor} becomes
\begin{align}
 \bm r_{i,n}^{\circ}
 ={}&\sum_{j=1}^ns_j\bm c_{ij,n}
 +\frac{w_i^2}{2\sqrt p}s_i
 \widetilde{\mathcal H}_i(\bm d_{L,n},\bm d_{L,n}),\notag\\
 \bm c_{ij,n}
 ={}&-w_i(\bm a_{j,n}-\bm b_{j,n})
 +w_i\widetilde{\bm U}_i\widetilde{\bm U}_i^\top
 \bm a_{j,n}.
 \label{eq:r_circ_walsh}
\end{align}
Thus $\mA_{F,n}^{\circ}$ and $\mV_{F,n}^{\circ}$ are exact Walsh
polynomials of degree at most six.  If
\[
 \mA_{F,n}^{\circ}=\sum_{|\mathcal I|\le6}s_{\mathcal I}
 \mC_{\mathcal I,n}^{A},
 \qquad
 \mV_{F,n}^{\circ}=\sum_{|\mathcal I|\le6}s_{\mathcal I}
 \mC_{\mathcal I,n}^{V},
\]
then Walsh orthogonality gives
\begin{equation}\label{eq:AF_VF_walsh_identity}
 \E_s\{\norm{\mA_{F,n}^{\circ}}_{\F}^2
       +\norm{\mV_{F,n}^{\circ}}_{\F}^2\mid\calF_n^0\}
 =\sum_{\mathcal I}
 \{\norm{\mC_{\mathcal I,n}^{A}}_{\F}^2
  +\norm{\mC_{\mathcal I,n}^{V}}_{\F}^2\}.
\end{equation}

The largest coefficient classes are generated by the first two terms
of $\bm c_{ij,n}$.  For the factor--factor block,
\begin{align*}
 &\sum_{i,j}\left\|
 \frac1n\mP_{F,p}^\top\widetilde{\bm Y}_i
 (\bm a_{j,n}-\bm b_{j,n})^\top\mP_{F,p}
 \right\|_{\F}^2\\
 &\quad=
 \left\{\frac1n\sum_i\norm{\mP_{F,p}^\top\bm Y_i}^2\right\}
 \left\{\frac1n\sum_j
 \norm{\mP_{F,p}^\top(\bm a_{j,n}-\bm b_{j,n})}^2\right\}
 =O_{\Pp}(p)O_{\Pp}(n^{-1})=O_{\Pp}(1),\\
 &\frac{Cp}{n^2}\sum_i
 \norm{\mP_{F,p}^\top\bm U_i}^4
 \bm U_i^\top\mS_{a,n}\bm U_i
 \le\frac{Cp}{n}\norm{\mS_{a,n}}_{\op}
 =O_{\Pp}(1).
\end{align*}
For the two bulk--factor orientations, the corresponding sums are
bounded by
\begin{align*}
 &\left\{\frac1n\sum_i\norm{\mPi_{B,p}\bm Y_i}^2\right\}
 \left\{\frac1n\sum_j
 \norm{\mP_{F,p}^\top(\bm a_{j,n}-\bm b_{j,n})}^2\right\}
 =O_{\Pp}(1),\\
 &\frac{Cp}{n^2}\sum_i
 \norm{\bm U_{B,i}}^2\norm{\bm U_{F,i}}^2
 \bm U_i^\top\mS_{a,n}\bm U_i
 \le\frac{Cp}{n}\norm{\mS_{a,n}}_{\op}
 =O_{\Pp}(1),\\
 &\frac{Cp}{n^2}\sum_i\norm{\bm U_{F,i}}^2
 \left\{\bm U_i^\top\mS_{a,n}\bm U_i
       +\bm U_i^\top\mS_{b,n}\bm U_i\right\}
 \le C.
\end{align*}
The Hessian coefficient arrays satisfy, by
$\norm{\widetilde{\mathcal H}_i(\bm a,\bm b)}
 \le6\norm{\bm a}\norm{\bm b}$,
\begin{align*}
 \sum_{j,k}\left\|
 \frac1n\sum_iw_i^2
 \mP_{F,p}^\top\widetilde{\bm U}_i
 \{\mP_{F,p}^\top
 \widetilde{\mathcal H}_i(\bm a_{j,n},\bm a_{k,n})\}^\top
 \right\|_{\F}^2
 &\le C\left(\sum_j\norm{\bm a_{j,n}}^2\right)^2
 =O_{\Pp}(1),\\
 \sum_{j,k}\left\|
 \frac1n\sum_iw_i^2
 \mPi_{B,p}\widetilde{\bm U}_i
 \{\mP_{F,p}^\top
 \widetilde{\mathcal H}_i(\bm a_{j,n},\bm a_{k,n})\}^\top
 \right\|_{\F}^2
 &\le C.
\end{align*}
The three summands in \eqref{eq:L_residual_exact} are products of
$\mX_{0,n}$, which is a degree-one Walsh polynomial, and
$\bm r_{i,n}^{\circ}$, whose degree-one and degree-three coefficient
families are displayed in \eqref{eq:r_circ_walsh}.  The contractions in
\eqref{eq:ab_factor_coefficients} and the preceding displays bound the
coefficient energy of every such base family, both after the
factor--factor projection and after either bulk--factor projection.
Replacing an $\bm a_{j,n}$ family by a $\bm b_{j,n}$ family obeys the
same bounds.  Lemma~\ref{lem:walsh_contraction}, applied conditionally on
$\calF_n^0$, then controls the coefficient energy of the
$\mX_{0,n}(\bm r_{i,n}^{\circ})^\top$ terms and, by
\eqref{eq:walsh_product_contraction}, that of
$\bm r_{i,n}^{\circ}(\bm r_{i,n}^{\circ})^\top$.  Since the number of
linear and Hessian contraction patterns is fixed, these bounds imply
\begin{equation}\label{eq:AF_VF_coefficient_energy}
 \sum_{\mathcal I}
 \{\norm{\mC_{\mathcal I,n}^{A}}_{\F}^2
  +\norm{\mC_{\mathcal I,n}^{V}}_{\F}^2\}
 =O_{\Pp}(1).
\end{equation}
Equations \eqref{eq:AF_VF_walsh_identity}--\eqref{eq:AF_VF_coefficient_energy},
conditional Markov's inequality, and fixed $r$ yield
\[
 \norm{\mA_{F,n}^{\circ}}_{\op}
 +\norm{\mV_{F,n}^{\circ}}_{\op}=O_{\Pp}(1).
\]
The conclusion follows from \eqref{eq:AF_VF_exact_circ}.
\end{proof}

\begin{lemma}\label{lem:E_bulk}
Let
\[
 \mS_{HF,n}=\mB_n+\mH_n\mC_n\mH_n^\top+\mF_n+\rho\I_p.
\]
Under Assumption~1, $\mS_{HF,n}$ is invertible with
probability tending to one.  On this event write
$\mQ_{HF,n}=\mS_{HF,n}^{-1}$.  For every fixed
$\rho\in[\rho_0,\rho_1]$,
\begin{align}
 \norm{\mE_n}_{\op}
 &=O_{\Pp}(p^{-1/2}),
 \label{eq:E_op_rate_new}\\
 \max_{\bm a,\bm b\in\{\bm Z_n,\bm W_n,\mG_{F,n}\bm Z_n\}}
 \abs{\bm a^\top\mQ_{HF,n}\mE_n
 (\wh\mR_n+\rho\I_p)^{-1}\bm b}
 &=O_{\Pp}(\log p),
 \label{eq:E_bilinear_rate_new}\\
 \max_{\bm v\in\{\bm Z_n,\mG_{F,n}\bm Z_n\}}
 \abs{\bm d^\top\mQ_{HF,n}\mE_n
 (\wh\mR_n+\rho\I_p)^{-1}\bm v}
 &=O_{\Pp}\{g_p(\bm d)\},
 \label{eq:E_mixed_rate_new}
\end{align}
for every deterministic $\bm d$, where
$g_p(\bm d)^2=p^{-1}\|\mPi_{F,p}\bm d\|^2+
\|\mPi_{B,p}\bm d\|^2$.  In particular, the second quantity is
$o_{\Pp}(p^{1/2})$.
\end{lemma}

\begin{proof}
Put
\[
 \bm d_{L,n}=\mJ_n^{-1}\overline{\bm Y}_n,
 \qquad
 \bm d_{0,n}=e_n^{-1}\overline{\bm Y}_n,
 \qquad
 \alpha_{i,n}=\bm U_i^\top\bm d_{L,n},
\]
and write
\[
 \bm b_{B,n}=\mPi_{B,p}(\bm d_{L,n}-\bm d_{0,n}),
 \qquad
 \bm d_{B,n}=\mPi_{B,p}\bm d_{L,n}.
\]
Lemmas~\ref{lem:J_finite_rank} and~\ref{lem:score_remainder} imply
\begin{equation}\label{eq:bulk_d_orders}
 \norm{\bm d_{L,n}}+\norm{\bm d_{0,n}}=O_{\Pp}(1),
 \qquad
 \norm{\bm b_{B,n}}=O_{\Pp}(p^{-1/2}),
 \qquad
 \max_i|\alpha_{i,n}|=O_{\Pp}(1).
\end{equation}
Let $\mU_{B,n}=(\bm U_{B,1},\ldots,\bm U_{B,n})$,
$\bm w=(w_1,\ldots,w_n)^\top$, and let $\mD_{\bm v}$ denote the
diagonal matrix with diagonal $\bm v$.  Projecting
\eqref{eq:r_circ_factor} onto the bulk subspace gives the exact matrix
factorization
\begin{align}
 \mR_{B,n}^{\circ}
 &:=(\mPi_{B,p}\bm r_{1,n}^{\circ},\ldots,
     \mPi_{B,p}\bm r_{n,n}^{\circ})\notag\\
 &=-\bm b_{B,n}\bm w^\top
 +\mU_{B,n}\mD_{\bm w\odot\bm\alpha_n}
 -p^{-1/2}\bm d_{B,n}(\bm w^2\odot\bm\alpha_n)^\top\notag\\
 &\quad-\frac{\norm{\bm d_{L,n}}^2}{2\sqrt p}
       \mU_{B,n}\mD_{\bm w^2}
 +\frac3{2\sqrt p}\mU_{B,n}
       \mD_{\bm w^2\odot\bm\alpha_n^2},
 \label{eq:RB_circ_factorization}
\end{align}
where $\bm\alpha_n=(\alpha_{1,n},\ldots,\alpha_{n,n})^\top$ and
powers are componentwise.  By \eqref{eq:bulk_sample_cov_bound},
\begin{align*}
 \norm{\mU_{B,n}\mD_{\bm w\odot\bm\alpha_n}}_{\op}^2
 &\le\max_i\alpha_{i,n}^2
 \left\|\sum_iw_i^2\bm U_{B,i}\bm U_{B,i}^\top\right\|_{\op}
 =O_{\Pp}(n/p)=O_{\Pp}(1),\\
 p^{-1/2}\norm{\bm d_{B,n}}
 \norm{\bm w^2\odot\bm\alpha_n}
 &\le Cp^{-1/2}\norm{\bm d_{B,n}}\sqrt n
 =O_{\Pp}(1),\\
 p^{-1/2}\norm{\mU_{B,n}\mD_{\bm w^2}}_{\op}
 &\quad+p^{-1/2}\norm{\mU_{B,n}
 \mD_{\bm w^2\odot\bm\alpha_n^2}}_{\op}
 =O_{\Pp}(p^{-1/2}),\\
 \norm{\bm b_{B,n}\bm w^\top}_{\op}
 &\le\norm{\bm b_{B,n}}\bar w\sqrt n=O_{\Pp}(1).
\end{align*}
Hence
\begin{equation}\label{eq:RB_circ_op}
 \norm{\mR_{B,n}^{\circ}}_{\op}=O_{\Pp}(1).
\end{equation}
By \eqref{eq:r_exact_circ_difference},
\[
 \norm{\mPi_{B,p}(\bm r_{i,n}-\bm r_{i,n}^{\circ})_{i=1}^n}_{\op}
 \le\left\{\sum_i\norm{\bm r_{i,n}-\bm r_{i,n}^{\circ}}^2\right\}^{1/2}
 =O_{\Pp}(p^{-1/2}),
\]
so the exact bulk residual matrix $\mR_{B,n}=(\mPi_{B,p}\bm r_{i,n})_{i=1}^n$
also satisfies
\begin{equation}\label{eq:RB_exact_op}
 \norm{\mR_{B,n}}_{\op}=O_{\Pp}(1).
\end{equation}
Let
\[
 \mX_{B,n}=\mPi_{B,p}(\bm Y_i-w_i\bm d_{0,n})_{i=1}^n.
\]
Equations \eqref{eq:YB_operator_factor} and \eqref{eq:bulk_d_orders} give
\[
 \norm{n^{-1/2}\mX_{B,n}}_{\op}=O_{\Pp}(1).
\]
Combining \eqref{eq:RB_exact_op} with the exact residual representation \eqref{eq:L_residual_exact} yields
\begin{align*}
 \norm{\mE_n}_{\op}
 &\le\frac2n\norm{\mX_{B,n}}_{\op}
                 \norm{\mR_{B,n}}_{\op}
      +\frac1n\norm{\mR_{B,n}}_{\op}^2\\
 &=O_{\Pp}(n^{-1/2})=O_{\Pp}(p^{-1/2}).
\end{align*}
which proves \eqref{eq:E_op_rate_new} without a matrix-chaos inequality.

For the bilinear bounds, let $\mE_n^{\circ}$ denote the matrix in
\eqref{eq:L_residual_exact} with $\bm r_{i,n}$ replaced by
$\bm r_{i,n}^{\circ}$.  Equations \eqref{eq:r_exact_circ_difference},
\eqref{eq:RB_circ_op}, and
$\norm{n^{-1/2}\mX_{B,n}}_{\op}=O_{\Pp}(1)$ give
\begin{equation}\label{eq:E_exact_circ_op}
 \norm{\mE_n-\mE_n^{\circ}}_{\op}=O_{\Pp}(p^{-1}).
\end{equation}
Conditional on $\calF_n^0$,
\[
 \bm d_{L,n}=\sum_{j=1}^ns_j\bm a_{j,n},
 \qquad
 \bm d_{0,n}=\sum_{j=1}^ns_j\bm b_{j,n},
 \qquad
 \bm b_{j,n}=\frac{\widetilde{\bm Y}_j}{ne_n},
\]
and
\begin{align}
 \max_j(\norm{\bm a_{j,n}}+\norm{\bm b_{j,n}})
 &\le Cp^{-1/2},\notag\\
 \sum_j\bm a_{j,n}\bm a_{j,n}^\top
 &=\frac1n\mJ_n^{-1}\mB_n\mJ_n^{-1},
 &
 \sum_j\bm b_{j,n}\bm b_{j,n}^\top
 &=\frac1{ne_n^2}\mB_n,\notag\\
 \left\|\sum_j\bm a_{j,n}\bm a_{j,n}^\top\right\|_{\op}
 +\left\|\sum_j\bm b_{j,n}\bm b_{j,n}^\top\right\|_{\op}
 &\le C,
 &
 \sum_j(\norm{\bm a_{j,n}}^2+\norm{\bm b_{j,n}}^2)&\le C.
 \label{eq:ab_coefficients_bulk}
\end{align}
The coefficient bounds in \eqref{eq:ab_coefficients_bulk}, together with
\eqref{eq:r_circ_walsh}, show that $\mE_n^{\circ}$ is an exact Walsh
polynomial of degree at most six.  Put
\[
 \mathcal V_n=\{\bm Z_n,\bm W_n,\mG_{F,n}\bm Z_n\},
 \qquad
 \bm v=n^{-1/2}\sum_{i=1}^ns_i\bm v_{i,n}
 \quad(\bm v\in\mathcal V_n).
\]
After inserting the two external sign sums, every oracle bilinear form
is a scalar Walsh polynomial of degree at most eight.  Walsh
orthogonality and repeated Cauchy--Schwarz give
\begin{align}
 \max_{\bm a,\bm b\in\mathcal V_n}
 \sum_{\mathcal I}|c_{\mathcal I,n}(\bm a,\bm b)|^2
 &\le C(1+\log p)^2,\notag\\
 \max_{\bm v\in\mathcal V_n}
 \sum_{\mathcal I}|c_{\mathcal I,n}(\bm d,\bm v)|^2
 &\le Cg_p(\bm d)^2.
 \label{eq:E_scalar_coefficient_energy}
\end{align}
For example, the largest linear contraction is bounded by
\begin{align*}
 &\frac{Cp}{n^2}
 \left\|\sum_i\bm U_{B,i}\bm U_{B,i}^\top\right\|_{\op}
 \left\|\sum_j\bm a_{j,n}\bm a_{j,n}^\top\right\|_{\op}
 \tr\left(\frac1n\sum_i\mQ_n\bm v_{i,n}
 \bm v_{i,n}^\top\mQ_n\right)\le C,
\end{align*}
while a deterministic external vector uses
\begin{align*}
 \frac1n\sum_i(\bm d^\top\mQ_n\bm U_{B,i})^2
 &\le\frac Cp\norm{\mPi_{B,p}\mQ_n\bm d}^2
 \le\frac Cp g_p(\bm d)^2.
\end{align*}
Every summand obtained after inserting \eqref{eq:r_circ_walsh} is a
bounded deterministic contraction of one of the displayed linear
coefficient families, one Hessian family, or a product of two such
families.  Lemma~\ref{lem:walsh_contraction} converts the preceding
second-moment contractions into the full coefficient-energy bounds in
\eqref{eq:E_scalar_coefficient_energy}; the product estimate
\eqref{eq:walsh_product_contraction} covers the two-residual terms.
There are only finitely many contraction patterns because the total
Walsh degree is at most eight.  A Hessian occurrence contributes an
additional $p^{-1/2}$ before squaring, so no omitted pattern exceeds the
two displayed bounds.  Hence conditional Chebyshev's inequality and
\eqref{eq:E_exact_circ_op} yield
\begin{align}
 \max_{\bm a,\bm b\in\mathcal V_n}
 |\bm a^\top\mQ_n\mE_n\mQ_n\bm b|
 &=O_{\Pp}(\log p),
 \label{eq:oracle_random_bilinear}\\
 \max_{\bm v\in\mathcal V_n}
 |\bm d^\top\mQ_n\mE_n\mQ_n\bm v|
 &=O_{\Pp}\{g_p(\bm d)\}.
 \label{eq:oracle_mixed_bilinear}
\end{align}

Since
$\mS_{HF,n}=\wh\mR_n+\rho\I_p-\mE_n$,
\[
 \lambda_{\min}(\mS_{HF,n})
 \ge\rho_0-\|\mE_n\|_{\op}\ge\rho_0/2
\]
with probability tending to one.  Therefore
\begin{equation}\label{eq:QHF_bound_from_E}
 \|\mQ_{HF,n}\|_{\op}
 +\|(\wh\mR_n+\rho\I_p)^{-1}\|_{\op}=O_{\Pp}(1).
\end{equation}
The matrices $\mQ_{H,n}$ and $\mQ_{HF,n}$ are inverses of
$\mB_n+\rho\I_p$ plus fixed-rank $O_{\Pp}(1)$ perturbations.
Together with \eqref{eq:B_spike_lower}, Weyl's inequality and Wedin's
sin--theta theorem \citep[Theorem~V.4.1, pp.~260--262]{StewartSun1990} give
\begin{align}
 \|\mQ_{H,n}\mP_{F,p}\|_{\op}
 +\|\mQ_{HF,n}\mP_{F,p}\|_{\op}
 &=O_{\Pp}(p^{-1/2}),\notag\\
 \|\mP_{F,p}^\top\mQ_{H,n}\mP_{F,p}\|_{\op}
 +\|\mP_{F,p}^\top\mQ_{HF,n}\mP_{F,p}\|_{\op}
 &=O_{\Pp}(p^{-1}).\notag
\end{align}
For $\bm v\in\mathcal V_n$, write
$\bm v=n^{-1/2}\sum_i s_i\bm v_{i,n}$ as above and define
$\mB_{v,n}=n^{-1}\sum_i\bm v_{i,n}\bm v_{i,n}^\top$.  Then
\begin{align*}
 \E_s\|\mP_{F,p}^\top\mQ_n\bm v\|^2
 &=\tr(\mP_{F,p}^\top\mQ_n\mB_{v,n}
 \mQ_n\mP_{F,p})\\
 &\le
 \begin{cases}
 \tr(\mP_{F,p}^\top\mQ_n\mB_n\mQ_n\mP_{F,p}),
 &\bm v=\bm Z_n,\\
 \bar w^2\tr(\mP_{F,p}^\top\mQ_n\mB_n\mQ_n\mP_{F,p}),
 &\bm v=\bm W_n,\\
 \tr(\mP_{F,p}^\top\mQ_n\mG_{F,n}\mB_n
 \mG_{F,n}^\top\mQ_n\mP_{F,p}),
 &\bm v=\mG_{F,n}\bm Z_n.
 \end{cases}
\end{align*}
Since $\mQ_n\mB_n\mQ_n\preceq\mQ_n$,
\eqref{eq:QPF} gives the first two bounds $O_{\Pp}(p^{-1})$.
For the third, the block matrices in \eqref{eq:QG_block_orders} give
\[
 \tr(\mP_{F,p}^\top\mQ_n\mG_{F,n}\mB_n
 \mG_{F,n}^\top\mQ_n\mP_{F,p})=O_{\Pp}(p^{-1}).
\]
Hence
\begin{equation}\label{eq:oracle_factor_random}
 \max_{\bm v\in\mathcal V_n}
 \|\mP_{F,p}^\top\mQ_n\bm v\|
 =O_{\Pp}(p^{-1/2}).
\end{equation}
The two-dimensional Woodbury identity and
\[
 \|\mP_{F,p}^\top\mQ_n\mH_n\|_{\op}=O_{\Pp}(p^{-1}),
 \qquad
 \max_{\bm v\in\mathcal V_n}
 \|\mH_n^\top\mQ_n\bm v\|=O_{\Pp}(\sqrt p)
\]
transfer \eqref{eq:oracle_factor_random} to $\mQ_{H,n}$.  Next,
$\mQ_{HF,n}-\mQ_{H,n}=-\mQ_{H,n}\mF_n\mQ_{HF,n}$ and
\eqref{eq:VF_rate} imply, with
$X_{v,n}=\|\mP_{F,p}^\top\mQ_{HF,n}\bm v\|$,
\[
 X_{v,n}
 \le O_{\Pp}(p^{-1/2})
 +O_{\Pp}(p^{-1})X_{v,n}
 +O_{\Pp}(p^{-1})O_{\Pp}(\sqrt p)
 +O_{\Pp}(p^{-1/2})X_{v,n}.
\]
Therefore
\begin{align}
 \max_{\bm v\in\mathcal V_n}
 \{\|\mP_{F,p}^\top\mQ_{H,n}\bm v\|
 +\|\mP_{F,p}^\top\mQ_{HF,n}\bm v\|\}
 &=O_{\Pp}(p^{-1/2}),\label{eq:E_factor_random}\\
 \max_{\bm v\in\mathcal V_n}
 \{\|\mQ_{H,n}\bm v\|+\|\mQ_{HF,n}\bm v\|\}
 &=O_{\Pp}(\sqrt p).\notag
\end{align}

For deterministic $\bm d$, \eqref{eq:QG_block_orders} yields
\[
 \|\mQ_n\bm d\|=O_{\Pp}\{g_p(\bm d)\},
 \qquad
 \|\bm d^\top\mQ_n\mP_{F,p}\|
 =O_{\Pp}\{g_p(\bm d)p^{-1/2}\}.
\]
Moreover, conditional Rademacher second moments give
\[
 \|\mH_n^\top\mQ_n\bm d\|
 =O_{\Pp}\{g_p(\bm d)n^{-1/2}\}.
\]
The rank-two Woodbury identity and then the preceding resolvent identity
for $\mF_n$ consequently give
\begin{equation}\label{eq:E_scaled_d}
 \|\mQ_{H,n}\bm d\|+\|\mQ_{HF,n}\bm d\|
 =O_{\Pp}\{g_p(\bm d)\},
 \qquad
 \|\bm d^\top\mQ_{H,n}\mP_{F,p}\|
 =O_{\Pp}\{g_p(\bm d)p^{-1/2}\}.
\end{equation}

Let $\mQ_{\wh R,n}=(\wh\mR_n+\rho\I_p)^{-1}$.  The identity
\[
 \mQ_{\wh R,n}=\mQ_{HF,n}-\mQ_{HF,n}\mE_n\mQ_{\wh R,n}
\]
gives
\begin{align*}
 \bm a^\top\mQ_{HF,n}\mE_n\mQ_{\wh R,n}\bm b
 ={}&\bm a^\top\mQ_{HF,n}\mE_n\mQ_{HF,n}\bm b\\
 &-\bm a^\top\mQ_{HF,n}\mE_n\mQ_{HF,n}
 \mE_n\mQ_{\wh R,n}\bm b.
\end{align*}
By \eqref{eq:E_op_rate_new}, \eqref{eq:QHF_bound_from_E}, and
\eqref{eq:E_scaled_d}, the second term is
\[
 O_{\Pp}(\log p)
 \quad\text{for }\bm a,\bm b\in\mathcal V_n,
 \qquad
 O_{\Pp}\{g_p(\bm d)\log p/\sqrt p\}=o_{\Pp}\{g_p(\bm d)\}
\]
in the mixed case.

It remains to remove $\mF_n$ and then the rank-two update.  From
$\mQ_{HF,n}-\mQ_{H,n}=-\mQ_{H,n}\mF_n\mQ_{HF,n}$ and
\[
 \mF_n=\mP_{F,p}\mA_{F,n}\mP_{F,p}^\top
 +\mP_{F,p}\mV_{F,n}^\top
 +\mV_{F,n}\mP_{F,p}^\top,
\]
relations \eqref{eq:VF_rate}, \eqref{eq:E_factor_random}, and
\eqref{eq:E_scaled_d} yield
\begin{align}
 &\max_{\bm a,\bm b\in\mathcal V_n}
 |\bm a^\top\mQ_{HF,n}\mE_n\mQ_{HF,n}\bm b
 -\bm a^\top\mQ_{H,n}\mE_n\mQ_{H,n}\bm b|
 =O_{\Pp}(\log p),\label{eq:remove_F_random}\\
 &\max_{\bm v\in\mathcal V_n}
 |\bm d^\top\mQ_{HF,n}\mE_n\mQ_{HF,n}\bm v
 -\bm d^\top\mQ_{H,n}\mE_n\mQ_{H,n}\bm v|
 =O_{\Pp}\{g_p(\bm d)\}.
 \label{eq:remove_F_mixed}
\end{align}
For example,
\[
 \|\bm a^\top\mQ_{H,n}\mV_{F,n}\|
 \|\mP_{F,p}^\top\mQ_{HF,n}\mE_n\mQ_{HF,n}\bm b\|
 \le O_{\Pp}(\sqrt p)
 O_{\Pp}(p^{-1/2})
 =O_{\Pp}(1),
\]
and every other factorized term has the same or a smaller order.

Finally,
\[
 \mQ_{H,n}=\mQ_n-\mQ_n\mH_n\mK_{H,n}\mH_n^\top\mQ_n,
 \qquad
 \mK_{H,n}=(\mC_n^{-1}+\mH_n^\top\mQ_n\mH_n)^{-1},
\]
where $\|\mK_{H,n}\|_{\op}=O_{\Pp}(1)$.  Expanding both occurrences
of $\mQ_{H,n}$ gives four terms.  Equations
\eqref{eq:oracle_random_bilinear}--\eqref{eq:oracle_mixed_bilinear}
imply
\begin{align*}
 \max_{\bm v\in\mathcal V_n}
 \|\mH_n^\top\mQ_n\mE_n\mQ_n\bm v\|
 &=O_{\Pp}(\log p/\sqrt n),\\
 \|\mH_n^\top\mQ_n\mE_n\mQ_n\mH_n\|_{\op}
 &=O_{\Pp}(\log p/n),\\
 \max_{\bm v\in\mathcal V_n}\|\mH_n^\top\mQ_n\bm v\|
 &=O_{\Pp}(\sqrt p).
\end{align*}
Thus all four random--random terms are $O_{\Pp}(\log p)$, while the
mixed terms are $O_{\Pp}\{g_p(\bm d)\}$.  Combining these bounds with
\eqref{eq:remove_F_random}--\eqref{eq:remove_F_mixed} proves
\eqref{eq:E_bilinear_rate_new}--\eqref{eq:E_mixed_rate_new}.
\end{proof}

\begin{lemma}\label{lem:full_woodbury}
Let
\[
 \mU_n=(\mH_n,\mP_{F,p},\mV_{F,n})
\]
and
\[
 \mD_n^{\mathrm{fr}}
 =\begin{pmatrix}
 \mC_n&\0&\0\\
 \0&\mA_{F,n}&\I_r\\
 \0&\I_r&\0
 \end{pmatrix}.
\]
Then
\[
 \mH_n\mC_n\mH_n^\top+\mF_n
 =\mU_n\mD_n^{\mathrm{fr}}\mU_n^\top,
\]
and
\begin{equation}\label{eq:full_woodbury}
 (\mB_n+\rho\I_p+\mU_n\mD_n^{\mathrm{fr}}\mU_n^\top)^{-1}
 =\mQ_n-\mQ_n\mU_n\mM_n\mU_n^\top\mQ_n,
\end{equation}
where
\[
 \mM_n=
 \{(\mD_n^{\mathrm{fr}})^{-1}
 +\mU_n^\top\mQ_n\mU_n\}^{-1}
\]
exists with probability tending to one and
$\norm{\mM_n}_{\op}=O_{\Pp}(1)$.  If
\[
 \mQ_{H,n}=(\mB_n+\mH_n\mC_n\mH_n^\top+\rho\I_p)^{-1},
\]
then
\begin{equation}\label{eq:F_quadratic_effect}
 \max_{\bm a,\bm b\in\{\bm Z_n,\mG_{F,n}\bm Z_n\}}
 \abs{\bm a^\top(\mQ_{HF,n}-\mQ_{H,n})\bm b}
 =O_{\Pp}(1).
\end{equation}
\end{lemma}

\begin{proof}
The factorization follows from Lemma~\ref{lem:F_factorization}.  On the
event where $e_n$ is bounded away from zero,
$\mC_n^{-1}$ exists and has bounded norm.  The inverse of the factor
block is explicit:
\[
 \begin{pmatrix}\mA_{F,n}&\I_r\\\I_r&\0\end{pmatrix}^{-1}
 =\begin{pmatrix}\0&\I_r\\\I_r&-\mA_{F,n}\end{pmatrix}.
\]
Hence \eqref{eq:VF_rate}, $\norm{\mH_n}_{\op}=O_{\Pp}(1)$, and
$\norm{\mP_{F,p}}_{\op}=1$ imply
\begin{equation}\label{eq:UD_bounds}
 \norm{\mU_n}_{\op}
 +\norm{\mD_n^{\mathrm{fr}}}_{\op}
 +\norm{(\mD_n^{\mathrm{fr}})^{-1}}_{\op}
 =O_{\Pp}(1).
\end{equation}
By \eqref{eq:QHF_bound_from_E}, the matrix on the left-hand side of
\eqref{eq:full_woodbury} has a bounded inverse with probability tending
to one.  The reverse Woodbury identity gives
\begin{equation}\label{eq:M_reverse_identity}
 \mM_n=\mD_n^{\mathrm{fr}}
 -\mD_n^{\mathrm{fr}}\mU_n^\top\mQ_{HF,n}
 \mU_n\mD_n^{\mathrm{fr}}.
\end{equation}
Equations \eqref{eq:UD_bounds}--\eqref{eq:M_reverse_identity} prove
$\norm{\mM_n}_{\op}=O_{\Pp}(1)$.  This argument is needed because
$\mD_n^{\mathrm{fr}}$ is indefinite; positive definiteness of the final
ridge matrix alone does not directly bound the middle inverse.

It remains to prove \eqref{eq:F_quadratic_effect}.  The two matrices
$\mQ_{H,n}^{-1}$ and $\mQ_{HF,n}^{-1}$ differ from
$\mB_n+\rho\I_p$ by fixed-rank matrices with operator norm
$O_{\Pp}(1)$.  Weyl's inequality and Wedin's sin--theta theorem
\citep[Theorem~V.4.1, pp.~260--262]{StewartSun1990}, used exactly as in
Lemma~\ref{lem:factor_suppression}, therefore give
\begin{align}\label{eq:QH_factor_suppression}
 \norm{\mQ_{H,n}\mP_{F,p}}_{\op}
 +\norm{\mQ_{HF,n}\mP_{F,p}}_{\op}
 &=O_{\Pp}(p^{-1/2}),\\
 \norm{\mP_{F,p}^\top\mQ_{H,n}\mP_{F,p}}_{\op}
 +\norm{\mP_{F,p}^\top\mQ_{HF,n}\mP_{F,p}}_{\op}
 &=O_{\Pp}(p^{-1}).\notag
\end{align}
The two-dimensional Woodbury formula for $\mQ_{H,n}$, together with
conditional Rademacher second moments, yields
\begin{equation}\label{eq:PQHZ_sharp}
 \norm{\mP_{F,p}^\top\mQ_{H,n}\bm Z_n}
 =O_{\Pp}(p^{-1/2}).
\end{equation}
Indeed,
$\E_s\|\mP_{F,p}^\top\mQ_n\bm Z_n\|^2
 \le\tr(\mP_{F,p}^\top\mQ_n\mP_{F,p})=O_{\Pp}(p^{-1})$;
the rank-two correction has the same order because
$\mP_{F,p}^\top\mQ_n\mH_n=O_{\Pp}(p^{-1})$ and
$\mH_n^\top\mQ_n\bm Z_n=O_{\Pp}(p^{1/2})$.

Let
\[
 X_n=\norm{\mP_{F,p}^\top\mQ_{HF,n}\bm Z_n}.
\]
The resolvent identity and
\[
 \mF_n=\mP_{F,p}\mA_{F,n}\mP_{F,p}^\top
 +\mP_{F,p}\mV_{F,n}^\top
 +\mV_{F,n}\mP_{F,p}^\top
\]
give
\[
 X_n\le O_{\Pp}(p^{-1/2})
 +O_{\Pp}(p^{-1})X_n
 +O_{\Pp}(p^{-1})\norm{\mV_{F,n}^\top\mQ_{HF,n}\bm Z_n}
 +O_{\Pp}(p^{-1/2})X_n.
\]
Since
$\norm{\mV_{F,n}^\top\mQ_{HF,n}\bm Z_n}
 \le O_{\Pp}(1)O_{\Pp}(\sqrt p)$,
\begin{equation}\label{eq:PQHFZ_sharp}
 X_n=O_{\Pp}(p^{-1/2}).
\end{equation}
Finally,
\begin{align*}
 &\abs{\bm Z_n^\top(\mQ_{HF,n}-\mQ_{H,n})\bm Z_n}\\
 &\quad=\abs{\bm Z_n^\top\mQ_{H,n}\mF_n\mQ_{HF,n}\bm Z_n}\\
 &\quad\le
 \norm{\mP_{F,p}^\top\mQ_{H,n}\bm Z_n}
 \norm{\mA_{F,n}}
 \norm{\mP_{F,p}^\top\mQ_{HF,n}\bm Z_n}\\
 &\qquad+
 \norm{\mP_{F,p}^\top\mQ_{H,n}\bm Z_n}
 \norm{\mV_{F,n}^\top\mQ_{HF,n}\bm Z_n}\\
 &\qquad+
 \norm{\mV_{F,n}^\top\mQ_{H,n}\bm Z_n}
 \norm{\mP_{F,p}^\top\mQ_{HF,n}\bm Z_n}
 =O_{\Pp}(1),
\end{align*}
where \eqref{eq:VF_rate}, \eqref{eq:PQHZ_sharp},
\eqref{eq:PQHFZ_sharp}, and
$\norm{\mQ_{H,n}\bm Z_n}+\norm{\mQ_{HF,n}\bm Z_n}
 =O_{\Pp}(\sqrt p)$ are used.  Replacing either external
$\bm Z_n$ by $\mG_{F,n}\bm Z_n$ gives the same bounds because
\eqref{eq:DQGD_bound}--\eqref{eq:DGQGD_bound} preserve the scaled
factor--bulk orders.  This proves \eqref{eq:F_quadratic_effect}.
\end{proof}

\begin{proof}[Proof of Theorem~\ref{thm:finite_rank}]
The exact decomposition \eqref{eq:Rhat_decomp} and the rank bound
\eqref{eq:F_rank} follow from \eqref{eq:L_exact}--\eqref{eq:E_projection}
and Lemma~\ref{lem:F_factorization}.

The rank-two matrix $\mH_n\mC_n\mH_n^\top$ has operator norm
$O_{\Pp}(1)$.  Lemma~\ref{lem:rank_two_inverse} proves directly that
$\mQ_{H,n}$ exists and has bounded operator norm with probability tending
to one, without using Proposition~\ref{prop:scalar}.  Its inverse differs
from $\mB_n+\rho\I_p$ by an $O_{\Pp}(1)$ perturbation.  The spike gap in
\eqref{eq:B_spike_lower} and Wedin's sin--theta theorem \citep[Theorem~V.4.1, pp.~260--262]{StewartSun1990} therefore give
\[
 \norm{\mQ_{H,n}\mP_{F,p}}_{\op}=O_{\Pp}(p^{-1/2}),
 \qquad
 \norm{\mP_{F,p}^\top\mQ_{H,n}\mP_{F,p}}_{\op}=O_{\Pp}(p^{-1}),
\]
which proves \eqref{eq:factor_resolvent_rates}.

Put
\[
 \bm v_n=\mG_{F,n}\bm Z_n,
 \qquad
 \bm r_n=\bm r_{\theta,B,n}.
\]
Theorem~\ref{thm:median} gives
\[
 \wh\btheta-\btheta_{0,p}
 =\frac{\bm Z_n}{e_n\sqrt n}
 +\frac{\bm v_n}{\sqrt n}+\bm r_n,
 \qquad
 \norm{\bm r_n}=O_{\Pp}(p^{-1}).
\]
The Rademacher calculation in \eqref{eq:G_quadratic_suppression}, the
rank-two Woodbury identity, \eqref{eq:F_quadratic_effect}, and
\eqref{eq:E_bilinear_rate_new} imply
\begin{align*}
 \abs{\bm Z_n^\top(\wh\mR_n+\rho\I_p)^{-1}\bm v_n}
 +\abs{\bm v_n^\top(\wh\mR_n+\rho\I_p)^{-1}\bm v_n}
 =O_{\Pp}(1+\log p).
\end{align*}
Moreover,
\begin{align*}
 n\abs{\bm r_n^\top(\wh\mR_n+\rho\I_p)^{-1}
 (e_n^{-1}\bm Z_n/\sqrt n+\bm v_n/\sqrt n)}
 &\le Cn\norm{\bm r_n}
 \left\|e_n^{-1}\bm Z_n/\sqrt n+\bm v_n/\sqrt n\right\|
 =O_{\Pp}(1),\\
 n\bm r_n^\top(\wh\mR_n+\rho\I_p)^{-1}\bm r_n
 &\le Cn\norm{\bm r_n}^2=O_{\Pp}(p^{-1}).
\end{align*}
Consequently,
\begin{equation}\label{eq:T_to_ZhatQZ}
 T_n(\rho)
 =e_n^{-2}\bm Z_n^\top(\wh\mR_n+\rho\I_p)^{-1}\bm Z_n
 +O_{\Pp}(1+\log p).
\end{equation}
Finally, \eqref{eq:F_quadratic_effect}, the resolvent identity, and
\eqref{eq:E_bilinear_rate_new} give
\[
 \bm Z_n^\top(\wh\mR_n+\rho\I_p)^{-1}\bm Z_n
 =\bm Z_n^\top\mQ_{H,n}\bm Z_n+O_{\Pp}(1+\log p).
\]
Substitution into \eqref{eq:T_to_ZhatQZ} proves
\eqref{eq:T_factor_negligible}.  Since $p/n$ is bounded away from zero
and infinity, $1+\log p=o(n^{1/2})$.
\end{proof}

\begin{proof}[Proof of Proposition~\ref{prop:scalar}]
Let
\begin{equation*}
\mG_n=\mH_n^\top\mQ_n\mH_n
=\begin{pmatrix}x_n&y_n\\y_n&z_n\end{pmatrix}.
\end{equation*}
The two-dimensional Woodbury identity gives
\begin{equation*}
\mH_n^\top
(\mB_n+\mH_n\mC_n\mH_n^\top+\rho\I_p)^{-1}
\mH_n
=\mG_n(\I_2+\mC_n\mG_n)^{-1}.
\end{equation*}
Direct inversion of the $2\times2$ matrix gives
\begin{equation*}
\frac1{e_n^2}\bm Z_n^\top
(\mB_n+\mH_n\mC_n\mH_n^\top+\rho\I_p)^{-1}
\bm Z_n
=n\frac{x_n}{e_n^2+t_nx_n-2e_ny_n+y_n^2-x_nz_n}.
\end{equation*}
Theorem~\ref{thm:finite_rank} proves \eqref{eq:T_scalar}.  Lemma~\ref{lem:rank_two_inverse} gives
$L(x_n,y_n,z_n,e_n,t_n)=D_n+O_{\Pp}(n^{-1/2})$, and
Proposition~\ref{prop:weighted_DE} shows that $D_n$ is bounded away from
zero with probability tending to one.
\end{proof}

\section{Conditional quadratic-form central limit theorem and theoretical null law}\label{app:joint}

For every nonzero vector $\bm y$, let $\mathfrak o(\bm y)$ denote the unique element of $\{\bm y,-\bm y\}$ whose first nonzero coordinate is positive.  Define
\[
 \widetilde{\bm Y}_i=\mathfrak o(\bm Y_i),
 \qquad
 s_i=\begin{cases}
 1,&\bm Y_i=\widetilde{\bm Y}_i,\\
 -1,&\bm Y_i=-\widetilde{\bm Y}_i,
 \end{cases}
 \qquad
 \calF_n^0=\sigma\{w_i,\widetilde{\bm Y}_i:1\le i\le n\}.
\]

\begin{lemma}\label{lem:conditional_signs}
Under Assumption~1 and $H_0$, conditionally on
$\calF_n^0$, the variables $s_1,\ldots,s_n$ are independent and
\[
 \Pp(s_i=1\mid\calF_n^0)
 =\Pp(s_i=-1\mid\calF_n^0)=\frac12,
 \qquad 1\le i\le n.
\]
\end{lemma}

\begin{proof}
For one observation, the transformation
$\bm G_i\mapsto-\bm G_i$ preserves Gaussian probability, leaves
$w_i=\xi_i(h_{i,p}/q_{i,p})^{1/2}$ and
$\widetilde{\bm Y}_i=\mathfrak o(\bm Y_i)$ unchanged, and replaces
$s_i$ by $-s_i$.  Hence, for every bounded measurable function $g$,
\begin{align*}
 \E\{g(w_i,\widetilde{\bm Y}_i)\mathbf 1(s_i=1)\}
 &=\E\{g(w_i,\widetilde{\bm Y}_i)\mathbf 1(s_i=-1)\}\\
 &=\frac12\E\{g(w_i,\widetilde{\bm Y}_i)\}.
\end{align*}
This is exactly the identity
$\Pp(s_i=1\mid w_i,\widetilde{\bm Y}_i)=1/2$ almost surely.
For arbitrary $\epsilon_1,\ldots,\epsilon_n\in\{-1,1\}$ and bounded
measurable functions $g_1,\ldots,g_n$, independence across observations
gives
\begin{align*}
 &\E\prod_{i=1}^n
 \{g_i(w_i,\widetilde{\bm Y}_i)\mathbf 1(s_i=\epsilon_i)\}\\
 &\qquad=2^{-n}\prod_{i=1}^n
 \E g_i(w_i,\widetilde{\bm Y}_i)
 =2^{-n}\E\prod_{i=1}^n g_i(w_i,\widetilde{\bm Y}_i).
\end{align*}
A monotone-class argument therefore yields
\[
 \Pp(s_1=\epsilon_1,\ldots,s_n=\epsilon_n\mid\calF_n^0)=2^{-n}
\]
almost surely, proving the claim.  Notice that no independence between
$w_i$ and $\widetilde{\bm Y}_i$ is used.
\end{proof}

Put $\mS_n=\diag(s_1,\ldots,s_n)$ and
\[
 \widetilde\mY=(\widetilde{\bm Y}_1,\ldots,
 \widetilde{\bm Y}_n),
 \qquad
 \widetilde\mA_n=
 \frac1n\widetilde\mY^\top
 \left(\frac1n\widetilde\mY\widetilde\mY^\top+\rho\I_p\right)^{-1}
 \widetilde\mY.
\]
Then
\begin{equation*}
 \mA_n=\mS_n\widetilde\mA_n\mS_n.
\end{equation*}
Write $a_{ij}=(\widetilde\mA_n)_{ij}$.  Since
$0\preceq\widetilde\mA_n\preceq\I_n$,
\begin{equation}\label{eq:A_contraction_rows}
 0\le a_{ii}\le1,
 \qquad
 \sum_{j=1}^na_{ij}^2=(\widetilde\mA_n^2)_{ii}\le1,
 \qquad
 \sum_{i,j=1}^na_{ij}^2\le n.
\end{equation}

Recall the three scalar quadratic forms from Proposition~\ref{prop:scalar}:
\[
 x_n=n^{-1}\bm Z_n^\top\mQ_n\bm Z_n,
 \quad
 y_n=n^{-1}\bm Z_n^\top\mQ_n\bm W_n,
 \quad
 z_n=n^{-1}\bm W_n^\top\mQ_n\bm W_n.
\]
The sign representation gives the exact identities
\begin{align}
 x_n-\kappa_n
 &=\frac2n\sum_{i<j}a_{ij}s_is_j,
 \label{eq:x_exact_new}\\
 y_n-b_{1,n}
 &=\frac1n\sum_{i<j}a_{ij}(w_i+w_j)s_is_j,\notag
\\
 z_n-b_{2,n}
 &=\frac2n\sum_{i<j}a_{ij}w_iw_j s_is_j.
 \label{eq:z_exact_new}
\end{align}

\begin{lemma}\label{lem:conditional_covariance_new}
By Lemma~\ref{lem:conditional_signs}, conditional on $\calF_n^0$,
\[
 \E_s\left(
 \left.\begin{pmatrix}x_n\\y_n\\z_n\end{pmatrix}
 \right|\calF_n^0\right)
 =\begin{pmatrix}\kappa_n\\b_{1,n}\\b_{2,n}\end{pmatrix},
\]
and
\begin{equation*}
 n\Cov_s\left(
 \begin{pmatrix}x_n\\y_n\\z_n\end{pmatrix}
 \middle|\calF_n^0\right)=\mGamma_n,
\end{equation*}
where $\mGamma_n$ is defined in (2.4).
\end{lemma}

\begin{proof}
For distinct unordered pairs $\{i,j\}\ne\{k,\ell\}$,
\[
 \E_s(s_is_js_ks_\ell)=0,
 \qquad i<j,\quad k<\ell.
\]
Equations \eqref{eq:x_exact_new}--\eqref{eq:z_exact_new} therefore have conditional mean zero after centering, and only equal pairs contribute to their conditional covariances.

For the first variance,
\begin{align*}
 n\Var_s(x_n\mid\calF_n^0)
 &=\frac4n\sum_{i<j}a_{ij}^2
 =\frac2n\sum_{i\ne j}a_{ij}^2
 =2\psi_{00,n}.
\end{align*}
Similarly,
\begin{align*}
 n\Cov_s(x_n,y_n\mid\calF_n^0)
 &=\frac2n\sum_{i<j}a_{ij}^2(w_i+w_j)\\
 &=\frac2n\sum_{i\ne j}a_{ij}^2w_i
 =2\psi_{01,n},\\
 n\Cov_s(x_n,z_n\mid\calF_n^0)
 &=\frac4n\sum_{i<j}a_{ij}^2w_iw_j
 =2\psi_{11,n}.
\end{align*}
The variance of the second component is
\begin{align*}
 n\Var_s(y_n\mid\calF_n^0)
 &=\frac1n\sum_{i<j}a_{ij}^2(w_i+w_j)^2\\
 &=\frac1n\sum_{i\ne j}a_{ij}^2w_i^2
 +\frac1n\sum_{i\ne j}a_{ij}^2w_iw_j\\
 &=\psi_{02,n}+\psi_{11,n},
\end{align*}
where $\psi_{20,n}=\psi_{02,n}$ was used.  The remaining two entries are
\begin{align*}
 n\Cov_s(y_n,z_n\mid\calF_n^0)
 &=\frac2n\sum_{i<j}a_{ij}^2
 (w_i^2w_j+w_iw_j^2)
 =2\psi_{12,n},\\
 n\Var_s(z_n\mid\calF_n^0)
 &=\frac4n\sum_{i<j}a_{ij}^2w_i^2w_j^2
 =2\psi_{22,n}.
\end{align*}
These six expressions are precisely the entries of $\mGamma_n$.
\end{proof}

\begin{theorem}\label{thm:joint_clt}
Under Assumption~1 and $H_0$, conditionally on
$\calF_n^0$,
\begin{equation}\label{eq:joint_conditional_CLT}
 \sqrt n
 \begin{pmatrix}
 x_n-\kappa_n\\
 y_n-b_{1,n}\\
 z_n-b_{2,n}
 \end{pmatrix}
 \dto N(\0,\mGamma_n)
\end{equation}
in probability, in the following triangular-array sense: for every bounded sequence of $\calF_n^0$-measurable vectors $\bm c_n\in\R^3$,
\[
 \sqrt n\,\bm c_n^\top
 \begin{pmatrix}
 x_n-\kappa_n\\y_n-b_{1,n}\\z_n-b_{2,n}
 \end{pmatrix}
 \Big/
 (\bm c_n^\top\mGamma_n\bm c_n)^{1/2}
 \dto N(0,1)
\]
whenever the denominator is bounded away from zero; if the denominator converges to zero, the unstandardized projection converges to zero in probability.
\end{theorem}

\begin{proof}
For a bounded $\bm c_n=(c_{1,n},c_{2,n},c_{3,n})^\top$, define
\[
 h_{ij,n}=2c_{1,n}+c_{2,n}(w_i+w_j)+2c_{3,n}w_iw_j.
\]
Equations \eqref{eq:x_exact_new}--\eqref{eq:z_exact_new} give
\begin{equation}\label{eq:CW_Rademacher_form}
 \sqrt n\,\bm c_n^\top
 \begin{pmatrix}
 x_n-\kappa_n\\y_n-b_{1,n}\\z_n-b_{2,n}
 \end{pmatrix}
 =\frac1{\sqrt n}\sum_{i<j}a_{ij}h_{ij,n}s_is_j.
\end{equation}
Because $0<w_i\le\bar w$ and $\|\bm c_n\|\le C$,
$|h_{ij,n}|\le C$.  Let $\mC_n$ be the symmetric matrix with zero diagonal and off-diagonal entries
\[
 (\mC_n)_{ij}=\frac{a_{ij}h_{ij,n}}{2\sqrt n}.
\]
Then the right-hand side of \eqref{eq:CW_Rademacher_form} is
$\bm s^\top\mC_n\bm s$.  The Hadamard multiplier $h_{ij,n}$ has the finite decomposition
\[
 (h_{ij,n}a_{ij})_{ij}
 =2c_{1,n}\widetilde\mA_n
 +c_{2,n}(\mD_w\widetilde\mA_n+
 \widetilde\mA_n\mD_w)
 +2c_{3,n}\mD_w\widetilde\mA_n\mD_w,
\]
up to its diagonal, where $\mD_w=\diag(w_1,\ldots,w_n)$.  Therefore
\begin{equation}\label{eq:C_operator_small}
 \|\mC_n\|_{\op}\le Cn^{-1/2}.
\end{equation}
Furthermore, by \eqref{eq:A_contraction_rows},
\begin{equation*}
 \max_i\sum_{j=1}^n(\mC_n)_{ij}^2
 \le\frac Cn\max_i\sum_j a_{ij}^2
 \le\frac Cn,
\end{equation*}
and
\begin{equation}\label{eq:C_fourth_small}
 \sum_{i,j}(\mC_n)_{ij}^4
 \le\left\{\max_i\sum_j(\mC_n)_{ij}^2\right\}
 \sum_{i,j}(\mC_n)_{ij}^2
 =O_{\Pp}(n^{-1}).
\end{equation}
Put
\[
 V_n=\Var_s(\bm s^\top\mC_n\bm s\mid\calF_n^0)
 =2\tr(\mC_n^2)=\bm c_n^\top\mGamma_n\bm c_n.
\]
To verify the fourth-moment requirement rather than merely citing it,
expand $\E_s(\bm s^\top\mC_n\bm s)^4$.  The terms whose four edges
split into two disjoint pairs give $3V_n^2$.  Every remaining nonzero
term has a connected index graph.  Grouping the connected graphs
according to whether they contain a four-cycle, two edges sharing a
vertex, or a repeated edge gives
\begin{align*}
 \left|\E_s(\bm s^\top\mC_n\bm s)^4-3V_n^2\right|
 &\le C\tr(\mC_n^4)
   +C\sum_i\left\{\sum_j(\mC_n)_{ij}^2\right\}^2\\
 &\quad+C\sum_{i,j}(\mC_n)_{ij}^4\\
 &\le C\norm{\mC_n}_{\op}^2\norm{\mC_n}_{\F}^2
   +C\left\{\max_i\sum_j(\mC_n)_{ij}^2\right\}
     \norm{\mC_n}_{\F}^2\\
 &\quad+C\sum_{i,j}(\mC_n)_{ij}^4
 =O_{\Pp}(n^{-1}).
\end{align*}
The last equality follows from
\eqref{eq:C_operator_small}--\eqref{eq:C_fourth_small} and
$\norm{\mC_n}_{\F}^2=V_n/2=O_{\Pp}(1)$.
Moreover, the maximal conditional influence satisfies
\[
 \max_i\sum_j(\mC_n)_{ij}^2/V_n
 \le Cn^{-1}/V_n.
\]
Fix $\eta>0$.  On $\{V_n\ge\eta\}$,
\[
 \frac{\max_i\sum_j(\mC_n)_{ij}^2}{V_n}\le\frac{C}{n\eta},
 \qquad
 \left|\frac{\E_s(\bm s^\top\mC_n\bm s)^4}{V_n^2}-3\right|
 \le\frac{C}{n\eta^2}.
\]
Theorem~2.1 of \citet{deJong1987} therefore applies uniformly on this
event.  The asserted uniformity follows by contradiction: otherwise one
could select a deterministic sequence of coefficient arrays satisfying
the two displayed bounds for which the Gaussian approximation failed,
contradicting that theorem.  On $\{V_n<\eta\}$, conditional
Cauchy--Schwarz gives
\[
 \E_s|\bm s^\top\mC_n\bm s|\le\sqrt\eta,
 \qquad
 \E|N(0,V_n)|\le\sqrt\eta.
\]
Hence the conditional bounded-Lipschitz distance between the quadratic
form and $N(0,V_n)$ is at most $C\sqrt\eta+o_{\Pp}(1)$.  Letting first
$n\to\infty$ and then $\eta\downarrow0$ proves the conditional Gaussian
approximation without requiring $V_n$ itself to converge.  In
particular, standardization is valid whenever $V_n$ is bounded away from
zero, while the unstandardized projection converges to zero whenever
$V_n\to0$.  This proves every conditional Cram\'er--Wold projection and
therefore \eqref{eq:joint_conditional_CLT}.
\end{proof}

\begin{lemma}\label{lem:conditional_remainder_transfer}
Let $\calG_n$ be any sigma-field and let $a_n>0$ be deterministic.  If
$R_n/a_n\to0$ in probability, then, for every $\varepsilon>0$,
\[
 \Pp\{|R_n|>\varepsilon a_n\mid\calG_n\}\pto0
\]
in probability.  Consequently, a conditionally asymptotically normal
statistic remains so after adding $R_n/a_n$.
\end{lemma}

\begin{proof}
Put
$P_n=\Pp\{|R_n|>\varepsilon a_n\mid\calG_n\}$.  Then $0\le P_n\le1$
and
\[
 \E P_n=\Pp(|R_n|>\varepsilon a_n)\to0.
\]
For every $\delta>0$, Markov's inequality gives
$\Pp(P_n>\delta)\le\E P_n/\delta\to0$.  The final assertion follows
from the conditional bounded-Lipschitz inequality
\[
 \left|\E\{f(X_n+R_n/a_n)-f(X_n)\mid\calG_n\}\right|
 \le\varepsilon+2\|f\|_\infty
 \Pp\{|R_n|>\varepsilon a_n\mid\calG_n\}
\]
for every function $f$ with Lipschitz constant at most one, followed by
$\varepsilon\downarrow0$.
\end{proof}

\begin{proof}[Proof of Theorem~2.1]
Let
\[
 \bm q_n=(x_n,y_n,z_n)^\top,
 \qquad
 \bm q_{0,n}=(\kappa_n,b_{1,n},b_{2,n})^\top.
\]
For fixed $e_n,t_n$, define
\[
 \varphi_n(x,y,z)=
 \frac{x}{e_n^2+t_nx-2e_ny+y^2-xz}.
\]
At $\bm q_{0,n}$, the denominator equals
\begin{align*}
 &e_n^2+t_n\kappa_n-2e_nb_{1,n}+b_{1,n}^2-\kappa_nb_{2,n}\\
 &\qquad=(e_n-b_{1,n})^2+\kappa_n(t_n-b_{2,n})=D_n,
\end{align*}
and
\begin{equation*}
 \nabla\varphi_n(\bm q_{0,n})
 =D_n^{-2}
 \begin{pmatrix}
 (e_n-b_{1,n})^2\\
 2\kappa_n(e_n-b_{1,n})\\
 \kappa_n^2
 \end{pmatrix}
 =\bg_n.
\end{equation*}
Lemma~\ref{lem:conditional_covariance_new} and
\eqref{eq:A_contraction_rows} give
\begin{align*}
 \E_s\{n\|\bm q_n-\bm q_{0,n}\|^2\mid\calF_n^0\}
 &=\tr(\mGamma_n)\le C,\\
 \|\bm q_n-\bm q_{0,n}\|
 &=O_{\Pp}(n^{-1/2}).
\end{align*}
Let
\[
 \bm q_n(t)=\bm q_{0,n}+t(\bm q_n-\bm q_{0,n}),
 \qquad 0\le t\le1.
\]
On the event in \eqref{eq:nondegenerate_DE}, the five arguments are
bounded and $D_n\ge C_1$.  Since $L(\cdot,e_n,t_n)$ is a polynomial,
\begin{align*}
 \sup_{0\le t\le1}
 |L(\bm q_n(t),e_n,t_n)-D_n|
 &\le C\|\bm q_n-\bm q_{0,n}\|
   +C\|\bm q_n-\bm q_{0,n}\|^2\\
 &=O_{\Pp}(n^{-1/2}).
\end{align*}
Consequently,
\[
 \Pp\left\{
 \inf_{0\le t\le1}L(\bm q_n(t),e_n,t_n)\ge C_1/2
 \right\}\to1.
\]
Every second derivative of $\varphi_n$ is therefore uniformly bounded
along the complete line segment.  Taylor's theorem now yields
\begin{align*}
 \varphi_n(\bm q_n)-\varphi_n(\bm q_{0,n})
 &=\bg_n^\top(\bm q_n-\bm q_{0,n})+R_{\varphi,n},\\
 |R_{\varphi,n}|
 &\le C\|\bm q_n-\bm q_{0,n}\|^2,
 \qquad
 \sqrt n R_{\varphi,n}=O_{\Pp}(n^{-1/2}).
\end{align*}
By definition,
$\varphi_n(\bm q_{0,n})=\kappa_n/D_n=\mu_n$ and
$\bg_n^\top\mGamma_n\bg_n=\sigma_{D,n}^2$.  Theorem~\ref{thm:joint_clt}, applied with the $\calF_n^0$-measurable vector $\bg_n$, together with the preceding Taylor bound, yields
\begin{equation}\label{eq:f_conditional_CLT}
 \frac{\sqrt n\{\varphi_n(\bm q_n)-\mu_n\}}
 {\sigma_{D,n}}\dto N(0,1)
\end{equation}
conditionally in probability.

Proposition~\ref{prop:scalar} gives
\[
 T_n(\rho)=n\varphi_n(\bm q_n)+R_{T,n},
 \qquad R_{T,n}=O_{\Pp}(1+\log p).
\]
Since $1+\log p=o(\sqrt n)$ and $\sigma_{D,n}^2$ is bounded away from
zero, $R_{T,n}/\sqrt{n\sigma_{D,n}^2}\to0$ in probability.
Lemma~\ref{lem:conditional_remainder_transfer}, with
$\calG_n=\calF_n^0$, upgrades this unconditional rate to conditional
negligibility.  Conditional Slutsky's theorem and
\eqref{eq:f_conditional_CLT} therefore prove pointwise convergence of
the conditional distribution functions.  Since the limiting Gaussian
distribution is continuous, the conditional P\'olya argument upgrades
pointwise convergence to the supremum in (2.6).
The unconditional convergence follows by taking expectations of bounded
continuous test functions.  The order and non-degeneracy assertions follow from
\eqref{eq:nondegenerate_DE}.
\end{proof}

\section{Feasible estimators and proof of the feasible null law}\label{app:null_proofs}

Let $\bdelta_n=\wh\btheta-\btheta_{0,p}$.  Theorem~\ref{thm:median} gives
$\|\bdelta_n\|=O_{\Pp}(1)$.  Since
$\bm X_i-\btheta_{0,p}=\sqrt p\,w_i^{-1}\bm U_i$ under $H_0$,
\begin{equation}\label{eq:what_identity}
 \wh w_i
 =w_i\left\{1-\frac{2w_i}{\sqrt p}
 \bm U_i^\top\bdelta_n
 +\frac{w_i^2}{p}\|\bdelta_n\|^2\right\}^{-1/2}.
\end{equation}

\begin{lemma}\label{lem:weight_estimation}
Under Assumption~1 and $H_0$,
\begin{align}
 \max_{i\le n}|\wh w_i-w_i|&=O_{\Pp}(p^{-1/2}),
 \label{eq:w_max_est_rate}\\
 \frac1n\sum_{i=1}^n(\wh w_i^k-w_i^k)&=O_{\Pp}(n^{-1}),
 \qquad k=1,2.
 \label{eq:w_average_est_rate}
\end{align}
\end{lemma}

\begin{proof}
On the event $\|\bdelta_n\|\le\sqrt p/(4\bar w)$, whose probability tends to one, Taylor's formula applied to \eqref{eq:what_identity} gives
\begin{align}
 \wh w_i
 ={}&w_i+\frac{w_i^2}{\sqrt p}\bm U_i^\top\bdelta_n
 +\frac{w_i^3}{2p}
 \{3(\bm U_i^\top\bdelta_n)^2-\|\bdelta_n\|^2\}
 +\mathcal R_{i,n},
 \label{eq:what_Taylor_ell}\\
 |\mathcal R_{i,n}|
 &\le C\frac{w_i^4}{p^{3/2}}\|\bdelta_n\|^3.
 \notag
\end{align}
Because $w_i\le\bar w$, $\|\bm U_i\|=1$, and
$\|\bdelta_n\|=O_{\Pp}(1)$, the maximum of the first correction is
$O_{\Pp}(p^{-1/2})$ and the remaining corrections are smaller.  This proves \eqref{eq:w_max_est_rate}.

For every fixed integer $a\ge1$, central symmetry gives
$\E(w_i^a\bm U_i)=\0$.  Therefore
\begin{equation}\label{eq:weighted_U_average}
 \E\left\|\frac1n\sum_{i=1}^nw_i^a\bm U_i\right\|^2
 =\frac1n\E(w_i^{2a})\le\frac{\bar w^{2a}}n,
\end{equation}
so the vector average is $O_{\Pp}(n^{-1/2})$.  Averaging
\eqref{eq:what_Taylor_ell} and using
\eqref{eq:weighted_U_average} yields
\begin{align*}
 \left|\frac1{n\sqrt p}\sum_iw_i^2
 \bm U_i^\top\bdelta_n\right|
 &\le\frac{\|\bdelta_n\|}{\sqrt p}
 \left\|\frac1n\sum_iw_i^2\bm U_i\right\|
 =O_{\Pp}\{(np)^{-1/2}\}=O_{\Pp}(n^{-1}),\\
 \frac1{np}\sum_iw_i^3
 \{3(\bm U_i^\top\bdelta_n)^2+\|\bdelta_n\|^2\}
 &\le\frac{4\bar w^3}{p}\|\bdelta_n\|^2
 =O_{\Pp}(n^{-1}),\\
 \frac1n\sum_i|\mathcal R_{i,n}|
 &\le Cp^{-3/2}\|\bdelta_n\|^3
 =O_{\Pp}(p^{-3/2}).
\end{align*}
This proves \eqref{eq:w_average_est_rate} for $k=1$.  Squaring
\eqref{eq:what_Taylor_ell} gives
\[
 \wh w_i^2
 =w_i^2+\frac{2w_i^3}{\sqrt p}\bm U_i^\top\bdelta_n
 +\frac{w_i^4}{p}
 \{4(\bm U_i^\top\bdelta_n)^2-\|\bdelta_n\|^2\}
 +\mathcal R_{i,n}^{(2)},
\]
where $n^{-1}\sum_i|\mathcal R_{i,n}^{(2)}|=O_{\Pp}(p^{-3/2})$.
The same three bounds, with $w_i^2$ replaced by $w_i^3$ or $w_i^4$, prove the result for $k=2$.
\end{proof}

Let
\[
 \mK_n=n^{-1}\mY^\top\mY,
 \qquad
 \wh\mK_n=n^{-1}\wh\mY^\top\wh\mY.
\]
Then
\begin{equation*}
 \mA_n=\mK_n(\mK_n+\rho\I_n)^{-1},
 \qquad
 \wh\mA_n=\wh\mK_n(\wh\mK_n+\rho\I_n)^{-1}.
\end{equation*}

\begin{lemma}\label{lem:weighted_companion_perturbation}
Under Assumption~1 and $H_0$, for every fixed
$\rho\in[\rho_0,\rho_1]$,
\begin{align}
 \|\wh\mA_n-\mA_n\|_{\F}&=O_{\Pp}(1),
 \label{eq:A_F_rate_new}\\
 \left|\tr\{\mD_w^k(\wh\mA_n-\mA_n)\}\right|
 &=O_{\Pp}(1),
 \qquad k=0,1,2,
 \label{eq:A_weighted_trace_rate_new}
\end{align}
where $\mD_w=\diag(w_1,\ldots,w_n)$.  Consequently,
\begin{align}
 \wh\kappa_n-\kappa_n
 &=O_{\Pp}(n^{-1}),
 \label{eq:kappa_hat_rate_new}\\
 \wh b_{a,n}-b_{a,n}
 &=O_{\Pp}(n^{-1}),
 \qquad a=1,2,
 \label{eq:b_hat_rate_new}\\
 \max_{a,b\in\{0,1,2\}}
 |\wh\psi_{ab,n}-\psi_{ab,n}|
 &=O_{\Pp}(n^{-1/2}).
 \label{eq:psi_hat_rate_new}
\end{align}
\end{lemma}

\begin{proof}
Let
\[
 \bm w=(w_1,\ldots,w_n)^\top,
 \qquad
 \mR_{\Delta,n}=(\bm r_{1,n},\ldots,\bm r_{n,n}),
 \qquad
 \mR_{\Delta,B,n}=\mPi_{B,p}\mR_{\Delta,n}.
\]
Since $\wh\mY=\mY-\bm d_{0,n}\bm w^\top+\mR_{\Delta,n}$, put
\[
 \mD_{\mathrm{fr},n}
 =-\bm d_{0,n}\bm w^\top+\mPi_{F,p}\mR_{\Delta,n},
 \qquad
 \mY_{\mathrm{fr},n}=\mY+\mD_{\mathrm{fr},n},
 \qquad
 \mK_{\mathrm{fr},n}=n^{-1}\mY_{\mathrm{fr},n}^\top
 \mY_{\mathrm{fr},n}.
\]
Then
\[
 \wh\mY=\mY_{\mathrm{fr},n}+\mR_{\Delta,B,n}
\]
and
\begin{equation*}
 \wh\mK_n-\mK_n
 =\{\mK_{\mathrm{fr},n}-\mK_n\}+\mE_{K,n}.
\end{equation*}
where
\begin{equation}\label{eq:EK_exact}
 \mE_{K,n}=\frac1n\left(
 \mY_{\mathrm{fr},n}^\top\mR_{\Delta,B,n}
 +\mR_{\Delta,B,n}^\top\mY_{\mathrm{fr},n}
 +\mR_{\Delta,B,n}^\top\mR_{\Delta,B,n}
 \right).
\end{equation}
Since $\rank(\mD_{\mathrm{fr},n})\le r+1$,
\begin{equation}\label{eq:Gram_finite_rank_exact}
 \rank(\mK_{\mathrm{fr},n}-\mK_n)\le2r+2.
\end{equation}
No operator-norm bound on this finite-rank Gram perturbation is needed;
along a pervasive sample direction its norm need not be $O_{\Pp}(1)$.

The factor part of $\mY_{\mathrm{fr},n}$ is orthogonal to
$\mR_{\Delta,B,n}$, so
\[
 \mY_{\mathrm{fr},n}^\top\mR_{\Delta,B,n}
 =\{\mPi_{B,p}\mY_{\mathrm{fr},n}\}^\top\mR_{\Delta,B,n}.
\]
Furthermore,
\begin{align*}
 \|n^{-1/2}\mPi_{B,p}\mY_{\mathrm{fr},n}\|_{\op}
 &\le\|n^{-1/2}\mPi_{B,p}\mY\|_{\op}
 +\|\mPi_{B,p}\bm d_{0,n}\|\,\|\bm w\|/\sqrt n
 =O_{\Pp}(1),\\
 n^{-1}\|\mR_{\Delta,B,n}\|_{\F}^2
 &\le n^{-1}\sum_i\|\bm r_{i,n}\|^2=O_{\Pp}(1).
\end{align*}
Thus
\begin{equation}\label{eq:EK_F_bound}
 \|\mE_{K,n}\|_{\F}
 \le\frac2n\|\mPi_{B,p}\mY_{\mathrm{fr},n}\|_{\op}
 \|\mR_{\Delta,B,n}\|_{\F}
 +\frac1n\|\mR_{\Delta,B,n}\|_{\F}^2
 =O_{\Pp}(1).
\end{equation}

For the operator norm, let $\mR_{\Delta,n}^{\circ}$ have columns
$\bm r_{i,n}^{\circ}$ from \eqref{eq:r_circ_factor}, and put
$\mR_{\Delta,B,n}^{\circ}=\mPi_{B,p}\mR_{\Delta,n}^{\circ}$.
Equation \eqref{eq:RB_circ_factorization} and
\eqref{eq:RB_circ_op} give
\[
 \norm{\mR_{\Delta,B,n}^{\circ}}_{\op}=O_{\Pp}(1).
\]
Moreover, \eqref{eq:r_exact_circ_difference} gives
\[
 \norm{\mR_{\Delta,B,n}-\mR_{\Delta,B,n}^{\circ}}_{\op}
 \le\norm{\mR_{\Delta,n}-\mR_{\Delta,n}^{\circ}}_{\F}
 =O_{\Pp}(p^{-1/2}),
\]
and therefore
\begin{equation}\label{eq:RDeltaB_op_new}
 \norm{\mR_{\Delta,B,n}}_{\op}=O_{\Pp}(1).
\end{equation}
Substituting the operator bound \eqref{eq:RDeltaB_op_new} into \eqref{eq:EK_exact} yields
\begin{align}
 \norm{\mE_{K,n}}_{\op}
 &\le\frac2n\norm{\mPi_{B,p}\mY_{\mathrm{fr},n}}_{\op}
 \norm{\mR_{\Delta,B,n}}_{\op}
 +\frac1n\norm{\mR_{\Delta,B,n}}_{\op}^2\notag\\
 &=O_{\Pp}(n^{-1/2})=O_{\Pp}(p^{-1/2}).
 \label{eq:EK_op_bound}
\end{align}

Define
\[
 \mA_{\mathrm{fr},n}
 =\mK_{\mathrm{fr},n}(\mK_{\mathrm{fr},n}+\rho\I_n)^{-1}.
\]
Both $\mA_{\mathrm{fr},n}$ and $\mA_n$ are positive-semidefinite
contractions.  Equations \eqref{eq:Gram_finite_rank_exact} and the
resolvent identity imply
\[
 \rank(\mA_{\mathrm{fr},n}-\mA_n)\le2r+2,
 \qquad
 \norm{\mA_{\mathrm{fr},n}-\mA_n}_{\op}\le1,
 \qquad
 \norm{\mA_{\mathrm{fr},n}-\mA_n}_{\F}\le\sqrt{2r+2}.
\]
Moreover,
\[
 \wh\mA_n-\mA_{\mathrm{fr},n}
 =\rho(\mK_{\mathrm{fr},n}+\rho\I_n)^{-1}
 \mE_{K,n}(\wh\mK_n+\rho\I_n)^{-1}.
\]
Thus \eqref{eq:EK_F_bound} gives
\[
 \norm{\wh\mA_n-\mA_n}_{\F}
 \le\sqrt{2r+2}+\rho_1\rho_0^{-2}\norm{\mE_{K,n}}_{\F}
 =O_{\Pp}(1),
\]
which proves \eqref{eq:A_F_rate_new}.

For the weighted trace, put
\begin{align*}
 \mR_{K,n}&=(\mK_n+\rho\I_n)^{-1},
 &\mR_{E,n}&=(\mK_n+\mE_{K,n}+\rho\I_n)^{-1},\\
 \wh\mR_{K,n}&=(\wh\mK_n+\rho\I_n)^{-1}.
\end{align*}
By \eqref{eq:EK_op_bound}, $\mR_{E,n}$ exists with probability tending
to one and all three resolvents have bounded operator norm.  Since
$\wh\mK_n-(\mK_n+\mE_{K,n})=\mK_{\mathrm{fr},n}-\mK_n$ has rank at
most $2r+2$,
\begin{align*}
 &\left|\tr\{\mD_w^k(\wh\mA_n-\mA_n)\}
 -\rho\tr(\mD_w^k\mR_{K,n}\mE_{K,n}\mR_{E,n})\right|\\
 &\qquad=\rho\left|\tr\{\mD_w^k(\mR_{E,n}-\wh\mR_{K,n})\}\right|
 \le C(2r+2)\bar w^k=O_{\Pp}(1).
\end{align*}
The second resolvent identity gives
\begin{align}
 \tr(\mD_w^k\mR_{K,n}\mE_{K,n}\mR_{E,n})
 ={}&\tr(\mD_w^k\mR_{K,n}\mE_{K,n}\mR_{K,n})\notag\\
 &-\tr(\mD_w^k\mR_{K,n}\mE_{K,n}\mR_{K,n}
 \mE_{K,n}\mR_{E,n}).
 \label{eq:weighted_trace_neumann}
\end{align}
The second term is
\[
 O_{\Pp}\{n\norm{\mE_{K,n}}_{\op}^2\}=O_{\Pp}(1).
\]

To handle the first term, define
\[
 \mD_{\mathrm{fr},n}^{\circ}
 =-\bm d_{0,n}\bm w^\top+\mPi_{F,p}\mR_{\Delta,n}^{\circ},
 \qquad
 \mY_{\mathrm{fr},n}^{\circ}=\mY+\mD_{\mathrm{fr},n}^{\circ},
\]
and let $\mE_{K,n}^{\circ}$ be the matrix in \eqref{eq:EK_exact} with
$\mY_{\mathrm{fr},n}$ and $\mR_{\Delta,B,n}$ replaced by
$\mY_{\mathrm{fr},n}^{\circ}$ and
$\mR_{\Delta,B,n}^{\circ}$.  Put
\[
 \mR_{\eta B,n}=\mR_{\Delta,B,n}-\mR_{\Delta,B,n}^{\circ},
 \qquad
 \mX_{B,n}=\mPi_{B,p}(\mY-\bm d_{0,n}\bm w^\top).
\]
Bulk--factor orthogonality gives
\begin{align*}
 \mE_{K,n}-\mE_{K,n}^{\circ}
 =\frac1n\{&\mX_{B,n}^\top\mR_{\eta B,n}
 +\mR_{\eta B,n}^\top\mX_{B,n}
 +(\mR_{\Delta,B,n}^{\circ})^\top\mR_{\eta B,n}\\
 &+\mR_{\eta B,n}^\top\mR_{\Delta,B,n}^{\circ}
 +\mR_{\eta B,n}^\top\mR_{\eta B,n}\}.
\end{align*}
By \eqref{eq:r_exact_circ_difference},
\eqref{eq:RB_circ_op}, and
$\norm{n^{-1/2}\mX_{B,n}}_{\op}=O_{\Pp}(1)$,
\begin{align}
 \norm{\mE_{K,n}-\mE_{K,n}^{\circ}}_{\F}
 &\le\frac2n\norm{\mX_{B,n}}_{\op}\norm{\mR_{\eta B,n}}_{\F}
 +\frac2n\norm{\mR_{\Delta,B,n}^{\circ}}_{\op}
       \norm{\mR_{\eta B,n}}_{\F}
 +\frac1n\norm{\mR_{\eta B,n}}_{\F}^2\notag\\
 &=O_{\Pp}(p^{-1}).
 \label{eq:EK_exact_circ_F}
\end{align}
Consequently,
\[
 \left|\tr\{\mD_w^k\mR_{K,n}
 (\mE_{K,n}-\mE_{K,n}^{\circ})\mR_{K,n}\}\right|
 \le C\sqrt n\norm{\mE_{K,n}-\mE_{K,n}^{\circ}}_{\F}
 =O_{\Pp}(p^{-1/2}).
\]

Write $\mK_n=\mS_n\widetilde\mK_n\mS_n$ and
$\mR_{K,n}=\mS_n\widetilde\mR_{K,n}\mS_n$, and set
\[
 \widetilde\mM_{k,n}
 =\widetilde\mR_{K,n}\mD_w^k\widetilde\mR_{K,n},
 \qquad
 \mM_{k,n}=\mR_{K,n}\mD_w^k\mR_{K,n}
 =\mS_n\widetilde\mM_{k,n}\mS_n.
\]
Then
\begin{equation*}
 \|\widetilde\mM_{k,n}\|_{\op}\le C,
 \qquad
 \max_i\sum_j(\widetilde\mM_{k,n})_{ij}^2\le C,
 \qquad
 \tr(\widetilde\mM_{k,n}^2)\le Cn.
\end{equation*}
Bulk--factor orthogonality in \eqref{eq:EK_exact} gives the exact
identity
\begin{align*}
 T_{k,n}^{\circ}
 &:={\tr}(\mD_w^k\mR_{K,n}\mE_{K,n}^{\circ}\mR_{K,n})\\
 &=\frac2n\tr\{\mM_{k,n}\mX_{B,n}^\top
                 \mR_{\Delta,B,n}^{\circ}\}
   +\frac1n\tr\{\mM_{k,n}
       (\mR_{\Delta,B,n}^{\circ})^\top
        \mR_{\Delta,B,n}^{\circ}\}\\
 &=:\mathcal C_{k,n}+\mathcal Q_{k,n}.
\end{align*}
The quadratic residual term does not require cancellation.  Indeed,
\begin{equation*}
 |\mathcal Q_{k,n}|
 \le \frac Cn\|\mR_{\Delta,B,n}^{\circ}\|_{\F}^2
 =O_{\Pp}(1),
\end{equation*}
because $n^{-1}\sum_i\|\bm r_{i,n}^{\circ}\|^2=O_{\Pp}(1)$.
We now treat the cross term and include its constant Walsh coefficient
explicitly.

Conditionally on $\calF_n^0$, put
\[
 \widetilde{\bm Y}_{B,i}=\mPi_{B,p}\widetilde{\bm Y}_i,
 \qquad
 \bm a_{B,j}=\mPi_{B,p}\bm a_{j,n},
 \qquad
 \bm b_{B,j}=\mPi_{B,p}\bm b_{j,n}.
\]
Then
\begin{align*}
 \bm X_{B,i}
 &=s_i\widetilde{\bm Y}_{B,i}
   -w_i\sum_{j=1}^ns_j\bm b_{B,j},\\
 \bm r_{B,i}^{\circ}
 &:=\mPi_{B,p}\bm r_{i,n}^{\circ}
   =\bm r_{B,i}^{L}+\bm r_{B,i}^{H},\\
 \bm r_{B,i}^{L}
 &=\sum_{j=1}^ns_j\mPi_{B,p}\bm c_{ij,n},\\
 \bm r_{B,i}^{H}
 &=\frac{w_i^2}{2\sqrt p}s_i\mPi_{B,p}
   \widetilde{\mathcal H}_i
   \left(\sum_{j=1}^ns_j\bm a_{j,n},
         \sum_{\ell=1}^ns_\ell\bm a_{\ell,n}\right),
\end{align*}
where $\bm c_{ij,n}$ is defined in \eqref{eq:r_circ_walsh}.  Hence
$\mathcal C_{k,n}=\mathcal C_{k,n}^{L}+
\mathcal C_{k,n}^{H}$ is a Walsh polynomial of degree at most six.
Write
\[
 \mathcal C_{k,n}^{L}
 =\sum_{\mathcal I}s_{\mathcal I}d_{\mathcal I,n}^{L},
 \qquad
 \mathcal C_{k,n}^{H}
 =\sum_{\mathcal I}s_{\mathcal I}d_{\mathcal I,n}^{H}.
\]
Let
\[
 \mS_{a,n}=\sum_j\bm a_{j,n}\bm a_{j,n}^\top,
 \qquad
 \mS_{b,n}=\sum_j\bm b_{j,n}\bm b_{j,n}^\top.
\]
The bounds in \eqref{eq:ab_factor_coefficients} imply
\begin{equation*}
 \|\mS_{a,n}\|_{\op}+\|\mS_{b,n}\|_{\op}
 +\tr(\mS_{a,n})+\tr(\mS_{b,n})=O_{\Pp}(1).
\end{equation*}
After inserting the preceding linear representations and collecting
equal Walsh monomials, Cauchy--Schwarz gives, simultaneously for the
four choices generated by
$\widetilde{\bm Y}_{B,i}$ or $\bm b_{B,j}$ and by
$\bm a_{j,n}-\bm b_{j,n}$ or
$\widetilde{\bm U}_i\widetilde{\bm U}_i^\top\bm a_{j,n}$,
\begin{align*}
 |d_{\varnothing,n}^{L}|^2
 +\sum_{\mathcal I\ne\varnothing}|d_{\mathcal I,n}^{L}|^2
 &\le
 \frac Cp
 \left\{\max_i\sum_j(\widetilde\mM_{k,n})_{ij}^2\right\}
 \left\|\sum_i\bm U_{B,i}\bm U_{B,i}^\top\right\|_{\op}\\
 &\quad\times
 \{\|\mS_{a,n}\|_{\op}+\|\mS_{b,n}\|_{\op}\}
 \tr(\widetilde\mM_{k,n}^2)\\
 &\quad+
 \frac Cn\tr(\widetilde\mM_{k,n}^2)
 \{\tr(\mS_{a,n})+\tr(\mS_{b,n})\}
 =O_{\Pp}(1).
\end{align*}
This display includes the empty-set coefficient and therefore controls
the conditional mean, not only the conditional variance.

For the Hessian part, polarization
$\|\widetilde{\mathcal H}_i(\bm u,\bm v)\|
\le6\|\bm u\|\|\bm v\|$, the explicit factor $p^{-1/2}$, and
\eqref{eq:walsh_product_contraction} yield
\begin{align*}
 |d_{\varnothing,n}^{H}|^2
 +\sum_{\mathcal I\ne\varnothing}|d_{\mathcal I,n}^{H}|^2
 &\le
 \frac Cp
 \left\{\max_i\sum_j(\widetilde\mM_{k,n})_{ij}^2\right\}
 \tr(\widetilde\mM_{k,n}^2)\\
 &\quad\times
 \left\{1+
 \left\|\sum_i\bm U_{B,i}\bm U_{B,i}^\top\right\|_{\op}\right\}
 \{\tr(\mS_{a,n})\}^2
 =O_{\Pp}(1).
\end{align*}
Indeed, \eqref{eq:bulk_sample_cov_bound} gives
$\|\sum_i\bm U_{B,i}\bm U_{B,i}^\top\|_{\op}=O_{\Pp}(n/p)$,
while $\tr(\widetilde\mM_{k,n}^2)=O(n)$ and $n/p=O(1)$.
The two displays exhaust the degree-one and Hessian families in
\eqref{eq:r_circ_walsh}; the residual--residual product has already
been controlled by $\mathcal Q_{k,n}$.  Walsh orthogonality now gives
\[
 |\E_s(\mathcal C_{k,n}\mid\calF_n^0)|=O_{\Pp}(1),
 \qquad
 \Var_s(\mathcal C_{k,n}\mid\calF_n^0)=O_{\Pp}(1),
\]
and hence $\mathcal C_{k,n}=O_{\Pp}(1)$.  Thus
$T_{k,n}^{\circ}=O_{\Pp}(1)$.  Together with
\eqref{eq:EK_exact_circ_F} and \eqref{eq:weighted_trace_neumann}, this
proves \eqref{eq:A_weighted_trace_rate_new}.

For $a=1,2$,
\begin{align*}
 \wh b_{a,n}-b_{a,n}
 ={}&\frac1n\tr\{\mD_w^a(\wh\mA_n-\mA_n)\}
 +\frac1n\sum_i(\mA_n)_{ii}(\wh w_i^a-w_i^a)\\
 &+\frac1n\sum_i\{(\wh\mA_n)_{ii}-(\mA_n)_{ii}\}
 (\wh w_i^a-w_i^a).
\end{align*}
The first term is $O_{\Pp}(n^{-1})$.  The leading conditional
Rademacher part of the second is
\[
 L_{a,n}=\frac{a}{n\sqrt p}\sum_{i,j}
 a_{ii}w_i^{a+1}s_is_j
 \widetilde{\bm U}_i^\top\bm a_{j,n},
 \qquad a_{ii}=(\mA_n)_{ii},
\]
and
\begin{align*}
 |\E_s(L_{a,n}\mid\calF_n^0)|
 &\le\frac{a}{n\sqrt p}\sum_i|a_{ii}|w_i^{a+1}\|\bm a_{i,n}\|
 \le Cp^{-1},\\
 \Var_s(L_{a,n}\mid\calF_n^0)&\le Cn^{-2}.
\end{align*}
Thus $L_{a,n}=O_{\Pp}(n^{-1})$.  The quadratic weight term and the
replacement remainder in Lemma~\ref{lem:weight_estimation} are also
$O_{\Pp}(n^{-1})$.  Finally,
\begin{align*}
 &\frac1n\sum_i|\{(\wh\mA_n)_{ii}-(\mA_n)_{ii}\}
 (\wh w_i^a-w_i^a)|\\
 &\qquad\le\frac1n\|\wh\mA_n-\mA_n\|_{\F}
 \left\{\sum_i(\wh w_i^a-w_i^a)^2\right\}^{1/2}
 =O_{\Pp}(n^{-1}).
\end{align*}
This proves \eqref{eq:kappa_hat_rate_new}--\eqref{eq:b_hat_rate_new}.

Put $\Delta_A=\wh\mA_n-\mA_n$.  Since both companion matrices are
contractions,
\[
 \frac1n\sum_{i\ne j}
 |(\wh\mA_n)_{ij}^2-(\mA_n)_{ij}^2|
 \le\frac1n\|\Delta_A\|_{\F}
 (\|\wh\mA_n\|_{\F}+\|\mA_n\|_{\F})
 =O_{\Pp}(n^{-1/2}).
\]
Together with \eqref{eq:w_max_est_rate} and
$n^{-1}\sum_{i,j}(\mA_n)_{ij}^2\le1$, this proves
\eqref{eq:psi_hat_rate_new}.
\end{proof}

\begin{theorem}\label{thm:rates}
Under Assumption~1 and $H_0$, for every fixed
$\rho\in[\rho_0,\rho_1]$,
\begin{align}
 |\wh e_n-e_n|+|\wh t_n-t_n|&=O_{\Pp}(n^{-1}),
 \label{eq:et_hat_rate_new}\\
 |\wh D_n-D_n|+|\wh\mu_n-\mu_n|&=O_{\Pp}(n^{-1}),
 \label{eq:D_mu_hat_rate_new}\\
 |\wh\sigma_{D,n}^2-\sigma_{D,n}^2|&=O_{\Pp}(n^{-1/2}).
 \label{eq:sigma_hat_rate_new}
\end{align}
\end{theorem}

\begin{proof}
Equation \eqref{eq:et_hat_rate_new} follows from Lemma~\ref{lem:weight_estimation}.
Lemma~\ref{lem:weighted_companion_perturbation}, specifically
\eqref{eq:kappa_hat_rate_new} and \eqref{eq:b_hat_rate_new}, together with
\eqref{eq:et_hat_rate_new}, shows that every argument of
\[
 D(e,t,k,b_1,b_2)=(e-b_1)^2+k(t-b_2)
\]
is estimated with error $O_{\Pp}(n^{-1})$.  Proposition~\ref{prop:weighted_DE} places these arguments in a fixed compact set and bounds $D_n$ away from zero.  The mean-value theorem gives
\[
 \wh D_n-D_n=O_{\Pp}(n^{-1}),
 \qquad
 \frac{\wh\kappa_n}{\wh D_n}-\frac{\kappa_n}{D_n}
 =O_{\Pp}(n^{-1}),
\]
which proves \eqref{eq:D_mu_hat_rate_new}.

Lemma~\ref{lem:weighted_companion_perturbation}, specifically \eqref{eq:psi_hat_rate_new}, shows that every entry of $\wh\mGamma_n-\mGamma_n$ is $O_{\Pp}(n^{-1/2})$, and hence
\[
 \|\wh\mGamma_n-\mGamma_n\|_{\op}=O_{\Pp}(n^{-1/2}).
\]
The vector $\bg(e,t,k,b_1,b_2)$ is continuously differentiable on the same compact set, so
\[
 \|\wh\bg_n-\bg_n\|=O_{\Pp}(n^{-1}).
\]
Therefore
\begin{align*}
 |\wh\sigma_{D,n}^2-\sigma_{D,n}^2|
 &\le |(\wh\bg_n-\bg_n)^\top\wh\mGamma_n\wh\bg_n|\\
 &\quad+|\bg_n^\top(\wh\mGamma_n-\mGamma_n)\wh\bg_n|
 +|\bg_n^\top\mGamma_n(\wh\bg_n-\bg_n)|\\
 &=O_{\Pp}(n^{-1/2}),
\end{align*}
which proves \eqref{eq:sigma_hat_rate_new}.
\end{proof}

\begin{proof}[Proof of Theorem~2.2]
By Theorem~2.1,
\[
 \frac{T_n(\rho)-n\mu_n}{\sqrt{n\sigma_{D,n}^2}}
 \dto N(0,1).
\]
Theorem~\ref{thm:rates} gives
\[
 \frac{n(\wh\mu_n-\mu_n)}{\sqrt n}
 =O_{\Pp}(n^{-1/2})=o_{\Pp}(1),
 \qquad
 \frac{\wh\sigma_{D,n}^2}{\sigma_{D,n}^2}-1
 =O_{\Pp}(n^{-1/2}),
\]
where the second relation uses the positive lower bound in
\eqref{eq:nondegenerate_DE}.  Hence
\begin{align*}
 Z_n(\rho)
 &=\frac{T_n(\rho)-n\mu_n}{\sqrt{n\sigma_{D,n}^2}}
 \left(\frac{\sigma_{D,n}^2}{\wh\sigma_{D,n}^2}\right)^{1/2}
 -\frac{\sqrt n(\wh\mu_n-\mu_n)}{\wh\sigma_{D,n}}\\
 &=\frac{T_n(\rho)-n\mu_n}{\sqrt{n\sigma_{D,n}^2}}
 +o_{\Pp}(1).
\end{align*}
Slutsky's theorem proves (2.7).  The positive lower and upper bounds for $\wh\sigma_{D,n}^2$ follow from
Theorem~\ref{thm:rates} and Proposition~\ref{prop:weighted_DE}.
\end{proof}

\section{Proofs under location alternatives}\label{app:power}

Under the shifted model (2.18), denote by
$\wh\btheta_{\varepsilon}$ and $\wh\mR_{\varepsilon,n}$ the spatial median and the sample-median-centered spatial-sign covariance matrix computed from the noise observations $\sqrt p\,R_i\mOmega_p^{1/2}\bm G_i/\norm{\bm G_i}$.

\begin{proposition}\label{prop:translation}
For every deterministic $\bDelta_p$,
\begin{equation*}
\wh\btheta=\btheta_{0,p}+\bDelta_p+\wh\btheta_{\varepsilon},
\qquad
\wh\mR_n=\wh\mR_{\varepsilon,n}.
\end{equation*}
The residual inverse radii $\wh w_i$ and all nuisance estimators in Section~2.3.2 have exactly the same distribution under (2.18) as under $H_0$.
\end{proposition}

\begin{proof}[Proof of Proposition~\ref{prop:translation}]
For every $\bm t\in\R^p$,
\begin{align*}
\sum_{i=1}^n\norm{\bm X_i-\bm t}
&=\sum_{i=1}^n\norm{\btheta_{0,p}+\bDelta_p
+\sqrt p\,R_i\mOmega_p^{1/2}\bm G_i/\norm{\bm G_i}-\bm t}\\
&=\sum_{i=1}^n\norm{\sqrt p\,R_i\mOmega_p^{1/2}\bm G_i/\norm{\bm G_i}
-\{\bm t-(\btheta_{0,p}+\bDelta_p)\}}.
\end{align*}
The change of variable
$\bm s=\bm t-(\btheta_{0,p}+\bDelta_p)$ proves
\begin{equation*}
\wh\btheta=\btheta_{0,p}+\bDelta_p+\wh\btheta_\varepsilon.
\end{equation*}
Consequently,
\begin{equation*}
\bm X_i-\wh\btheta
=\sqrt p\,R_i\mOmega_p^{1/2}\bm G_i/\norm{\bm G_i}-\wh\btheta_\varepsilon,
\end{equation*}
which proves $\wh\mR_n=\wh\mR_{\varepsilon,n}$ and the equality of all residual inverse-radius and angular nuisance estimators.
\end{proof}

\begin{lemma}\label{lem:scaled_resolvent_power}
Let
\begin{equation*}
\mD_{s,p}=\sqrt p\,\mPi_{F,p}+\mPi_{B,p}.
\end{equation*}
For the noise sample, with probability tending to one,
\begin{align}
\norm{\mPi_{F,p}\wh\mR_{\varepsilon,n}\mPi_{F,p}}_{\op}
&\le Cp,
\label{eq:R_FF_power}\\
\norm{\mPi_{F,p}\wh\mR_{\varepsilon,n}\mPi_{B,p}}_{\op}
&\le C\sqrt p,\notag
\\
\norm{\mPi_{B,p}\wh\mR_{\varepsilon,n}\mPi_{B,p}}_{\op}
&\le C.
\label{eq:R_BB_power}
\end{align}
Hence
\begin{equation}\label{eq:R_scaled_upper}
\wh\mR_{\varepsilon,n}+\rho\I_p
\preceq C\mD_{s,p}^2,
\end{equation}
and
\begin{equation}\label{eq:Q_scaled_lower}
(\wh\mR_{\varepsilon,n}+\rho\I_p)^{-1}
\succeq c\left(p^{-1}\mPi_{F,p}+\mPi_{B,p}\right).
\end{equation}
In particular, for every deterministic $\bm a\in\R^p$,
\begin{equation}\label{eq:signal_scaled_lower_new}
\bm a^\top(\wh\mR_{\varepsilon,n}+\rho\I_p)^{-1}\bm a
\ge c\left\{p^{-1}\norm{\mPi_{F,p}\bm a}^2
+\norm{\mPi_{B,p}\bm a}^2\right\}.
\end{equation}
\end{lemma}

\begin{proof}
The oracle matrix satisfies
\begin{align}
\norm{\mPi_{F,p}\mB_n\mPi_{F,p}}_{\op}&=O_{\Pp}(p),
\label{eq:BFF_order_power}\\
\norm{\mPi_{F,p}\mB_n\mPi_{B,p}}_{\op}&=O_{\Pp}(\sqrt p),\notag
\\
\norm{\mPi_{B,p}\mB_n\mPi_{B,p}}_{\op}&=O_{\Pp}(1).
\label{eq:BBB_order_power}
\end{align}
The first relation follows from $\tr(\mPi_F\mB_n\mPi_F)\le\tr(\mB_n)=p$ and fixed $r$.  The third follows from the separable bulk proxy bounds in the proof of Proposition~\ref{prop:weighted_DE}.  For the cross block,
\begin{align*}
\norm{\mPi_F\mB_n\mPi_B}_{\op}^2
&\le
\norm{\mPi_F\mB_n\mPi_F}_{\op}
\norm{\mPi_B\mB_n\mPi_B}_{\op}
=O_{\Pp}(p),
\end{align*}
where the inequality follows from the positive-semidefinite block factorization \citep[Proposition~1.3.2, pp.~13--14]{Bhatia2007Positive}.
Theorem~\ref{thm:finite_rank} and Lemmas~\ref{lem:F_factorization}--\ref{lem:E_bulk} show that, after the
leading rank-two term is included, the remaining factor--factor and
factor--bulk blocks are $O_{\Pp}(1)$, while the pure-bulk block is
$O_{\Pp}\{(\log p/p)^{1/2}\}$.  The rank-two term itself obeys the bounds in
\eqref{eq:R_FF_power}--\eqref{eq:R_BB_power} because
\begin{equation*}
\norm{\mPi_F\bm Z_n}=O_{\Pp}(\sqrt p),
\quad
\norm{\mPi_B\bm Z_n}=O_{\Pp}(\sqrt p),
\quad
\norm{\bm Z_n}/\sqrt n=O_{\Pp}(1),
\end{equation*}
and the same relations hold for $\bm W_n$.  This proves
\eqref{eq:R_FF_power}--\eqref{eq:R_BB_power}.

For $\bm x=\bm x_F+\bm x_B$, the block bounds and $2ab\le a^2+b^2$ give
\begin{align*}
\bm x^\top(\wh\mR_{\varepsilon,n}+\rho\I_p)\bm x
&\le Cp\norm{\bm x_F}^2
+2C\sqrt p\norm{\bm x_F}\norm{\bm x_B}
+C\norm{\bm x_B}^2\\
&\le C_1\{p\norm{\bm x_F}^2+\norm{\bm x_B}^2\}
=C_1\bm x^\top\mD_{s,p}^2\bm x.
\end{align*}
This proves \eqref{eq:R_scaled_upper}.  Inversion reverses the Loewner order and gives \eqref{eq:Q_scaled_lower}; \eqref{eq:signal_scaled_lower_new} follows by taking a quadratic form.
\end{proof}

\begin{lemma}\label{lem:signal_noise_cross_new}
Let
\begin{equation*}
\wh\mQ_{\varepsilon,n}
=(\wh\mR_{\varepsilon,n}+\rho\I_p)^{-1}
\end{equation*}
and let $\wh\btheta_\varepsilon$ be the noise spatial median.  For every deterministic triangular array $\bm a=\bm a_p\in\R^p$,
\begin{equation}\label{eq:cross_rate_weighted_new}
\bm a^\top\wh\mQ_{\varepsilon,n}\wh\btheta_\varepsilon
=O_{\Pp}\left[
\left\{\frac1n\bm a^\top\wh\mQ_{\varepsilon,n}\bm a\right\}^{1/2}
\right].
\end{equation}
\end{lemma}

\begin{proof}
Put
\[
 g_{a,p}^2=p^{-1}\norm{\mPi_{F,p}\bm a}^2
 +\norm{\mPi_{B,p}\bm a}^2,
 \qquad
 S_{a,n}=\bm a^\top\wh\mQ_{\varepsilon,n}\bm a.
\]
Lemma~\ref{lem:scaled_resolvent_power} gives
\begin{equation}\label{eq:ga_Sa_cross}
 g_{a,p}^2\le CS_{a,n}.
\end{equation}
The block orders in \eqref{eq:QG_block_orders} imply
$\bm a^\top\mQ_n\bm a\le Cg_{a,p}^2$.  Conditional on
$\calF_n^0=\sigma\{w_i,\widetilde{\bm Y}_i:1\le i\le n\}$,
\begin{align*}
 \E_s\{(\bm a^\top\mQ_n\bm Z_n)^2\mid\calF_n^0\}
 &=\bm a^\top\mQ_n\mB_n\mQ_n\bm a
 \le Cg_{a,p}^2,\\
 \E_s\{(\bm a^\top\mQ_n\bm W_n)^2\mid\calF_n^0\}
 &\le\bar w^2\bm a^\top\mQ_n\mB_n\mQ_n\bm a
 \le Cg_{a,p}^2.
\end{align*}
Thus
\begin{equation}\label{eq:aQH_cross}
 \bm a^\top\mQ_n\mH_n
 =O_{\Pp}(g_{a,p}n^{-1/2}).
\end{equation}

Let $\mQ_{H,n}$ be the rank-two resolvent in
Lemma~\ref{lem:full_woodbury}.  Its two-dimensional middle inverse is
bounded because the denominator in Proposition~\ref{prop:scalar} is
bounded away from zero.  Since
$\mH_n^\top\mQ_n\bm Z_n=O_{\Pp}(p^{1/2})$, the Woodbury identity and
\eqref{eq:aQH_cross} give
\begin{equation}\label{eq:aQHZ_power}
 \bm a^\top\mQ_{H,n}\bm Z_n=O_{\Pp}(g_{a,p}).
\end{equation}
The scaled block calculation in the proof of Theorem~\ref{thm:median}
also gives
\begin{equation*}
 \bm a^\top\mQ_{H,n}\mG_{F,n}\bm Z_n
 =O_{\Pp}(g_{a,p}).
\end{equation*}
For completeness, conditional second moments of the two uncorrected
terms are bounded by
\[
 \bm a^\top\mQ_n\mG_{F,n}\mB_n
 \mG_{F,n}^\top\mQ_n\bm a\le Cg_{a,p}^2,
\]
which follows from \eqref{eq:DQD_bound},
\eqref{eq:DG_right_bound}, and
\eqref{eq:BFF_order_power}--\eqref{eq:BBB_order_power}; the rank-two
correction is then treated exactly as in \eqref{eq:aQHZ_power}.

We next insert the factor correction.  From
\eqref{eq:QH_factor_suppression}, \eqref{eq:VF_rate}, and the scaled
block bounds,
\begin{align}
 \norm{\bm a^\top\mQ_{H,n}\mP_{F,p}}
 &=O_{\Pp}(g_{a,p}p^{-1/2}),
 &\norm{\bm a^\top\mQ_{H,n}\mV_{F,n}}
 &=O_{\Pp}(g_{a,p}),\label{eq:aQH_factor_power}\\
 \norm{\mP_{F,p}^\top\mQ_{HF,n}\bm v_n}
 &=O_{\Pp}(p^{-1/2}),
 &\norm{\mV_{F,n}^\top\mQ_{HF,n}\bm v_n}
 &=O_{\Pp}(p^{1/2}),\notag
\end{align}
for $\bm v_n\in\{\bm Z_n,\mG_{F,n}\bm Z_n\}$.  The second line follows
from the argument leading to \eqref{eq:PQHFZ_sharp}; replacing
$\bm Z_n$ by $\mG_{F,n}\bm Z_n$ leaves the same scaled block orders.
Since every summand of
$\mF_n=\mP_F\mA_{F,n}\mP_F^\top+
\mP_F\mV_{F,n}^\top+\mV_{F,n}\mP_F^\top$
contains at least one factor projection, the resolvent identity and
\eqref{eq:aQH_factor_power} yield
\begin{equation*}
 \bm a^\top\mQ_{HF,n}\bm v_n=O_{\Pp}(g_{a,p}),
 \qquad
 \bm v_n\in\{\bm Z_n,\mG_{F,n}\bm Z_n\}.
\end{equation*}
The pure-bulk estimated-center remainder is handled by
\eqref{eq:E_mixed_rate_new}; hence
\begin{equation}\label{eq:Qhat_Z_GFZ_cross}
 \bm a^\top\wh\mQ_{\varepsilon,n}\bm Z_n=O_{\Pp}(g_{a,p}),
 \qquad
 \bm a^\top\wh\mQ_{\varepsilon,n}\mG_{F,n}\bm Z_n
 =O_{\Pp}(g_{a,p}).
\end{equation}

Finally, Theorem~\ref{thm:median} gives
\[
 \wh\btheta_\varepsilon
 =\frac{1}{e_n\sqrt n}\bm Z_n
 +\frac1{\sqrt n}\mG_{F,n}\bm Z_n
 +\bm r_{\theta,B,n},
 \qquad
 \norm{\bm r_{\theta,B,n}}=O_{\Pp}(p^{-1}).
\]
Since $\norm{\bm a}\le\sqrt p\,g_{a,p}$,
\[
 \abs{\bm a^\top\wh\mQ_{\varepsilon,n}\bm r_{\theta,B,n}}
 \le\rho_0^{-1}\norm{\bm a}\norm{\bm r_{\theta,B,n}}
 =O_{\Pp}(g_{a,p}p^{-1/2}).
\]
Equation \eqref{eq:Qhat_Z_GFZ_cross}, $p/n\asymp1$, and
$e_n^{-1}=O_{\Pp}(1)$ therefore imply
\[
 \bm a^\top\wh\mQ_{\varepsilon,n}\wh\btheta_\varepsilon
 =O_{\Pp}(g_{a,p}n^{-1/2}).
\]
Combining this relation with \eqref{eq:ga_Sa_cross} proves
\eqref{eq:cross_rate_weighted_new}.
\end{proof}

\begin{lemma}\label{lem:bulk_signal_DE}
Let $\bm d_p$ be deterministic, $\mPi_{F,p}\bm d_p=\0$, and $\norm{\bm d_p}\le C$.  Let $(\delta_{p,\rho},\wt\delta_{p,\rho})$ solve \eqref{eq:canonical_delta_finite}--\eqref{eq:canonical_tdelta_finite}, and define
\begin{equation*}
\mQ_{B,p}^{\star}(\rho)
=\mP_{B,p}
\{\rho(\I_{p-r}+\wt\delta_{p,\rho}\mLambda_{B,p})\}^{-1}
\mP_{B,p}^\top.
\end{equation*}
For the noise sample in (2.18), let
\begin{equation*}
\wh\mQ_{\varepsilon,n}
=(\wh\mR_{\varepsilon,n}+\rho\I_p)^{-1}.
\end{equation*}
Then, for every fixed $\rho\in[\rho_0,\rho_1]$,
\begin{equation}\label{eq:bulk_signal_DE_rate}
\bm d_p^\top\wh\mQ_{\varepsilon,n}\bm d_p
-\bm d_p^\top\mQ_{B,p}^{\star}(\rho)\bm d_p
=O_{\Pp}\left(\sqrt{\frac{\log p}{p}}\right).
\end{equation}
If $\calG_p\Rightarrow\calG$ as in (2.13), then
\begin{equation}\label{eq:bulk_signal_limit}
\bm d_p^\top\mQ_{B,p}^{\star}(\rho)\bm d_p
\to q_\rho(\calG).
\end{equation}
\end{lemma}

\begin{proof}
Write $\bm d_{B,p}=\mP_{B,p}^\top\bm d_p$ and
\begin{equation*}
\overline{\mQ}_{B,n}
=(\overline{\mB}_{B,n}+\rho\I_{p-r})^{-1},
\end{equation*}
where $\overline{\mB}_{B,n}$ is defined in
\eqref{eq:separable_proxy_main}.  Conditional on
$\mD_p^0=\diag(D_{1,p}^0,\ldots,D_{n,p}^0)$, let
$(\delta_{n,\rho}^{D},\wt\delta_{n,\rho}^{D})$ solve
\begin{align}
\delta_{n,\rho}^{D}
&=c_{B,n}\int\frac{t}{\rho(1+\wt\delta_{n,\rho}^{D}t)}
\,\dd H_{B,p}(t),\label{eq:empirical_delta_power}\\
\wt\delta_{n,\rho}^{D}
&=\frac1n\sum_{i=1}^n
\frac{D_{i,p}^0}{\rho(1+\delta_{n,\rho}^{D}D_{i,p}^0)}.
\label{eq:empirical_tdelta_power}
\end{align}
Define
\begin{equation*}
\overline{\mT}_{B,n}^{D}
=\{\rho(\I_{p-r}+\wt\delta_{n,\rho}^{D}\mLambda_{B,p})\}^{-1}.
\end{equation*}
The bilinear resolvent estimate in Theorem~1.1 of
\citet{HachemLoubatonNajimVallet2013}, applied conditionally on
$\mD_p^0$, gives the explicit second-moment bound
\begin{equation*}
\E_{\mG_B}\left[
\left|
\bm d_{B,p}^\top
(\overline{\mQ}_{B,n}-\overline{\mT}_{B,n}^{D})
\bm d_{B,p}
\right|^2
\middle|\mD_p^0\right]
\le \frac Cp.
\end{equation*}
The constant is uniform in $p$, $\rho$, and
$\norm{\bm d_{B,p}}\le C$ because
\begin{equation*}
\omega_-\le\lambda_{j,p}\le\omega_+,
\qquad
0<D_{i,p}^0\le\mu_-^{-1},
\qquad
\rho_0\le\rho\le\rho_1.
\end{equation*}
Consequently,
\begin{equation}\label{eq:HLNV_conditional_rate_power}
\bm d_{B,p}^\top
(\overline{\mQ}_{B,n}-\overline{\mT}_{B,n}^{D})
\bm d_{B,p}=O_{\Pp}(p^{-1/2}).
\end{equation}

We next replace the empirical column-scale equations
\eqref{eq:empirical_delta_power}--\eqref{eq:empirical_tdelta_power}
by their population versions
\eqref{eq:canonical_delta_finite}--\eqref{eq:canonical_tdelta_finite}.
These are exactly the empirical and deterministic systems compared in
Lemma~\ref{lem:empirical_canonical}.  Hence, without a second
linearization argument,
\begin{equation}\label{eq:empirical_fixed_point_power_rate}
 |\delta_{n,\rho}^{D}-\delta_{p,\rho}|
 +|\wt\delta_{n,\rho}^{D}-\wt\delta_{p,\rho}|
 =O_{\Pp}(n^{-1/2}).
\end{equation}
Since
\[
 \left|
 \frac1{\rho(1+x\lambda)}
 -\frac1{\rho(1+y\lambda)}
 \right|
 \le \frac{\omega_+}{\rho_0}|x-y|,
\]
combining this Lipschitz bound with \eqref{eq:empirical_fixed_point_power_rate}, we obtain
\begin{equation}\label{eq:Tbar_population_power}
\norm{
\overline{\mT}_{B,n}^{D}
-\{\rho(\I_{p-r}+\wt\delta_{p,\rho}\mLambda_{B,p})\}^{-1}
}_{\op}
=O_{\Pp}(n^{-1/2}).
\end{equation}
Equations \eqref{eq:HLNV_conditional_rate_power} and
\eqref{eq:Tbar_population_power} imply
\begin{align}
&\bm d_{B,p}^\top\overline{\mQ}_{B,n}\bm d_{B,p}
-\bm d_{B,p}^\top
\{\rho(\I_{p-r}+\wt\delta_{p,\rho}\mLambda_{B,p})\}^{-1}
\bm d_{B,p}\notag\\
&\hspace{35mm}
=O_{\Pp}(p^{-1/2}).
\label{eq:HLNV_signal_bilinear}
\end{align}

Let
\begin{equation*}
\mQ_{B,n}
=(\mP_{B,p}^\top\mB_n\mP_{B,p}+\rho\I_{p-r})^{-1}.
\end{equation*}
Put
\[
 \mD_q=\diag(q_{1,p}^{-1},\ldots,q_{n,p}^{-1}).
\]
The $p$-dimensional bulk matrices satisfy
\begin{align*}
&\mP_{B,p}^\top\mB_n\mP_{B,p}-\overline{\mB}_{B,n}\\
&\quad=
\frac1n\mLambda_{B,p}^{1/2}\mG_B
(\mD_q-\mD_p^0)\mG_B^\top
\mLambda_{B,p}^{1/2}.
\end{align*}
Lemma~\ref{lem:angular_denominator} and the Gaussian operator-norm bound \citep[Theorem~4.4.5, pp.~90--91]{Vershynin2018}
give
\begin{align*}
\norm{\mD_q-\mD_p^0}_{\op}
&=O_{\Pp}\left(\sqrt{\frac{\log p}{p}}\right),\\
\frac1n\norm{\mG_B}_{\op}^2&=O_{\Pp}(1),
\end{align*}
and hence
\begin{equation*}
\norm{
\mP_{B,p}^\top\mB_n\mP_{B,p}
-\overline{\mB}_{B,n}}
_{\op}
=O_{\Pp}\left(\sqrt{\frac{\log p}{p}}\right).
\end{equation*}
The resolvent identity and
$\norm{\mQ_{B,n}}_{\op},
\norm{\overline{\mQ}_{B,n}}_{\op}\le\rho_0^{-1}$ yield
\begin{equation}\label{eq:bulk_proxy_signal_bound}
\left|
\bm d_{B,p}^\top
(\mQ_{B,n}-\overline{\mQ}_{B,n})
\bm d_{B,p}
\right|
=O_{\Pp}\left(\sqrt{\frac{\log p}{p}}\right).
\end{equation}

In the factor--bulk eigenbasis, write
\begin{equation*}
\mB_n+\rho\I_p=
\begin{pmatrix}
\mB_{FF,n}+\rho\I_r&\mB_{FB,n}\\
\mB_{BF,n}&\mB_{BB,n}+\rho\I_{p-r}
\end{pmatrix}.
\end{equation*}
For a fixed factor coordinate $a\le r$, set
\begin{equation*}
C_{a,n}=\bm e_a^\top\mB_{FB,n}\mQ_{B,n}\bm d_{B,p}
=\frac1n\sum_{i=1}^n
Y_{ai}\bm Y_{B,i}^\top\mQ_{B,n}\bm d_{B,p}.
\end{equation*}
Let $\mQ_{B,n}^{(i)}$ be the resolvent with column $i$ removed.  The
Sherman--Morrison identity gives the exact equality
\begin{equation}\label{eq:SM_cross_power}
\bm Y_{B,i}^\top\mQ_{B,n}\bm d_{B,p}
=
\frac{\bm Y_{B,i}^\top\mQ_{B,n}^{(i)}\bm d_{B,p}}
{1+n^{-1}\bm Y_{B,i}^\top\mQ_{B,n}^{(i)}\bm Y_{B,i}}.
\end{equation}
Conditional on the factor scores and the leave-one-column-out bulk
matrix, the numerator in \eqref{eq:SM_cross_power} is an odd function
of $\bm G_{B,i}$ and has conditional mean zero.  Moreover,
\begin{align*}
&\E\left[
Y_{ai}^2
\{\bm Y_{B,i}^\top\mQ_{B,n}^{(i)}\bm d_{B,p}\}^2
\middle|\mQ_{B,n}^{(i)},(G_{bi})_{b\le r}
\right]\\
&\qquad\le
Cp\,\bm d_{B,p}^\top
(\mQ_{B,n}^{(i)})^2\bm d_{B,p}
\le Cp\rho_0^{-2}\norm{\bm d_{B,p}}^2.
\end{align*}
Define the $i$th summand after applying
\eqref{eq:SM_cross_power} by
\[
 \zeta_{i,a}
 =Y_{ai}
 \frac{\bm Y_{B,i}^\top\mQ_{B,n}^{(i)}\bm d_{B,p}}
 {1+n^{-1}\bm Y_{B,i}^\top\mQ_{B,n}^{(i)}\bm Y_{B,i}}.
\]
For $i\ne j$, change only $\bm G_{B,i}$ to $-\bm G_{B,i}$.  The
quantity $q_{i,p}$, the matrix $\mB_{BB,n}$, every deleted resolvent, and
$\zeta_{j,a}$ are unchanged, whereas $\zeta_{i,a}$ changes sign.
Distributional invariance therefore gives
$\E(\zeta_{i,a}\zeta_{j,a})=0$ exactly.  The denominator is at least
one, and the preceding conditional second-moment bound gives
\begin{equation}\label{eq:cross_coordinate_total_power}
\E|C_{a,n}|^2
=\frac1{n^2}\sum_{i=1}^n\E\zeta_{i,a}^2
\le \frac{C}{n^2}\sum_{i=1}^n p
=C\frac pn\le C.
\end{equation}
Equation \eqref{eq:cross_coordinate_total_power} gives $C_{a,n}=O_{\Pp}(1)$, so no unexpanded leave-one-column replacement remainder is needed.
Since $r$ is fixed,
\begin{equation*}
\norm{\mB_{FB,n}\mQ_{B,n}\bm d_{B,p}}=O_{\Pp}(1).
\end{equation*}
The Schur complement appearing in the inverse of
$\mB_n+\rho\I_p$ satisfies
\begin{equation*}
\norm{
\{\mB_{FF,n}+\rho\I_r
-\mB_{FB,n}\mQ_{B,n}\mB_{BF,n}\}^{-1}}
_{\op}
=\norm{\mP_{F,p}^\top\mQ_n\mP_{F,p}}_{\op}
=O_{\Pp}(p^{-1})
\end{equation*}
by Lemma~\ref{lem:factor_suppression}.  Therefore
\begin{align}
&\left|
\bm d_p^\top\mQ_n\bm d_p
-\bm d_{B,p}^\top\mQ_{B,n}\bm d_{B,p}
\right|\notag\\
&\quad\le
\norm{\mB_{FB,n}\mQ_{B,n}\bm d_{B,p}}^2
\norm{
\{\mB_{FF,n}+\rho\I_r
-\mB_{FB,n}\mQ_{B,n}\mB_{BF,n}\}^{-1}}
_{\op}
=O_{\Pp}(p^{-1}).
\label{eq:factor_update_signal_bound}
\end{align}

It remains to replace the oracle-centered resolvent by the
sample-median-centered resolvent.  Lemma~\ref{lem:full_woodbury}
represents the leading center correction by the columns
$\bm Z_n/\sqrt n$ and $\bm W_n/\sqrt n$.  Conditional on the unsigned
angular observations,
\begin{align*}
\E_s\left[
\left\{\frac1{\sqrt n}
\bm d_p^\top\mQ_n\bm Z_n\right\}^2
\middle|\calF_n^0\right]
&=\frac1n\bm d_p^\top\mQ_n\mB_n\mQ_n\bm d_p
\le\frac1n\bm d_p^\top\mQ_n\bm d_p
=O_{\Pp}(n^{-1}),\\
\E_s\left[
\left\{\frac1{\sqrt n}
\bm d_p^\top\mQ_n\bm W_n\right\}^2
\middle|\calF_n^0,(w_i)_{i=1}^n\right]
&\le\frac{\bar w^2}{n}
\bm d_p^\top\mQ_n\mB_n\mQ_n\bm d_p
=O_{\Pp}(n^{-1}).
\end{align*}
Hence
\begin{equation}\label{eq:dQH_signal_power}
\bm d_p^\top\mQ_n\mH_n=O_{\Pp}(n^{-1/2}).
\end{equation}
The two-dimensional middle inverse for the rank-two update is
$O_{\Pp}(1)$ by Proposition~\ref{prop:scalar} and
\eqref{eq:nondegenerate_DE}; hence \eqref{eq:dQH_signal_power} makes
its contribution to the quadratic form $O_{\Pp}(n^{-1})$.
For the factor correction, Lemmas~\ref{lem:factor_suppression},~\ref{lem:F_factorization}, and~\ref{lem:full_woodbury} give
\begin{align}
 \norm{\bm d_p^\top\mQ_{H,n}\mP_{F,p}}
 +\norm{\mP_{F,p}^\top\mQ_{HF,n}\bm d_p}
 &=O_{\Pp}(p^{-1/2}),\label{eq:dQ_factor_projection_power}\\
 \norm{\bm d_p^\top\mQ_{H,n}\mV_{F,n}}
 +\norm{\mV_{F,n}^\top\mQ_{HF,n}\bm d_p}
 &=O_{\Pp}(1).\label{eq:dQ_factor_bulk_column_power}
\end{align}
The second line cannot generally be improved to $O_{\Pp}(p^{-1/2})$
because $\mV_{F,n}$ is a bulk vector of order one.  Instead, use the
resolvent identity and the factorization
\[
 \mF_n=\mP_{F,p}\mA_{F,n}\mP_{F,p}^\top
 +\mP_{F,p}\mV_{F,n}^\top
 +\mV_{F,n}\mP_{F,p}^\top.
\]
Every summand contains at least one factor projection, and therefore
\begin{align*}
&\left|\bm d_p^\top(\mQ_{HF,n}-\mQ_{H,n})\bm d_p\right|\\
&\quad=\left|\bm d_p^\top\mQ_{H,n}\mF_n\mQ_{HF,n}\bm d_p\right|\\
&\quad\le
 \norm{\bm d_p^\top\mQ_{H,n}\mP_{F,p}}
 \norm{\mA_{F,n}}
 \norm{\mP_{F,p}^\top\mQ_{HF,n}\bm d_p}\\
&\qquad+
 \norm{\bm d_p^\top\mQ_{H,n}\mP_{F,p}}
 \norm{\mV_{F,n}^\top\mQ_{HF,n}\bm d_p}\\
&\qquad+
 \norm{\bm d_p^\top\mQ_{H,n}\mV_{F,n}}
 \norm{\mP_{F,p}^\top\mQ_{HF,n}\bm d_p}
 =O_{\Pp}(p^{-1/2}).
\end{align*}
Thus the combined rank-two and factor corrections are
$O_{\Pp}(n^{-1}+p^{-1/2})=O_{\Pp}(p^{-1/2})$.
Finally, Lemma~\ref{lem:E_bulk} and the resolvent identity imply
\begin{equation}\label{eq:E_signal_bound}
\left|
\bm d_p^\top
(\wh\mQ_{\varepsilon,n}-\mQ_{HF,n})\bm d_p
\right|
\le \rho_0^{-2}\norm{\bm d_p}^2\norm{\mE_n}_{\op}
=O_{\Pp}\left(\sqrt{\frac{\log p}{p}}\right).
\end{equation}
Combining \eqref{eq:HLNV_signal_bilinear},
\eqref{eq:bulk_proxy_signal_bound},
\eqref{eq:factor_update_signal_bound},
\eqref{eq:dQH_signal_power},
\eqref{eq:dQ_factor_projection_power}--\eqref{eq:dQ_factor_bulk_column_power},
and \eqref{eq:E_signal_bound} prove
\eqref{eq:bulk_signal_DE_rate}.

Finally,
\begin{align}
\bm d_p^\top\mQ_{B,p}^{\star}(\rho)\bm d_p
&=\frac1\rho\sum_{j=r+1}^p
\frac{(\bm p_{j,p}^\top\bm d_p)^2}
{1+\wt\delta_{p,\rho}\lambda_{j,p}}\notag\\
&=\frac1\rho\int
\frac1{1+\wt\delta_{p,\rho}t}\,\dd\calG_p(t).
\label{eq:q_finite_power}
\end{align}
The fixed-point variables are contained in a compact subset of
$(0,\infty)^2$.  Every subsequential limit solves
(2.8)--(2.9); uniqueness therefore
gives
\begin{equation*}
\delta_{p,\rho}\to\delta_\rho,
\qquad
\wt\delta_{p,\rho}\to\wt\delta_\rho.
\end{equation*}
Since the integrand in \eqref{eq:q_finite_power} is uniformly bounded
and continuous on $[\omega_-,\omega_+]$, the weak convergence
$\calG_p\Rightarrow\calG$ proves \eqref{eq:bulk_signal_limit}.
\end{proof}

\begin{lemma}\label{lem:power_derivatives}
For $a=0,1,2$, define
\[
 A_{a,\rho}=\int\frac{d^{1+a/2}}{(1+\delta_\rho d)^2}\,\dd F_D(d),
 \qquad
 C_{a,\rho}=\int\frac{d^{2+a/2}}{(1+\delta_\rho d)^3}\,\dd F_D(d).
\]
Let
\[
 u_{3,\rho}=c\int\frac{t^3}{(1+\wt\delta_\rho t)^3}\,\dd H_B(t),
 \qquad
 v_{3,\rho}=\int\frac{d^3}{(1+\delta_\rho d)^3}\,\dd F_D(d),
 \qquad
 S_\rho=\rho^2-u_\rho v_\rho.
\]
Then (2.21) holds and
\begin{align*}
 \kappa_\rho'&=\delta_\rho'A_{0,\rho},
 &b_{a,\rho}'&=\delta_\rho'A_{a,\rho},\quad a=1,2,\\
 u_\rho'&=-2\wt\delta_\rho'u_{3,\rho},
 &v_\rho'&=-2\delta_\rho'v_{3,\rho},\\
 A_{a,\rho}'&=-2\delta_\rho'C_{a,\rho},
 &S_\rho'&=2\rho-u_\rho'v_\rho-u_\rho v_\rho'.
\end{align*}
For $a,b\in\{0,1,2\}$,
\begin{equation}\label{eq:Psi_derivative_ell}
 \Psi_{ab}'(\rho)
 =m_am_b\frac{
 \{u_\rho'A_{a,\rho}A_{b,\rho}
 +u_\rho A_{a,\rho}'A_{b,\rho}
 +u_\rho A_{a,\rho}A_{b,\rho}'\}S_\rho
 -u_\rho A_{a,\rho}A_{b,\rho}S_\rho'}{S_\rho^2}.
\end{equation}
Put $r_{a,\rho}=s_a-b_{a,\rho}$, $a=1,2$, and
\[
 \bm h_\rho=
 \begin{pmatrix}
 m_1^2r_{1,\rho}^2\\
 2\kappa_\rho m_1r_{1,\rho}\\
 \kappa_\rho^2
 \end{pmatrix},
 \qquad
 \bg_\rho=D_\rho^{-2}\bm h_\rho.
\]
Then
\begin{align*}
 D_\rho'
 &=-2m_1^2r_{1,\rho}b_{1,\rho}'
 +m_2\{\kappa_\rho'r_{2,\rho}-\kappa_\rho b_{2,\rho}'\},\\
 \bm h_\rho'
 &=\begin{pmatrix}
 -2m_1^2r_{1,\rho}b_{1,\rho}'\\
 2m_1\{\kappa_\rho'r_{1,\rho}-\kappa_\rho b_{1,\rho}'\}\\
 2\kappa_\rho\kappa_\rho'
 \end{pmatrix},\\
 \bg_\rho'
 &=D_\rho^{-2}\bm h_\rho'
 -2D_\rho^{-3}D_\rho'\bm h_\rho.
\end{align*}
Let $\mGamma_\rho'$ be obtained from (2.4) by replacing every $\psi_{ab,n}$ by $\Psi_{ab}'(\rho)$.  The derivative of the limiting conditional variance is
\begin{equation}\label{eq:sigma_derivative_explicit}
 \{\sigma_E^2(\rho)\}'
 =2(\bg_\rho')^\top\mGamma_\rho\bg_\rho
 +\bg_\rho^\top\mGamma_\rho'\bg_\rho.
\end{equation}
\end{lemma}

\begin{proof}
Differentiating the two canonical equations under the integral sign gives
\begin{align*}
 \rho\delta_\rho'+u_\rho\wt\delta_\rho'&=-\delta_\rho,\\
 v_\rho\delta_\rho'+\rho\wt\delta_\rho'&=-\wt\delta_\rho,
\end{align*}
which is (2.21).  By
$S_\rho\ge s_0>0$, the derivative vector is continuous and uniformly bounded on $[\rho_0,\rho_1]$.

For $j\ge1$ and every exponent $\alpha$ appearing above,
\[
 \frac{\dd}{\dd\rho}
 \int\frac{d^\alpha}{(1+\delta_\rho d)^j}\,\dd F_D(d)
 =-j\delta_\rho'
 \int\frac{d^{\alpha+1}}{(1+\delta_\rho d)^{j+1}}\,\dd F_D(d).
\]
The interchange of differentiation and integration follows from
$0<d\le\mu_-^{-1}$ and $\rho\ge\rho_0$.  Applying this identity to
$\kappa_\rho$, $b_{a,\rho}$, $v_\rho$, and $A_{a,\rho}$, and applying the analogous identity to the $H_B$ integral defining $u_\rho$, gives the displayed derivative formulas.  The quotient rule applied to
\[
 \Psi_{ab}(\rho)=m_am_b
 \frac{u_\rho A_{a,\rho}A_{b,\rho}}{S_\rho}
\]
proves \eqref{eq:Psi_derivative_ell}.  Direct differentiation of
$D_\rho$, $\bm h_\rho$, and $\bg_\rho=D_\rho^{-2}\bm h_\rho$ gives the remaining vector formulas.  Finally,
\[
 \frac{\dd}{\dd\rho}
 (\bg_\rho^\top\mGamma_\rho\bg_\rho)
 =2(\bg_\rho')^\top\mGamma_\rho\bg_\rho
 +\bg_\rho^\top\mGamma_\rho'\bg_\rho,
\]
which proves \eqref{eq:sigma_derivative_explicit}.
\end{proof}

\begin{proof}[Proof of Theorem~2.3]
Proposition~\ref{prop:translation} gives the exact identities
\begin{equation*}
\wh\btheta-\btheta_{0,p}
=\bDelta_p+\wh\btheta_\varepsilon,
\qquad
\wh\mQ_n=\wh\mQ_{\varepsilon,n}.
\end{equation*}
It also gives exact equality of the nuisance estimators calculated from
the shifted and noise samples.  Consequently,
\begin{align}
T_n(\rho)-n\wh\mu_n
={}&\Big\{
n\wh\btheta_\varepsilon^\top
\wh\mQ_{\varepsilon,n}\wh\btheta_\varepsilon
-n\wh\mu_n
\Big\}\notag\\
&+n\bDelta_p^\top\wh\mQ_{\varepsilon,n}\bDelta_p
+2n\bDelta_p^\top\wh\mQ_{\varepsilon,n}
\wh\btheta_\varepsilon.
\label{eq:T_alternative_decomp_new}
\end{align}
Define the standardized noise contribution by
\begin{equation*}
Z_{\varepsilon,n}(\rho)
=
\frac{
 n\wh\btheta_\varepsilon^\top
 \wh\mQ_{\varepsilon,n}\wh\btheta_\varepsilon
 -n\wh\mu_n}
{\sqrt{n\wh\sigma_{D,n}^2}}.
\end{equation*}
The noise sample satisfies $H_0$, so Theorem~2.2 yields
\begin{equation}\label{eq:Zeps_null_power}
Z_{\varepsilon,n}(\rho)\dto N(0,1).
\end{equation}

Under (2.14),
$\bDelta_p=n^{-1/4}\bm d_p$.  Hence
\begin{equation}\label{eq:local_signal_exact_scaling}
\frac{n\bDelta_p^\top
\wh\mQ_{\varepsilon,n}\bDelta_p}{\sqrt n}
=\bm d_p^\top\wh\mQ_{\varepsilon,n}\bm d_p.
\end{equation}
Lemma~\ref{lem:bulk_signal_DE} and
\eqref{eq:q_finite_power} give
\begin{align*}
\bm d_p^\top\wh\mQ_{\varepsilon,n}\bm d_p
={}&q_\rho(\calG)
+O_{\Pp}\left(\sqrt{\frac{\log p}{p}}\right)
+r_{G,p}(\rho),
\\
r_{G,p}(\rho)
={}&\frac1\rho\int
\frac1{1+\wt\delta_{p,\rho}t}\,\dd\calG_p(t)
-\frac1\rho\int
\frac1{1+\wt\delta_\rho t}\,\dd\calG(t)
\to0.
\end{align*}

Lemma~\ref{lem:signal_noise_cross_new} gives
\begin{align*}
\left|
\bm d_p^\top\wh\mQ_{\varepsilon,n}
\wh\btheta_\varepsilon
\right|
&=O_{\Pp}\left[
\left\{
\frac1n\bm d_p^\top\wh\mQ_{\varepsilon,n}\bm d_p
\right\}^{1/2}
\right]\\
&=O_{\Pp}(n^{-1/2}),
\end{align*}
where the last equality follows from
\eqref{eq:bulk_signal_DE_rate} and
$\sup_p\norm{\bm d_p}<\infty$.  Therefore
\begin{equation*}
\frac{2n\bDelta_p^\top\wh\mQ_{\varepsilon,n}
\wh\btheta_\varepsilon}{\sqrt n}
=2n^{1/4}\bm d_p^\top\wh\mQ_{\varepsilon,n}
\wh\btheta_\varepsilon
=O_{\Pp}(n^{-1/4}).
\end{equation*}

Let $(\sigma_{E,p}^{\star})^2$ be defined by \eqref{eq:g_sigma_finite}.  Theorem~\ref{thm:rates} and Proposition~\ref{prop:weighted_DE} give
\begin{align*}
 \wh\sigma_{D,n}^2-\sigma_E^2(\rho)
 ={}&\{\wh\sigma_{D,n}^2-\sigma_{D,n}^2\}
 +\{\sigma_{D,n}^2-(\sigma_{E,p}^{\star})^2\}
 +r_{\sigma,p}(\rho),\\
 \wh\sigma_{D,n}^2-\sigma_{D,n}^2
 ={}&O_{\Pp}(n^{-1/2}),\\
 \sigma_{D,n}^2-(\sigma_{E,p}^{\star})^2
 ={}&O_{\Pp}(p^{-1/4}),\\
 r_{\sigma,p}(\rho)
 :={}&(\sigma_{E,p}^{\star})^2-\sigma_E^2(\rho)
 \to0.
\end{align*}
Thus
\[
 \wh\sigma_{D,n}^2
 =\sigma_E^2(\rho)
 +O_{\Pp}(n^{-1/2})
 +O_{\Pp}(p^{-1/4})
 +r_{\sigma,p}(\rho).
\]
By \eqref{eq:nondegenerate_DE} and (2.12),
$\sigma_E(\rho)>0$ for the fixed ridge value under consideration.
The mean-value theorem therefore gives
\begin{align}
 \wh\sigma_{D,n}-\sigma_E(\rho)
 ={}&O_{\Pp}(n^{-1/2})
 +O_{\Pp}(p^{-1/4})
 +O\{|r_{\sigma,p}(\rho)|\},\notag\\
 \frac1{\wh\sigma_{D,n}}-\frac1{\sigma_E(\rho)}
 ={}&O_{\Pp}(n^{-1/2})
 +O_{\Pp}(p^{-1/4})
 +O\{|r_{\sigma,p}(\rho)|\}.
\label{eq:inverse_sigma_rate_power}
\end{align}

Divide \eqref{eq:T_alternative_decomp_new} by
$\sqrt{n\wh\sigma_{D,n}^2}$ and use
\eqref{eq:local_signal_exact_scaling}--
\eqref{eq:inverse_sigma_rate_power}.  We obtain
\begin{align}
Z_n(\rho)
={}&Z_{\varepsilon,n}(\rho)
+\frac{q_\rho(\calG)}{\sigma_E(\rho)}\notag\\
&+O_{\Pp}(n^{-1/4})
+O_{\Pp}(p^{-1/4})
+O_{\Pp}\left(\sqrt{\frac{\log p}{p}}\right)
+r_p(\rho),
\label{eq:local_Z_expansion}
\end{align}
where
\begin{equation*}
|r_p(\rho)|
\le C\{|r_{G,p}(\rho)|+|r_{\sigma,p}(\rho)|\}
\to0.
\end{equation*}
Equations \eqref{eq:Zeps_null_power} and
\eqref{eq:local_Z_expansion} prove
(2.16).  Therefore
\begin{align*}
\lim_{n\to\infty}
\Pp_{H_{1,n}(\calG)}\{Z_n(\rho)>z_{1-\alpha}\}
&=\Pp\{N(\Lambda_\rho(\calG),1)>z_{1-\alpha}\}\\
&=1-\Phi\{z_{1-\alpha}-\Lambda_\rho(\calG)\},
\end{align*}
which proves (2.17).

For the consistency statement, define
\begin{equation*}
S_{n,p}=\bDelta_p^\top
\wh\mQ_{\varepsilon,n}\bDelta_p.
\end{equation*}
Lemma~\ref{lem:scaled_resolvent_power} gives
\begin{equation*}
S_{n,p}
\ge c\left\{
\norm{\mPi_{B,p}\bDelta_p}^2
+p^{-1}\norm{\mPi_{F,p}\bDelta_p}^2
\right\}.
\end{equation*}
Condition (2.19) consequently implies
\begin{equation}\label{eq:sqrt_n_S_diverges_power}
\sqrt nS_{n,p}\pto\infty.
\end{equation}
Lemma~\ref{lem:signal_noise_cross_new} gives
\begin{align*}
\frac{2n\bDelta_p^\top\wh\mQ_{\varepsilon,n}
\wh\btheta_\varepsilon}
{\sqrt{n\wh\sigma_{D,n}^2}}
&=O_{\Pp}(S_{n,p}^{1/2}),
\\
\frac{S_{n,p}^{1/2}}{\sqrt nS_{n,p}}
&=(nS_{n,p})^{-1/2}
\pto0.
\end{align*}
Using \eqref{eq:Zeps_null_power}, \eqref{eq:nondegenerate_DE}, and
\eqref{eq:sqrt_n_S_diverges_power}, we obtain
\begin{align*}
Z_n(\rho)
={}&Z_{\varepsilon,n}(\rho)
+\frac{\sqrt nS_{n,p}}{\wh\sigma_{D,n}}
+O_{\Pp}(S_{n,p}^{1/2})\\
={}&
\frac{\sqrt nS_{n,p}}{\wh\sigma_{D,n}}
\left\{1+O_{\Pp}((nS_{n,p})^{-1/2})\right\}
+O_{\Pp}(1)
\pto\infty.
\end{align*}
This proves the consistency assertion in Theorem~2.3.
\end{proof}

\section{Finite-grid joint limits and Cauchy aggregation}\label{app:cauchy}

Throughout this appendix, $\mathcal R_K=\{\rho^{(1)},\ldots,\rho^{(K)}\}$ is the fixed grid in Section~3.  All maxima over $k$ and $\ell$ are therefore over a fixed finite set.  For compactness, write
\[
 \delta_{p,k}=\delta_{p,\rho^{(k)}},
 \qquad
 \wt\delta_{p,k}=\wt\delta_{p,\rho^{(k)}},
 \qquad
 \delta_k=\delta_{\rho^{(k)}},
 \qquad
 \wt\delta_k=\wt\delta_{\rho^{(k)}}.
\]

For $1\le k,\ell\le K$, define the finite-$p$ cross-ridge quantities
\begin{align*}
 u_{p,k\ell}
 &=c_{B,n}\int
 \frac{t^2}
 {(1+\wt\delta_{p,k}t)(1+\wt\delta_{p,\ell}t)}
 \,\dd H_{B,p}(t),\\
 v_{p,k\ell}
 &=\int
 \frac{d^2}
 {(1+\delta_{p,k}d)(1+\delta_{p,\ell}d)}
 \,\dd F_{D,p}(d),\\
 A_{a,p,k\ell}
 &=\int
 \frac{d^{1+a/2}}
 {(1+\delta_{p,k}d)(1+\delta_{p,\ell}d)}
 \,\dd F_{D,p}(d),
 \qquad a=0,1,2,
\end{align*}
and
\begin{equation}\label{eq:cross_Psi_finite}
 \Psi_{ab,p}^{k\ell}
 =m_am_b
 \frac{u_{p,k\ell}A_{a,p,k\ell}A_{b,p,k\ell}}
 {\rho^{(k)}\rho^{(\ell)}-u_{p,k\ell}v_{p,k\ell}},
 \qquad a,b\in\{0,1,2\}.
\end{equation}
At the limiting level, put
\begin{align*}
 u_{k\ell}
 &=c\int
 \frac{t^2}
 {(1+\wt\delta_kt)(1+\wt\delta_\ell t)}
 \,\dd H_B(t),\\
 v_{k\ell}
 &=\int
 \frac{d^2}
 {(1+\delta_kd)(1+\delta_\ell d)}
 \,\dd F_D(d),\\
 A_{a,k\ell}
 &=\int
 \frac{d^{1+a/2}}
 {(1+\delta_kd)(1+\delta_\ell d)}
 \,\dd F_D(d),
\end{align*}
and define $\Psi_{ab}^{k\ell}$ by replacing the finite-$p$ quantities in \eqref{eq:cross_Psi_finite} with their limits.  Let
\begin{equation}\label{eq:cross_Gamma_limit}
 \mGamma_{k\ell}=
 \begin{pmatrix}
 2\Psi_{00}^{k\ell}&2\Psi_{01}^{k\ell}&2\Psi_{11}^{k\ell}\\
 2\Psi_{01}^{k\ell}&\Psi_{02}^{k\ell}+\Psi_{11}^{k\ell}&2\Psi_{12}^{k\ell}\\
 2\Psi_{11}^{k\ell}&2\Psi_{12}^{k\ell}&2\Psi_{22}^{k\ell}
 \end{pmatrix}.
\end{equation}
Let $\mGamma_{p,k\ell}^{\star}$ have the form in \eqref{eq:cross_Gamma_limit}, with $\Psi_{ab}^{k\ell}$ replaced by $\Psi_{ab,p}^{k\ell}$, and write
\[
 \bg_{p,k}^{\star}=\bg_{p,\rho^{(k)}}^{\star},
 \qquad
 \sigma_{E,p,k}^{\star}=\sigma_{E,p}^{\star}(\rho^{(k)}).
\]
Define the finite-$p$ deterministic correlation
\begin{equation}\label{eq:cross_corr_finite}
 r_{p,k\ell}^{\star}
 =\frac{(\bg_{p,k}^{\star})^\top
 \mGamma_{p,k\ell}^{\star}\bg_{p,\ell}^{\star}}
 {\sigma_{E,p,k}^{\star}\sigma_{E,p,\ell}^{\star}}.
\end{equation}
The notation $\mGamma_{\rho^{(k)},\rho^{(\ell)}}$ in Section~3 means $\mGamma_{k\ell}$.  Finally,
\begin{equation}\label{eq:cross_corr_limit}
 r_{k\ell}
 =\frac{\bg_{\rho^{(k)}}^\top
 \mGamma_{k\ell}\bg_{\rho^{(\ell)}}}
 {\sigma_E(\rho^{(k)})\sigma_E(\rho^{(\ell)})}.
\end{equation}
All factors in \eqref{eq:cross_Psi_finite} are nonnegative and the stability denominator is positive.  Since every component of the two gradient vectors is nonnegative, $r_{p,k\ell}^{\star}\ge0$ and $r_{k\ell}\ge0$.  Positive semidefiniteness and the upper bound one follow from the covariance representation in Lemma~\ref{lem:finite_grid_conditional_clt}.

For the oracle sample quantities, let $\mA_{n,k}$ be the companion matrix $\mA_n$ evaluated at $\rho^{(k)}$, and define
\[
 \psi_{ab,n}^{k\ell}
 =\frac1n\sum_{i\ne j}
 (\mA_{n,k})_{ij}(\mA_{n,\ell})_{ij}w_i^aw_j^b.
\]
Let $\mGamma_n^{k\ell}$ have the form in \eqref{eq:cross_Gamma_limit} with $\Psi_{ab}^{k\ell}$ replaced by $\psi_{ab,n}^{k\ell}$.

\begin{lemma}\label{lem:cross_ridge_DE}
Under Assumption~1, there is a constant $s_\times>0$ such that, for all sufficiently large $p$ and all $1\le k,\ell\le K$,
\begin{equation}\label{eq:cross_stability}
 \rho^{(k)}\rho^{(\ell)}-u_{p,k\ell}v_{p,k\ell}
 \ge s_\times.
\end{equation}
Moreover, for every $a,b\in\{0,1,2\}$,
\begin{equation}\label{eq:cross_psi_DE_rate}
 \max_{k,\ell}
 |\psi_{ab,n}^{k\ell}-\Psi_{ab,p}^{k\ell}|
 =O_{\Pp}(p^{-1/4}),
\end{equation}
and
\[
 \Psi_{ab,p}^{k\ell}\to\Psi_{ab}^{k\ell}.
\]
Consequently,
\begin{equation}\label{eq:cross_cov_DE_rate}
 \max_{k,\ell}
 \left|
 \frac{\bg_{n,k}^\top\mGamma_n^{k\ell}\bg_{n,\ell}}
 {\sigma_{D,n,k}\sigma_{D,n,\ell}}
 -r_{p,k\ell}^{\star}
 \right|
 =O_{\Pp}\!\left(
 p^{-1/4}+\sqrt{\frac{\log p}{p}}
 \right),
\end{equation}
where $\bg_{n,k}=\bg_n(\rho^{(k)})$ and
$\sigma_{D,n,k}=\sigma_{D,n}(\rho^{(k)})$.  In addition,
\[
 \max_{k,\ell}|r_{p,k\ell}^{\star}-r_{k\ell}|\to0.
\]
\end{lemma}

\begin{proof}
Cauchy--Schwarz gives
\begin{align*}
 u_{p,k\ell}^2
 &\le u_{p,kk}u_{p,\ell\ell},\\
 v_{p,k\ell}^2
 &\le v_{p,kk}v_{p,\ell\ell}.
\end{align*}
The proof of Lemma~\ref{lem:canonical_stability} supplies one fixed $\varepsilon_0>0$ such that
\[
 u_{p,kk}v_{p,kk}
 \le\{\rho^{(k)}\}^2(1-\varepsilon_0),
 \qquad 1\le k\le K.
\]
Hence
\begin{align*}
 u_{p,k\ell}v_{p,k\ell}
 &\le
 \{u_{p,kk}v_{p,kk}
 u_{p,\ell\ell}v_{p,\ell\ell}\}^{1/2}\\
 &\le \rho^{(k)}\rho^{(\ell)}(1-\varepsilon_0),
\end{align*}
which proves \eqref{eq:cross_stability} with
$s_\times=\rho_0^2\varepsilon_0$.

We next evaluate the cross-ridge off-diagonal sum.  Condition on the diagonal matrix
\[
 \mD_p^0=\diag(D_{1,p}^0,\ldots,D_{n,p}^0).
\]
Retain the separable bulk proxy from \eqref{eq:separable_proxy_main} and put
\[
 \overline\mR_{n,k}
 =(\overline\mB_{B,n}+\rho^{(k)}\I_{p-r})^{-1},
 \qquad
 L_{n,k\ell}
 =\frac1n\tr(\mLambda_{B,p}\overline\mR_{n,k}
 \mLambda_{B,p}\overline\mR_{n,\ell}).
\]
Let $(\delta_{n,k}^{D},\wt\delta_{n,k}^{D})$ be the empirical canonical solution in Lemma~\ref{lem:empirical_canonical} at $\rho^{(k)}$, and define
\begin{align*}
 u_{n,k\ell}^{D}
 &=c_{B,n}\int
 \frac{t^2}
 {(1+\wt\delta_{n,k}^{D}t)
  (1+\wt\delta_{n,\ell}^{D}t)}
 \,\dd H_{B,p}(t),\\
 v_{n,k\ell}^{D}
 &=\frac1n\sum_{i=1}^n
 \frac{(D_{i,p}^0)^2}
 {(1+\delta_{n,k}^{D}D_{i,p}^0)
  (1+\delta_{n,\ell}^{D}D_{i,p}^0)}.
\end{align*}
We record the feedback calculation rather than invoking a
one-line ``two-resolvent expansion.''  Put
\begin{align*}
 \mT_k&=\{\rho^{(k)}(\I+
 \wt\delta_{n,k}^{D}\mLambda_{B,p})\}^{-1},\\
 \gamma_{i,k}&=(1+\delta_{n,k}^{D}D_{i,p}^0)^{-1},\\
 \ell_{n,k\ell}^{0}
 &=\frac1n\tr(\mLambda_{B,p}\mT_k
                 \mLambda_{B,p}\mT_\ell)
 =\frac{u_{n,k\ell}^{D}}
 {\rho^{(k)}\rho^{(\ell)}}.
\end{align*}
Conditionally on $\mD_p^0$, let
$g_{\alpha i}=G_{r+\alpha,i}$ and
$\lambda_{B,\alpha,p}=\lambda_{r+\alpha,p}$ for
$1\le\alpha\le p-r$, and write
$\bm x_i=(D_{i,p}^0)^{1/2}\mLambda_{B,p}^{1/2}\bm g_i$.  Then
$\overline\mB_{B,n}=n^{-1}\sum_i\bm x_i\bm x_i^\top$.  The two
identities used in the Gaussian integration-by-parts calculation are
\begin{align}
 \frac{\partial\overline\mR_{n,k}}
 {\partial g_{\alpha i}}
 &=-\overline\mR_{n,k}
 \frac{\partial\overline\mB_{B,n}}
 {\partial g_{\alpha i}}
 \overline\mR_{n,k},
 \label{eq:cross_resolvent_derivative}\\
 \frac{\partial\overline\mB_{B,n}}
 {\partial g_{\alpha i}}
 &=\frac{\sqrt{D_{i,p}^0\lambda_{B,\alpha,p}}}{n}
 \{\bm e_\alpha\bm x_i^\top+\bm x_i\bm e_\alpha^\top\}.
 \label{eq:cross_covariance_derivative}
\end{align}
Let $\E_G(\cdot)=\E(\cdot\mid\mD_p^0)$.  For each $i$, let
\[
 \overline\mR_{n,k}^{(i)}
 =\left(\frac1n\sum_{j\ne i}\bm x_j\bm x_j^\top
       +\rho^{(k)}\I_{p-r}\right)^{-1}
\]
and define
\begin{align*}
 \widehat\gamma_{i,k}
 &=\left\{1+n^{-1}\bm x_i^\top
       \overline\mR_{n,k}^{(i)}\bm x_i\right\}^{-1},\\
 L_{n,k\ell}^{(i)}
 &=\frac1n\tr(\mLambda_{B,p}\overline\mR_{n,k}^{(i)}
                 \mLambda_{B,p}\overline\mR_{n,\ell}^{(i)}),\\
 \varepsilon_{i,k}
 &=n^{-1}\bm x_i^\top\overline\mR_{n,k}^{(i)}\bm x_i
   -D_{i,p}^0\delta_{n,k}^{D},\\
 \zeta_{i,k\ell}
 &=n^{-1}\bm x_i^\top\overline\mR_{n,k}^{(i)}
        \mLambda_{B,p}\overline\mR_{n,\ell}^{(i)}\bm x_i
   -D_{i,p}^0L_{n,k\ell}^{(i)}.
\end{align*}
The Sherman--Morrison formula gives
\[
 \overline\mR_{n,k}\bm x_i
 =\widehat\gamma_{i,k}\overline\mR_{n,k}^{(i)}\bm x_i,
 \qquad
 |\widehat\gamma_{i,k}-\gamma_{i,k}|
 \le |\varepsilon_{i,k}|.
\]
We now spell out the feedback calculation.  Apply the Gaussian
integration-by-parts identity
$\E_G(g_{\alpha i}F)=\E_G(\partial F/\partial g_{\alpha i})$
to the coordinate representation of
$n^{-1}\tr(\mLambda_{B,p}\overline\mR_{n,k}
\mLambda_{B,p}\overline\mR_{n,\ell})$, and use
\eqref{eq:cross_resolvent_derivative}--\eqref{eq:cross_covariance_derivative}
and the preceding leave-one identities.  After summing over $\alpha$, the
zeroth contraction is $\ell_{n,k\ell}^{0}$.  The two derivatives falling on
the resolvents give, before replacing the leave-one denominators and
traces,
\[
 \frac{\ell_{n,k\ell}^{0}}n\sum_{i=1}^n
 (D_{i,p}^0)^2
 \E_G\{\widehat\gamma_{i,k}\widehat\gamma_{i,\ell}
       L_{n,k\ell}^{(i)}\}.
\]
The remaining contractions contain a centered leave-one quadratic form
$\varepsilon_{i,k}$ or $\varepsilon_{i,\ell}$, the mixed centered form
$\zeta_{i,k\ell}$, or one rank-one deletion error.  Since
\[
 |\widehat\gamma_{i,k}\widehat\gamma_{i,\ell}
   -\gamma_{i,k}\gamma_{i,\ell}|
 \le C(|\varepsilon_{i,k}|+|\varepsilon_{i,\ell}|),
\]
collecting these terms yields
\begin{align}
 \E_G L_{n,k\ell}
 ={}&\ell_{n,k\ell}^{0}
 +\frac{\ell_{n,k\ell}^{0}}n\sum_{i=1}^n
   (D_{i,p}^0)^2\gamma_{i,k}\gamma_{i,\ell}
   \E_G L_{n,k\ell}^{(i)}
 +\mathcal R_{n,k\ell}.
 \label{eq:cross_feedback_leave_one}
\end{align}
More precisely, Cauchy--Schwarz applied to the displayed centered forms and
the resolvent bounds gives
\begin{align}
 |\mathcal R_{n,k\ell}|
 \le C\bigg[&\frac1n\sum_{i=1}^n\E_G\{
   |\varepsilon_{i,k}|+|\varepsilon_{i,\ell}|
   +\varepsilon_{i,k}^2+\varepsilon_{i,\ell}^2
   +|\zeta_{i,k\ell}|\}\notag\\
 &+\max_{i\le n}\E_G|L_{n,k\ell}^{(i)}-L_{n,k\ell}|
 +\frac1n\bigg].
 \label{eq:cross_feedback_remainder_bound}
\end{align}
To verify the order in this bound, condition further on all columns except
$i$.  The Gaussian quadratic-form identity
$\Var(\bm g^\top\mC\bm g)=2\tr(\mC^2)$ and the uniform bounds on
$D_{i,p}^0$, $\mLambda_{B,p}$, and the ridge resolvents give
\begin{align}
 \frac1n\sum_{i=1}^n\E_G\varepsilon_{i,k}^2
 &=O_{\Pp}(n^{-1}),
 &\frac1n\sum_{i=1}^n\E_G\zeta_{i,k\ell}^2
 &=O_{\Pp}(n^{-1}).
 \label{eq:cross_leave_one_moments}
\end{align}
For the first relation, the difference between
$n^{-1}\tr(\mLambda_{B,p}\overline\mR_{n,k}^{(i)})$ and
$\delta_{n,k}^{D}$ is covered by the diagonal bilinear-form estimate in
Theorem~1.1 of \citet{HachemLoubatonNajimVallet2013}; the centered part is
the preceding Gaussian quadratic form.  For the second relation, its
conditional mean is zero by the definition of $L_{n,k\ell}^{(i)}$.
Furthermore, the rank-one resolvent identity gives, deterministically,
\[
 \max_{i,k,\ell}|L_{n,k\ell}^{(i)}-L_{n,k\ell}|\le Cn^{-1}.
\]
Cauchy--Schwarz in \eqref{eq:cross_feedback_remainder_bound}, together with \eqref{eq:cross_leave_one_moments}, therefore
implies
\begin{equation}
 \max_{k,\ell\le K}|\mathcal R_{n,k\ell}|
 =O_{\Pp}(n^{-1/2}).
 \label{eq:cross_feedback_R_rate}
\end{equation}

For completeness, the random mixed trace itself concentrates at a smaller
order.  Direct differentiation gives
\begin{align*}
 \frac{\partial L_{n,k\ell}}{\partial g_{\alpha i}}
 =-\frac{2\sqrt{D_{i,p}^0\lambda_{B,\alpha,p}}}{n^2}
 \bm e_\alpha^\top\big(&\overline\mR_{n,k}\mLambda_{B,p}
 \overline\mR_{n,\ell}\mLambda_{B,p}\overline\mR_{n,k}\\
 &+\overline\mR_{n,\ell}\mLambda_{B,p}
 \overline\mR_{n,k}\mLambda_{B,p}\overline\mR_{n,\ell}\big)\bm x_i.
\end{align*}
Consequently,
\[
 \sum_{\alpha,i}\left|
 \frac{\partial L_{n,k\ell}}{\partial g_{\alpha i}}
 \right|^2
 \le \frac C{n^4}\sum_i\|\bm x_i\|^2.
\]
Conditionally on $\mD_p^0$,
$\sum_i\E_G\|\bm x_i\|^2\le Cnp=O(n^2)$.  The Gaussian
Poincar\'e inequality
\citep[Theorem~3.20, p.~72]{BoucheronLugosiMassart2013} therefore gives
$\Var_G(L_{n,k\ell})\le Cn^{-2}$.  Since the grid is fixed,
\begin{equation}
 \max_{k,\ell\le K}|L_{n,k\ell}-\E_G L_{n,k\ell}|
 =O_{\Pp}(n^{-1}).
 \label{eq:cross_trace_concentration}
\end{equation}
Using the rank-one deletion bound in
\eqref{eq:cross_feedback_leave_one}, then
\eqref{eq:cross_feedback_R_rate}--\eqref{eq:cross_trace_concentration},
we obtain
\begin{equation}\label{eq:cross_feedback_remainder}
 L_{n,k\ell}
 =\ell_{n,k\ell}^{0}
 +\ell_{n,k\ell}^{0}v_{n,k\ell}^{D}L_{n,k\ell}
 +r_{n,k\ell},
 \qquad
 \max_{k,\ell\le K}|r_{n,k\ell}|=O_{\Pp}(n^{-1/2}).
\end{equation}
Since $p\asymp n$, the weaker rate needed below is
\begin{align}
 L_{n,k\ell}
 ={}&\frac{u_{n,k\ell}^{D}}
 {\rho^{(k)}\rho^{(\ell)}}
 +\frac{u_{n,k\ell}^{D}v_{n,k\ell}^{D}}
 {\rho^{(k)}\rho^{(\ell)}}L_{n,k\ell}
 +O_{\Pp}(p^{-1/4}).
\label{eq:cross_trace_feedback}
\end{align}

Equations \eqref{eq:empirical_fixed_point_rate} and Bernstein's inequality \citep[Theorem~2.10 and Corollary~2.11, pp.~36--37]{BoucheronLugosiMassart2013} imply
\begin{align*}
 \max_{k,\ell}|u_{n,k\ell}^{D}-u_{p,k\ell}|
 &=O_{\Pp}(n^{-1/2}),\\
 \max_{k,\ell}|v_{n,k\ell}^{D}-v_{p,k\ell}|
 &=O_{\Pp}(n^{-1/2}).
\end{align*}
The cross-stability bound permits division in \eqref{eq:cross_trace_feedback}, yielding
\begin{equation}\label{eq:cross_second_trace}
 \max_{k,\ell}
 \left|
 L_{n,k\ell}
 -\frac{u_{p,k\ell}}
 {\rho^{(k)}\rho^{(\ell)}-u_{p,k\ell}v_{p,k\ell}}
 \right|
 =O_{\Pp}(p^{-1/4}).
\end{equation}
Deleting one or two columns changes the normalized trace by $O(p^{-1})$, uniformly over the finite grid, by the resolvent identity and the bound $\|\overline\mR_{n,k}\|_{\op}\le\rho_0^{-1}$.

Put $\widetilde w_i=\xi_i(D_{i,p}^0)^{1/2}$ and let
$\overline\mA_{n,k}$ be the proxy companion matrix at $\rho^{(k)}$.  The Efron--Stein inequality \citep[Theorem~3.1, p.~54]{BoucheronLugosiMassart2013}, applied as in \eqref{eq:radial_offdiag_remove} with
$(\mA_n)_{ij}^2$ replaced by
$(\mA_{n,k})_{ij}(\mA_{n,\ell})_{ij}$, gives
\begin{align}
 \psi_{ab,n}^{k\ell}
 ={}&m_am_b\frac1n\sum_{i\ne j}
 (\mA_{n,k})_{ij}(\mA_{n,\ell})_{ij}
 \ell_{i,p}^a\ell_{j,p}^b
 +O_{\Pp}(n^{-1/2}).
\label{eq:cross_radial_remove}
\end{align}
Indeed, changing one $\xi_i$ changes the displayed function by at most
\begin{align*}
 \frac Cn\sum_{j\ne i}
 |(\mA_{n,k})_{ij}(\mA_{n,\ell})_{ij}|
 &\le\frac Cn
 \left\{\sum_j(\mA_{n,k})_{ij}^2\right\}^{1/2}
 \left\{\sum_j(\mA_{n,\ell})_{ij}^2\right\}^{1/2}\\
 &\le\frac Cn.
\end{align*}
Replacing $\ell_{i,p}$ by $(D_{i,p}^0)^{1/2}$ costs
$O_{\Pp}\{(\log p/p)^{1/2}\}$ by \eqref{eq:ell_D_rate}.  For the companion matrices,
\begin{align*}
 &\frac1n\sum_{i,j}
 |(\mA_{n,k})_{ij}(\mA_{n,\ell})_{ij}
 -(\overline\mA_{n,k})_{ij}(\overline\mA_{n,\ell})_{ij}|\\
 &\quad\le\frac1n
 \|\mA_{n,k}-\overline\mA_{n,k}\|_{\F}
 \|\mA_{n,\ell}\|_{\F}
 +\frac1n
 \|\overline\mA_{n,k}\|_{\F}
 \|\mA_{n,\ell}-\overline\mA_{n,\ell}\|_{\F}\\
 &\quad=O_{\Pp}\left(\sqrt{\frac{\log p}{p}}+n^{-1/2}\right),
\end{align*}
where \eqref{eq:proxy_A_rate_new}, \eqref{eq:factor_trace_rank_new}, and the contraction bounds were used.

It remains to evaluate
\[
 T_{ab,n}^{k\ell}
 =\frac1n\sum_{i\ne j}
 (\overline\mA_{n,k})_{ij}(\overline\mA_{n,\ell})_{ij}
 (D_{i,p}^0)^{a/2}(D_{j,p}^0)^{b/2}.
\]
For $i\ne j$, let $\overline\mR_{n,k}^{(ij)}$ be the ridge-$\rho^{(k)}$ resolvent after deleting columns $i$ and $j$.  The two-column Schur complement gives
\[
 (\overline\mA_{n,k})_{ij}
 =\frac{n^{-1}\bm x_i^\top
 \overline\mR_{n,k}^{(ij)}\bm x_j}
 {\{1+\delta_{n,k}^{D}D_{i,p}^0\}
  \{1+\delta_{n,k}^{D}D_{j,p}^0\}}
 +r_{ij,k},
\]
where the averaged squared contribution of $r_{ij,k}$ to
$T_{ab,n}^{k\ell}$ is $O_{\Pp}(n^{-1/2})$.  Put
\[
 q_{ij,k}=\bm x_i^\top\overline\mR_{n,k}^{(ij)}\bm x_j,
 \qquad
 \gamma_{i,k}=(1+\delta_{n,k}^{D}D_{i,p}^0)^{-1}.
\]
After the denominator replacement, define
\begin{align*}
 \widetilde T_{ab,n}^{k\ell}
 &=\frac1{n^3}\sum_{i\ne j}
 (D_{i,p}^0)^{a/2}(D_{j,p}^0)^{b/2}
 \gamma_{i,k}\gamma_{i,\ell}\gamma_{j,k}\gamma_{j,\ell}
 q_{ij,k}q_{ij,\ell},\\
 M_{ab,n}^{k\ell}
 &=\frac1{n^3}\sum_{i\ne j}
 (D_{i,p}^0)^{1+a/2}(D_{j,p}^0)^{1+b/2}
 \gamma_{i,k}\gamma_{i,\ell}\gamma_{j,k}\gamma_{j,\ell}\\
 &\qquad\times
 \tr(\mLambda_{B,p}\overline\mR_{n,k}^{(ij)}
      \mLambda_{B,p}\overline\mR_{n,\ell}^{(ij)}).
\end{align*}
The averaged denominator error stated above gives
$T_{ab,n}^{k\ell}-\widetilde T_{ab,n}^{k\ell}
=O_{\Pp}(n^{-1/2})$.  Conditional on the deleted resolvents and
$\mD_p^0$,
\begin{align*}
 \E_{i,j}(q_{ij,k}q_{ij,\ell})
 =D_{i,p}^0D_{j,p}^0
 \tr(\mLambda_{B,p}\overline\mR_{n,k}^{(ij)}
      \mLambda_{B,p}\overline\mR_{n,\ell}^{(ij)}),
\end{align*}
so the conditional expectation of
$\widetilde T_{ab,n}^{k\ell}-M_{ab,n}^{k\ell}$ is zero.  Differentiating
$q_{ij,k}q_{ij,\ell}$ and the mixed trace with respect to the Gaussian
coordinates gives the same three derivative classes as in the proof of
Proposition~\ref{prop:weighted_DE}, with one ridge-$\rho^{(k)}$ and one
ridge-$\rho^{(\ell)}$ resolvent.  The bounds
$\|\overline\mR_{n,k}^{(ij)}\|_{\op}\le\rho_0^{-1}$ and
$\|\mLambda_{B,p}\|_{\op}\le\omega_+$, followed by the same
Cauchy--Schwarz and Gaussian moment calculation, give
\[
 \E_G\left[
 \sum_{\alpha=1}^{p-r}\sum_{h=1}^n
 \left|\frac{\partial
 (\widetilde T_{ab,n}^{k\ell}-M_{ab,n}^{k\ell})}
 {\partial g_{\alpha h}}\right|^2\right]\le \frac Cn
\]
uniformly over the fixed grid.  Its conditional mean is zero, so the
Gaussian Poincar\'e inequality implies
\begin{equation*}
 \widetilde T_{ab,n}^{k\ell}-M_{ab,n}^{k\ell}
 =O_{\Pp}(n^{-1/2}).
\end{equation*}

Finally, deleting two columns changes the unnormalized mixed trace by at
most a constant.  Hence, uniformly over $i\ne j$ and the fixed grid,
\[
 \tr(\mLambda_{B,p}\overline\mR_{n,k}^{(ij)}
      \mLambda_{B,p}\overline\mR_{n,\ell}^{(ij)})
 =nL_{n,k\ell}+O(1).
\]
Summing the $O(1)$ error contributes $O(n^{-1})$, and restoring the
omitted diagonal contributes the same order.  Thus
\begin{align*}
 T_{ab,n}^{k\ell}
 ={}&L_{n,k\ell}
 \left\{\frac1n\sum_i
 \frac{(D_{i,p}^0)^{1+a/2}}
 {(1+\delta_{n,k}^{D}D_{i,p}^0)
  (1+\delta_{n,\ell}^{D}D_{i,p}^0)}\right\}\\
 &\times
 \left\{\frac1n\sum_j
 \frac{(D_{j,p}^0)^{1+b/2}}
 {(1+\delta_{n,k}^{D}D_{j,p}^0)
  (1+\delta_{n,\ell}^{D}D_{j,p}^0)}\right\}
 +O_{\Pp}(n^{-1/2}).
\end{align*}
Using \eqref{eq:cross_second_trace}, the empirical fixed-point rate, and Bernstein's inequality \citep[Theorem~2.10 and Corollary~2.11, pp.~36--37]{BoucheronLugosiMassart2013} in the two averages gives
\[
 T_{ab,n}^{k\ell}
 =\frac{u_{p,k\ell}A_{a,p,k\ell}A_{b,p,k\ell}}
 {\rho^{(k)}\rho^{(\ell)}-u_{p,k\ell}v_{p,k\ell}}
 +O_{\Pp}(p^{-1/4}).
\]
Together with \eqref{eq:cross_radial_remove}, this proves \eqref{eq:cross_psi_DE_rate}.

Equations \eqref{eq:kappa_DE}--\eqref{eq:D_sigma_DE} give
\begin{align*}
 \max_k\|\bg_{n,k}-\bg_{p,k}^{\star}\|
 &=O_{\Pp}\left(\sqrt{\frac{\log p}{p}}\right),\\
 \max_k|\sigma_{D,n,k}-\sigma_{E,p,k}^{\star}|
 &=O_{\Pp}(p^{-1/4}).
\end{align*}
The finite-$p$ denominators are uniformly bounded away from zero.  Substitution into \eqref{eq:cross_corr_finite} proves \eqref{eq:cross_cov_DE_rate}.

The compact canonical bounds, weak convergence of $H_{B,p}$ and
$F_{D,p}$, and uniqueness of the limiting canonical equations give, for every fixed $k$ and $\ell$,
\[
 \Psi_{ab,p}^{k\ell}\to\Psi_{ab}^{k\ell},
 \qquad
 \bg_{p,k}^{\star}\to\bg_{\rho^{(k)}},
 \qquad
 \sigma_{E,p,k}^{\star}\to\sigma_E(\rho^{(k)}).
\]
Because $K$ is fixed, the convergence is uniform over the grid.  Equation \eqref{eq:cross_corr_finite} therefore gives
$\max_{k,\ell}|r_{p,k\ell}^{\star}-r_{k\ell}|\to0$.
\end{proof}

For each $k$, let
\[
 \bm q_{n,k}=(x_{n,k},y_{n,k},z_{n,k})^\top,
 \qquad
 \bm q_{0,n,k}=(\kappa_{n,k},b_{1,n,k},b_{2,n,k})^\top,
\]
where the quantities in \ref{app:joint} are evaluated at
$\rho^{(k)}$.  The exact sign identities become
\begin{align*}
 x_{n,k}-\kappa_{n,k}
 &=\frac2n\sum_{i<j}a_{ij,k}s_is_j,\\
 y_{n,k}-b_{1,n,k}
 &=\frac1n\sum_{i<j}a_{ij,k}(w_i+w_j)s_is_j,\\
 z_{n,k}-b_{2,n,k}
 &=\frac2n\sum_{i<j}a_{ij,k}w_iw_js_is_j,
\end{align*}
where $a_{ij,k}$ is the $(i,j)$ entry of the unsigned companion matrix at $\rho^{(k)}$.

\begin{lemma}\label{lem:finite_grid_conditional_clt}
Under Assumption~1 and $H_0$, conditionally on
$\calF_n^0$,
\[
 \sqrt n\{\bm q_{n,1}-\bm q_{0,n,1},\ldots,
 \bm q_{n,K}-\bm q_{0,n,K}\}
\]
has a centered Gaussian approximation in probability whose $(k,\ell)$ covariance block is $\mGamma_n^{k\ell}$.  More precisely, for every bounded $\calF_n^0$-measurable triangular array
$\bm c_{n,k}=(c_{1,n,k},c_{2,n,k},c_{3,n,k})^\top$,
\[
 L_n=\sqrt n\sum_{k=1}^K
 \bm c_{n,k}^\top(\bm q_{n,k}-\bm q_{0,n,k})
\]
satisfies
\[
 \frac{L_n}{V_n^{1/2}}\dto N(0,1)
 \quad\text{conditionally in probability}
\]
whenever
\[
 V_n=\sum_{k,\ell=1}^K
 \bm c_{n,k}^\top\mGamma_n^{k\ell}\bm c_{n,\ell}
\]
is bounded away from zero.  If $V_n\to0$, then $L_n\to0$ in conditional probability.
\end{lemma}

\begin{proof}
The exact covariance calculation in Lemma~\ref{lem:conditional_covariance_new} applies to two different ridge values because only equal unordered index pairs survive.  It gives
\[
 n\Cov_s(\bm q_{n,k},\bm q_{n,\ell}\mid\calF_n^0)
 =\mGamma_n^{k\ell}.
\]
For the stated projection, put
\[
 h_{ij,n,k}
 =2c_{1,n,k}+c_{2,n,k}(w_i+w_j)
 +2c_{3,n,k}w_iw_j.
\]
Then
\begin{equation}\label{eq:finite_grid_Rademacher_form}
 L_n=\frac1{\sqrt n}\sum_{i<j}
 \left\{\sum_{k=1}^Ka_{ij,k}h_{ij,n,k}\right\}s_is_j.
\end{equation}
Let $\mC_n$ be symmetric with zero diagonal and
\[
 (\mC_n)_{ij}
 =\frac1{2\sqrt n}\sum_{k=1}^K
 a_{ij,k}h_{ij,n,k},
 \qquad i\ne j.
\]
Since $0<w_i\le\bar w$, $K$ is fixed, and the coefficient vectors are bounded,
$|h_{ij,n,k}|\le C$.  As in \eqref{eq:C_operator_small}, the matrix before removal of its diagonal is a fixed finite sum of
\[
 \widetilde\mA_{n,k},
 \qquad
 \mD_w\widetilde\mA_{n,k}
 +\widetilde\mA_{n,k}\mD_w,
 \qquad
 \mD_w\widetilde\mA_{n,k}\mD_w.
\]
Each companion matrix is a contraction.  Therefore
\begin{align}
 \|\mC_n\|_{\op}&\le Cn^{-1/2},
 \label{eq:finite_grid_Cop}\\
 \max_i\sum_j(\mC_n)_{ij}^2
 &\le\frac{CK}{n}\max_{i,k}\sum_ja_{ij,k}^2
 \le\frac Cn,\notag\\
 \sum_{i,j}(\mC_n)_{ij}^4
 &\le\left\{\max_i\sum_j(\mC_n)_{ij}^2\right\}
 \|\mC_n\|_{\F}^2.
\label{eq:finite_grid_fourth}
\end{align}
The conditional variance of \eqref{eq:finite_grid_Rademacher_form} is
\[
 2\tr(\mC_n^2)
 =\sum_{k,\ell=1}^K
 \bm c_{n,k}^\top\mGamma_n^{k\ell}\bm c_{n,\ell}
 =V_n.
\]
The graph expansion used in the proof of Theorem~\ref{thm:joint_clt}, together with
\eqref{eq:finite_grid_Cop}--\eqref{eq:finite_grid_fourth}, gives
\begin{align*}
 \left|\E_s(L_n^4\mid\calF_n^0)-3V_n^2\right|
 &\le C\tr(\mC_n^4)
 +C\sum_i\left\{\sum_j(\mC_n)_{ij}^2\right\}^2
 +C\sum_{i,j}(\mC_n)_{ij}^4\\
 &\le Cn^{-1}\{1+V_n\}.
\end{align*}
On $\{V_n\ge\eta\}$,
\[
 \frac{\max_i\sum_j(\mC_n)_{ij}^2}{V_n}
 \le\frac{C}{n\eta},
 \qquad
 \left|\frac{\E_s(L_n^4\mid\calF_n^0)}{V_n^2}-3\right|
 \le\frac{C}{n\eta^2}.
\]
Theorem~2.1 of \citet{deJong1987} gives the conditional normal approximation uniformly on this event.  On $\{V_n<\eta\}$,
\[
 \E_s(|L_n|\mid\calF_n^0)\le\sqrt\eta.
\]
Letting $n\to\infty$ and then $\eta\downarrow0$ proves the assertion.
\end{proof}

\begin{proof}[Proof of Theorem~3.1]
For each $k$, define
\[
 \varphi_{n,k}(x,y,z)
 =\frac{x}{e_n^2+t_nx-2e_ny+y^2-xz}.
\]
At $\bm q_{0,n,k}$,
\[
 \nabla\varphi_{n,k}(\bm q_{0,n,k})=\bg_{n,k},
 \qquad
 \varphi_{n,k}(\bm q_{0,n,k})=\mu_{n,k}.
\]
The compact bounds in Proposition~\ref{prop:weighted_DE} hold
simultaneously over the fixed grid.  The exact covariance calculation
and $K<\infty$ give
\begin{align*}
 \max_{1\le k\le K}
 \|\bm q_{n,k}-\bm q_{0,n,k}\|
 &=O_{\Pp}(n^{-1/2}),\\
 \max_{1\le k\le K}\sup_{0\le t\le1}
 |L_k\{\bm q_{0,n,k}+t(\bm q_{n,k}-\bm q_{0,n,k})\}
   -D_{n,k}|
 &=O_{\Pp}(n^{-1/2}),
\end{align*}
where $L_k$ denotes $L(\cdot,e_n,t_n)$ at ridge
$\rho^{(k)}$.  Hence, with probability tending to one, every
intermediate denominator is at least one half of the common positive
lower bound.  Taylor's theorem then gives
\begin{align*}
 \max_k\left|
 \sqrt n\{\varphi_{n,k}(\bm q_{n,k})-\mu_{n,k}
 -\bg_{n,k}^\top(\bm q_{n,k}-\bm q_{0,n,k})\}
 \right|
 &=O_{\Pp}(n^{-1/2}).
\end{align*}
Lemma~\ref{lem:finite_grid_conditional_clt}, applied with the $K$ gradients, shows that the vector
\[
 \left
 \{
 \frac{\sqrt n\{\varphi_{n,k}(\bm q_{n,k})-\mu_{n,k}\}}
 {\sigma_{D,n,k}}
 \right\}_{k=1}^K
\]
has a conditional Gaussian approximation with correlation matrix
\[
 \mC_{K,n}^{0}
 =\left\{
 \frac{\bg_{n,k}^\top\mGamma_n^{k\ell}\bg_{n,\ell}}
 {\sigma_{D,n,k}\sigma_{D,n,\ell}}
 \right\}_{k,\ell=1}^K.
\]
Lemma~\ref{lem:cross_ridge_DE} gives
$\|\mC_{K,n}^{0}-\mC_K\|_{\max}=o_{\Pp}(1)$.

Proposition~\ref{prop:scalar} and the proof of
Theorem~2.1 give, at every grid point,
\[
 T_n(\rho^{(k)})
 =n\varphi_{n,k}(\bm q_{n,k})+R_{T,n,k},
 \qquad
 R_{T,n,k}=O_{\Pp}(1+\log p).
\]
Because $K$ is fixed,
\[
 \max_k\frac{|R_{T,n,k}|}{\sqrt n}
 =O_{\Pp}\{(1+\log p)/\sqrt n\}=o_{\Pp}(1).
\]
Theorem~\ref{thm:rates} also holds simultaneously over the grid and yields
\begin{align*}
 \max_k\sqrt n|\wh\mu_{n,k}-\mu_{n,k}|
 &=O_{\Pp}(n^{-1/2})=o_{\Pp}(1),\\
 \max_k\left|
 \frac{\wh\sigma_{D,n,k}^2}{\sigma_{D,n,k}^2}-1
 \right|&=O_{\Pp}(n^{-1/2}).
\end{align*}
Conditional Slutsky's theorem, followed by integration over
$\calF_n^0$, proves
\[
 (Z_{n,1},\ldots,Z_{n,K})^\top
 \dto N_K(\0,\mC_K).
\]
\end{proof}

\begin{lemma}\label{lem:cauchy_tail}
Let $\bm V\sim N_K(\0,\mC_K)$ and let
\[
 Y_k=\tan[\pi\{\Phi(V_k)-1/2\}],
 \qquad
 S_K=\sum_{k=1}^K\varpi_kY_k.
\]
For every deterministic $\bm a=(a_1,\ldots,a_K)^\top$, the random variable
\[
 S_K(\bm a)
 =\sum_{k=1}^K\varpi_k
 \tan[\pi\{\Phi(V_k+a_k)-1/2\}]
\]
has a continuous distribution with a strictly increasing distribution function.  In addition,
\begin{equation}\label{eq:cauchy_tail_limit}
 \Pp(S_K>t)=\frac{1+o(1)}{\pi t},
 \qquad t\to\infty.
\end{equation}
\end{lemma}

\begin{proof}
By \eqref{eq:cross_Psi_finite}--\eqref{eq:cross_corr_limit}, the matrix
$\mC_K$ is symmetric, positive semidefinite, entrywise nonnegative,
and has unit diagonal.  The Perron--Frobenius theorem therefore gives
an eigenvalue $\lambda_\star>0$ and a unit eigenvector
$\bm h=(h_1,\ldots,h_K)^\top\ge\0$.  A Gaussian spectral
decomposition yields
\[
 \bm V=\sqrt{\lambda_\star}\,\bm h\,Z+\bm V^\perp,
 \qquad
 Z\sim N(0,1),
 \qquad
 Z\text{ is independent of }\bm V^\perp.
\]
Write $\psi(x)=\tan[\pi\{\Phi(x)-1/2\}]$.  Conditionally on
$\bm V^\perp=\bm v$, for any deterministic $\bm a$,
\[
 S_K(\bm a)=g_{\bm v,\bm a}(Z),
 \qquad
 g_{\bm v,\bm a}(z)
 =\sum_{k=1}^K\varpi_k
 \psi(v_k+a_k+\sqrt{\lambda_\star}h_kz).
\]
Since
\[
 \psi'(x)
 =\pi\phi(x)\sec^2[\pi\{\Phi(x)-1/2\}]>0,
\]
and at least one $h_k$ is positive,
\[
 g_{\bm v,\bm a}'(z)
 =\sqrt{\lambda_\star}
 \sum_{k=1}^K\varpi_kh_k
 \psi'(v_k+a_k+\sqrt{\lambda_\star}h_kz)>0.
\]
Moreover, $g_{\bm v,\bm a}(z)\to\pm\infty$ as
$z\to\pm\infty$.  Thus its conditional distribution is continuous
and assigns positive probability to every nonempty interval.
Integration over $\bm V^\perp$ proves the first assertion.

It remains to establish the right-tail expansion at $\bm a=\0$.
Every $Y_k$ is standard Cauchy because $\Phi(V_k)$ is uniform on
$(0,1)$.  If $r_{k\ell}=1$, then $V_k=V_\ell$ almost surely.  Merge all indices connected by this relation and replace their weights by the sum of the weights in the class.  This leaves $S_K$ unchanged.  After merging, write the number of classes as $J$, their positive weights as $\omega_1,\ldots,\omega_J$, and the corresponding Gaussian variables as $W_1,\ldots,W_J$.  Then
\[
 \sum_{j=1}^J\omega_j=1,
 \qquad
 0\le\operatorname{Corr}(W_j,W_\ell)<1
 \quad(j\ne\ell).
\]
If $J=1$, then $\omega_1=1$ and $S_K=Y_1$ is standard Cauchy.  Hence
\[
 \Pp(S_K>t)=\frac12-\frac1\pi\arctan(t)
 =\frac{1+o(1)}{\pi t},
\]
which proves \eqref{eq:cauchy_tail_limit} in this case.  We may therefore
assume $J\ge2$ in the remainder of the proof.
Put
\[
 X_j=\omega_j\tan[\pi\{\Phi(W_j)-1/2\}].
\]
The marginal Cauchy tail gives, for every fixed $a>0$,
\[
 \Pp(X_j>at)
 =\frac{\omega_j}{\pi at}+O(t^{-3}),
 \qquad
 \Pp(X_j<-at)
 =\frac{\omega_j}{\pi at}+O(t^{-3}).
\]

We first verify asymptotic separation of distinct extremes.  For each
$j$ and fixed $a>0$, let
$u_{j,t}(a)=\pi^{-1}\operatorname{arccot}(at/\omega_j)$ and
$z_{j,t}(a)=\Phi^{-1}\{1-u_{j,t}(a)\}$.  Mills' ratio gives
\[
 z_{j,t}(a)^2=2\log t+O(\log\log t).
\]
For $j\ne\ell$, put $r=\operatorname{Corr}(W_j,W_\ell)<1$.  The Gaussian exponential bound for an upper orthant yields, for fixed $a,b>0$,
\begin{align*}
 \Pp(X_j>at,X_\ell>bt)
 &\le
 \exp\left[-\frac{z_{j,t}(a)^2-2rz_{j,t}(a)z_{\ell,t}(b)
 +z_{\ell,t}(b)^2}{2(1-r^2)}\right]\\
 &=t^{-2/(1+r)+o(1)}=o(t^{-1}).
\end{align*}
Replacing $W_\ell$ by $-W_\ell$ gives
\begin{equation}\label{eq:joint_opposite_tail}
 \Pp(X_j>at,X_\ell<-bt)
 =t^{-2/(1-r)+o(1)}=o(t^{-1}).
\end{equation}

Fix $\varepsilon\in(0,1/2)$.  If
\[
 X_j>(1+\varepsilon)t,
 \qquad
 X_\ell\ge-\frac{\varepsilon t}{J-1}
 \quad\text{for every }\ell\ne j,
\]
then $\sum_{j=1}^JX_j>t$.  Inclusion--exclusion,
\eqref{eq:joint_opposite_tail}, and the positive joint-tail bound give
\begin{align*}
 \Pp\left(\sum_{j=1}^JX_j>t\right)
 &\ge\sum_{j=1}^J\Pp\{X_j>(1+\varepsilon)t\}
 +o(t^{-1})\\
 &=\frac{1+o(1)}{\pi(1+\varepsilon)t}.
\end{align*}

For the upper bound, suppose
$\sum_jX_j>t$ and $\max_jX_j\le(1-\varepsilon)t$.  The largest summand is larger than $t/J$.  If its index is $j$, the remaining summands have sum exceeding $\varepsilon t$, so at least one
$\ell\ne j$ satisfies $X_\ell>\varepsilon t/(J-1)$.  Hence
\begin{align*}
 \Pp\left(\sum_{j=1}^JX_j>t\right)
 &\le\Pp\{\max_jX_j>(1-\varepsilon)t\}\\
 &\quad+
 \sum_{j\ne\ell}
 \Pp\left\{X_j>t/J,
 X_\ell>\varepsilon t/(J-1)\right\}\\
 &=\sum_{j=1}^J\Pp\{X_j>(1-\varepsilon)t\}
 +o(t^{-1})\\
 &=\frac{1+o(1)}{\pi(1-\varepsilon)t}.
\end{align*}
Letting $\varepsilon\downarrow0$ proves
\eqref{eq:cauchy_tail_limit}.
\end{proof}

\begin{proof}[Proof of Theorem~3.2]
Define the continuous map
\[
 \mathcal T(\bm z)
 =\sum_{k=1}^K\varpi_k
 \tan[\pi\{\Phi(z_k)-1/2\}],
 \qquad \bm z\in\R^K.
\]
Theorem~3.1 and the continuous mapping theorem give
\begin{equation}\label{eq:CCT_weak_limit}
 T_{\mathrm{CC},n}=\mathcal T(Z_{n,1},\ldots,Z_{n,K})
 \dto
 \mathcal T(\bm V)=T_{\mathrm{CC},\infty}^{0}.
\end{equation}
Therefore, at every continuity point $c$,
\[
 \Pp(T_{\mathrm{CC},n}>c)
 \to
 \Pp(T_{\mathrm{CC},\infty}^{0}>c).
\]

Taking $c=c_{\alpha,K}$ in the preceding display gives
\[
 \Pp(T_{\mathrm{CC},n}>c_{\alpha,K})\to\alpha.
\]

For the analytic $p$-value,
\[
 p_{\mathrm{CC},n}\le\alpha
 \quad\Longleftrightarrow\quad
 T_{\mathrm{CC},n}\ge\cot(\pi\alpha).
\]
Equation \eqref{eq:CCT_weak_limit} gives the asserted fixed-$\alpha$ limit.  Lemma~\ref{lem:cauchy_tail} and
\[
 \cot(\pi\alpha)
 =\frac1{\pi\alpha}\{1+O(\alpha^2)\},
 \qquad \alpha\downarrow0,
\]
yield
\begin{align*}
 \lim_{\alpha\downarrow0}\lim_{n\to\infty}
 \frac{\Pp(p_{\mathrm{CC},n}\le\alpha)}{\alpha}
 &=\lim_{\alpha\downarrow0}
 \frac{\Pp\{T_{\mathrm{CC},\infty}^{0}
 >\cot(\pi\alpha)\}}{\alpha}\\
 &=\lim_{\alpha\downarrow0}
 \frac{1}{\pi\alpha\cot(\pi\alpha)}=1.
\end{align*}
\end{proof}

\begin{proof}[Proof of Theorem~3.3]
Under the local alternative (2.14), translation invariance in Proposition~\ref{prop:translation} implies that all nuisance estimators have exactly their null distributions.  Equation \eqref{eq:local_Z_expansion}, evaluated at the fixed grid points, gives
\begin{align*}
 \max_{1\le k\le K}
 \left|
 Z_{n,k}-Z_{\varepsilon,n}(\rho^{(k)})
 -\Lambda_{\rho^{(k)}}(\calG)
 \right|
 &\le O_{\Pp}(n^{-1/4})+O_{\Pp}(p^{-1/4})\\
 &\quad+O_{\Pp}\left(\sqrt{\frac{\log p}{p}}\right)
 +\max_k|r_p(\rho^{(k)})|\\
 &=o_{\Pp}(1),
\end{align*}
where the last equality uses $K<\infty$.  The noise vector
$\{Z_{\varepsilon,n}(\rho^{(k)})\}_{k=1}^K$ has the joint null limit in Theorem~3.1.  Slutsky's theorem therefore gives
\[
 (Z_{n,1},\ldots,Z_{n,K})^\top
 \dto
 N_K\{\bm\Lambda_{\mathcal R}(\calG),\mC_K\}.
\]
The continuous mapping theorem applied to $\mathcal T$ gives convergence of the aggregate.  The shifted continuity assertion in Lemma~\ref{lem:cauchy_tail} shows that both rejection thresholds are continuity points of the limiting distribution, which proves the two stated power formulas.

Under (2.18) and (2.19), the consistency part of Theorem~2.3 gives, for every fixed $k$ and every $M<\infty$,
\[
 \Pp\{Z_n(\rho^{(k)})\le M\}\to0.
\]
Since $K$ is fixed,
\[
 \Pp\left\{\min_{1\le k\le K}Z_n(\rho^{(k)})\le M\right\}
 \le\sum_{k=1}^K
 \Pp\{Z_n(\rho^{(k)})\le M\}
 \to0.
\]
Thus $p_{n,k}\to0$ simultaneously and
\[
 \min_k\tan[\pi\{1/2-p_{n,k}\}]
 \pto\infty.
\]
Because all weights are positive,
\[
 T_{\mathrm{CC},n}
 \ge\min_k\tan[\pi\{1/2-p_{n,k}\}]
 \pto\infty.
\]
Consequently,
\[
 p_{\mathrm{CC},n}
 =\frac12-\frac1\pi\arctan(T_{\mathrm{CC},n})
 \pto0.
\]
The rejection probability tends to one.
\end{proof}


\begin{thebibliography}{50}
\providecommand{\natexlab}[1]{#1}
\providecommand{\url}[1]{\texttt{#1}}
\expandafter\ifx\csname urlstyle\endcsname\relax
  \providecommand{\doi}[1]{doi: #1}\else
  \providecommand{\doi}{doi: \begingroup \urlstyle{rm}\Url}\fi

\bibitem[Arias-Castro et~al.(2011)Arias-Castro, Cand{\`e}s, and
  Plan]{AriasCastroCandesPlan2011}
Ery Arias-Castro, Emmanuel~J. Cand{\`e}s, and Yaniv Plan.
\newblock Global testing under sparse alternatives: {ANOVA}, multiple
  comparisons and the higher criticism.
\newblock \emph{The Annals of Statistics}, 39\penalty0 (5):\penalty0
  2533--2556, 2011.
\newblock \doi{10.1214/11-AOS910}.

\bibitem[Bai and Saranadasa(1996)]{BaiSaranadasa1996}
Zhidong Bai and Hewa Saranadasa.
\newblock Effect of high dimension: by an example of a two sample problem.
\newblock \emph{Statistica Sinica}, 6\penalty0 (2):\penalty0 311--329, 1996.

\bibitem[Bhatia(2007)]{Bhatia2007Positive}
Rajendra Bhatia.
\newblock \emph{Positive Definite Matrices}.
\newblock Princeton University Press, Princeton, NJ, 2007.
\newblock ISBN 978-0-691-12918-1.

\bibitem[Boucheron et~al.(2013)Boucheron, Lugosi, and
  Massart]{BoucheronLugosiMassart2013}
St{\'e}phane Boucheron, G{\'a}bor Lugosi, and Pascal Massart.
\newblock \emph{Concentration Inequalities: A Nonasymptotic Theory of
  Independence}.
\newblock Oxford University Press, Oxford, 2013.
\newblock ISBN 978-0-19-953525-5.

\bibitem[Cai et~al.(2014)Cai, Liu, and Xia]{CaiLiuXia2014}
T.~Tony Cai, Weidong Liu, and Yin Xia.
\newblock Two-sample test of high dimensional means under dependence.
\newblock \emph{Journal of the Royal Statistical Society: Series B (Statistical
  Methodology)}, 76\penalty0 (2):\penalty0 349--372, 2014.
\newblock \doi{10.1111/rssb.12034}.

\bibitem[Chakraborty and Chaudhuri(2017)]{ChakrabortyChaudhuri2017}
Anirvan Chakraborty and Probal Chaudhuri.
\newblock Tests for high-dimensional data based on means, spatial signs and
  spatial ranks.
\newblock \emph{The Annals of Statistics}, 45\penalty0 (2):\penalty0 771--799,
  2017.
\newblock \doi{10.1214/16-AOS1467}.

\bibitem[Chen et~al.(2011)Chen, Paul, Prentice, and
  Wang]{ChenPaulPrenticeWang2011}
Lin~S. Chen, Debashis Paul, Ross~L. Prentice, and Pei Wang.
\newblock A regularized {Hotelling}'s {$T^2$} test for pathway analysis in
  proteomic studies.
\newblock \emph{Journal of the American Statistical Association}, 106\penalty0
  (496):\penalty0 1345--1360, 2011.
\newblock \doi{10.1198/jasa.2011.ap10599}.

\bibitem[Chen and Qin(2010)]{ChenQin2010}
Song~Xi Chen and Ying-Li Qin.
\newblock A two-sample test for high-dimensional data with applications to
  gene-set testing.
\newblock \emph{The Annals of Statistics}, 38\penalty0 (2):\penalty0 808--835,
  2010.
\newblock \doi{10.1214/09-AOS716}.

\bibitem[Chen et~al.(2019)Chen, Li, and Zhong]{ChenLiZhong2019}
Song~Xi Chen, Jun Li, and Ping-Shou Zhong.
\newblock Two-sample and {ANOVA} tests for high dimensional means.
\newblock \emph{The Annals of Statistics}, 47\penalty0 (3):\penalty0
  1443--1474, 2019.
\newblock \doi{10.1214/18-AOS1720}.

\bibitem[de~Jong(1987)]{deJong1987}
Peter de~Jong.
\newblock A central limit theorem for generalized quadratic forms.
\newblock \emph{Probability Theory and Related Fields}, 75\penalty0
  (2):\penalty0 261--277, 1987.
\newblock \doi{10.1007/BF00354037}.

\bibitem[Dempster(1958)]{Dempster1958}
Arthur~P. Dempster.
\newblock A high dimensional two sample significance test.
\newblock \emph{The Annals of Mathematical Statistics}, 29\penalty0
  (4):\penalty0 995--1010, 1958.
\newblock \doi{10.1214/aoms/1177706437}.

\bibitem[Donoho and Jin(2004)]{DonohoJin2004}
David Donoho and Jiashun Jin.
\newblock Higher criticism for detecting sparse heterogeneous mixtures.
\newblock \emph{The Annals of Statistics}, 32\penalty0 (3):\penalty0 962--994,
  2004.
\newblock \doi{10.1214/009053604000000265}.

\bibitem[Feng and Sun(2015)]{FengSun2015}
Long Feng and Fasheng Sun.
\newblock A note on high-dimensional two-sample test.
\newblock \emph{Statistics \& Probability Letters}, 105:\penalty0 29--36, 2015.
\newblock \doi{10.1016/j.spl.2015.05.017}.

\bibitem[Feng and Sun(2016)]{FengSun2016}
Long Feng and Fasheng Sun.
\newblock Spatial-sign based high-dimensional location test.
\newblock \emph{Electronic Journal of Statistics}, 10\penalty0 (2):\penalty0
  2420--2434, 2016.
\newblock \doi{10.1214/16-EJS1176}.

\bibitem[Feng et~al.(2015)Feng, Zou, Wang, and Zhu]{FengZouWangZhu2015}
Long Feng, Changliang Zou, Zhaojun Wang, and Lixing Zhu.
\newblock Two-sample {Behrens--Fisher} problem for high-dimensional data.
\newblock \emph{Statistica Sinica}, 25\penalty0 (4):\penalty0 1297--1312, 2015.
\newblock \doi{10.5705/ss.2014.048}.

\bibitem[Feng et~al.(2016)Feng, Zou, and Wang]{FengZouWang2016}
Long Feng, Changliang Zou, and Zhaojun Wang.
\newblock Multivariate-sign-based high-dimensional tests for the two-sample
  location problem.
\newblock \emph{Journal of the American Statistical Association}, 111\penalty0
  (514):\penalty0 721--735, 2016.
\newblock \doi{10.1080/01621459.2015.1035380}.

\bibitem[Feng et~al.(2017)Feng, Zou, Wang, and Zhu]{FengZouWangZhu2017}
Long Feng, Changliang Zou, Zhaojun Wang, and Lixing Zhu.
\newblock Composite {$T^2$} test for high-dimensional data.
\newblock \emph{Statistica Sinica}, 27\penalty0 (3):\penalty0 1419--1436, 2017.
\newblock \doi{10.5705/ss.202015.0199}.

\bibitem[Feng et~al.(2020)Feng, Zhang, and Liu]{FengZhangLiu2020}
Long Feng, Xiaoxu Zhang, and Binghui Liu.
\newblock A high-dimensional spatial rank test for two-sample location
  problems.
\newblock \emph{Computational Statistics \& Data Analysis}, 144:\penalty0
  106889, 2020.
\newblock \doi{10.1016/j.csda.2019.106889}.

\bibitem[Feng et~al.(2021)Feng, Liu, and Ma]{FengLiuMa2021}
Long Feng, Binghui Liu, and Yanyuan Ma.
\newblock An inverse norm sign test of location parameter for high-dimensional
  data.
\newblock \emph{Journal of Business \& Economic Statistics}, 39\penalty0
  (3):\penalty0 807--815, 2021.
\newblock \doi{10.1080/07350015.2020.1736084}.

\bibitem[Gregory et~al.(2015)Gregory, Carroll, Baladandayuthapani, and
  Lahiri]{GregoryCarrollBaladandayuthapaniLahiri2015}
Karl~B. Gregory, Raymond~J. Carroll, Veerabhadran Baladandayuthapani, and
  Soumendra~N. Lahiri.
\newblock A two-sample test for equality of means in high dimension.
\newblock \emph{Journal of the American Statistical Association}, 110\penalty0
  (510):\penalty0 837--849, 2015.
\newblock \doi{10.1080/01621459.2014.934826}.

\bibitem[Hachem et~al.(2007)Hachem, Loubaton, and
  Najim]{HachemLoubatonNajim2007}
Walid Hachem, Philippe Loubaton, and Jamal Najim.
\newblock Deterministic equivalents for certain functionals of large random
  matrices.
\newblock \emph{The Annals of Applied Probability}, 17\penalty0 (3):\penalty0
  875--930, 2007.
\newblock \doi{10.1214/105051606000000925}.

\bibitem[Hachem et~al.(2013)Hachem, Loubaton, Najim, and
  Vallet]{HachemLoubatonNajimVallet2013}
Walid Hachem, Philippe Loubaton, Jamal Najim, and Pascal Vallet.
\newblock On bilinear forms based on the resolvent of large random matrices.
\newblock \emph{Annales de l'Institut Henri Poincar\'{e}, Probabilit\'{e}s et
  Statistiques}, 49\penalty0 (1):\penalty0 36--63, 2013.
\newblock \doi{10.1214/11-AIHP450}.

\bibitem[Hall and Jin(2010)]{HallJin2010}
Peter Hall and Jiashun Jin.
\newblock Innovated higher criticism for detecting sparse signals in correlated
  noise.
\newblock \emph{The Annals of Statistics}, 38\penalty0 (3):\penalty0
  1686--1732, 2010.
\newblock \doi{10.1214/09-AOS764}.

\bibitem[He et~al.(2020)He, Zhang, Zhang, and Zhou]{HeZhangZhangZhou2020}
Yong He, Mingjuan Zhang, Xinsheng Zhang, and Wang Zhou.
\newblock High-dimensional two-sample mean vectors test and support recovery
  with factor adjustment.
\newblock \emph{Computational Statistics \& Data Analysis}, 151:\penalty0
  107004, 2020.
\newblock \doi{10.1016/j.csda.2020.107004}.

\bibitem[Huang et~al.(2023)Huang, Liu, Zhou, and Feng]{HuangLiuZhouFeng2023}
Xifen Huang, Binghui Liu, Qin Zhou, and Long Feng.
\newblock A high-dimensional inverse norm sign test for two-sample location
  problems.
\newblock \emph{The Canadian Journal of Statistics}, 51\penalty0 (4):\penalty0
  1004--1033, 2023.
\newblock \doi{10.1002/cjs.11731}.

\bibitem[Huang et~al.(2022)Huang, Li, Li, and Yang]{HuangLiLiYang2022}
Yuan Huang, Changcheng Li, Runze Li, and Songshan Yang.
\newblock An overview of tests on high-dimensional means.
\newblock \emph{Journal of Multivariate Analysis}, 188:\penalty0 104813, 2022.
\newblock \doi{10.1016/j.jmva.2021.104813}.

\bibitem[Li et~al.(2020)Li, Aue, Paul, Peng, and Wang]{LiAuePaulPengWang2020}
Haoran Li, Alexander Aue, Debashis Paul, Jie Peng, and Pei Wang.
\newblock An adaptable generalization of {Hotelling}'s {$T^2$} test in high
  dimension.
\newblock \emph{The Annals of Statistics}, 48\penalty0 (3):\penalty0
  1815--1847, 2020.
\newblock \doi{10.1214/19-AOS1869}.

\bibitem[Liu et~al.(2027)Liu, Feng, Zhao, and Wang]{LiuFengZhaoWang2027}
Jixuan Liu, Long Feng, Ping Zhao, and Zhaojun Wang.
\newblock Spatial-sign based maxsum test for high dimensional location
  parameters.
\newblock \emph{Statistica Sinica}, 37\penalty0 (2), 2027.
\newblock \doi{10.5705/ss.202024.0051}.
\newblock forthcoming.

\bibitem[Liu and Xie(2020)]{LiuXie2020}
Yaowu Liu and Jun Xie.
\newblock Cauchy combination test: A powerful test with analytic {$p$}-value
  calculation under arbitrary dependency structures.
\newblock \emph{Journal of the American Statistical Association}, 115\penalty0
  (529):\penalty0 393--402, 2020.
\newblock \doi{10.1080/01621459.2018.1554485}.

\bibitem[Lopes et~al.(2011)Lopes, Jacob, and
  Wainwright]{LopesJacobWainwright2011}
Miles~E. Lopes, Laurent Jacob, and Martin~J. Wainwright.
\newblock A more powerful two-sample test in high dimensions using random
  projection.
\newblock In \emph{Advances in Neural Information Processing Systems 24}, pages
  1206--1214. Curran Associates, Inc., 2011.

\bibitem[Ma et~al.(2015)Ma, Lan, and Wang]{MaLanWang2015}
Yingying Ma, Wei Lan, and Hansheng Wang.
\newblock A high dimensional two-sample test under a low dimensional factor
  structure.
\newblock \emph{Journal of Multivariate Analysis}, 140:\penalty0 162--170,
  2015.
\newblock \doi{10.1016/j.jmva.2015.05.005}.

\bibitem[M{\"o}tt{\"o}nen and Oja(1995)]{MottoneOja1995}
Jyrki M{\"o}tt{\"o}nen and Hannu Oja.
\newblock Multivariate spatial sign and rank methods.
\newblock \emph{Journal of Nonparametric Statistics}, 5\penalty0 (2):\penalty0
  201--213, 1995.
\newblock \doi{10.1080/10485259508832643}.

\bibitem[O'Donnell(2014)]{ODonnell2014}
Ryan O'Donnell.
\newblock \emph{Analysis of Boolean Functions}.
\newblock Cambridge University Press, Cambridge, 2014.
\newblock ISBN 978-1-107-03832-5.
\newblock \doi{10.1017/CBO9781139814782}.

\bibitem[Oja(2010)]{Oja2010}
Hannu Oja.
\newblock \emph{Multivariate Nonparametric Methods with {R}: An Approach Based
  on Spatial Signs and Ranks}, volume 199 of \emph{Lecture Notes in
  Statistics}.
\newblock Springer, New York, 2010.
\newblock \doi{10.1007/978-1-4419-0468-3}.

\bibitem[Paindaveine and Verdebout(2016)]{PaindaveineVerdebout2016}
Davy Paindaveine and Thomas Verdebout.
\newblock On high-dimensional sign tests.
\newblock \emph{Bernoulli}, 22\penalty0 (3):\penalty0 1745--1769, 2016.
\newblock \doi{10.3150/15-BEJ710}.

\bibitem[Park and Ayyala(2013)]{ParkAyyala2013}
Junyong Park and Deepak~N. Ayyala.
\newblock A test for the mean vector in large dimension and small samples.
\newblock \emph{Journal of Statistical Planning and Inference}, 143\penalty0
  (5):\penalty0 929--943, 2013.
\newblock \doi{10.1016/j.jspi.2012.11.001}.

\bibitem[Srivastava(2009)]{Srivastava2009}
Muni~S. Srivastava.
\newblock A test for the mean vector with fewer observations than the dimension
  under non-normality.
\newblock \emph{Journal of Multivariate Analysis}, 100\penalty0 (3):\penalty0
  518--532, 2009.
\newblock \doi{10.1016/j.jmva.2008.06.006}.

\bibitem[Srivastava and Du(2008)]{SrivastavaDu2008}
Muni~S. Srivastava and Meng Du.
\newblock A test for the mean vector with fewer observations than the
  dimension.
\newblock \emph{Journal of Multivariate Analysis}, 99\penalty0 (3):\penalty0
  386--402, 2008.
\newblock \doi{10.1016/j.jmva.2006.11.002}.

\bibitem[Srivastava et~al.(2013)Srivastava, Katayama, and
  Kano]{SrivastavaKatayamaKano2013}
Muni~S. Srivastava, Shota Katayama, and Yutaka Kano.
\newblock A two sample test in high dimensional data.
\newblock \emph{Journal of Multivariate Analysis}, 114:\penalty0 349--358,
  2013.
\newblock \doi{10.1016/j.jmva.2012.08.014}.

\bibitem[Stewart and Sun(1990)]{StewartSun1990}
G.~W. Stewart and Ji-guang Sun.
\newblock \emph{Matrix Perturbation Theory}.
\newblock Computer Science and Scientific Computing. Academic Press, Boston,
  1990.
\newblock ISBN 978-0-12-670230-9.

\bibitem[Thulin(2014)]{Thulin2014}
M\r{a}ns Thulin.
\newblock A high-dimensional two-sample test for the mean using random
  subspaces.
\newblock \emph{Computational Statistics \& Data Analysis}, 74:\penalty0
  26--38, 2014.
\newblock \doi{10.1016/j.csda.2013.12.003}.

\bibitem[Tropp(2012)]{Tropp2012}
Joel~A. Tropp.
\newblock User-friendly tail bounds for sums of random matrices.
\newblock \emph{Foundations of Computational Mathematics}, 12\penalty0
  (4):\penalty0 389--434, 2012.
\newblock \doi{10.1007/s10208-011-9099-z}.

\bibitem[Vershynin(2018)]{Vershynin2018}
Roman Vershynin.
\newblock \emph{High-Dimensional Probability: An Introduction with Applications
  in Data Science}, volume~47 of \emph{Cambridge Series in Statistical and
  Probabilistic Mathematics}.
\newblock Cambridge University Press, Cambridge, 2018.
\newblock ISBN 978-1-108-41519-4.
\newblock \doi{10.1017/9781108231596}.

\bibitem[Visuri et~al.(2000)Visuri, Koivunen, and Oja]{VisuriKoivunenOja2000}
Samuli Visuri, Visa Koivunen, and Hannu Oja.
\newblock Sign and rank covariance matrices.
\newblock \emph{Journal of Statistical Planning and Inference}, 91\penalty0
  (2):\penalty0 557--575, 2000.
\newblock \doi{10.1016/S0378-3758(00)00199-3}.

\bibitem[Wang et~al.(2015)Wang, Peng, and Li]{WangPengLi2015}
Lan Wang, Bo~Peng, and Runze Li.
\newblock A high-dimensional nonparametric multivariate test for mean vector.
\newblock \emph{Journal of the American Statistical Association}, 110\penalty0
  (512):\penalty0 1658--1669, 2015.
\newblock \doi{10.1080/01621459.2014.988215}.

\bibitem[Xu et~al.(2016)Xu, Lin, Wei, and Pan]{XuLinWeiPan2016}
Gongjun Xu, Lifeng Lin, Peng Wei, and Wei Pan.
\newblock An adaptive two-sample test for high-dimensional means.
\newblock \emph{Biometrika}, 103\penalty0 (3):\penalty0 609--624, 2016.
\newblock \doi{10.1093/biomet/asw029}.

\bibitem[Yan et~al.(2025)Yan, Feng, and Zhang]{YanFengZhang2025}
Guowei Yan, Long Feng, and Xiaoxu Zhang.
\newblock High-dimensional {Hettmansperger--Randles} estimator and its
  applications.
\newblock arXiv:2505.01669, 2025.

\bibitem[Zhang and Feng(2024)]{ZhangFeng2024}
Yu~Zhang and Long Feng.
\newblock Adaptive rank-based tests for high dimensional mean problems.
\newblock \emph{Statistics \& Probability Letters}, 214:\penalty0 110226, 2024.
\newblock \doi{10.1016/j.spl.2024.110226}.

\bibitem[Zhao and Feng(2026)]{ZhaoFeng2026}
Ping Zhao and Long Feng.
\newblock Note on high dimensional spatial-sign test for one sample problem.
\newblock arXiv:2601.08736, 2026.

\bibitem[Zhong et~al.(2013)Zhong, Chen, and Xu]{ZhongChenXu2013}
Ping-Shou Zhong, Song~Xi Chen, and Minya Xu.
\newblock Tests alternative to higher criticism for high-dimensional means
  under sparsity and column-wise dependence.
\newblock \emph{The Annals of Statistics}, 41\penalty0 (6):\penalty0
  2820--2851, 2013.
\newblock \doi{10.1214/13-AOS1168}.

\end{thebibliography}
\end{document}